\documentclass[11pt]{scrartcl}

\usepackage[utf8]{inputenc}
\usepackage[T1]{fontenc}
\usepackage{lmodern}
\usepackage{alphabeta}

\usepackage{titling}
\usepackage{xspace}
\usepackage{soul} 

\usepackage{mathrsfs}

 \usepackage[a4paper, top=25.4mm, bottom=25.4mm, left=25.4mm, right=25.4mm, includefoot]{geometry}

\DeclareSectionCommand[%
level=4,
indent=0pt,
beforeskip=1ex plus 1ex minus .2ex,
afterskip=-1em,
font={},
tocindent=7em,
tocnumwidth=4.1em,
counterwithin=subsubsection
]{paragraph}

\usepackage[inline,shortlabels]{enumitem}

\usepackage{dsfont}
\usepackage{setspace}

\usepackage{bm}
\usepackage{microtype}
\usepackage{amsmath}
\usepackage{amssymb}
\usepackage{amsfonts}
\usepackage{mathtools}
\usepackage[mathscr]{euscript}
\usepackage{wasysym}

\usepackage[hyphens]{url}
\usepackage{amsthm}

\usepackage{tikz}
\usepackage{xcolor}
\usepackage{subcaption}
\usepackage{graphicx}
\usepackage{wrapfig}

\usepackage{cite}

\usepackage{hyperref}
\usepackage{zref-clever}

\usepackage{nicefrac}
\usepackage[textwidth=1.5in]{todonotes}

\newcommand{\cref}[1]{\zcref{#1}}
\renewcommand{\autoref}[1]{\zcref{#1}}

\definecolor{mediumslateblue}{rgb}{0.48, 0.41, 0.93}
\definecolor{lightsalmonpink}{rgb}{1.0, 0.6, 0.6}
\definecolor{amaranth}{rgb}{0.9, 0.17, 0.31}
\definecolor{jonquil}{rgb}{0.98, 0.85, 0.37}
\definecolor{orangepeel}{rgb}{1.0, 0.62, 0.0}
\definecolor{darkspringgreen}{rgb}{0.09, 0.45, 0.27}
\colorlet{myGreen}{green!50!black}
\colorlet{myLightgreen}{green}
\colorlet{myRed}{red!90!black}
\definecolor{myBlue}{rgb}{0.25, 0.0, 1.0}
\definecolor{myLightBlue}{rgb}{0.39, 0.58, 0.93}
\colorlet{myViolet}{myBlue!55!myRed}
\definecolor{myOrange}{rgb}{1.0, 0.66, 0.07}
\definecolor{ceruleanblue}{rgb}{0.16, 0.32, 0.75}
\definecolor{CornflowerBlue}{rgb}{0.39, 0.58, 0.93}
\definecolor{DarkGoldenrod}{rgb}{0.72, 0.53, 0.04}
\definecolor{BritishRacingGreen}{rgb}{0.0, 0.26, 0.15}
\definecolor{DarkMagenta}{rgb}{0.55, 0.0, 0.55}
\definecolor{AO}{rgb}{0.0, 0.5, 0.0}
\definecolor{BostonUniversityRed}{rgb}{0.8, 0.0, 0.0}
\definecolor{myRed}{rgb}{0.8, 0.0, 0.0}
\definecolor{DarkMidnightBlue}{rgb}{0.0, 0.2, 0.4}
\definecolor{DarkTangerine}{rgb}{1.0, 0.66, 0.07}
\definecolor{AppleGreen}{rgb}{0.55, 0.71, 0.0}
\definecolor{BrightUbe}{rgb}{0.82, 0.62, 0.91}
\definecolor{Amethyst}{rgb}{0.6, 0.4, 0.8}
\definecolor{DarkGray}{rgb}{0.52, 0.52, 0.51}
\definecolor{Gray}{rgb}{0.66, 0.66, 0.66}
\definecolor{BananaYellow}{rgb}{1.0, 0.88, 0.21}
\definecolor{Amber}{rgb}{1.0, 0.75, 0.0}
\definecolor{LightGray}{rgb}{0.83, 0.83, 0.83}
\definecolor{PrincetonOrange}{rgb}{1.0, 0.56, 0.0}
\definecolor{DeepCarrotOrange}{rgb}{0.91, 0.41, 0.17}
\definecolor{CarrotOrange}{rgb}{0.93, 0.57, 0.13}
\definecolor{MidnightBlue}{rgb}{0.1, 0.1, 0.44}
\definecolor{Magenta}{rgb}{0.50, 0.0, 0.50}
\definecolor{BrightPink}{rgb}{1.0, 0.0, 0.5}
\definecolor{BrilliantRose}{rgb}{1.0, 0.33, 0.64}
\definecolor{ChromeYellow}{rgb}{1.0, 0.65, 0.0}
\definecolor{HotMagenta}{rgb}{1.0, 0.11, 0.81}
\definecolor{Auburn}{rgb}{0.43, 0.21, 0.1}
\definecolor{BrightTurquoise}{rgb}{0.03, 0.91, 0.87}
\definecolor{DarkCyan}{rgb}{0.0, 0.55, 0.55}
\definecolor{DarkGray}{rgb}{0.66, 0.66, 0.66}
\definecolor{DimGray}{rgb}{0.41, 0.41, 0.41}
\definecolor{DarkBananaYellow}{RGB}{240,181,67}

\vfuzz \hfuzz
\setlist[itemize]{topsep=0pt,partopsep=0pt,itemsep=0pt,parsep=0pt}
\setlist[itemize,1]{label={\small\textbullet}}
\setlist[itemize,2]{label={\tiny\textbullet}}
\setlist[itemize,3]{label=$\cdot$}
\setlist[enumerate]{topsep=0pt,partopsep=0pt,itemsep=0pt,parsep=0pt}
\setlist[enumerate,1]{label=\roman*)}
\setlist[enumerate,2]{label=\alph*)}
\setlist[enumerate,3]{label=\arabic*)}

\hypersetup{
colorlinks=true,
linkcolor=AO!65!black,
citecolor=AO!65!black,
urlcolor=AppleGreen!65!black,
bookmarksopen=true,
bookmarksnumbered,
bookmarksopenlevel=2,
bookmarksdepth=3
}

\newenvironment{claimproof}[1][Proof of claim]
  {\begin{proof}[#1]}
  {\end{proof}}

\theoremstyle{definition}

\newtheorem{environment}{Environment}[section]

\newtheorem{lemma}[environment]{Lemma}
\AddToHook{env/lemma/begin}{\zcsetup{countertype={environment=lemma}}}
\zcRefTypeSetup{lemma}{
Name-sg = Lemma ,
name-sg = Lemma ,
Name-pl = Lemmas ,
name-pl = Lemmas ,
}

\newtheorem*{lemma*}{Lemma}
\AddToHook{env/lemma*/begin}{\zcsetup{countertype={environment=lemma*}}}
\zcRefTypeSetup{lemma*}{
Name-sg = Lemma ,
name-sg = Lemma ,
Name-pl = Lemmas ,
name-pl = Lemmas ,
}

\newtheorem{proposition}[environment]{Proposition}
\AddToHook{env/proposition/begin}{\zcsetup{countertype={environment=proposition}}}
\zcRefTypeSetup{proposition}{
Name-sg = Proposition ,
name-sg = Proposition ,
Name-pl = Propositions ,
name-pl = Propositions ,
}

\newtheorem{corollary}[environment]{Corollary}
\AddToHook{env/corollary/begin}{\zcsetup{countertype={environment=corollary}}}
\zcRefTypeSetup{corollary}{
Name-sg = Corollary ,
name-sg = Corollary ,
Name-pl = Corollaries ,
name-pl = Corollaries ,
}

\newtheorem{theorem}[environment]{Theorem}
\AddToHook{env/theorem/begin}{\zcsetup{countertype={environment=theorem}}}
\zcRefTypeSetup{theorem}{
Name-sg = Theorem ,
name-sg = Theorem ,
Name-pl = Theorems ,
name-pl = Theorems ,
}

\newtheorem*{theorem*}{Theorem}
\AddToHook{env/theorem*/begin}{\zcsetup{countertype={environment=theorem*}}}
\zcRefTypeSetup{theorem*}{
Name-sg = Theorem ,
name-sg = Theorem ,
Name-pl = Theorems ,
name-pl = Theorems ,
}

\AddToHook{env/conjecture/begin}{\zcsetup{countertype={environment=conjecture}}}
\zcRefTypeSetup{conjecture}{
Name-sg = Conjecture ,
name-sg = Conjecture ,
Name-pl = Conjectures ,
name-pl = Conjectures ,
}

\newtheorem*{hypothesis*}{Hypothesis}
\AddToHook{env/hypothesis*/begin}{\zcsetup{countertype={environment=hypothesis*}}}
\zcRefTypeSetup{hypothesis*}{
Name-sg = Hypothesis ,
name-sg = Hypothesis ,
Name-pl = Hypotheses ,
name-pl = Hypotheses ,
}

\newtheorem{observation}[environment]{Observation}
\AddToHook{env/observation/begin}{\zcsetup{countertype={environment=observation}}}
\zcRefTypeSetup{observation}{
Name-sg = Observation ,
name-sg = Observation ,
Name-pl = Observations ,
name-pl = Observations ,
}

\AddToHook{env/example/begin}{\zcsetup{countertype={environment=example}}}
\zcRefTypeSetup{example}{
Name-sg = Example ,
name-sg = Example ,
Name-pl = Examples ,
name-pl = Examples ,
}

\AddToHook{env/remark/begin}{\zcsetup{countertype={environment=remark}}}
\zcRefTypeSetup{remark}{
Name-sg = Remark ,
name-sg = Remark ,
Name-pl = Remarks ,
name-pl = Remarks ,
}

\zcRefTypeSetup{equation}{
Name-sg = Equation ,
name-sg = Equation ,
Name-pl = Equations ,
name-pl = Equations ,
}

\zcRefTypeSetup{chapter}{
Name-sg = Chapter ,
name-sg = Chapter ,
Name-pl = Chapters ,
name-pl = Chapters ,
}

\zcRefTypeSetup{section}{
Name-sg = Section ,
name-sg = Section ,
Name-pl = Sections ,
name-pl = Sections ,
}

\zcRefTypeSetup{algorithm}{
Name-sg = Algorithm ,
name-sg = Algorithm ,
Name-pl = Algorithms ,
name-pl = Algorithms ,
}

\AddToHook{env/notation/begin}{\zcsetup{countertype={environment=notation}}}
\zcRefTypeSetup{notation}{
Name-sg = Notation ,
name-sg = Notation ,
Name-pl = Notations ,
name-pl = Notations ,
}

\AddToHook{env/question/begin}{\zcsetup{countertype={environment=question}}}
\zcRefTypeSetup{question}{
Name-sg = Question ,
name-sg = Question ,
Name-pl = Questions ,
name-pl = Questions ,
}

\AddToHook{env/problem/begin}{\zcsetup{countertype={environment=problem}}}
\zcRefTypeSetup{problem}{
Name-sg = Problem ,
name-sg = Problem ,
Name-pl = Problems ,
name-pl = Problems ,
}

\newtheorem{definition}[environment]{Definition}
\AddToHook{env/definition/begin}{\zcsetup{countertype={environment=definition}}}
\zcRefTypeSetup{definition}{
Name-sg = Definition ,
name-sg = Definition ,
Name-pl = Definitions ,
name-pl = Definitions ,
}

\zcRefTypeSetup{figure}{
Name-sg = Figure ,
name-sg = Figure ,
Name-pl = Figures ,
name-pl = Figures ,
}

\newtheoremstyle{beautypalour}
{3pt}
{3pt}
{}
{}
{\itshape}
{.}
{.5em}
{}
\theoremstyle{beautypalour}

\newtheorem{claim}{Claim}[environment]
\AddToHook{env/claim/begin}{\zcsetup{countertype={environment=claim}}}
\zcRefTypeSetup{claim}{
Name-sg = Claim ,
name-sg = Claim ,
Name-pl = Claims ,
name-pl = Claims ,
}

\AddToHook{env/beautifulclaim/begin}{\zcsetup{countertype={environment=beautifulclaim}}}
\zcRefTypeSetup{beautifulclaim}{
Name-sg = Claim ,
name-sg = Claim ,
Name-pl = Claims ,
name-pl = Claims ,
}

\usetikzlibrary{calc}
\usetikzlibrary{fit}
\usetikzlibrary{decorations}
\usetikzlibrary{decorations.pathmorphing}
\usetikzlibrary{decorations.text}
\usetikzlibrary{shapes,hobby}

\tikzset{
	position/.style args={#1:#2 from #3}{
		at=($(#3)+(#1:#2)$)
	}
}

\tikzset{
  v:main/.style = {draw, circle, scale=0.8, thick,fill=black,inner sep=0.7mm},
  v:ghost/.style = {inner sep=0pt,scale=1},
  >={latex},
  e:marker/.style = {line width=8.5pt,line cap=round,opacity=0.35,color=DarkGoldenrod},
  e:main/.style = {line width=1pt},
}

\newcommand{\Coordinates}{%
    \def\R{8}

    \draw[step=.25, very thin, gray!40] (-\R,-\R) grid (\R,\R);

    \draw[step=1, thin, gray!70] (-\R,-\R) grid (\R,\R);

    \draw[->, thick] (-\R,0) -- (\R+0.35,0) node[right] {$x$};
    \draw[->, thick] (0,-\R) -- (0,\R+0.35) node[above] {$y$};

    \foreach \x in {-8,-5,...,8}{
        \draw (\x,0.08) -- (\x,-0.08);
        \ifnum\x=0\else
            \node[below,font=\scriptsize] at (\x,0) {\x};
        \fi
    }

    \foreach \y in {-6,-5,...,6}{
        \draw (0.08,\y) -- (-0.08,\y);
        \ifnum\y=0\else
            \node[left,font=\scriptsize] at (0,\y) {\y};
        \fi
    }

}

\newcommand{\N}{\mathbb{N}}

\newcommand{\LLL}{\mathcal{L}}

\definecolor{BrilliantRose}{rgb}{1.0, 0.33, 0.64}
\definecolor{MidnightBlack}{rgb}{0.1,0.1,.34}
\definecolor{MidnightBlue}{rgb}{0.1,0.1,0.43}
\definecolor{Black}{rgb}{0,0, 0}
\definecolor{Blue}{rgb}{0, 0 ,1}
\definecolor{Red}{rgb}{1, 0 ,0}
\definecolor{White}{rgb}{1, 1, 1}
\definecolor{grey}{rgb}{.6, .6, .6}
\definecolor{Mygreen}{rgb}{.0, .7, .0}
\definecolor{Yellow}{rgb}{.55,.55,0}
\definecolor{Mustard}{rgb}{1.0, 0.86, 0.35}
\definecolor{applegreen}{rgb}{0.55, 0.71, 0.0}
\definecolor{darkturquoise}{rgb}{0.0, 0.81, 0.82}
\definecolor{celestialblue}{rgb}{0.29, 0.59, 0.82}
\definecolor{green_yellow}{rgb}{0.68, 1.0, 0.18}
\definecolor{crimsonglory}{rgb}{0.75, 0.0, 0.2}
\definecolor{darkmagenta}{rgb}{0.30, 0.0, 0.30}
\definecolor{magenta}{rgb}{0.50, 0.0, 0.50}
\definecolor{internationalorange}{rgb}{1.0, 0.31, 0.0}
\definecolor{darkorange}{rgb}{1.0, 0.55, 0.0}
\definecolor{ao}{rgb}{0.0, 0.5, 0.0}
\definecolor{awesome}{rgb}{1.0, 0.13, 0.32}
\definecolor{darkcyan}{rgb}{0.0, 0.50, 0.50}
\definecolor{violet}{rgb}{0.93, 0.51, 0.93}
\definecolor{brown}{rgb}{0.65, 0.16, 0.16}
\definecolor{orange}{rgb}{1.0, 0.65, 0.0}
\definecolor{cornflowerblue}{rgb}{0.39, 0.58, 0.93}

\newcommand{\Acal}{\mathcal{A}}

\newcommand{\Ccal}{\mathcal{C}}

\newcommand{\Pcal}{\mathcal{P}}

\newcommand{\Tcal}{\mathcal{T}}

\newcommand{\Nbbb}{\mathbb{N}}

\newcommand{\remove}[1]{}

\newcommand{\tw}{\mathsf{tw}\xspace}%
\newcommand{\poly}{\mathbf{poly}\xspace}%

\newcommand{\p}{\mathsf{p}\xspace}%

\usepackage{xparse}

\predate{}
\date{}
\postdate{}

 \ifdefined\sth
 \else 
 \newcommand{\sth}{\mathrel : }
 \fi
\ifdefined\bd
\else
\newcommand{\bd}{\mathsf{bd}}
\fi

\newcommand{\R}{\mathbb{R}}

\newcommand{\hh}{\end{document}}

\newcommand{\yes}{\textsf{yes}}

\DeclareFontEncoding{LS1}{}{}
\DeclareFontSubstitution{LS1}{stix}{m}{n}
\DeclareSymbolFont{symbolsstix}{LS1}{stixscr}{m}{n}
\SetSymbolFont{symbolsstix}{bold}{LS1}{stixscr}{b}{n}
\DeclareMathSymbol{\mathvisiblespace}{0}{symbolsstix}{"B6}

\RequirePackage{stmaryrd}
\usepackage{textcomp}
\DeclareUnicodeCharacter{2286}{\subseteq}
\DeclareUnicodeCharacter{2192}{\ifmmode\to\else\textrightarrow\fi}
\DeclareUnicodeCharacter{2203}{\ensuremath\exists}
\DeclareUnicodeCharacter{183}{\cdot}
\DeclareUnicodeCharacter{2200}{\forall}
\DeclareUnicodeCharacter{2264}{\leq}
\DeclareUnicodeCharacter{2265}{\geq}
\DeclareUnicodeCharacter{8614}{\mathbin{\mapsto}}
\DeclareUnicodeCharacter{8656}{\Leftarrow}
\DeclareUnicodeCharacter{8657}{\Uparrow}
\DeclareUnicodeCharacter{8658}{\Rightarrow}
\DeclareUnicodeCharacter{8659}{\Downarrow}
\DeclareUnicodeCharacter{8669}{\rightsquigarrow}
\newcommand{\eqdef}{\stackrel{{\scriptsize\rm def}}{=}}
\DeclareUnicodeCharacter{8797}{\eqdef}
\DeclareUnicodeCharacter{8870}{\vdash}
\DeclareUnicodeCharacter{8873}{\Vdash}
\DeclareUnicodeCharacter{22A7}{\models}
\DeclareUnicodeCharacter{9121}{\lceil}
\DeclareUnicodeCharacter{9123}{\lfloor}
\DeclareUnicodeCharacter{9124}{\rceil}
\DeclareUnicodeCharacter{2208}{\in}
\DeclareUnicodeCharacter{9126}{\rfloor}
\DeclareUnicodeCharacter{9655}{\triangleright}
\DeclareUnicodeCharacter{9665}{\triangleleft}
\DeclareUnicodeCharacter{9671}{\diamond}
\DeclareUnicodeCharacter{9675}{\circ}
\DeclareUnicodeCharacter{10178}{\bot}
\DeclareUnicodeCharacter{10229}{\longleftarrow}
\DeclareUnicodeCharacter{10230}{\longrightarrow}
\DeclareUnicodeCharacter{10231}{\longleftrightarrow}
\DeclareUnicodeCharacter{10232}{\Longleftarrow}
\DeclareUnicodeCharacter{10233}{\Longrightarrow}
\DeclareUnicodeCharacter{10234}{\Longleftrightarrow}
\DeclareUnicodeCharacter{10236}{\longmapsto}
\DeclareUnicodeCharacter{10238}{\Longmapsto} 
\DeclareUnicodeCharacter{10503}{\Mapsto}    
\DeclareUnicodeCharacter{10971}{\mathrel{\not\hspace{-0.2em}\cap}}
\DeclareUnicodeCharacter{65294}{\ldotp}
\DeclareUnicodeCharacter{65372}{\mid}

\title{Spanning Paths and Cycles: Structural Limitations of the Irrelevant Vertex Technique}

\predate{}
\date{}
\postdate{}

\preauthor{}
\DeclareRobustCommand{\authorthing}{
	\begin{center}
		Dimitrios M.\@ Thilikos\thanks{Supported by French National Research Agency (ANR) under project GODASse ANR-24-CE48-4377 and under the France 2030 grant reference number ANR-24-RRII-0002 operated by the Inria Quadrant Program.}~\thanks{\href{mailto:sedthilk@thilikos.info}{sedthilk@thilikos.info}} \\
		{\small LIRMM, Université de Montpellier, CNRS, Montpellier, France} \\
		  \medskip
		Sebastian Wiederrecht\thanks{Supported by the Institute for Basic Science (IBS-R029-C1).}~\thanks{\href{mailto:sebastian.wiederrecht@gmail.com}{wiederrecht@kaist.ac.kr}} \\
		{\small School of Computing, KAIST, South Korea} \\
\end{center}}
\author{\authorthing}
\postauthor{}

\begin{document}
\maketitle

\begin{abstract}
\noindent The Irrelevant Vertex Technique is one of the cornerstones of algorithmic graph theory, underlying Robertson and Seymour's algorithm for \textsc{Disjoint Paths} and much of the algorithmic Graph Minors theory.
We show that, in the setting of spanning routing, this technique exhibits an exact combinatorial limitation.

Unlike classical routing problems, spanning routing is not governed by the number of distinguished vertices but by the way they are distributed throughout the graph.
Here, the input is a triple $(G,R,\mathcal{T})$ where $(G,R)$ is an annotated graph and $\mathcal{T}$ is a set of terminal pairs.
The goal is to determine if $G$ contains a family of internally disjoint paths connecting the pairs in $\mathcal{T}$ such that the union of the paths spans the entire set $R$.
We identify a new structural parameter of annotated graphs, called $\mathsf{depth}_2$, that measures precisely this phenomenon.
Our main result is a complete combinatorial dichotomy:
for every red-minor-closed class of annotated graphs, the Irrelevant Vertex Technique applies to \textsc{Spanning Disjoint Paths} \textsl{if and only if} the class has bounded $\mathsf{depth}_2$.
Thus $\mathsf{depth}_2$ forms the exact structural boundary between classes where the Robertson-Seymour paradigm survives and those where it  breaks down.
Our proof combines a new local structure theorem for annotated graphs of bounded $\mathsf{depth}_2$ with a spanning analogue of the celebrated Vital Linkage Theorem.
The resulting algorithm solves \textsc{Spanning Disjoint Paths} in time $2^{2^{\poly(k+d)}}\cdot n^2$ where $d$ denotes the $\mathsf{depth}_2$ of the input instance.
We provide matching lower bounds which show that beyond bounded $\mathsf{depth}_2$ no irrelevant-vertex rule can exist, even on planar graphs.
In particular, $\mathsf{depth}_2$ is the exact combinatorial barrier for the Irrelevant Vertex Technique under spanning constraints.

\end{abstract}

\let\sc\itshape
\thispagestyle{empty}

\newpage

\newpage
\thispagestyle{empty}
\tableofcontents

\newpage

\setcounter{page}{1}


\newcommand{\TC}{\textsc{Terminal Cyclability}\xspace}
\newcommand{\HCp}{\textsc{Hamiltonian Cycle}\xspace}
 
\section{Introduction}\label{sec:introduction}

A guiding principle in algorithmic graph theory is that the tractability of a
graph problem on a class of graphs is governed by the substructures the class
is allowed to contain. \HCp is the textbook example. On any class of bounded
treewidth it is solvable in linear time by Courcelle's theorem
\cite{Courcelle1990Monadic}, while on planar graphs it is already
\textsf{NP}-complete \cite{GareyJohnsonTarjan1976PlanarHC}. Writing
$\textsc{HC}(\mathcal{C})$ for the restriction of \HCp to a minor-closed class
$\mathcal{C}$, these two facts, together with the Grid Theorem
\cite{RobertsonS1986GraphMinorsV}, yield a clean dichotomy:
$\textsc{HC}(\mathcal{C})$ is polynomial-time solvable if and only if
$\mathcal{C}$ excludes some grid as a minor, equivalently, if and only if
$\mathcal{C}$ has bounded treewidth \cite{RobertsonS1986GraphMinorsV}. Grids are thus the precise
\emph{dichotomy delimiters}, under minors, of polynomial solvability for \HCp.
The same mechanism -- tractability below a treewidth threshold via
Courcelle's theorem plus hardness above it via \textsf{NP}-hardness in planar graphs --
delineates the complexity of a large family of graph problems.
Many algorithmic results that extend beyond bounded treewidth exploit this structural picture through the \emph{Irrelevant Vertex Technique}, which repeatedly reduces instances of large treewidth to equivalent instances of bounded treewidth before applying dynamic programming. 
This paradigm underlies Robertson and Seymour's algorithm for \textsc{Disjoint Paths} \cite{RobertsonSeymour1995DisjointPaths} and many subsequent developments in algorithmic graph minors.
Despite its remarkable generality, surprisingly little is understood about the precise conditions under which the technique itself applies.
Rather than seeking a dichotomy in the sense of classical complexity theory, in this paper we ask the following fundamental question:
\begin{eqnarray}
\begin{minipage}{13cm}
\textsl{What are the exact structural limitations of the Irrelevant Vertex Technique?}
\end{minipage}\label{questi0}
\end{eqnarray}

The search for such limitations naturally leads to \textsl{annotated graphs}, where a distinguished set of vertices forms part of the input.
An \emph{annotated graph} is a pair $(G,R)$ where $G$ is a graph and $R\subseteq V(G)$, depicted in \textcolor{BostonUniversityRed}{red}, is a set of annotated vertices.
In the \TC problem the input is an annotated graph $(G,R)$ and the question is whether $G$
contains a simple cycle passing through every red vertex. When $R=V(G)$, this is exactly \HCp. 
Given a class $\mathcal{A}$ of annotated graphs, we write $\textsc{TC}(\mathcal{A})$ for the restriction of \TC to the annotated graphs in
$\mathcal{A}$.
Passing from graphs to annotated graphs sharpens the lens: the natural containment relation is no longer the minor relation but its annotated refinement, the \emph{red-minor} relation \cite{ProtopapasThilikosWiederrecht2025ColorfulMinors,GorskyPW2026Quickly}.
Here, each branch set representing a red vertex is required to contain at least one red vertex \cite{ThilikosWiederrecht2025Bidimensionality,ProtopapasThilikosWiederrecht2025ColorfulMinors} of the host graph.
If the Irrelevant Vertex Technique admits an exact structural characterisation in this setting, it should therefore be governed not by excluded minors, but by excluded red-minors.

Cycles through prescribed vertices are a classical theme in graph theory, studied since the 1960s \cite{Dirac1960Abstrakten,BondyLovasz1981Cycles,
HaggkvistThomassen1982Circuits,Gould2009Cycles}.
The flavour of the problem is older than its complexity.
In Lewis Carroll's word game Doublets, first published in Vanity Fair in 1879, one transforms one word into another through intermediate words while being required to visit prescribed words along the way \cite{Carroll1879Doublets}.
Viewed graph-theoretically, Carroll's puzzle asks for a path through a distinguished set of vertices, making it a charming historical precursor of \TC (see also the related discussion by Bj\"orklund, Husfeldt, and Taslaman in \cite{BjorklundHusfeldtTaslaman2012Shortest}).

From the algorithmic point of view, however, the problem has been always studied under the assumption that the prescribed set is small.
Robertson and Seymour's algorithm for \textsc{Disjoint Paths} implies polynomial-time solvability for every fixed number $k=|R|$ of prescribed vertices \cite{RobertsonSeymour1995DisjointPaths}.
Since then, a long line of work has focused on improving the dependence on $k$, leading to algebraic algorithms \cite{BjorklundHusfeldtTaslaman2012Shortest}, kernelization results \cite{Wahlstrom2013Tutte}, and, most recently, an almost-optimal $2^{\mathbf{O}(\sqrt{k}\log k)}\cdot n$-time algorithm -- again relying on the Irrelevant Vertex Technique -- together with a near-linear kernel for planar graphs \cite{GahlawatRathodZehavi2025TCycle}.
Throughout this line of work, the parameter governing tractability is the size of the prescribed set.

The present paper departs from this regime entirely.
We place \textsl{no restriction} on the size of $R$; indeed, it may be even linear in $|V(G)|$.
Instead, we show that, for spanning routing problems, the structural quantity governing the applicability of the Irrelevant Vertex Technique is not the number of annotated vertices, but the way they are distributed throughout the graph.

The difficulty is not that spanning routing appears algorithmically harder than classical routing.
Rather, it is that one of the central paradigms of algorithmic graph theory -- the Irrelevant Vertex Technique -- suddenly ceases to behave in the way we have come to expect.
To better understand this phenomenon, we briefly recall the reduction paradigm itself.

The Irrelevant Vertex Technique repeatedly replaces an instance by an equivalent induced sub-instance of bounded treewidth on which  dynamic programming becomes applicable.
Thus it may be viewed as a structural data-reduction scheme whose combinatorial guarantee is the following.

We say that a problem $\Pi$, whose instances are graphs $G$ or annotated graphs $(G,R)$, has the
\emph{irrelevant vertex property} if every instance is equivalent to an induced sub-instance of bounded treewidth:
an induced subgraph $G'$ of $G$ -- with $R\subseteq V(G')$ in the annotated case -- of bounded treewidth that is a \textsf{yes}-instance if and only if the original is.
The reduction proceeds by repeatedly identifying a vertex whose deletion preserves the answer -- an \emph{irrelevant vertex} -- until bounded treewidth is reached.
Whenever such irrelevant vertices can be found efficiently, the reduction yields a polynomial-time preprocessing algorithm with a provable structural guarantee.
Once the treewidth is bounded, dynamic programming solves the problem in time linear in $n$, with the treewidth bound absorbed into the constant factor.
The quantity controlled by this data-reduction is the \emph{treewidth} of the reduced instance rather than its size and it is this bounded width that renders the instance tractable.
A central question in algorithmic graph theory is therefore to identify the conditions under which a problem enjoys the irrelevant vertex property.

For problems on \textsl{graphs}, this paradigm is remarkably robust.
Indeed, it often seems almost unconditional: with no restriction on the input, the irrelevant vertex property holds for the \textsc{Disjoint Paths} and the \textsc{Minor Checking} problem, as well as their common generalisation, the \textsc{Folio} problem \cite{RobertsonSeymour1995DisjointPaths} (see also \cite{CavallaroGKTW2026Optimal}).
For many further problems, excluding a minor is enough to recover it: 
recently, Sau, Stamoulis, and Thilikos \cite{SauStamoulisThilikos2025CMSO} extended the paradigm to a broad family of graph problems definable in a certain fragment of monadic second-order logic.

For \emph{annotated} problems this behaviour breaks down.

The obstacle is no longer the graph structure alone, but also the interaction between the graph and its annotation.
Excluding any non-planar pattern -- regardless of the annotation -- still permits every planar annotated graph and, on fully annotated planar instances \TC, is exactly \HCp, hence \textsf{NP}-hard. 
This brings us to the question driving this paper.
For a red-minor-closed class $\mathcal{A}$ of annotated graphs, we arrive at a more precise formulation of (\ref{questi0}):
\begin{eqnarray}
\begin{minipage}{13cm}
\textsl{What condition should $\mathcal{A}$ satisfy in order for \textsc{TC}$(\mathcal{A})$ to have the irrelevant vertex property?}
\end{minipage}\label{questi}
\end{eqnarray}
Typically, a precise answer to such a question amounts to finding an annotated graph
parameter whose boundedness across $\mathcal{A}$ is \textsl{both} necessary and
sufficient for the desired property.
Since Terminal Cyclability specialises to Hamiltonian Cycle on fully annotated graphs, where the governing parameter is treewidth, it is natural to search for the answer among annotated analogues of treewidth.
As treewidth is characterised by the largest grid minor, such analogues arise by prescribing how the grid interacts with the annotation, i\@e.\@ the red vertices.

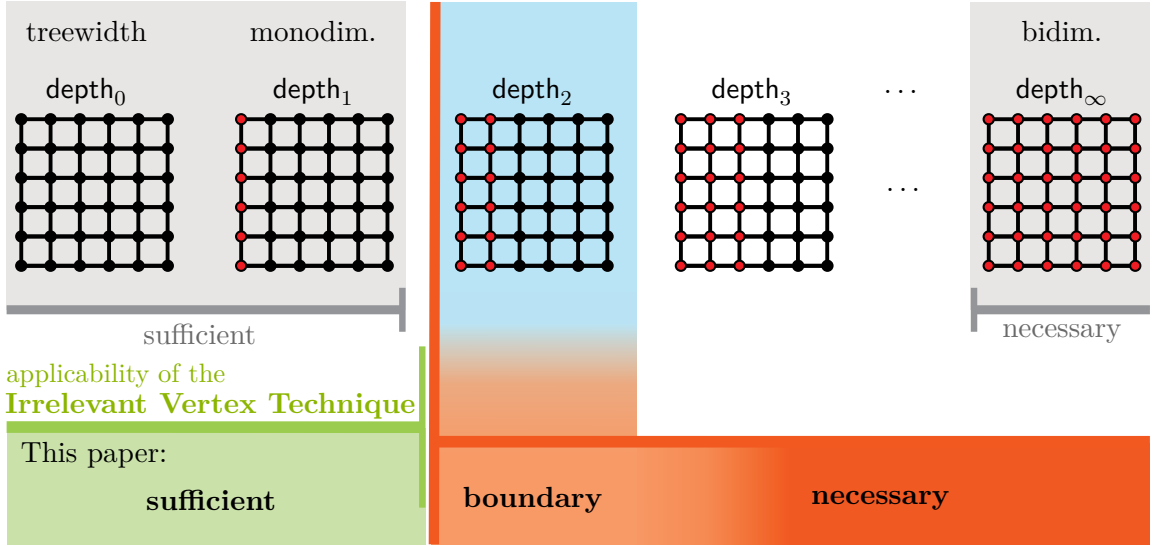
\begin{figure}[ht]
 \centering
 \begin{tikzpicture}

 \pgfdeclarelayer{background}
		\pgfdeclarelayer{foreground}
			
		\pgfsetlayers{background,main,foreground}

 \begin{pgfonlayer}{background}
 \pgftext{\includegraphics[width=15.5cm]{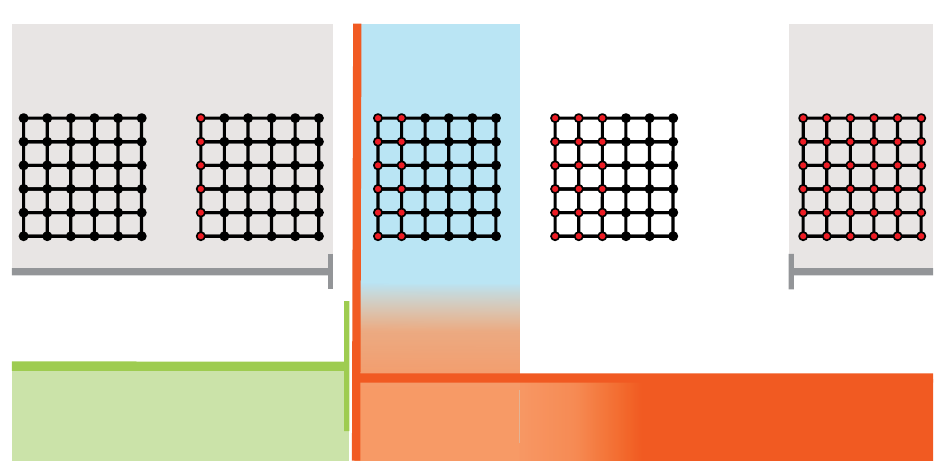}} at (C.center);
 \end{pgfonlayer}{background}
			
 \begin{pgfonlayer}{main}
 \node (C) [v:ghost] {};
 
 \end{pgfonlayer}{main}
 
 \begin{pgfonlayer}{foreground}

    \node (tw) [v:ghost] at (-6.5,3.1) {treewidth};

    \node (monodim) [v:ghost] at (-3.5,3.1) {monodim.\@};

    \node (bidim) [v:ghost] at (6.4,3.1) {bidim.\@};

    \node (depth_0) [v:ghost] at (-6.5,2.3) {$\mathsf{depth}_0$};

    \node (depth_1) [v:ghost] at (-3.5,2.3) {$\mathsf{depth}_1$};

    \node (depth_2) [v:ghost] at (-0.6,2.3) {$\mathsf{depth}_2$};

    \node (depth_3) [v:ghost] at (2.3,2.3) {$\mathsf{depth}_3$};

    \node (depth_dots) [v:ghost] at (4.3,2.3) {$\dots$};

    \node (depth_infty) [v:ghost] at (6.4,2.3) {$\mathsf{depth}_{\infty}$};

    \node (dots) [v:ghost] at (4.3,1) {$\mathbf{\dots}$};

    \node (sufficient_0) [v:ghost] at (-5,-0.9) {\textcolor{DimGray}{sufficient}};

    \node (necessary_0) [v:ghost] at (6.4,-0.9) {\textcolor{DimGray}{necessary}};

    \node (applicable) [v:ghost] at (-6.11,-1.45) {{\textcolor{AppleGreen}{\small applicability of the}}};
    
    \node (Irrelevant) [v:ghost] at (-4.86,-1.85) {\textbf{\textcolor{AppleGreen}{Irrelevant Vertex Technique}}};

    \node (ThisPaper) [v:ghost] at (-6.4,-2.5) {{This paper:}};

    \node (sufficient_1) [v:ghost] at (-4.85,-3.1) {\textbf{sufficient}};

    \node (boundary) [v:ghost] at (-0.6,-3.1) {\textbf{boundary}};

    \node (necessary_1) [v:ghost] at (4,-3.1) {\textbf{necessary}};
 \end{pgfonlayer}{foreground}

 \end{tikzpicture}
 \caption{The hierarchy of annotated grid parameters extending treewidth. Previous work (highlighted in \textcolor{black!60}{gray}) established that bounded $\mathsf{depth}_1$ suffices to guarantee tractability of \TC, via a reduction to the bounded-terminal regime where irrelevant-vertex methods are available, while unbounded $\mathsf{depth}_{\infty}$ already exhibits \textsf{NP}-hardness. This paper shows that the exact boundary where the Irrelevant Vertex Technique may be applied is reached at $\mathsf{depth}_2$, thereby establishing a combinatorial dichotomy.}
 \label{fig_DepthHierarchy}
\end{figure}

This yields an infinite hierarchy of parameters as illustrated in \zcref{fig_DepthHierarchy}. 
For each fixed integer $r\geq 0$, the \emph{$r$-outer-annotated $(k\times k)$-grid} is obtained from the $(k\times k)$-grid by colouring the $r$ leftmost columns red.
The parameter $\mathsf{depth}_r(G,R)$ is then defined as the maximum $k$ for which an $r$-outer-annotated $(k\times k)$-grid is a red-minor of $(G,R)$.
In the limit, we set $\mathsf{depth}_\infty(G,R)$ by requiring all vertices of the grid to be red.
Notice that the case $r=0$ simply asks for the largest grid-minor in the graph while disregarding annotation entirely.
Hence $\mathsf{depth}_0$ is functionally equivalent to treewidth itself.
As illustrated in \zcref{fig_DepthHierarchy}, reddening additional columns only makes the pattern harder to realize.
Consequently, the parameters form a genuine hierarchy -- the \textsf{depth}\emph{-hierarchy} -- extending treewidth to annotated graphs, whose all levels collapse to treewidth whenever $R=V(G)$.

The two extremes of this hierarchy were already understood and carry various names.
The parameter $\mathsf{depth}_1$ appears as \emph{monodimensionality} in the work of Sau, Schirrmacher, Siebertz, Stamoulis, Thilikos, and Vigny \cite{sau2026modelcheckinglowmonodimensionality} while $\mathsf{depth}_{\infty}$ was introduced by Thilikos and Wiederrecht \cite{ThilikosWiederrecht2025Bidimensionality} as \emph{bidimensionality} and was also used in the definition of the fragment $\mathsf{CMSO/tw+dp}$ of {Counting Monadic Second Order Logic}, introduced in  \cite{SauStamoulisThilikos2025CMSO}.
Both names are inspired by the distribution of the red vertices within the grid.
One may therefore view the parameters $\mathsf{depth}_r$ as successive levels of monodimensionality: the red vertices remain essentially one-dimensional inside the grid, but occupy an increasingly thick band.

From the perspective of the irrelevant vertex property, the two extremes already determine the landscape.
We have already observed $\mathsf{depth}_0$ to be the same as treewidth.
Moreover, bounded $\mathsf{depth}_1$ --equivalently bounded \textsl{torso treewidth} (see \cite{JansenSwennenhuis2024SteinerTree,HodorLaMicekRambaud2026ApexForest}), or \textsl{monodimensionality} -- is sufficient for \textsc{Terminal Cyclability} to enjoy the irrelevant vertex property \cite{ProtopapasThilikosWiederrecht2025ColorfulMinors}.
At the other end, unbounded $\mathsf{depth}_{\infty}$ -- known as \emph{bidimensionality} (see \cite{ThilikosWiederrecht2025Bidimensionality,ProtopapasThilikosWiederrecht2025ColorfulMinors}) -- already contains every planar annotated graph and therefore exhibits \textsf{NP}-hardness 
which renders bounded $\mathsf{depth}_{\infty}$ a necessary condition for \textsc{Terminal Cyclability} to enjoy the irrelevant vertex property.
Thus bounded $\mathsf{depth}_1$ is sufficient, bounded $\mathsf{depth}_{\infty}$ is necessary, and \textsl{the exact boundary must lie somewhere between them}.

Our main result shows that the exact boundary is reached just one step beyond monodimensionality.

\begin{theorem}\label{thm:intro-TC}
Let $\mathcal{A}$ be a red-minor-closed class of annotated graphs. Then
$\textsc{TC}(\mathcal{A})$ has the irrelevant vertex property if and only if
$\mathsf{depth}_2$ is bounded in $\mathcal{A}$, equivalently, if and only if
$\mathcal{A}$ excludes some $2$-outer-annotated  grid as a red-minor.
\end{theorem}

\zcref{thm:intro-TC} is merely the visible face of a more general phenomenon.
The underlying combinatorial dichotomy concerns spanning routing in its full generality.

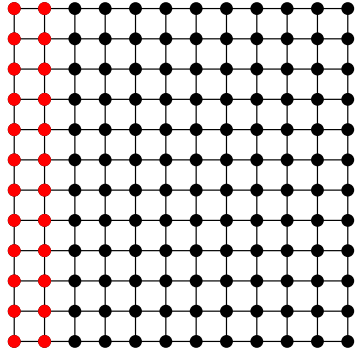
\begin{figure}
\begin{center}
\begin{tikzpicture}[
 scale=0.4,
 vertex/.style={circle, fill, inner sep=1.7pt},
 edge/.style={black, line width=0.45pt}
]

\def\N{11}

\foreach \i in {0,...,\N} {
 \draw[edge] (\i,0) -- (\i,\N);
 \draw[edge] (0,\i) -- (\N,\i);
}

\foreach \x in {0,...,\N} {
 \foreach \y in {0,...,\N} {
 \node[vertex, black] at (\x,\y) {};
 }
}

\foreach \y in {0,...,\N} {
 \node[vertex, red] at (0,\y) {};
 \node[vertex, red] at (1,\y) {};
}

\end{tikzpicture}
\end{center}
\caption{The $2$-outer-annotated  $(12\times 12)$-grid. The annotated vertices are
depicted in \textcolor{red}{red}.}
\label{fig:two-outer-grid-intro}
\end{figure}

\paragraph{From cycles to linkages.}
An annotated graph $(G,R)$ records a distinguished set $R$; instead of asking that a single cycle covers $R$, we may prescribe a pattern of connections and ask that the routing realising it covers $R$.
A natural generalisation of \textsc{Terminal Cyclability} is the following spanning routing problem:
In \textsc{Spanning Disjoint Paths} (\textsc{SDP}) the input is an annotated graph $(G,R)$ together with terminal pairs $\mathcal{T}=\{(s_1,t_1),\dots,(s_k,t_k)\}$ and the task is to find pairwise vertex-disjoint paths $P_1,\dots,P_k$, forming a \emph{linkage} of $G$,  such that $P_i$ joins $s_i$ to $t_i$ for every $i\in[k]$ and, moreover, $R\subseteq V(P_1\cup\cdots\cup P_k)$.
A solution is thus a linkage that realises the prescribed pattern \emph{and} spans the annotation.
If $R=\varnothing$ this is the classical \textsc{Disjoint Paths} problem; if $R=V(G)$ and $k=1$ it asks for a Hamiltonian path between two prescribed vertices and both \TC and \HCp reduce to it.
In particular \textsc{SDP} is \textsf{NP}-complete already in restricted forms \cite{Karp1972Reducibility,GareyJohnsonTarjan1976PlanarHC}, so the question is not whether unrestricted \textsc{SDP} is tractable, but whether the treewidth-reduction machinery behind \textsc{Disjoint Paths} survives the spanning requirement.

\paragraph{Disjoint Paths and the Irrelevant Vertex Technique.}
The \textsc{Disjoint Paths} problem is one of the cornerstones of algorithmic
graph theory. For $k=2$ its tractability is already classical
\cite{Seymour1980Disjoint,Shiloach1980Polynomial,Thomassen19802Linked} and in
Graph Minors XIII Robertson and Seymour proved that, for every fixed $k$, it is
solvable in time $f(k)\cdot n^{3}$ \cite{RobertsonSeymour1995DisjointPaths}. This
theorem is a basic primitive behind minor testing, rooted minor detection, folio
computation, and a large part of the algorithmic theory of graph minors
\cite{LokshtanovSaurabhZehavi2020Efficient}; its running time was later improved
to quadratic \cite{KawarabayashiKobayashiReed2012Quadratic}, i.e., one 
running in $f(k)\cdot n^2$ time for some (huge) function $f$. 
The latter was recently improved in two different ways. 
The one is the almost-linear, for fixed $k$, algorithm of 
Korhonen,  Pilipczuk, and  Stamoulis in \cite{KorhonenPilipczukStamoulis2024AlmostLinear}
and the other is  the  $2^{2^{\poly(b)}\cdot \poly(k)}n^2$ time algorithm of Cavallaro,  Gorsky,  Kreutzer,  Thilikos, and
 Wiederrecht  in \cite{CavallaroGKTW2026Optimal}, where $b$ is the bidimensionlity of the terminals.
 Moreover, the planar case of \textsc{Disjoint Paths} has been  a perennial testing ground for sharper techniques 
\cite{DingSchrijverSeymour1992Planar,Schrijver1994DirectedPlanar,
AdlerEtAl2017PlanarIrrelevant,LokshtanovEtAl2025PlanarDP,
ChoOhOh2023PlanarDP,WlodarczykZehavi2023PlanarKernels}.

The structural engine behind all of these results is precisely the irrelevant vertex
technique described above, which for \textsc{Disjoint Paths} applies
unconditionally: above a treewidth threshold depending only on $k$, some
non-terminal vertex is always irrelevant. Establishing the existence of that
vertex is among the deepest parts of the series: the Vital Linkage Theorem of
Graph Minors XXI \cite{RobertsonSeymour2009UniqueLinkages} and the
Irrelevant Vertex Lemma of Graph Minors XXII
\cite{RobertsonSeymour2012Irrelevant}, later given a shorter proof by
Kawarabayashi and Wollan \cite{KawarabayashiWollan2010Shorter}, are what justify
the algorithm of Graph Minors XIII. Since then, irrelevant vertices have become a
standard interface between structural graph theory and algorithms
\cite{AdlerEtAl2017PlanarIrrelevant,LokshtanovSaurabhZehavi2020Efficient,
KorhonenPilipczukStamoulis2024AlmostLinear,SauStamoulisThilikos2025CMSO,
CavallaroGKTW2026Optimal}.

For the spanning variant, the obstacle is that $R$ is not assumed to be small and
its size alone is too crude a measure: the intuition is that a large red set is harmless when it can be
confined to the outer interface of a grid, yet a red set reaching only two columns
into a large grid already compels every spanning linkage to thread the grid's deep 
interior. What is needed is a structural criterion on $(G,R)$ that captures
exactly when large treewidth still forces an irrelevant vertex outside the
terminals and outside $R$ -- and the parameter that does so is $\mathsf{depth}_2$,
the same parameter that governs \cref{thm:intro-TC}.

\paragraph{The dichotomy.}
The irrelevant vertex property of the first part specialises to
$\textsc{SDP}(\mathcal{A})$, the restriction of \textsc{Spanning Disjoint
Paths} to instances  where $(G,R)\in\mathcal{A}$, with irrelevance understood in the
following concrete sense, the spanning analogue of the classical notion.

\begin{definition}\label{def:intro-irrelevant}
Given some red-minor-closed class of annotated graphs $\Acal$,
we say that $\textsc{SDP}(\mathcal{A})$ has the \emph{irrelevant vertex property}
if there exists a  function $f:\Nbbb\to\Nbbb$ such that, 
for every instance $(G,R,\Tcal)$, there is an equivalent one $(G',R,\Tcal)$
where $G'$ is an induced subgraph of $G$ containing the annotation $R$ and the terminals in $\Tcal$ and 
moreover $\tw(G')\leq f(|\Tcal|)$.
\end{definition}

Our main result determines exactly when $\textsc{SDP}(\mathcal{A})$ has the irrelevant vertex property. The criterion is the same as the one of \cref{thm:intro-TC}. 

\begin{theorem}\label{thm:intro-main-dichotomy}
Let $\mathcal{A}$ be a red-minor-closed class of annotated graphs. Then  $\textsc{SDP}(\mathcal{A})$ has the irrelevant vertex property if and
only if $\mathcal{A}$ has bounded $\mathsf{depth}_2$, equivalently, if and only
if $\mathcal{A}$ excludes some $2$-outer-annotated  grid as a red-minor.
\end{theorem}

\zcref{thm:intro-main-dichotomy} identifies when a natural spanning variant of \textsc{Disjoint Paths} has the irrelevant vertex property and the contrast with the unannotated world is the substance of the statement.
There, the property is in effect always present: for \textsc{Disjoint Paths} and its relatives large treewidth alone exposes an irrelevant vertex, while this is also the case, under the exclusion of a minor, for $\mathsf{CMSO/tw+dp}$-definable problems \cite{SauStamoulisThilikos2025CMSO}
(in fact $\mathsf{CMSO/tw+dp}$ is defined using disjoint paths as a core predicate). 
For the spanning annotated problem that we consider no such blanket guarantee holds: even under the exclusion of a red-minor the property may fail, and $\mathsf{depth}_2$ marks the exact line between the two cases.
The spanning requirement is, as far as we know, the first natural setting in which the Irrelevant Vertex Technique meets a precise non-trivial combinatorial limit.

\paragraph{The positive side.}
On a red-minor-closed class of bounded $\mathsf{depth}_2$, the Robertson--Seymour paradigm is fully restored for a genuinely spanning problem: an instance of \textsc{SDP} can be reduced, by deleting irrelevant vertices, to an equivalent instance of bounded treewidth. The resulting algorithm is summarised by the following.

\begin{theorem}
\label{sdp_times}
There is an algorithm that given an instance  $(G,R,\Tcal)$
of \textsc{SDP} either outputs a solution or reports that a solution does not exist in  $2^{2^{\poly(k+d)}}\cdot n^{2}$ time, where $k\coloneqq |\Tcal|$ and $d\coloneqq\mathsf{depth}_2(G,R)$.
\end{theorem} 

As \TC can be seen as a special case of \textsc{SDP}, where $k=3$, \cref{sdp_times}  also implies an $2^{2^{\poly(d)}}\cdot n^{2}$ time algorithm for \TC. Also, Carroll's puzzle, being   \textsc{SDP} with $k=1$, can be solved in the same running time
where $G$ expresses the ``graph of words'' and $R$ are the prescribed words in it.

The decisive feature of \cref{sdp_times} is that the dependence is on the number $k$ of paths and on $d$ and \emph{not} on $|R|$: the algorithm applies even when the set to be spanned is arbitrarily large, so that what governs the complexity is the placement of $R$ (measured by $d\coloneqq\mathsf{depth}_2(G,R)$) rather than its size.
To our knowledge this is the first parametrization of \textsc{SDP} -- and, through \cref{thm:intro-TC}, of \TC -- that does not charge for the distinguished set; the bounded-distinguished-set regime is recovered as the special case in which $d$ is automatically small.
At the technical heart of the positive direction is a spanning analogue of the \emph{Vital Linkage Theorem}:
a vital instance of \textsc{SDP} of $\mathsf{depth}_2$ at most $d$ has treewidth $2^{\poly(k+d)}$ 
where vitality asks the spanning linkage realising $\mathcal{T}$ to be unique and to span the entire vertex set of the graph (see \cref{thm_findIrrelevantVertexFlatWall}).

\paragraph{The negative side.}
The $\mathsf{depth}_2$ criterion is tight in the strongest sense we could ask for.
If $\mathcal{A}$ has unbounded $\mathsf{depth}_2$ then arbitrarily large $2$-outer-annotated  grids occur and no irrelevant-vertex rule of the above form can exist -- already for a single terminal pair and already on \emph{planar}
graphs.
The obstruction is best seen through $R$-spanning vital linkages.
A classical vital linkage spans all of $V(G)$ and is the unique linkage with its pattern and the unique-linkage theorem forces its host to have bounded treewidth \cite{RobertsonSeymour2009UniqueLinkages}; an $R$-spanning vital linkage is required to be unique only \emph{subject to} spanning $R$.
We construct, for $k=1$, planar annotated graphs of arbitrarily large treewidth carrying an $R$-spanning vital linkage.
Hence the bounded-treewidth phenomenon for vital linkages does not survive the replacement of $V(G)$ by $R$ and no treewidth threshold can
guarantee an irrelevant vertex.
Returning to \TC, the same construction applies already to cycles through a prescribed set: the irrelevant-vertex method provably breaks down once $R$ carries arbitrarily large $2$-outer-annotated  grids, so the dichotomy of \cref{thm:intro-TC} is tight in the same strong sense.
Already for \textsc{Terminal Cyclability}, the situation of the prescribed set, rather than its size, decides whether the technique applies.
Taken together, these results identify $\mathsf{depth}_2$ as the exact combinatorial boundary of the Irrelevant Vertex Technique for spanning routing problems.

\begin{figure}[ht]
 \centering
 \begin{tikzpicture}

 \pgfdeclarelayer{background}
		\pgfdeclarelayer{foreground}
			
		\pgfsetlayers{background,main,foreground}

 \begin{pgfonlayer}{background}
 \pgftext{\includegraphics[width=7.6cm]{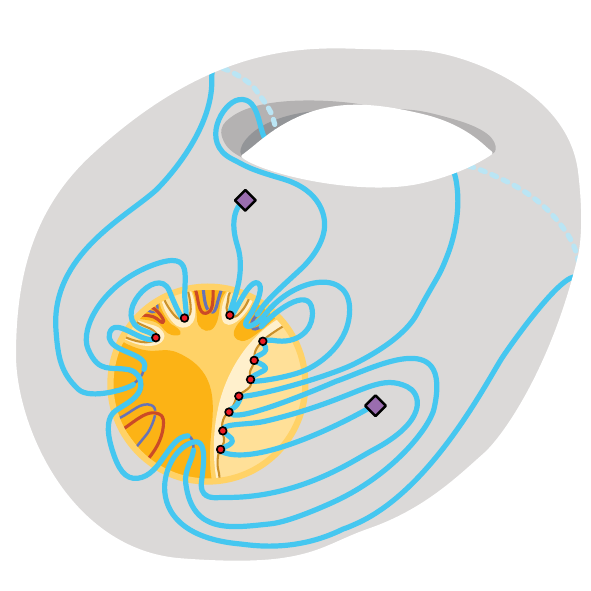}} at (C.center);
 \end{pgfonlayer}{background}
			
 \begin{pgfonlayer}{main}
 \node (C) [v:ghost] {};
 
 \end{pgfonlayer}{main}
 
 \begin{pgfonlayer}{foreground}
 \end{pgfonlayer}{foreground}

 \end{tikzpicture}
 \caption{The local structure of an annotated graph excluding a $2$-outer-annotated grid as a red-minor. The \textcolor{DarkBananaYellow}{yellow} disc is a vortex which is partitioned into non-planar (darker) and planar (lighter) regions. The red vertices lie on the outer face of the planar regions of the vortex.
 In \textcolor{CornflowerBlue}{blue} we depict a path spanning all red vertices while connecting the two \textcolor{Amethyst}{purple} terminals.}
 \label{fig_Depth2Structure}
\end{figure}

\paragraph{Why $\mathsf{depth}_2$?}
Before proceeding, we briefly explain why $\mathsf{depth}_2$ turns out to be the correct parameter and why, in retrospect, one should expect it to appear.

The central obstacle in extending the Irrelevant Vertex Technique to spanning routing is no longer the topology of the graph alone, but the interaction between topology and annotation.
In the Graph Minor Structure Theorem, every graph of sufficiently large treewidth is described in terms of pieces that are almost embeddable into a surface.
The only genuinely complicated parts of such a decomposition are the \emph{vortices}: regions where the graph may exhibit arbitrary non-planar behaviour while remaining attached to the embedded part through a controlled interface.
For spanning problems, these vortices acquire a second role.
Besides carrying the topological complexity of the graph, they may also contain the annotated vertices that every solution is required to visit.
Understanding how these two sources of complexity interact is precisely the difficulty addressed in this paper.

Our local structure theorem reveals that excluding a $2$-outer-annotated grid as a red-minor forces these two phenomena to separate.
On the one hand, bounded $\mathsf{depth}_2$ implies bounded bidimensionality, and therefore every annotated vertex must already lie either in the apex set or inside a vortex~\cite{ThilikosWiederrecht2025Bidimensionality,ProtopapasThilikosWiederrecht2025ColorfulMinors}.
On the other hand, and this is the new structural insight, each vortex itself admits a decomposition into two qualitatively different types of regions.
Some regions may contain the full non-planar complexity of typical vortices but are guaranteed to contain no red vertices.
The remaining regions are essentially planar and, moreover, all red vertices contained in such a region are confined to a single face.
Thus the topology and the annotation become disentangled inside every vortex: non-planar behaviour occurs where there are no annotated vertices, while red vertices occur only in regions that are topologically simple.

This phenomenon is not accidental.
The $2$-outer-annotated grid is precisely the annotated analogue of the shallow-vortex grid introduced by Thilikos and Wiederrecht~\cite{ThilikosW2024Killing}, which characterises when vortices can no longer be simplified inside the Graph Minor Structure Theorem.
Viewed from this perspective, $\mathsf{depth}_2$ is not merely another member of the depth hierarchy but the first parameter that detects whether topology and annotation can still be separated inside vortices.

This separation is ultimately what makes our positive result possible.
Once every annotated vertex lies on a bounded collection of facial regions, spanning linkages can be analysed using topological arguments that are impossible in general.
See \zcref{fig_Depth2Structure} for an illustration.
Conversely, the $2$-outer-annotated grid is exactly the first configuration in which annotated vertices can no longer be confined to a bounded number of faces.
The additional freedom created by this failure is precisely what enables the construction underlying our lower bound.
In this sense, the dichotomy of \zcref{thm:intro-main-dichotomy} is already encoded in the local geometry of vortices: bounded $\mathsf{depth}_2$ is exactly the point where topology and annotation cease to interfere with one another.

\paragraph{Beyond the Irrelevant Vertex Technique.}
\zcref{thm:intro-main-dichotomy} settles the applicability of the Irrelevant Vertex Technique for \textsc{SDP}.
It does not, however, settle the complexity of the problem itself.
Our dichotomy concerns the boundary of a fundamental algorithmic paradigm rather than the boundary of polynomial-time solvability.

In particular, it remains open whether \textsc{SDP} is already \textsf{NP}-hard for some fixed value of $k$ on annotated graphs of unbounded $\mathsf{depth}_2$.
A positive answer would establish $\mathsf{depth}_2$ not only as the exact boundary of the Irrelevant Vertex Technique, but also as the exact structural boundary of polynomial-time solvability.

There are, however, reasons to believe that the picture may be more subtle.
One possibility is that algorithmic paradigms fundamentally different from the Irrelevant Vertex Technique -- for example algebraic methods in the spirit of Bj\"orklund, Husfeldt, and Taslaman~\cite{BjorklundHusfeldtTaslaman2012Shortest} -- may succeed beyond classes of bounded $\mathsf{depth}_2$.
Another is that $\mathsf{depth}_2$ already marks the boundary of fixed-parameter tractability, so that algorithms beyond it necessarily leave the \textsf{FPT} regime and require genuinely different techniques.
Recent progress on \textsf{XP} algorithms exploiting the structure of vortices~\cite{FioriniJWY2025IntegerPrograms} may provide a promising starting point for such an approach.

Whether either of these possibilities can be realised remains, in our view, one of the most intriguing questions raised by this work.

\paragraph{Related work.}
\cref{thm:intro-main-dichotomy} should be compared with recent quantitative work
on irrelevant vertices. Cavallaro, Gorsky, Kreutzer, Thilikos, and Wiederrecht
prove a general irrelevant-vertex theorem for the $(k,d)$-\textsc{Folio} problem
\cite{CavallaroGKTW2026Optimal}: if the terminals are drawn from a set $R$ of
bidimensionality $b$, the irrelevant-vertex threshold has the form
$2^{\poly(b+d)}\cdot\poly(k)$, so the distribution of the distinguished set
controls the quantitative cost of the machinery. Our result is aligned with this
viewpoint but goes one step further for the spanning problem: $\mathsf{depth}_2$
is not only a quantitative parameter in the running time, but the
\emph{qualitative} sharp barrier for the applicability of the technique itself.
More broadly, the paper fits into the developing theory of rooted, annotated, and
colorful graph parameters: rooted grid minors show how a distinguished set can be
forced onto a grid \cite{MarxSeymourWollan2017RootedGrid}; torso-type parameters
measure the contribution of the annotation to treewidth after annotation-free
parts are compressed
\cite{JansenSwennenhuis2024SteinerTree,HodorLaMicekRambaud2026ApexForest,
FioriniEtAl2025FaceCovers}; annotated treewidth and bidimensionality measure
grid minors controlled by the annotation
\cite{ThilikosWiederrecht2025Bidimensionality,SauStamoulisThilikos2025CMSO,
GorskyPW2026Quickly}; and colorful minors provide the
structural language for several possibly overlapping annotated sets
\cite{ProtopapasThilikosWiederrecht2025ColorfulMinors}.

\subsection{Proof outline and organisation}
The positive direction rests on three ingredients. A local structure theorem for
annotated graphs of bounded $\mathsf{depth}_2$ (\cref{sec_structureTheorem}) shows that,
away from large clique minors, the annotation can be combed onto a few layers of
a flat wall. This feeds the spanning unique-linkage theorem
(\cref{sec_VitalTheorem}), which bounds the treewidth of vital instances of
bounded $\mathsf{depth}_2$, building on the vital-linkage bounds of
\cite{CavallaroGKTW2026Optimal} and the flat-wall machinery of
\cite{SauST2024More,ProtopapasThilikosWiederrecht2025ColorfulMinors,GorskyPW2026Quickly}. Together
they yield the irrelevant vertex (\cref{sec_Irrelevant}) and the algorithm
(\cref{sec_algorithm}). The matching lower bound, including its planarity, is
given in \cref{sec_CounterExample}.
Below, we give a slightly more in-depth overview of our proof and the techniques used.

\paragraph{The local structure of annotated graphs of small $\mathsf{depth}_2$.}
The main structural contribution of this paper is a ``local'' structure theorem for annotated graphs with bounded $\mathsf{depth}_2$.
Here the term ``local'' means ``with respect to a large wall'' -- or tangle if one prefers the language of Robertson and Seymour's Graph Minors Series.
The difference to typical structure theorems from the realm of graph minors is that, instead of describing the structure of the entire graphs in terms of clique sums or a tree-decomposition, we only describe the structure of the part of the graph which is highly connected to a given wall which is part of the input.
Historically, such local structure theorems have turned out to be the more important and powerful counterpart of global theorems -- see for example the introduction of \cite{RobertsonS2003GraphMinorsXVI}.
Our structure theorem describes the local structure of annotated graphs of small $\mathsf{depth}_2$ in terms of a so-called almost embedding on a surface of bounded Euler-genus under the deletion of a bounded size apex set and with the exception of a bounded number of vortices -- each of bounded depth.
What sets the theorem apart from comparable structure theorems \cite{ProtopapasThilikosWiederrecht2025ColorfulMinors,GorskyPW2026Quickly} is a more refined treatment of the red vertices:
All red vertices must be confined to the vortices and into the apex set. Moreover, we show that each vortex can be partitioned into non-planar areas which are void of red vertices, and (essentially) planar areas, each with a single face that hosts all red vertices contained in the area.
This ``facial'' behaviour is key for our algorithmic applications and the most technical part in our structural result.
\smallskip

Our local structure theorem is substantially stronger than a refinement of the Flat Wall Theorem.
Its proof genuinely requires the full machinery of Robertson and Seymour's structure theory, combining almost embeddings with the recently developed theory of colorful minors.
Although in this paper we use the theorem only as an ingredient in the irrelevant-vertex algorithm, we expect it to have independent applications beyond the present work.
\smallskip

To be able to find this structure, we, in particular, require an annotated refinement of the so-called Two Paths Theorem which may be of independent interest (see for example \zcref{thm_FacialSocieties}).
We base our proof on recent structure theorems for annotated graphs of bounded bidimensionality -- since bounded $\mathsf{depth}_2$ implies bounded $\mathsf{depth}_{\infty}$.
Indeed, we use the structure theorem with polynomial bounds due to Gorsky, Protopapas, and Wiederrecht \cite{GorskyPW2026Quickly}.
From here, we start analysing the structure of the vortices with respect to the red vertices inside.
In order to reach the final state we show that a technique introduced by Thilikos and Wiederrecht \cite{ThilikosW2024Killing} known as ``vortex killing'' can be applied to either find a large $2$-outer-annotated grid as a red-minor, or we may increase the apex set slightly in order to remove a certain substructure we call a ``candle'' from the vortices.
The absence of candles is then the core property that enables us to push the red vertices remaining onto a single facial boundary.

A conceptual difficulty is, that we are working with almost embeddings and not real embeddings.
This is unavoidable whenever one tries to prove a local structure theorem in the context of graph minors.
However, this also makes it hard to properly define the notion of a face.
Therefore, we introduce a property we call ``facial'' for annotated graphs in almost embeddings and make use of a special kind of curve to express the property we need.

\paragraph{A Vital Spanning Linkage Function.}
The fundamental engine behind Robertson and Seymour's algorithm for the \textsc{$k$-Disjoint Paths} problem is the so-called \emph{unique linkage function} \cite{RobertsonSeymour2009UniqueLinkages}.
Let us denote this function by $\beta$.
We say that an instance of the \textsc{$k$-Disjoint Paths} problem is \emph{vital} if it has a unique solution and this solution uses all vertices in the graph.
The \emph{Vital Linkage Theorem} of Robertson and Seymour \cite{RobertsonSeymour2009UniqueLinkages} (see \cite{CavallaroGKTW2026Optimal} for a more recent proof with better bounds) says that any vital instance of the \textsc{$k$-Disjoint Paths Problem} has treewidth at most $\beta(k)$.
The Irrelevant Vertex Technique itself is then a non-trivial and technical consequence of the Vital Linkage Theorem.
The route we are taking is that we show on vital instances of \textsc{$k$-Spanning Disjoint Paths} that, if $\mathsf{depth}_2$ is bounded, one can reduce the instance to an equivalent vital instance of \textsc{$k$-Disjoint Paths} where the Vital Linkage Theorem bounds the treewidth.
This allows us to prove that vital instances of \textsc{$k$-Spanning Disjoint Paths} of small $\mathsf{depth}_2$ have small treewidth.

In order to implement this reduction, we show that -- under the assumption of bounded $\mathsf{depth}_2$ -- any part of the solution that interacts with the red vertices in a meaningful way must be radically limited.
At their core, most of our arguments amount to the observation that the order in which the red vertices are visited, as well as the distribution of the red vertices onto the paths, are mostly irrelevant to the feasibility of the solution.
This, in particular, means that even a surprisingly small amount of bidimensional infrastructure is often enough to reroute a solution to avoid some vertex and thereby contradict the vitality assumption.
  
\paragraph{Irrelevant vertices.}
With a spanning version of the Vital Linkage Function established, we are now able to proceed with the design of the actual algorithm.
The final algorithm will only make use of the so-called \emph{Flat Wall Theorem} (see for example \cite{RobertsonSeymour1995DisjointPaths,KawarabayashiTW2018New}) and does not require the full power of our local structure theorem for annotated graphs of bounded $\mathsf{depth}_2$.
To prove the existence of an irrelevant vertex -- always under the assumption that $\mathsf{depth}_2$ is bounded -- we follow a well-established route.
We prove a variant of the so-called ``Annulus Combing Lemma'' of Golovach, Stamoulis, and Thilikos \cite{GolovachST2023Combing} for the setting of spanning linkages.
Roughly speaking the Annulus Combing Lemma says that, given a large enough cylindrical grid which is ``almost embedded'' into an annulus such that most of the remaining graph -- including the terminals -- can be confined into the two holes of the annulus, then any solution to the \textsc{$k$-Disjoint Paths} problem can be transformed to one that intersects the embedded part in the columns of the cylindrical grid only.
The main advantage of this piece of technology is, that it allows to decompose any given solution into two independent instances of the \textsc{$f(k)$-Disjoint Paths} problem.
Establishing a variant for the \textsc{$k$-Spanning Disjoint Paths} problem allows us to implement the following strategy:

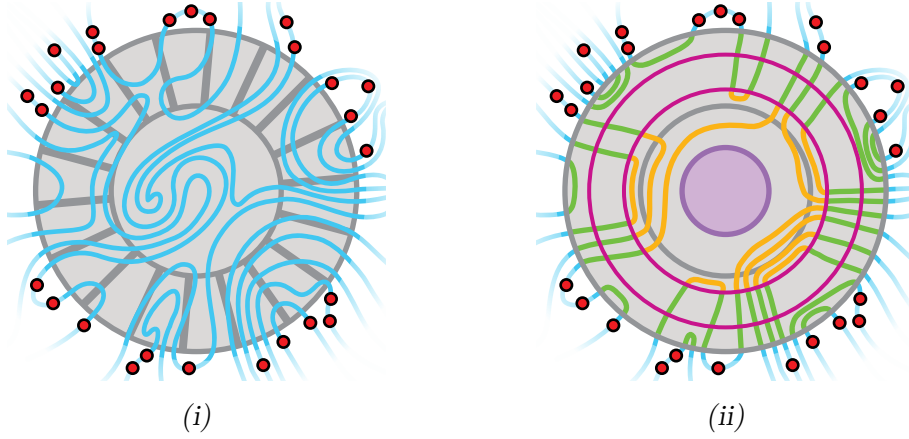
\begin{figure}[ht]
 \centering
 \begin{tikzpicture}

 \pgfdeclarelayer{background}
		\pgfdeclarelayer{foreground}
			
		\pgfsetlayers{background,main,foreground}
			
 \begin{pgfonlayer}{main}
 \node (C) [v:ghost] {};

 \node(L) [v:ghost] at (-3.5,0) {
 \begin{tikzpicture}

 \pgfdeclarelayer{background}
		 \pgfdeclarelayer{foreground}
			
		 \pgfsetlayers{background,main,foreground}

 \begin{pgfonlayer}{background}
 \pgftext{\includegraphics[width=5cm]{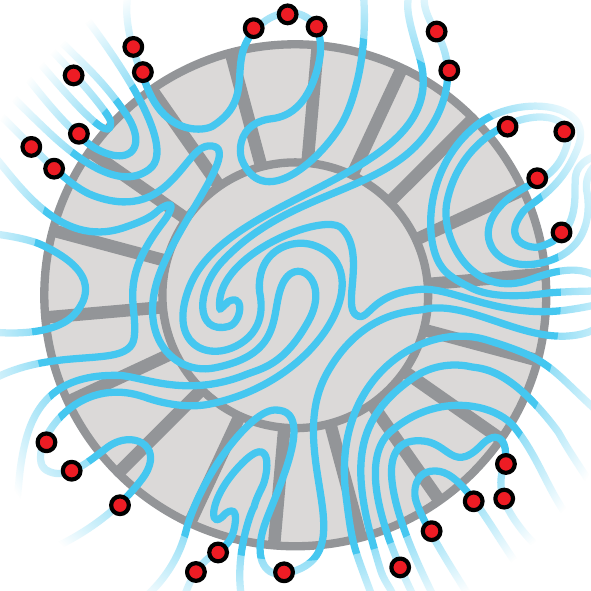}} at (C.center);
 \end{pgfonlayer}{background}
			
 \begin{pgfonlayer}{main}
 \node (C) [v:ghost] {};
 
 \end{pgfonlayer}{main}
 
 \begin{pgfonlayer}{foreground}
 \end{pgfonlayer}{foreground}

 \end{tikzpicture}
 };

 \node(M) [v:ghost] at (3.5,0) {
 \begin{tikzpicture}

 \pgfdeclarelayer{background}
		 \pgfdeclarelayer{foreground}
			
		 \pgfsetlayers{background,main,foreground}

 \begin{pgfonlayer}{background}
 \pgftext{\includegraphics[width=5cm]{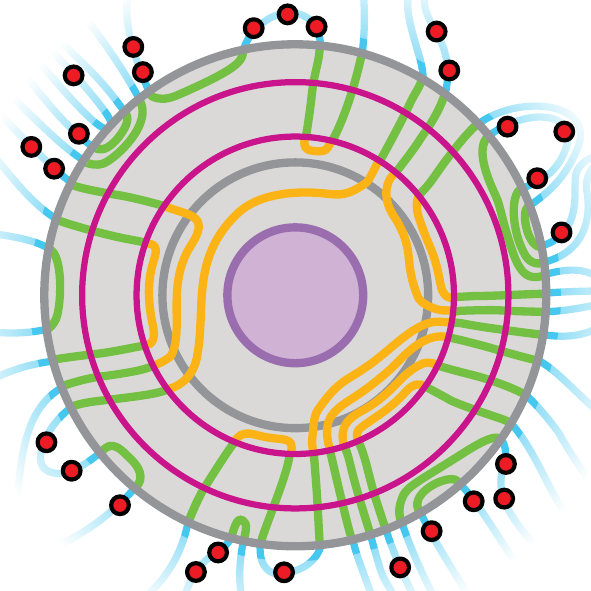}} at (C.center);
 \end{pgfonlayer}{background}
			
 \begin{pgfonlayer}{main}
 \node (C) [v:ghost] {};
 
 \end{pgfonlayer}{main}
 
 \begin{pgfonlayer}{foreground}
 \end{pgfonlayer}{foreground}

 \end{tikzpicture}
 };

 \node (i) [v:ghost] at (-3.5,-3) {\textsl{(i)}};
 \node (ii) [v:ghost] at (3.5,-3) {\textsl{(ii)}};

 \end{pgfonlayer}{main}
 
 \begin{pgfonlayer}{foreground}
 \end{pgfonlayer}{foreground}

 \end{tikzpicture}
 \caption{(i) An instance of the \textsc{$k$-Disjoint Paths} problem partially drawn in a disc ``guarded'' by a large cylindrical wall.
 (ii) The outcome of an application of the Annulus Combing Lemma followed by an application of known techniques for redrawing linkages in a disc. The result is an equivalent solution to the problem which avoids the central part of the wall.}
 \label{fig_RedrawModel}
\end{figure}

Once we have found a large flat wall -- one may think of a flat wall as a big grid-like subgraph $W$ of a graph $G$ such that $W$ itself together with all attachments of $G-W$ to the interior parts of $W$ can be drawn\footnote{Technically this is not a drawing but an ``almost embedding'', but for the sake of this informal description, the intuition of a drawing is sufficient. See \zcref{subsec_grpahMinors} for the formal definitions.} in a disc.
Since $\mathsf{depth}_2$ is bounded, there must exist a large subwall $W'$ of $W$ whose interior attachments do not contain any red vertices.
Another application of pigeonhole yields a still huge subwall $W''$ of $W'$ such that now also no terminal attaches to the interior.
Notice now that the inner part of $W''$ is surrounded by a large cylindrical grid which allows for the application of the Annulus Combing Lemma.
See \zcref{fig_RedrawModel} for an illustration.
To prove the spanning version of the Annulus Combing Lemma, we adapt the recent proof due to Cavallaro, Gorsky, Kreutzer, Thilikos, and Wiederrecht \cite{CavallaroGKTW2026Optimal} to our setting.
A key observation is that this is the only part where the Spanning Vital Linkage Function is explicitly applied.
Once the Spanning Annulus Combing Lemma is proven, we have reached a point where neither our structure theorem nor the Spanning Vital Linkage Function are explicitly necessary any more.

\paragraph{The lower bound.}
To prove the lower bound we give a recursive construction of gadgets which contain vital instances of \textsc{$k$-Spanning Disjoint Paths} for $k \in \{ 1,2\}$ while having increasing treewidth and being red-minors of sufficiently large $2$-outer annotated grids.
The idea follows more or less from the way the existence of the Spanning Unique Linkage function is proven:
A crucial part of the upper bound is to consider a vital solution and then decompose the paths in this solution into families of homotopic subpaths in a surface which have all of their endpoints on certain parts of the boundary -- this is really why we require the structure theorem for annotated graphs of bounded $\mathsf{depth}_2$.
Here, we encounter the same thing.
However, while bounded $\mathsf{depth}_2$ implies that all red vertices sit -- essentially -- on the boundary of faces, if $\mathsf{depth}_2$ is unbounded, there are ways to ``skip over'' some of the red vertices with additional paths.
The ability to skip over can be exploited to push parts of the solution deeper and deeper into a family of concentric cycles without generating any chance to reroute.

\begin{figure}[ht]
 \centering
 \begin{tikzpicture}

 \pgfdeclarelayer{background}
		\pgfdeclarelayer{foreground}
			
		\pgfsetlayers{background,main,foreground}

 \begin{pgfonlayer}{background}
 \pgftext{\includegraphics[width=15.5cm]{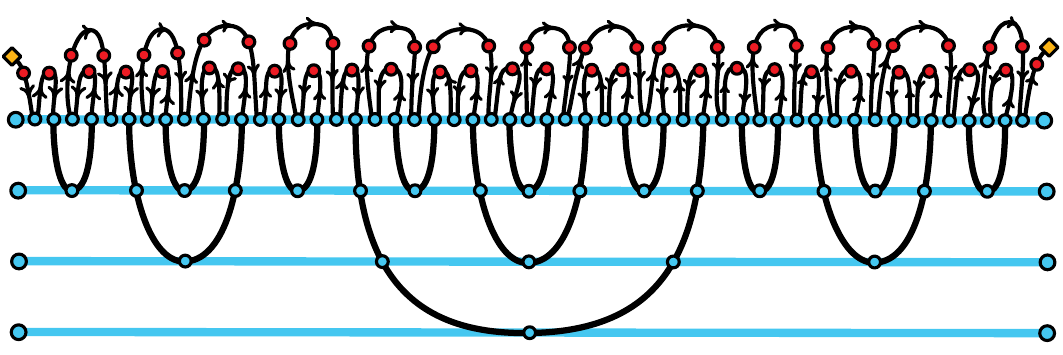}} at (C.center);
 \end{pgfonlayer}{background}
			
 \begin{pgfonlayer}{main}
 \node (C) [v:ghost] {};
 
 \end{pgfonlayer}{main}
 
 \begin{pgfonlayer}{foreground}
 \end{pgfonlayer}{foreground}

 \end{tikzpicture}
 \caption{The recursive construction of depth $4$ that shows that within the class of planar annotated graphs that are red-minors of $2$-outer annotated grids there exist examples which are vital instances of \textsc{$1$-Spanning Disjoint Paths} of unbounded treewidth.}
 \label{fig_Forward4Gadget}
\end{figure}

\zcref{fig_Forward4Gadget} shows an example of one of our gadgets.
The two 
\textcolor{DarkBananaYellow}{yellow}
vertices are the terminals and the black path is the unique solution.
The arrows in the figure indicate in which direction the solution must traverse the edges drawn above the top \textcolor{CornflowerBlue}{blue} line when going from left to right.
Notice that there are indeed $2$ levels above the top \textcolor{CornflowerBlue}{blue} line and those $2$ levels can be used to skip forward and backward.
The red vertices along these oriented edges are precisely what forces the entire solution to be unique:
Each red vertex has degree $2$, which means  that all of these ``subdivided edges'' must be used by any solution.
Moreover, every second one of the vertices on the top \textcolor{CornflowerBlue}{blue} line is incident with two of these forced edges while every other vertex along that line is incident to one of the black paths intruding deeper into the lower levels.
The construction in \zcref{fig_Forward4Gadget} is of a ``forward'' type and has a relatively simple recursive structure.

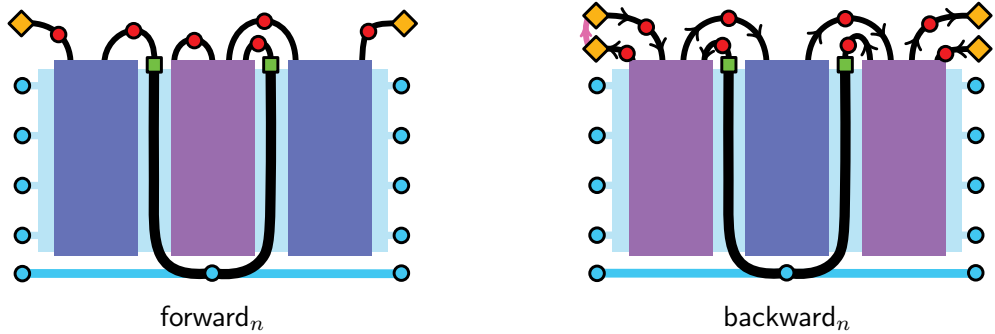
\begin{figure}[ht]
 \centering
 \begin{tikzpicture}

 \pgfdeclarelayer{background}
		\pgfdeclarelayer{foreground}
			
		\pgfsetlayers{background,main,foreground}
			
 \begin{pgfonlayer}{main}
 \node (C) [v:ghost] {};

 \node(L) [v:ghost] at (-3.8,0) {
 \begin{tikzpicture}

 \pgfdeclarelayer{background}
		 \pgfdeclarelayer{foreground}
			
		 \pgfsetlayers{background,main,foreground}

 \begin{pgfonlayer}{background}
 \pgftext{\includegraphics[width=5.5cm]{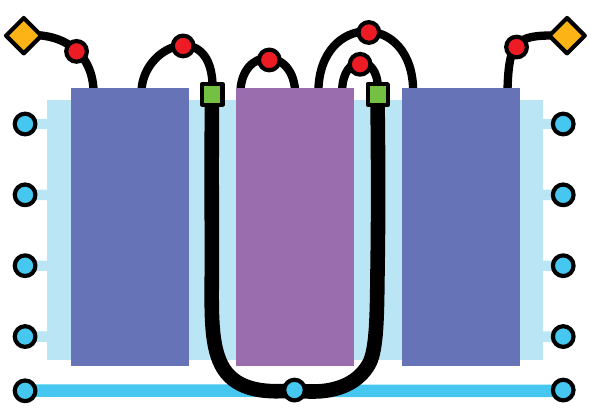}} at (C.center);
 \end{pgfonlayer}{background}
			
 \begin{pgfonlayer}{main}
 \node (C) [v:ghost] {};
 
 \end{pgfonlayer}{main}
 
 \begin{pgfonlayer}{foreground}
 \end{pgfonlayer}{foreground}

 \end{tikzpicture}
 };

 \node(M) [v:ghost] at (3.8,0) {
 \begin{tikzpicture}

 \pgfdeclarelayer{background}
		 \pgfdeclarelayer{foreground}
			
		 \pgfsetlayers{background,main,foreground}

 \begin{pgfonlayer}{background}
 \pgftext{\includegraphics[width=5.5cm]{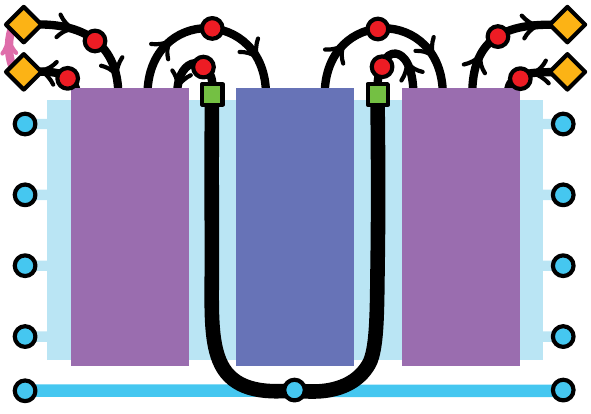}} at (C.center);
 \end{pgfonlayer}{background}
			
 \begin{pgfonlayer}{main}
 \node (C) [v:ghost] {};
 
 \end{pgfonlayer}{main}
 
 \begin{pgfonlayer}{foreground}
 \end{pgfonlayer}{foreground}

 \end{tikzpicture}
 };

 \node (i) [v:ghost] at (-3.8,-2.3) {$\mathsf{forward}_n$};
 \node (ii) [v:ghost] at (3.8,-2.3) {$\mathsf{backward}_n$};

 \end{pgfonlayer}{main}
 
 \begin{pgfonlayer}{foreground}
 \end{pgfonlayer}{foreground}

 \end{tikzpicture}
 \caption{Schematic illustrations of the gadgets of the forward type -- on the left -- and the backward type -- on the right. In \textcolor{MidnightBlue}{dark blue} we mark copies of the lower level forward gadget, while the \textcolor{Amethyst}{purple} boxes mark copies of the lower level backward gadget.}
 \label{fig_GeneralGadgets}
\end{figure}

In fact, we can divide the gadget in \zcref{fig_Forward4Gadget} into four parts:
The black path that goes all the way down to the bottom \textcolor{CornflowerBlue}{blue} line, the part that is drawn to the left of it, to its right, and inside of it.
The left and the right most parts are themselves smaller gadgets of the forward type, while the one in the middle is being entered from the right and also leaves to the right again.
This is a gadget of the ``backward'' type.
In \zcref{fig_GeneralGadgets} we give a schematic illustration of the recursive construction of both types of gadgets.
Due to the simple inductive construction, proving that these gadgets possess the desired properties is surprisingly straightforward.

\section{Preliminaries}\label{sec_prelim}
We introduce basic concepts as well as more advanced ideas from the toolkit of graph minors.

\subsection{Basics}
We open with a short array of definitions for commonly used concepts and notation.
By $\mathbb{N}$ we denote the set of non-negative integers.
Given any two integers $a,b\in\mathbb{N}$, we write $[a,b]$ for the set $\{z\in\mathbb{N} ~\!\colon\!~ a\leq z\leq b\}.$
Notice that the set $[a,b]$ is empty whenever $a>b.$
For any positive integer $c$ we set $[c]\coloneqq [1,c].$

\paragraph{Treewidth.}
A \emph{tree-decomposition} for a graph $G$ is a pair $(T,\beta)$ such that $T$ is a tree, $\beta\colon V(T)\to 2^{V(G)}$ assigns to each node of $T$ a subset of the vertices of $G$ known as a \emph{bag}, $\bigcup_{t\in V(T)}G[\beta(t)] = G$, and for each $v\in V(G)$, the set $\{ t\in V(T) \mid v\in\beta(t) \}$ is connected.
The \emph{width} of $(T,\beta)$ is defined as $\max_{t\in V(T)} |\beta(t)|-1$.

The \emph{treewidth} of a graph $G$, denoted by $\mathsf{tw}(G)$, is the smallest integer $k$ such that $G$ has a tree-decomposition of width at most~$k$.

\paragraph{Red-minors and bidimensionality.}
Let $H$ and $G$ be graphs.
A \emph{minor model} of $H$ in $G$ is a family of connected subgraphs $\{ G_v\}_{v \in V(H)} \cup \{ G_e\}_{e \in E(H)}$ of $G$ such that
\begin{enumerate}
    \item $V(G_v) \cap V(G_u) = \emptyset$ for all $u \neq v \in V(H)$,
    \item $|V(G_u) \cap V(G_{uv})| = 1$ for all $uv \in E(G)$,
    \item $V(G_u) \cap V(G_{e}) = \emptyset$ if $u \notin e$ for all $u \in V(G)$ and $e \in E(G)$,
    \item $G_e$ is a path in $G$ for every $e \in E(G)$, and
    \item $G_{e}$ and $G_{e'}$ are internally vertex-disjoint for all $e \neq e' \in E(G)$.
\end{enumerate}
In case $G_e$ is isomorphic to $K_2$ we do note require $G_e$ to be part of the model explicitly.
This also means that $\{ G_v\}_{v \in V(H)}$ may be a minor model of $H$ in $G$ if for all $uv \in E(G)$, there is an edge in $G$ between $V(G_u)$ and $V(G_v)$.
Most of the time we will be working with this simpler version of a model, but sometimes it will be convenient to provide explicit models for edges.

An annotated graph $(H,R_H)$ is a \emph{red-minor} of an annotated graph $(G,R_G)$ if there exists a minor-model $\{ G_v\}_{v\in V(H)}$ of $H$ in $G$ such that $R_G \cap V(G_v) \neq \emptyset$ for all $v\in R_H$.

The \emph{bidimensionality} of an annotated graph $(G,R)$, denoted by $\mathsf{bidim}(G,R)$, is the largest integer $k$ such that there exists $(\Gamma_k,V(\Gamma_k))$ as a red-minor in $(G,R)$ where $\Gamma_k$ denotes the $(k \times k)$-grid.

\paragraph{Paths, linkages, and vitality.}
A \emph{linkage} $\mathcal{L}$ in a graph $G$ is a set of pairwise vertex-disjoint paths.
Given the set $T \subseteq V(G)$ of endpoints the endpoints of the paths in $\mathcal{L}$, we say that $\mathcal{L}$ a \emph{$T$-linkage}.
Based on $T$, we define the \emph{pattern of $\mathcal{L}$} as $\tau(\mathcal{L}) \coloneqq \{\{s,t\} \mid \text{ there is }L\in \LLL \text{ with endpoints } s \text{ and } t\}$.
In some cases, we are already given a set $\mathcal{T} = \{ (s_i,t_i) \mid i\in[k]\}$ of terminal pairs.
In such a situation, we say that a linkage $\mathcal{L}$ is a \emph{$\mathcal{T}$-linkage} if $|\mathcal{L}| = k$ and for every $i\in[k]$ there is a path $L \in \mathcal{L}$ with endpoints $s_i$ and $t_i$. 

A linkage $\mathcal{L}$ in a graph $G$ is said to be \emph{vital} if $V(\mathcal{L}) = V(G)$ and there does not exist a linkage $\mathcal{L}' \neq \mathcal{L}$ in $G$ with the same pattern.
\smallskip

We say that a path $P$ in $G$ is \emph{internally disjoint} from a set $X \subseteq V(G)$ if $V(P) \cap X$ does not contain any vertex of $P$ that is not an endpoint vertex.
Given a graph $G$ and two subsets $A, B \subseteq V(G)$, an \emph{$A$-path} in $G$ is a path with both endpoints in $A$ and internally disjoint from $A$, and an \emph{$A$-$B$-path} is a path with one endpoint in $A$, the other in $B$, and internally disjoint from $A \cup B$.
An \emph{$A$-$B$-linkage} in $G$ is a linkage consisting of $A$-$B$ paths.
If $H$ is a subgraph of $G$, an \emph{$H$-path} is a $V(H)$-path of length at least one with no edge in $E(H)$.
\smallskip

A core result on graphs with vital linkages is the main theorem of Graph Minors XXI by Robertson and Seymour \cite{RobertsonSeymour2009UniqueLinkages} stating that a graph with a vital linkage $\mathcal{L}$ must have treewidth bounded in some function depending only on $|\mathcal{L}|$.
We state here the recent result of Cavallaro, Gorsky, Kreutzer, Thilikos, and Wiederrecht \cite{CavallaroGKTW2026Optimal} which provides almost optimal bounds for the treewidth of a graph with a vital linkage.

\begin{proposition}[Cavallaro, Gorsky, Kreutzer, Thilikos, and Wiederrecht \cite{CavallaroGKTW2026Optimal}]\label{prop:VitalLinkage}
There exists a function $\beta \colon \mathbb{N}^2 \to \mathbb{N}$ such that, for every graph, if there exists a vital $T$-linkage for some $T \subseteq V(G)$, then $G$ has treewidth at most $\beta(k,b)$ where $k = |T|$ and $b = \mathsf{bidim}(G,T)$.
Moreover, $\beta(k,b) \in \poly(k) \cdot 2^{\poly(b)}$.
\end{proposition}

\paragraph{Surfaces.}
By a \emph{surface} we mean a compact $2$-dimensional manifold $\Sigma$ with or without boundary.
Let $X \subseteq \Sigma$.
The boundary and interior of $X$ will be denoted $\bd(X)$ and $\mathsf{int}(X)$, respectively.
The topological closure of $X$ is denoted by $\mathsf{cl}(X)$.
If $\Sigma$ is a surface with boundary then we refer to the connected components of $\bd(\Sigma)$ as its \emph{cuffs}.
\smallskip

Let $\gamma \sth [0,1]_{\mathbb{R}}\to \Sigma$ be continuous and injective on $[0,1[_{\mathbb{R}}$.
We say that the image $\gamma([0,1]_{\mathbb{R}})$ of $\gamma$ is a \emph{curve} in $\Sigma$ and $\gamma(0),\gamma(1)$ are its \emph{endpoints}.
A curve is \emph{closed} if its endpoints agree.
A curve with both endpoints on the boundary of $\Sigma$ is called a \emph{boundary curve}.

Two curves $C, C' \subseteq \Sigma$ are \emph{homeomorphic} if there is a homeomorphism $h \sth \Sigma \rightarrow \Sigma$ which maps $C$ to $C'$. 
$C$ and $C'$ are \emph{homotopic} if there is a continuous function $f \sth \Sigma \times [0,1]_{\mathbb{R}} \rightarrow \Sigma$ such that $f(x, 0) = C(x)$ and $f(x, 1) = C'(x)$ and for every $t\in [0,1]_{\mathbb{R}} $ the function $f_t(x)\coloneqq f(x,t)$ maps the boundary of $\Sigma$ to the boundary of $\Sigma$.
We call $f$ a \emph{homotopy}.
Homotopy induces an equivalence relation on the set of curves on a surface $\Sigma$.
We refer to the equivalence classes of homotopy as \emph{homotopy classes} and to the class containing a curve $C \subseteq \Sigma$ as its \emph{homotopy type}.

We will require the following lemma on the number of homotopy types.

\begin{proposition}[Cavallaro, Gorsky, Kreutzer, Thilikos, and Wiederrecht \cite{CavallaroGKTW2026Optimal}]\label{prop_TypeCounting}
 Let $\Sigma$ be a connected surface of genus $g$ with $b \geq 1$ boundary
 components.
 Let $\mathcal{L}$ be a set of pairwise disjoint curves on $\Sigma$ whose endpoints are all on the boundary of $\Sigma$.

 Then the number of homotopy types realised by curves in $\mathcal{L}$ is $\leq 1$ in
 case $b=1$ and $g=0$, it is $\leq 3$ in case $b=2$ and $g=0$, and otherwise it is bounded by $4b + 6g - 6$.
\end{proposition}

Let $\Sigma$ be a connected surface. A curve $C \subseteq \Sigma$ is \emph{contractible} if one of the components of $\Sigma - C$ is homeomorphic to an open disc.
Otherwise it is \emph{non-contractible}.

\subsection{Tools more specific to graph minor structure theory}\label{subsec_grpahMinors}
We continue by briefly introducing some key concepts for handling the general structure of $H$-minor-free graphs.
The definitions we introduce here originate fromt he work of Robertson and Seymour \cite{RobertsonS1999GraphMinorsXVII}.
They were later refined by Kawarabayashi, Thomas, and Wollan \cite{KawarabayashiTW2021Quickly} and then adopted and reiterated by several authors (see for example \cite{ThilikosW2024Excluding, PaulPTW2024Obstructionsa, PaulPTW2025Local, GorskySW2025Polynomial}).

\paragraph{Brambles.}
Let $G$ be a graph.
A \emph{bramble} in $G$ is a collection of connected subgraphs $\mathcal{B}$ of $G$ such that for all $B_1,B_2 \in \mathcal{B}$ either $V(B_1) \cap V(B_2) \neq \emptyset$ or there exists an edge with one end in $B_1$ and the other in $B_2$.
A \emph{hitting set} for a bramble $\mathcal{B}$ is a set $S \subseteq V(G)$ such that $S \cap V(B) \neq \emptyset$ for all $B \in \mathcal{B}$.
The \emph{order} of a bramble $\mathcal{B}$ is the minimum size of a hitting set for $\mathcal{B}$.

\begin{proposition}[Reed \cite{Reed1997Tree}]\label{prop_brambles}
Let $k \geq 0$ be an integer and $G$ be a graph.
If $G$ has a bramble of order $k+1$, then $\mathsf{tw}(G) \geq k$.
\end{proposition}

\paragraph{Separations and tangles.}
Let $G$ be a graph and $k$ be a positive integer.
A separation of $G$ is a pair $(A,B)$ such that $A,B \subseteq V(G)$, $A -\cup B = V(G)$, and there is no edge with one end in $A \setminus B$ and the other in $B \setminus A$.
The \emph{order} of a separation $(A,B)$ of $G$ is defined as $|A \cap B|$.
We denote by $\mathcal{S}_k(G)$ the collection of all separations $(A,B)$ of order less than $k$ in $G$.

An \emph{orientation} of $\mathcal{S}_k(G)$ is a set $\mathcal{O}$ such that for all $(A,B)\in\mathcal{S}_k(G)$ exactly one of $(A,B)$ and $(B,A)$ belongs to $\mathcal{O}$. 
A \emph{tangle} of order $k$ in $G$ is an orientation $\mathcal{T}$ of $\mathcal{S}_k(G)$ such that for all $(A_1,B_1),(A_2,B_2),(A_3,B_3)\in\mathcal{T}$, it holds that $G[A_1]\cup G[A_2]\cup G[A_3]\neq G$.
If $\mathcal{T}$ is a tangle and $(A,B)\in\mathcal{T}$ we call $A$ the \emph{small side} and $B$ the \emph{big side} of $(A,B)$.

Let $G$ be a graph and $\mathcal{T}$ and $\mathcal{D}$ be tangles of $G$.
We say that $\mathcal{D}$ is a \emph{truncation} of $\mathcal{T}$ if $\mathcal{D}\subseteq\mathcal{T}$.
\medskip

Let $G$ and $H$ be graphs as well as $\mathcal{T}$ be a tangle in $G$.
We say that a minor-model $\mu$ of $H$ in $G$ is \emph{controlled} by $\mathcal{T}$ if there does not exist a separation $(A,B)\in\mathcal{T}$ of order less than $|V(H)|$ and an $x \in V(H)$ such that $\mu(x)\subseteq A\setminus B$.

Let $\mathcal{T}$ be a tangle of order at least $1$ in $G$.
The \emph{$\mathcal{T}$-big component} of $G$ is the unique component $J$ of $G$ such that $V(J) \subseteq B$ for all $(A,B) \in \mathcal{T}$. where $A \cap B = \emptyset$. 

\paragraph{Meshes.}
Let $n,m$ be integers with $n,m\geq 2$.
A \emph{$(n\times m)$-mesh} is a graph $M$ which is the union of paths $M=P_1\cup\cdots\cup P_n\cup Q_1\cup \cdots \cup Q_m$ where
    \begin{itemize}
        \item $P_1,\cdots,P_n$ are pairwise vertex-disjoint, and $Q_1,\cdots,Q_m$ are pairwise vertex-disjoint.
        \item for every $i\in [n]$ and $j\in [m]$, the intersection $P_i\cap Q_j$ induces a path,
        \item each $P_i$ is a  $V(Q_1)$-$V(Q_m)$-path intersecting the paths $Q_1,\cdots Q_m$ in the given order, and each $Q_j$ is a $V(P_1)$-$V(P_m)$-path intersecting the paths $P_1,\cdots, P_h$ in the given order. 
    \end{itemize}
We say that the paths $P_1,\cdots,P_n$ are the \emph{horizontal paths}, and the paths $Q_1,\cdots,Q_m$ are the \emph{vertical paths}.
The union $P_{1} \cup P_{n} \cup Q_{1} \cup Q_{m}$ is a cycle called the \emph{perimeter} of $M$.
The unique cycle in the union $P_{i} \cup P_{i+1} \cup Q_{j} \cup Q_{j+1}$, where $i \in [n - 1]$ and $j \in [m - 1]$, is called a \emph{brick} of $M$.
A mesh $M'$ is a \emph{submesh} of a mesh $M$ if every horizontal (vertical) path of $M'$ is a subpath of a horizontal (vertical) path $M$, respectively.
We write \emph{$n$-mesh} as a shorthand for an $(n \times n)$-mesh.

Let $r \in \mathbb{N}$ with $r\geq 3$, let $G$ be a graph, and $M$ be an $r$-mesh in $G$.
Let $\mathcal{T}_M$ be the orientation of $\mathcal{S}_r$ such that for every $(A,B)\in\mathcal{T}_M$, the set $B\setminus A$ contains the vertex set of both a horizontal and a vertical path of $M$, we call $B$ the \emph{$M$-majority side} of $(A,B)$.
Then $\mathcal{T}_M$ is the tangle \emph{induced} by $M$.
If $\mathcal{T}$ is a tangle in $G$, we say that $\mathcal{T}$ \emph{controls} the mesh $M$ if $\mathcal{T}_M$ is a truncation of $\mathcal{T}$.

The significance of meshes lies in the fact that every mesh has large treewidth and, reversely, every graph of large treewidth must contain a big mesh.
This theorem is originally due to Robertson and Seymour \cite{RobertsonS1986GraphMinorsV} but we state here the polynomial version of Chuzhoy and Tan \cite{ChuzhoyT2021TowardsTight}.

\begin{proposition}[Chuzhoy, Tan \cite{ChuzhoyT2021TowardsTight}]\label{prop_GridThm}
There exists a function $\mathsf{mesh}_{\ref{prop_GridThm}} \colon \mathbb{N} \to \mathbb{N}$ such that for every $k \in \mathbb{N}$ with $k \geq 2$ and every graph $G$, if $\mathsf{tw}(G) \geq \mathsf{mesh}_{\ref{prop_GridThm}}(k)$, then $G$ contains a $(k \times k)$-mesh as a subgraph.

Moreover, it holds that $\mathsf{mesh}_{\ref{prop_GridThm}}(k) \in \mathbf{O}(k^9 \mathsf{poly log} k)$.
\end{proposition}

\paragraph{Paintings in surfaces.}
A \emph{painting} in a surface $\Sigma$ is a pair $\Gamma = (U,N)$, where $N \subseteq U \subseteq \Sigma$, $N$ is finite, $U \setminus N$ has a finite number of arcwise-connected components, called \emph{cells} of $\Gamma$, and for every cell $c$, the closure of $c$, denoted by $\overline{c}$, is a closed disc where $N_\Gamma(c) \coloneqq \overline{c} \cap N \subseteq \mathsf{bd}(\overline{c})$.
If $|N_\Gamma(c)| \geq 4$, the cell $c$ is called a \emph{vortex}.
We further let $N(\Gamma) \coloneqq N$, let $U(\Gamma) \coloneqq U$, and let $C(\Gamma)$ be the set of all cells of $\Gamma$.
\medskip

Any given painting $\Gamma = (U,N)$ defines a hypergraph with $N$ as its vertices and the set of closures of the cells of $\Gamma$ as its edges.
Accordingly, we call $N$ the \emph{nodes} of $\Gamma$.

\paragraph{$\Sigma$-renditions.}
Let $G$ be a graph and $\Sigma$ be a surface.
A \emph{$\Sigma$-rendition} of $G$ is a triple $\rho = (\Gamma, \sigma, \pi)$, where
\begin{itemize}
    \item $\Gamma$ is a painting in $\Sigma$,
    \item for each cell $c \in C(\Gamma)$, $\sigma(c)$ is a subgraph of $G$, and
    \item $\pi \colon N(\Gamma) \to V(G)$ is an injection,
\end{itemize}
such that
\begin{description}
    \item[R1] $G = \bigcup_{c \in C(\Gamma)}\sigma(c)$,
    \item[R2] for all distinct $c,c' \in C(\Gamma)$, the graphs $\sigma(c)$ and $\sigma(c')$ are edge-disjoint,
    \item[R3] $\pi(N_\Gamma(c)) \subseteq V(\sigma(c))$ for every cell $c \in C(\Gamma)$, and
    \item[R4] for every cell $c \in C(\Gamma)$, we have $V(\sigma(c) \cap \bigcup_{c' \in C(\Gamma) \setminus \{ c \}} (\sigma(c'))) \subseteq \pi(N_\Gamma(c))$.
\end{description}
We write $N(\rho)$ for the set $N(\Gamma)$, let $N_\rho(c) = N_\Gamma(c)$ for all $c \in C(\Gamma)$, and similarly, we lift the set of cells from $C(\Gamma)$ to $C(\rho)$.
If it is clear from the context which $\rho$ is meant, we will sometimes simply write $N(c)$ instead of $N_\rho(c)$, and if the $\Sigma$-rendition $\rho$ for $G$ is understood from the context, we usually identify the sets $\pi(N(\rho))$ and $N(\rho)$ along $\pi$ for ease of notation.

\paragraph{Blank renditions.}
Let $\rho$ be a $\Sigma$-rendition of an annotated graph $(G,R)$.
If $$\pi(N(\rho)) \cup \bigcup \{ V(\sigma(c)) \colon c \in C(\rho) \text{ and } c \text{ is not a vortex}\}$$ is disjoint from $R$, we call $\rho$ a \emph{blank rendition (of $(G,R)$)}.

\paragraph{Societies.}
Let $\Omega$ be a cyclic ordering of the elements of some set which we denote by $V(\Omega)$.
A \emph{society} is a pair $(G,\Omega)$, where $G$ is a graph and $\Omega$ is a cyclic ordering with $V(\Omega)\subseteq V(G)$.
For a given set $S \subseteq V(\Omega)$ a vertex $s \in S$ is an \emph{endpoint} of $S$ if there exists a vertex $t \in V(\Omega) \setminus S$ that immediately precedes or succeeds $s$ in $\Omega$.
We call $S$ a \emph{segment} of $\Omega$ if $S$ has two or less endpoints.

Let $(G,\Omega)$ be a society and let $\Sigma$ be a surface with one boundary component $B$ homeomorphic to the unit circle.
A \emph{rendition} of $(G,\Omega)$ in $\Sigma$ is a $\Sigma$-rendition $\rho$ of $G$ such that the image under $\pi_{\rho}$ of $N(\rho) \cap B$ is $V(\Omega)$ and $\Omega$ is one of the two cyclic orderings of $V(\Omega)$ defined by the way the points of $\pi_{\rho}(V(\Omega))$ are arranged in the boundary $B$.

\paragraph{Traces of paths and cycles.}
Let $\rho$ be a $\Sigma$-rendition of a graph $G$.
For every cell $c \in C(\rho)$ with $|N_\rho(c)| = 2$, we select one of the components of $\mathsf{bd}(c) - N_\rho(c)$.
This selection will be called a \emph{tie-breaker in $\rho$}, and we assume that every rendition comes equipped with a tie-breaker.

Let $G$ be a graph and $\rho$ be a $\Sigma$-rendition of $G$.
Let $Q$ be a cycle or path in $G$ that uses no edge of $\sigma(c)$ for every vortex $c \in C(\rho)$.
We say that $Q$ is \emph{grounded} if either $P$ is a path with both endpoints in $N(\rho)$, or $Q$ is a cycle that uses edges of $\sigma(c_{1})$ and $\sigma(c_{2})$ for two distinct cells $c_{1}, c_{2} \in C(\rho).$
If $Q$ is grounded we define the \emph{trace} of $Q$ as follows.
Let $P_1,\dots,P_k$ be distinct maximal subpaths of $Q$ such that $P_i$ is a subgraph of $\sigma(c)$ for some cell $c$.
Fix $i \in [k]$.
The maximality of $P_i$ implies that its endpoints are $\pi(n_1)$ and $\pi(n_2)$ for distinct nodes $n_1,n_2 \in N(\rho)$.
If $|N_\rho(c)| = 2$, let $L_i$ be the component of $\mathsf{bd}(c) - \{ n_1,n_2 \}$ selected by the tie-breaker, and if $|N_\rho(c)| = 3$, let $L_i$ be the component of $\mathsf{bd}(c) - \{ n_1,n_2 \}$ that is disjoint from $N_\rho(c)$.
We define $L_i'$ by pushing $L_i$ slightly so that it is disjoint from all cells in $C(\rho)$, while maintaining that the resulting curves intersect only at a common endpoint.
The \emph{trace} of $Q$ is defined to be the curve $\gamma$ obtained from $\bigcup_{i\in[k]} L_i'$ by slightly pushing it in order to ensure that $\gamma$ intersects the painting of $\rho$ in nodes only.
If $Q$ is a cycle, its trace thus the homeomorphic image of the unit circle, and otherwise, it is an arc in $\Sigma$ with both endpoints in $N(\rho)$.

\paragraph{Aligned discs and grounded subgraphs.}
Let $G$ be a graph and let $\rho = (\Gamma, \sigma, \pi)$ be a $\Sigma$-rendition of $G$. 
We say that a 2-connected subgraph $H$ of $G$ is \emph{grounded (in $\rho$)} if every cycle in $H$ is grounded and no vertex of $H$ is drawn by $\Gamma$ in a vortex of $\rho$.
A disc in $\Sigma$ is called \emph{$\rho$-aligned} if its boundary only intersects $\Gamma$ in nodes.
If $H$ is planar, we say that it is \emph{flat in $\rho$} if there exists a $\rho$-aligned disc $\Delta \subseteq \Sigma$ which contains all cells $c \in C(\rho)$ with $E(\sigma(c)) \cap E(H) \neq \emptyset$ and $\Delta$ does not contain any vortices of $\Gamma$.

For any $\rho$-aligned disc $\Delta$, we call the subgraph of $G$ that is drawn by $\Gamma$ onto $\Delta$ the \emph{crop of $G$ by $\Delta$ (in $\rho$)}.
Furthermore, the \emph{restriction $\delta'$ of $\rho$ by $\Delta$} is defined as the $\Delta$-rendition that consists of the restriction of both $\Gamma$, $\sigma$, and $\pi$ to $\Delta$.

This allows to define a society associated to $\Delta$ as follows.
Let $V(\Omega_{\Delta})$ be the set of all vertices whose corresponding nodes are drawn in the boundary of $\Delta$ and let $\Omega_{\Delta}$ be the cyclic ordering of $V(\Omega_{\Delta})$ obtained by traversing along the boundary of $\Delta$ in the anticlockwise direction.
Now, let $G_{\Delta}$ be the crop of $G$ by $\Delta$.
We call the society $(G_{\Delta}, \Omega_{\Delta})$ the \emph{$\Delta$-society (in $\rho$)}.
If $\rho$ is clear from the context, we do not mention it.
We also call the restriction of $\rho$ by $\Delta$, the \emph{restriction of $\rho$ to $(G_{\Delta}, \Omega_{\Delta})$.}

Let $\mathcal{P}$ be an $X$-$Y$-linkage in $G$ such that $X \cap V(G_{\Delta}) \subseteq V(\Omega_{\Delta})$ and $Y \subseteq V(G_{\Delta})$ and assume that each path in $\mathcal{P}$ is grounded in $\rho$.
Then we define the \emph{$\Delta$-truncation (in  $\rho$)} of $\mathcal{P}$ to be the $V(\Omega_{\Delta})$-$Y$-linkage in $G_{\Delta}$ which consists of the minimal $V(\Omega_{\Delta})$-$Y$-subpaths of the paths in $\mathcal{P}.$

Let $\rho$ be a rendition of a society $(G, \Omega)$ in the disc $\Delta.$
Given a cycle $C \subseteq G$ that is grounded in $\rho$ we define the \emph{$C$-disc (in $\rho$)} as the unique $\rho$-aligned disc $\Delta' \subseteq \Delta$ bounded by the trace of $C$ in $\rho.$
We also use the terms \emph{$C$-society (in $\rho$)} to denote the $\Delta'$-society in $\rho$ and \emph{$C$-truncation (in $\rho$)} to denote the $\Delta'$-truncation in $\rho$ of an appropriately defined linkage in $G.$

\paragraph{Inner and outer graphs of cycles and aligned curves.}
Let $G$ be a graph with a $\Sigma$
-rendition $\rho$ and let $\gamma$ be a (closed) curve in $\Sigma$,
We say that $\gamma$ is \emph{$\rho$-aligned} if it intersects $\rho$ in nodes only.
That is, $\gamma$ must be disjoint from all cells of $\rho$.

Let $(G, \Omega)$ be a society with a $\Sigma$-rendition $\rho.$
Further, let $C$ be a grounded cycle or a $\rho$-aligned closed curve whose trace bounds a disc $\Delta_C$ and the $\Delta_C$-society $(G', \Omega').$
We call $G' \cup C$ the \emph{inner graph of $C$ (in $\rho$)} and call $G'$ itself the \emph{proper inner graph of $C$ (in $\rho$).}
Let $B = \pi(N(\rho) \cap \mathsf{bd}(\Delta_C)).$
We define the \emph{proper outer graph of $C$ (in $\rho$)} as $G'' \coloneqq G[B \cup (V(G) \setminus V(G'))]$ and call $G'' \cup C$ the \emph{outer graph of $C$ (in $\rho$).} 

\paragraph{Transactions in societies.}
Let $(G, \Omega)$ be a society. 
A \emph{transaction} in $(G, \Omega)$ is an $A$-$B$-linkage for disjoint segments $A, B$ of $\Omega$ consisting of $V(\Omega)$-paths. 
The inclusion-wise minimal segments $X$ and $Y$ of $\Omega$ for which $\mathcal{P}$ is an $X$-$Y$-linkage are called the \emph{end segments} of $\mathcal{P}$ in $(G, 
\Omega)$.

\paragraph{Cylindrical renditions.}
A rendition of a society $(G, \Omega)$ in the disc with a unique vortex $c_{0}$ is called a \emph{cylindrical rendition} of $(G, \Omega)$ \emph{around} $c_{0}$.

If $(G, \Omega)$ is a society with a cylindrical rendition $\rho$ around a vortex $c_{0}$ and $\mathcal{P}$ is a transaction in $(G, \Omega)$, we call $\mathcal{P}$ \emph{exposed} if for every path $p \in \mathcal{P}$ there exists an edge $e \in E(P) \cap \sigma(c_{0})$.

\paragraph{Nests and radial linkages.}
Let $\rho$ be a rendition of a society $(G, \Omega)$ in a disc $\Delta$.
A \emph{nest (in $\rho$)} is a set of disjoint cycles $\mathcal{C} = \{ C_{1}, \ldots, C_{s} \}$ in $G$ such that each of them is grounded in $\rho$, and if $\Delta_{i}$ is the $C_{i}$-disc for $i \in [s]$, then every vortex of $\rho$ is contained in $\Delta_{1}$ and $\Delta_{1} \subseteq \ldots \subseteq \Delta_{s} \subseteq \Delta$.
We call $C_{1}$ the \emph{inner cycle} of $\mathcal{C}$ and $C_{s}$ the \emph{outer cycle} of $\mathcal{C}$ respectively.
Moreover, we call a $V(\Omega)$-$V(C_{1})$-linkage $\mathcal{R}$ a \emph{radial linkage (in $\rho$) for $\mathcal{C}$} if all paths in $\mathcal{R}$ are grounded in $\rho$ and internally disjoint from $V(\Omega).$

If $(G, \Omega)$ is a society with a nest $\mathcal{C}$ in a rendition $\rho$ of $(G, \Omega)$ in a disc, we say that a radial linkage $\mathcal{R}$ for $\mathcal{C}$ is \emph{orthogonal to $\mathcal{C}$} if for all $C \in \mathcal{C}$ and all $R \in \mathcal{R}$ the graph $C \cap R$ is a path.
Similarly, we say that a transaction $\mathcal{P}$ in $(G, \Omega)$ is \emph{orthogonal to $\mathcal{C}$} if for all $C \in \mathcal{C}$ and all $P \in \mathcal{P}$ the graph $C \cap P$ consists of exactly two paths.

\paragraph{Depth of vortices.}
Let $G$ be a graph and $\rho$ be a $\Sigma$-rendition of $G$ with a vortex cell $c_0.$
Notice that $c_0$ defines a society $(\sigma(c_0), \Omega_{c_0})$, where $V(\Omega_{c_0})$ is the set of vertices of $G$ corresponding $N_\rho(c_0).$
The ordering $\Omega_{c_{0}}$ is obtained by traversing along the boundary of the closure of $c_0$ in anti-clockwise direction.
We call $(\sigma(c_0),\Omega_{c_0})$ as obtained above the \emph{vortex society} of $c_0.$

We define the \emph{depth} of a society $(G, \Omega)$ as the maximum cardinality of a transaction in $(G, \Omega).$
The \emph{depth} of the vortex $c_{0}$ is thereby defined as the depth of its vortex society.

Given a $\Sigma$-rendition $\rho$ with vortices, we define the \emph{breadth of $\rho$} as the number of vortex cells of $\rho$ and the \emph{depth of $\rho$} as the maximum depth of its vortex societies.

\subsection{The structure of annotated graphs of small bidimensionality}\label{subsec_redMinors}

Next, we introduce the structure theorem of Gorsky, Protopapas, and Wiederrecht \cite{GorskyPW2026Quickly} establishing polynomial bounds for the structure of annotated graphs of small bidimensionality.
We slightly modify the original statement by removing some amount of information that will not be necessary for the purpose of this paper.

\begin{proposition}[Gorsky, Protopapas, and Wiederrecht \cite{GorskyPW2026Quickly}]\label{prop:localstructure}
There exist functions $\mathsf{apex}_{\ref{prop:localstructure}},\mathsf{depth}_{\ref{prop:localstructure}} \colon \mathbb{N}^3 \to \mathbb{N}$ and $\mathsf{mesh}_{\ref{prop:localstructure}} \colon \mathbb{N}^4 \to \mathbb{N}$ such that for all integers $k,t \geq 1$, $r \geq 4$, and $w \geq 3$, every annotated graph $(G,R)$ with $\mathsf{bidim}(G,R) \leq r$, and every $\mathsf{mesh}_{\ref{prop:localstructure}}(r, k, t, w)$-mesh $M \subseteq G$ one of the following holds.
\begin{enumerate}

 \item $(G,R)$ has a blank $K_{\lfloor \nicefrac{5t}{2} \rfloor}$-minor model $\psi$ and a separation $(X,Y)$ of order at most $t-1$ such that the $\mathcal{T}_\psi$-big component of $G - (X \cap Y)$ is blank,
 
 \item $(G, R)$ has a red $K_t$-minor model controlled by $M$, or
 \item there exists a set $A \subseteq V(G)$ with $|A| \leq \mathsf{apex}_{\ref{prop:localstructure}}(k, t, r)$ and a surface $\Sigma$ of genus less than $9t^2$ such that $(G - A, R \setminus A)$ has a blank $\Sigma$-rendition $\rho$ with breadth at most $\nicefrac{3}{2}(t-1)(3t-4) + r(r-1)-3$, depth at most $\mathsf{depth}_{\ref{prop:localstructure}}(k, t, r)$, and there exists a vortex-free, $\rho$-aligned disc $\Delta \subseteq \Sigma$ such that the restriction of $\rho$ to $\Delta$ contains a flat $w$-wall.
 Moreover, if we denote by $\mathcal{V}$ the vortices of $\rho$, then for every $v \in \mathcal{V}$ there exists a $\rho$-aligned disc $\Delta_{v}$ such that
 \begin{itemize}
    \item $\Delta_v$ is disjoint from $\Delta$ and from $\Delta_{v'}$ for all $v' \in \mathcal{V}\setminus\{ v\}$, and
    \item there exists a nest of order $k$ in the restriction of $\rho$ to $\Delta_v$.
 \end{itemize}
\end{enumerate}
Moreover, it holds that $\mathsf{apex}_{\ref{prop:localstructure}}(k,t,r), \mathsf{depth}_{\ref{prop:localstructure}}(t,r) \in \poly(k+t+r)$ and $\mathsf{mesh}_{\ref{prop:localstructure}}(r,k,t,w) \in \poly(k+t+r) + \mathbf{O}(t^2w)$.

There also exists an algorithm that, given $t$, $k$, $r$, $w$, a graph $G$, and a mesh $M$ as above as input finds one of these outcomes in time $\poly(t+k+r+w)|E(G)|^3$.
\end{proposition}

\section{The local structure of annotated graphs with small depth$_2$}\label{sec_structureTheorem}

The goal of this section is a further refinement of \zcref{thm_localstructure} towards a local structure theorem for annotated graphs of small $\mathsf{depth}_2$.
To fully express our structural main theorem, we introduce some additional concepts as it is now necessary to further analyse the structure of vortices.

\paragraph{Linear societies.}
In the following we consider a slightly different concept of societies.
Let $(G,\Omega)$ be a society and $x \in V(\Omega)$.
The \emph{$x$-linearisation} of $\Omega$ is the linear order $\Lambda_{\Omega,x}$ of $V(\Omega)$ where $x$ is the minimum, the maximum is $z \in V(\Omega)$ which is the predecessor of $x$ in $\Omega$ and $a,b \in V(\Omega)$ satisfy $a \leq_{\Lambda_{\Omega,x}} b$ if and only if there exists a segment $S$ of $\Omega$ with ends $a$ and $b$ that does not contain $x$ in its interior.
A \emph{linearisation} of $\Omega$ is a linear order $\Lambda$ of $V(\Omega)$ such that there is $x \in \Omega$ for which $\Lambda = \Lambda_{\Omega,x}$.

It turns out that in most situations, it does not make a huge difference when considering a cyclic order or a linearisation of it for the sake of structure theory.
See \cite{PaulPTW2024Obstructionsa} for some discussion on the topic.
Thus, in the following a \emph{linear society} is a pair $(G,\Lambda)$ where $(G,\Omega)$ is a society and $\Lambda$ is a linearisation of $\Omega$.
We will drop the suffix ``linear'' in most cases and only stress the linearity of our societies where necessary.
However, we strive to use $\Lambda$ for linear societies and $\Omega$ for regular societies.
We also lift the notion of segments from cyclic orders to linear orders in the natural way.

A linear society $(G,\Lambda)$ has a \emph{rendition in a disc} if the society $(G,\Omega)$ such that $(G,\Lambda)$ is a linearisation of $(G,\Omega)$ has a rendition in a disc.
Similarly, we lift all other notions relating to societies and renditions to linear societies by simply requiring the property to hold for $(G,\Omega)$ where $(G,\Lambda)$ is a linearisation.

\paragraph{Candles and cropped societies.}
Let $(G,\Lambda)$ be a linear society.

Let $R \subseteq V(G)$.
An \emph{$R$-candle} is a pair $P,Q$ of vertex-disjoint paths where the endpoints of $P$ are $x_1, x_2 \in V(\Lambda)$, $Q$ has an endpoint $y_1 \in V(\Lambda)$ and one endpoint in $R$ -- we allow those endpoints to coincide -- and where $x_1 , y_1$, and $x_2$ occur in $\Lambda$ in the order listed.
See \zcref{fig_CandleExample} for an example.
We call the vertices $x_1,y_1$, and $x_2$ the \emph{roots} of the candle and the endpoint of $Q$ in $R$ its \emph{flame}.
Moreover, $Q$ is referred to as the \emph{candle stick} and we usually write $\mathcal{X} = \{ P,Q\}$ for an $R$-candle.
In slight abuse of notation, we sometimes identify $\mathcal{X}$ and the graph $P \cup Q$.

\begin{figure}[ht]
 \centering
 \begin{tikzpicture}

 \pgfdeclarelayer{background}
		\pgfdeclarelayer{foreground}
			
		\pgfsetlayers{background,main,foreground}

 \begin{pgfonlayer}{background}
 \pgftext{\includegraphics[width=8cm]{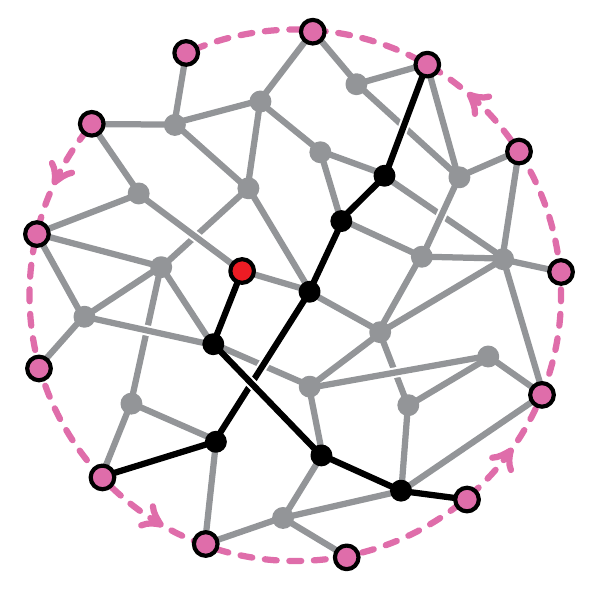}} at (C.center);
 \end{pgfonlayer}{background}
			
 \begin{pgfonlayer}{main}
 \node (C) [v:ghost] {};
 
 \end{pgfonlayer}{main}
 
 \begin{pgfonlayer}{foreground}
 \end{pgfonlayer}{foreground}

 \end{tikzpicture}
 \caption{A candle in a linear society.}
 \label{fig_CandleExample}
\end{figure}

Let $S$ be a segment of $\Lambda$.
We denote by $\Lambda(S)$ the restriction of $\Lambda$ to $S$ where all vertices of $S$ have the same order in $\Lambda(S)$ as in $\Lambda$.
The linear society $(G_S,\Lambda(S))$ is the \emph{$S$-crop} of $(G,\Lambda)$ where $G_S$ is the graph obtained from $G-(V(\Lambda)\setminus S)$ by deleting all components that do not contain at least one vertex of $S$.

\begin{figure}[ht]
 \centering
 \begin{tikzpicture}

 \pgfdeclarelayer{background}
		\pgfdeclarelayer{foreground}
			
		\pgfsetlayers{background,main,foreground}

 \begin{pgfonlayer}{background}
 \pgftext{\includegraphics[width=7cm]{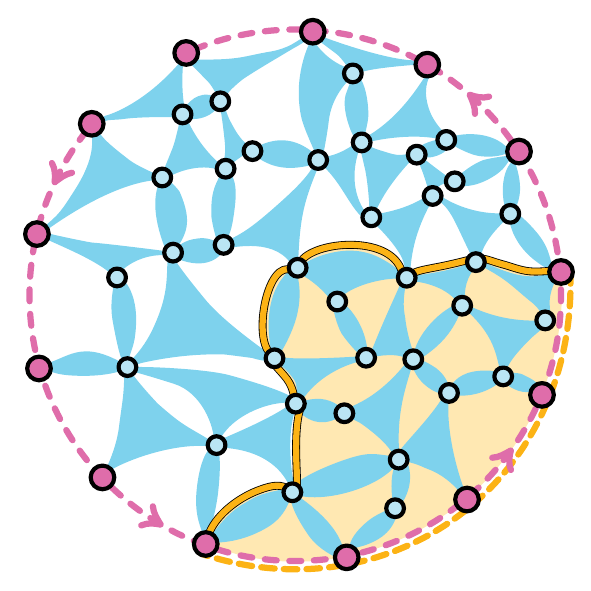}} at (C.center);
 \end{pgfonlayer}{background}
			
 \begin{pgfonlayer}{main}
 \node (C) [v:ghost] {};
 
 \end{pgfonlayer}{main}
 
 \begin{pgfonlayer}{foreground}
 \end{pgfonlayer}{foreground}

 \end{tikzpicture}
 \caption{A society with a vortex-free rendition in a disc together with a rooted curve $\gamma$ depicted in \textcolor{DarkBananaYellow}{yellow}. The base of $\gamma$ is marked as the \textcolor{DarkBananaYellow}{yellow} segment and the shaded area marks the inner disc of $\gamma$.}
 \label{fig_DiscsOfGamma}
\end{figure}

\paragraph{Rooted $\rho$-aligned curves.}
Let $(G,\Lambda)$ be a society and $\rho$ be a cylindrical rendition of $(G,\Lambda)$ in a disc $\Delta$ with vortex $c_0$.
We say that a $\rho$-aligned curve $\gamma$ is \emph{rooted} if both ends of $\gamma$ belong to $V(\Lambda)$ and $\gamma$ is otherwise disjoint from $V(\Lambda)$.

If $\gamma$ is a rooted $\rho$-aligned curve, then $\gamma$ partitions $\Lambda$ into three possibly empty segments where precisely one of those three segments contains both ends of $\gamma$.
We call this segment the \emph{base} of $\gamma$.

Notice that a rooted $\rho$-aligned curve together with the boundary of $\Delta$ decomposes $\Delta$ into two discs which intersect precisely in $\gamma$.
The base of $\gamma$ is contained fully in precisely one of those two discs and we refer to this disc as the \emph{inner disc} of $\gamma$ while the other disc is the \emph{outer disc} of $\gamma$.
See \zcref{fig_DiscsOfGamma} for an illustration.
The \emph{inner graph} of a rooted $\rho$-aligned curve is the inner graph of the boundary of its inner disc and its \emph{outer graph} is the inner graph of the boundary of its outer disc.

Let $\gamma$ be a $\rho$-aligned curve.
We say that $\gamma$ is a \emph{frontier} if: 
\begin{enumerate}
    \item The curve $\gamma$ can be decomposed into rooted $\rho$-aligned curves $\gamma_1,\dots,\gamma_k$ such that $\bigcup_{i \in [k]}\gamma_i = \gamma$  and for any two $i \neq j$, $\gamma_i \cap \gamma_j$ is either empty or $|i - j|=1$ and the intersection is a common endpoint.
    \item For each $i\in[k]$, the inner disc of $\gamma_i$ does not contain any of the $\gamma_j$, $j \in [k] \setminus \{ i\}$.
\end{enumerate}
The inner graph of a frontier $\gamma$ is the union of the inner graphs of the $\gamma_i$ and the outer graph of $\gamma$ is the common intersection of all outer graphs of all the $\gamma_i$.
See \zcref{fig_Frontier} for an illustration.

\begin{figure}[ht]
 \centering
 \begin{tikzpicture}

 \pgfdeclarelayer{background}
		\pgfdeclarelayer{foreground}
			
		\pgfsetlayers{background,main,foreground}

 \begin{pgfonlayer}{background}
 \pgftext{\includegraphics[width=7cm]{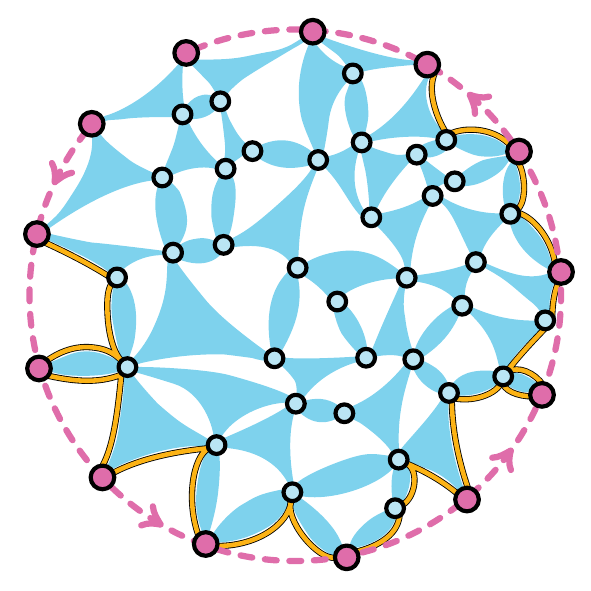}} at (C.center);
 \end{pgfonlayer}{background}
			
 \begin{pgfonlayer}{main}
 \node (C) [v:ghost] {};
 
 \end{pgfonlayer}{main}
 
 \begin{pgfonlayer}{foreground}
 \end{pgfonlayer}{foreground}

 \end{tikzpicture}
 \caption{A frontier -- depicted in \textcolor{DarkBananaYellow}{yellow} -- in a vortex-free rendition in a disc of a linear society.}
 \label{fig_Frontier}
\end{figure}

\paragraph{Candle-free societies.}
A set $R \subseteq V(G)$ is \emph{candle-free} in $(G,\Lambda)$ if there exists a vortex-free rendition of $(G,\Lambda)$ in a disc and there is no $R$-candle in $(G,\Lambda)$.

\paragraph{Facial sets in societies.}
Let $(G,\Lambda)$ be a linear society where $(G,R)$ is an annotated graph and $\rho$ be a vortex-free rendition of $(G,\Lambda)$.
We say that $\rho$ is \emph{facial} if there exists a frontier $\gamma$ such that
\begin{itemize}
    \item $R$ is contained entirely in the outer graph of $\gamma$, and
    \item there does not exist a $V(\gamma)$-candle in $(G,\Lambda)$.
\end{itemize}
We say that $\gamma$ is an \emph{ultimate frontier} of $(G,\Lambda)$.

\paragraph{Shallow renditions.}
Let $k$ be a non-negative integer, $(G,R)$ be an annotated graph and $\Sigma$ be a surface.
We say that a $\Sigma$-rendition $\rho$ of $G$ is \emph{$k$-shallow} if
\begin{enumerate}
    \item $\rho$ is blank, and
    \item for each vortex $v$ of $\rho$ there exists a linearisation $(\sigma(v),\Lambda_v)$ of its  vortex society such that there exists a partition of $\Lambda_v$ into at most $k$ segments $S_1,\dots,S_k$ such that, if we denote by $(G_{S_i},\Lambda_{S_i})$ is the $S_i$-crop of $(\sigma(v),\Lambda_v)$ for each $i \in [k]$,
    \begin{itemize}
        \item for all $i \neq j \in[k]$, $G_{S_i}$ and $G_{S_j}$, and
        \item for all $i \in [k]$ either $V(G_{S_i}) \cap R = \emptyset$ -- and we call $S_i$ a \emph{deep segment}, or $R \cap V(G_{S_i})$ is facial in $(G_{S_i},\Lambda_{S_i})$ -- and we call $S_i$ a \emph{facial segment}.
    \end{itemize}
\end{enumerate}

With these definitions, we have everything in place to state our local structure theorem for annotated graphs of small $\mathsf{depth}_2$.

\begin{theorem}\label{thm_localstructure}
There exist functions $\mathsf{apex}_{\ref{thm_localstructure}},\mathsf{depth}_{\ref{thm_localstructure}} \colon \mathbb{N}^2 \to \mathbb{N}$, and $\mathsf{mesh}_{\ref{thm_localstructure}} \colon \mathbb{N}^3 \to \mathbb{N}$ such that for all integers $t \geq 1$, $r \geq 4$, and $w \geq 3$, every annotated graph $(G,R)$ with $\mathsf{depth}_2(G,R) \leq r$, and every $\mathsf{mesh}_{\ref{thm_localstructure}}(r, t, w)$-mesh $M \subseteq G$ one of the following holds.
\begin{enumerate}

 \item $(G,R)$ has a blank $K_{\lfloor \nicefrac{5t}{2} \rfloor}$-minor model $\psi$ and a separation $(X,Y)$ of order at most $t-1$ such that the $\mathcal{T}_\psi$-big component of $G - (X \cap Y)$ is blank,
 
 \item $(G, R)$ has a red $K_t$-minor model controlled by $M$, or
 \item there exists a set $A \subseteq V(G)$ with $|A| \leq \mathsf{apex}_{\ref{thm_localstructure}}(t, r)$ and a surface $\Sigma$ of genus less than $9t^2$ such that $(G - A, R \setminus A)$ has a blank and $r$-shallow $\Sigma$-rendition $\rho$ with breadth at most $\nicefrac{3}{2}(t-1)(3t-4) + r(r-1)-3$, depth at most $\mathsf{depth}_{\ref{thm_localstructure}}(t, r)$, and there exists a vortex-free, $\rho$-aligned disc $\Delta \subseteq \Sigma$ such that the restriction of $\rho$ to $\Delta$ contains a flat $w$-wall.
\end{enumerate}
Moreover, it holds that $\mathsf{apex}_{\ref{thm_localstructure}}(t,r), \mathsf{depth}_{\ref{thm_localstructure}}(t,r) \in \poly(t+r)$ and $\mathsf{mesh}_{\ref{thm_localstructure}}(r,t,w) \in \poly(t+r) + \mathbf{O}(t^2w)$.

There also exists an algorithm that, given $t$, $r$, $w$, a graph $G$, and a mesh $M$ as above as input finds one of these outcomes in time $\poly(t+r+w)|G|^6$.
\end{theorem}

Notice that the main difference between \zcref{prop:localstructure} and \zcref{thm_localstructure} is the addition of the condition on $\rho$ in the third outcome to be shallow.
This shallowness is precisely the distinguishing factor between annotated graphs of small $\mathsf{depth}_2$ and those of bounded bidimensionality:
It provides further control about the interaction of the annotation with the vortices.

Our proof for \zcref{thm_localstructure} is divided into two major parts.

In \textbf{Part 1} we investigate the behaviour of candles in societies and establish general structure theorems for situations where candles occur versus situations where candles are absent.
We also show that, in the case where we have a society containing a large nest around a single vortex and we sequentially find many disjoint candles, then these candles together with the nest give rise to a large candle mesh which implies that the annotated graph underlying the society has large $\mathsf{depth}_2$.
This is a core observation that will later allow us to bound the number of disjoint candles we can find.

In \textbf{Part 2}, we adapt ideas from the work of Thilikos and Wiederrecht \cite{ThilikosW2024Killing} on handling structures rooted at the boundaries of bounded depth vortices to show that each vortex in the outcome of \zcref{prop:localstructure} either contains many disjoint candles -- implying that the annotated graph has large $\mathsf{depth}_2$ by using the tools developed in \textbf{Part 1} -- or there exists a small set of vertices whose deletion removes all candles from all vortices.
The later outcome will then imply the final structure required to prove \zcref{thm_localstructure}.

\subsection{Part 1: Placing candles}\label{subsec_Step1}
We begin with \textbf{Part 1}.
Our goal is to establish a sequence of general structural observations on candles in linear societies with underlying annotated graphs.
There are several core lemmas we need to extract in this first part, those include:
\begin{itemize}
    \item establishing a link between the absence of candles and the property of being facial in societies with vortex-free renditions in a disc,
    \item showing that many candles rooted at the top of a large mesh give rise to a big $2$-outer-annotated grid,
    \item a toolset for routing and redrawing candles under certain conditions, and finally
    \item proving that any linear society is either facial, has a candle, or is blank, where the last outcome is distinct from the first in the case where the society does not have a vortex-free rendition in a disc.
\end{itemize}

\paragraph{Facial societies.}
We start with a somewhat independent result describing the structure of linear societies without $R$-candles.

\begin{theorem}\label{thm_FacialSocieties}
Let $(G,\Lambda)$ be a society where $(G,R)$ is an annotated graph and $(G,\Lambda)$ has a vortex-free rendition in a disc.
Then $(G,\Lambda)$ is facial if and only if $(G,\Lambda)$ does not contain an $R$-candle.
\end{theorem}

Our main concern towards proving \zcref{thm_FacialSocieties} is to understand what the absence of an $R$-candle in a society implies if we already know that the society has a vortex-free rendition in a disc. 
Here, a fundamental theorem from Graph Minors becomes relevant: the \emph{Two Paths Theorem}.

We will later require the following subroutine for the \textsc{$2$-Disjoint Paths} problem.

\begin{proposition}[Kawarabayashi, Li, and Reed \cite{KawarabayashiLR2015Connectivity}]\label{prop_TwoPaths}
There exists a linear-time algorithm for the \textsc{$2$-Disjoint Paths} problem.
\end{proposition}

Let $(G,\Omega)$ be a society. 
We say that $(G,\Omega)$ \emph{has a cross} if there exist four distinct vertices $s_1,s_2,t_1,t_2$ appearing in $\Omega$ in the order listed together with two vertex-disjoint paths $P_1$ and $P_2$ in $G$ where $P_i$ joins $s_i$ and $t_i$ for each $i\in[2]$.
We also refer to the two paths $P_1$ and $P_2$ as a \emph{cross} of $(G,\Omega)$.

The Two Paths Theorem can be seen as a structural variant of \zcref{prop_TwoPaths} and it is the engine underlying the algorithm.

\begin{proposition}[Two Paths Theorem \cite{Jung1970Verallgemeinerung,Seymour1980Disjoint,Shiloach1980Polynomial,Thomassen19802Linked,RobertsonS1990Graph}]\label{prop_TwoPathsProper}
A society $(G,\Omega)$ has no cross if and only if it has a vortex-free rendition in a disc.
\end{proposition}

Indeed, the algorithm from \zcref{prop_TwoPaths} actually finds either a vortex-free rendition in a disc for $(G,\Omega)$ or a cross in $(G,\Omega)$ in linear time.

\begin{proof}[Proof of \zcref{thm_FacialSocieties}]
Let us first assume that $(G,\Lambda)$ has an $R$-candle $\mathcal{X}=\{ P,Q\}$ where $Q$ is the candle stick and show that it cannot be facial.

For this, assume that $(G,\Lambda)$ is indeed facial and let $\rho$ be a vortex-free rendition in the disc $\Delta$ together with a frontier $\gamma$ that meets the requirements of the definition of facial.
Let $r \in V(Q)\cap R$ be the flame of $\mathcal{X}$.
Now consider the trace $\pi$ of $P$ and notice that $\pi$ is a rooted $\rho$-aligned curve.
Let $\Delta'$ denote the inner disc of $\pi$ and $G'$ be its inner graph.
Then $Q \subseteq G'$.
With $\gamma$ being a frontier, we know that there is a decomposition of $\gamma$ into rooted $\rho$-aligned curves $\gamma_1,\dots,\gamma_k$ as in the definition of frontier.
For each $i\in[k]$ let $\Delta_i$ denote the outer disc of $\gamma_i$ and let $\fullmoon$ be the closure of $\Delta \setminus \big(\bigcap_{i \in [k]}\Delta_i\big)$.
It follows that there exists a subpath $Q'$ of $Q$ starting on $V(\Lambda)$ and ending on $V(\gamma)$ which is entirely contained in $\fullmoon$.
Then $\{ P,Q'\}$ is a $V(\gamma)$-candle in $(G,\Lambda)$ which is impossible.
\medskip

So now we may assume that $(G,\Lambda)$ has no $R$-candle and we need to prove that $(G,\Lambda)$ is facial.

Towards this goal let $G'$ be the graph obtained from $G$ by introducing a new vertex $r$ such that $N_{G'}(r) = R$.
Then let $\Omega'$ be the cyclic order of $V(\Lambda) \cup \{r\}$ obtained from $\Lambda$ by declaring $r$ the successor of the maximum of $\Lambda$ and the predecessor of the minimum of $\Lambda$.

\begin{claim}\label{claim_flatCandle}
There exists an $R$-candle in $(G,\Lambda)$ if and only if there exists a cross in $(G',\Omega')$.
\end{claim}

\begin{claimproof}
Suppose there is a cross $P_1,P_2$ in $(G',\Omega')$.
First of all notice that $r$ must be an endpoint of one of the two paths $P_i$ as otherwise $P_1$ and $P_2$ form a cross in $(G,\Lambda)$ which contradicts our assumption that $(G,\Lambda)$ has a vortex-free rendition in a disc because of \zcref{prop_TwoPathsProper}.
Hence, without loss of generality we may assume that $r$ is an endpoint of $P_1$.
But if we let $r'$ be the neighbour of $r$ on $P_1$, then we know that $r' \in R$ and thus $P_1-r$ is the candle stick of the $R$-candle $\{ P_2,P_1-r\}$ in $(G,\Lambda)$.

For the reverse direction let $\{ P,Q\}$ be an $R$-candle in $(G,\Lambda)$ where $Q$ is the candle stick and $r'$ is the flame.
Then $r'$ is adjacent to $r$ in $G'$ and thus, $P$ and $Qr$ form a cross in $(G',\Omega')$.
\end{claimproof}

By our assumption, we know that $(G,\Lambda)$ has no $R$-candle.
Hence, \zcref{claim_flatCandle} tells us that $(G',\Omega')$ has no cross and therefore, due to \zcref{prop_TwoPathsProper}, $(G',\Omega')$ has a vortex-free rendition $\rho'$ in a disc $\Delta$.

Notice that for every vertex $r' \in R$ there exists a cell $c$ of $\rho'$ such that $r \in N(c)$ and $r' \in V(\sigma(c))$.
Let $\rho$ be the vortex-free rendition of $(G,\Lambda)$ in $\Delta$ obtained from $\rho'$ by deleting $r$ from $G'$, removing $r$ from the set of nodes of $\rho'$ and removing, from each cell $c$ of $\rho'$ with $r \in N_{\rho'}(c)$ an open disc with radius $\varepsilon$ for some $\varepsilon > 0$ small enough such that the disc is disjoint from all other nodes of $\rho'$, all cells apart from $c$, and its removal leaves what remains of $c$ to be simply connected.
Let us call these newly created cells the \emph{residual cells} of $\rho$.
That is, a cell $c$ of $\rho$ is residual if and only if it was created by the process above from a cell $c'$ of $\rho'$ and thus $r \in N_{\rho}(c')$.
Notice that, as a result, every cell that had $r$ on its boundary in $\rho'$ now corresponds to a cell $c'$ in $\rho$ with at most two nodes in its boundary.
It follows that for every cell $c'$ of $\rho$ corresponding to a cell $c$ in $\rho'$ which had $r$ on its boundary, and every node $v \in N_{\rho}(c')$, there is now a curve $\gamma_{v}$ with one end in $r$ -- which belongs to the boundary of $\Delta$ -- and the other being $v$ such that the only intersection of $\gamma_{v}$ and $\rho$ is $v$.
Indeed, the curves $\gamma_{v}$ may be chosen in such a way that they are all pairwise disjoint except for their shared point $r$.
This implies that there is a linear ordering $\lambda$ of the nodes $v$ on the boundaries of the residual cells of $\rho$ obtained by tracing an arc in $\Delta$ from a point on the boundary of $\Delta$ right before $r$ to a point on the boundary of $\Delta$ right after $r$ -- where ``before'' and ``after'' are chosen by following along the boundary of $\Delta$ in clockwise direction.
This gives us the right to define the following frontier for $\rho$:
Let $\psi$ be a curve obtained by starting at the minimum of $\Lambda$, then tracing through the nodes of the residual cells of $\rho$ in the order induced by $\lambda$, and finally ending in the maximum of $\Lambda$ such that
\begin{itemize}
    \item $\psi$ intersects $\rho$ only in the extrema of $\Lambda$ and the nodes of the residual cells of $\rho$, and
    \item all residual cells of $\rho$ are contained in the outer graph of $\psi$.
\end{itemize}
This is possible because $|N_{\rho}(c)| \leq 2$ for all residual cells of $\rho$.
Moreover, now $\sigma_{\rho}(c)$ is contained in the outer graph of $\psi$ if and only if $c$ is a residual cell of $\rho$.
Hence, by definition of $\rho$ it follows that all of $R$ is contained in the outer graph of $\psi$.
Moreover, since all residual cells of $\rho$ have at most two nodes on their boundary each, and all of those nodes are on $\psi$, it follows that there does not exist a $V(\psi)$-candle in $(G,\Lambda)$.
This can also be seen through the fact that any such candle would imply the existence of a cross in $(G',\Omega')$.
Hence $\psi$ is an ultimate frontier of $(G,\Lambda)$ and thus, $(G,\Lambda)$ is facial as desired.
\end{proof}

\paragraph{From candles to two red columns.}
Before we proceed looking for candles in a vortex, we need to show that these candles ultimately yield want we want.
What do we want though?
Our goal is to show that whenever $\mathsf{depth}_2$ is bounded, we can ``clean'' any vortex by deleting a small set of vertices.
Here ``cleaning'' essentially means the removal of all candles in the vortex.
If however, we find many candles, we should be able to find a witness that $\mathsf{depth}_2$ is large.

To this end, let $k \geq 1$ be an integer.
A \emph{$k$-candle mesh} is an annotated graph obtained from a $(3k \times 3k)$-mesh $M$ as follows.
Let the vertical paths of $M$ be $Q_1,\dots,Q_{3k}$ and the horizontal paths of $M$ be $P_1,\dots,P_{3k}$.
For each $i \in [3k]$, let us denote by $x_i$ the endpoint of $Q_i$ on $P_1$.
For each $j\in[k]$ we now add a path $L_j$ with endpoints $x_{3(j-1)+1}$ and $3j$ whose internal vertices are do not belong to $M$ or any of the other paths added to the graph.
Then, for each $j\in[k]$ we add a path $B_j$ which is vertex-disjoint from the entire graph as constructed so far except for one of its endpoints -- this endpoint is $x_{3(j-1) + 2}$.
For each $j \in[k]$ let $r_j$ be the endpoint of $B_j$ that is not $x_{3(j-1) + 2}$ and let $R \coloneqq \{ r_j \mid j \in [k] \}$.
If we let $H$ be the graph constructed above, then $(H,R)$ is now an annotated graph we call a $k$-candle mesh.
See \zcref{fig_CandleMesh} for an illustration.
We refer to the vertices $r_j$ as the \emph{flames} and the paths $B_j$ as the \emph{candle sticks}.

\begin{figure}[ht]
 \centering
 \begin{tikzpicture}

 \pgfdeclarelayer{background}
		\pgfdeclarelayer{foreground}
			
		\pgfsetlayers{background,main,foreground}

 \begin{pgfonlayer}{background}
 \pgftext{\includegraphics[width=7cm]{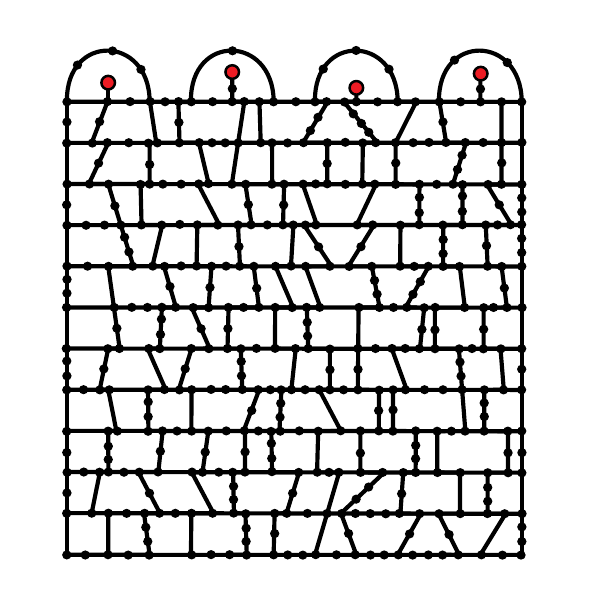}} at (C.center);
 \end{pgfonlayer}{background}
			
 \begin{pgfonlayer}{main}
 \node (C) [v:ghost] {};
 
 \end{pgfonlayer}{main}
 
 \begin{pgfonlayer}{foreground}
 \end{pgfonlayer}{foreground}

 \end{tikzpicture}
 \caption{A $4$-candle mesh.}
 \label{fig_CandleMesh}
\end{figure}

\begin{lemma}\label{lemma_2outerInCandles}
For every integer $k \geq 1$, every $2k$-candle mesh contains the $2$-outer $(k \times k)$-grid as a red-minor.
\end{lemma}

\begin{proof}
Let $H$ be a $k$-candle mesh.
We collect the flames of $H$ into pairs of two as follows.
Let $r_1,\dots,r_{2k}$ be the flames of $H$ numbered as in the definition of $k$-candle meshes.
For each $i\in[k]$ let $R_i \coloneqq \{ r_{2(i-1) + 1}, r_{2i}\}$.
For each $i \in [k]$, we will use $r_{2(i-1) + 1}$ for the $i$th vertex in the first column while $r_{2i}$ will be used for the $i$th vertex in the second column.
Recall from the definition of $2k$-candle meshes that each flame $r_i$ comes together with a path $B_i$.
Let us now trace a path $D_1'$ in $H$ in the union of the first horizontal path $P_1$ and the $L_i$ with $i \in[k]$ and $i$ being even such that $D_1'$ is disjoint from end endpoints of the candle sticks $B_i$ with even $i$.
The path $D_1'$ can now be broken into $k$ disjoint subpaths whose union contains all vertices of $D_1'$ and such that each subpath contains an endpoint of precisely one candle stick $B_i$ where $i$ is odd.
The path $D_1'$ together with the odd $B_i$ will make up the first row of the $(k \times k)$-grid we are building, we denote the resulting subgraph of $H$ by $D_1$.

For the second row take the union $D_2$ of the third horizontal path $P_3$ together with the extension of the even candlesticks along the even vertical paths until they meet $P_3$.
As before, $D_2$ can now be partitioned into precisely $k$ pieces whose union covers the entire vertex set of $D_2$ such that each piece contains precisely one of the even flames.
Notice that, by our construction, there also exist vertical paths in $H$, each connecting the $i$th part of $D_1$ to the $i$th part of $D_2$ for all $i\in[k]$.
This creates the complete first two columns of the $2$-outer $(k \times k)$-grid.

Notice that it is now straightforward to add the remaining $k-2$ columns to our minor model, as we have only used up the first three horizontal paths of the $(2k \times 2k)$-mesh within $H$ to capture the annotated vertices.
See \zcref{fig_CandlesTo2Outer} for a sketch of a slightly more expensive construction of the minor model and an indication on how the remaining $k-2$ columns may be created.
\end{proof}

\begin{figure}[ht]
 \centering
 \begin{tikzpicture}

 \pgfdeclarelayer{background}
		\pgfdeclarelayer{foreground}
			
		\pgfsetlayers{background,main,foreground}

 \begin{pgfonlayer}{background}
 \pgftext{\includegraphics[width=9cm]{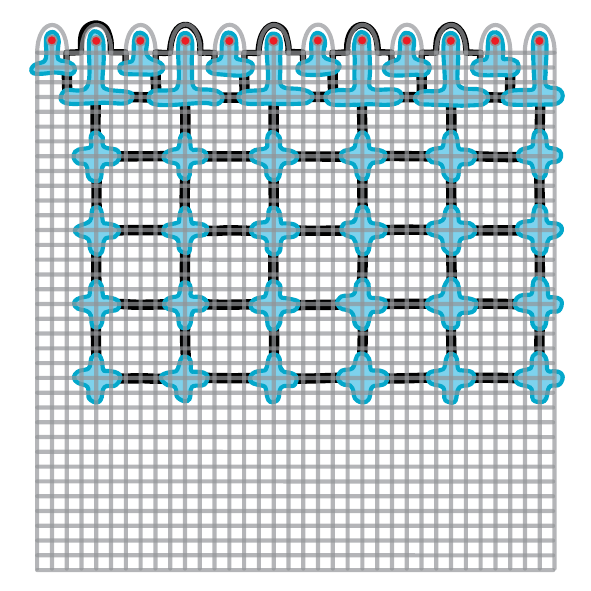}} at (C.center);
 \end{pgfonlayer}{background}
			
 \begin{pgfonlayer}{main}
 \node (C) [v:ghost] {};
 
 \end{pgfonlayer}{main}
 
 \begin{pgfonlayer}{foreground}
 \end{pgfonlayer}{foreground}

 \end{tikzpicture}
 \caption{A $2$-outer-annotated $(6 \times 6)$-grid as a red-minor in a $12$-candle mesh.
 Notice that the minor model indicated above has been chosen to enhance visibility and is not necessarily the most efficient choice.
 Still, as one can infer from the number of red vertices, no larger $2$-outer-annotated grid can exist in a $12$-candle mesh.
 In particular, the construction used here is slightly different from the one explained in the proof of \zcref{lemma_2outerInCandles}. The differences are only minor though and chosen purely to enhance visibility in the figure.}
 \label{fig_CandlesTo2Outer}
\end{figure}

\paragraph{Taming candles.}
With \zcref{thm_FacialSocieties,lemma_2outerInCandles} we have enough justification why candles are the correct object to be looking for.
Next we require a way to tame the way candles may appear inside of a vortex.

In later stages of our proof we will be confronted with the following situation.
We start with a cylindrical rendition $\rho$ in a disc $\Delta$ of some society $(G,\Lambda)$ containing a big nest $\mathcal{C} = \{ C_1,\dots,C_s \}$ around a vortex $c_0$ of bounded depth.
We then select -- in a later stage -- a rooted $\rho$-align curve $\gamma$ with the goal to capture an $R$-candle in the inner graph of $\gamma$.
However, $\gamma$ does not have to be well behaved with respect to the nest.
Let $\Delta_{\gamma} \subseteq \Delta$ be the inner disc of $\gamma$ and let $\rho_{\gamma}$ be the rendition of the inner graph $G_{\gamma}$ of $\gamma$ into $\Delta_{\gamma}$ inherited from $\rho$.
For each $i\in[n]$, $C_i \cap G_{\gamma}$ is now either still a cycle, or a collection of disjoint subpaths $Q$ of $C_i$ where each such $Q$ has both endpoints in $V(\gamma)$ -- where $V(\gamma)$ denotes the set of vertices corresponding to the nodes of $\rho$ contained in $\gamma$.
Let us call such a subpath $Q \subseteq C_i$ an \emph{arc} of $C_i$.

\begin{observation}\label{obs_arcsInVortices}
Let $Q$ be an arc of some cycle $C_i \in \mathcal{C}$.
Then there exists a subcurve $\psi_Q$ of $\gamma$ such that, when replacing $\psi_Q$ with the trace of $Q$ in $\gamma$, the resulting $\rho$-aligned curve $\gamma_Q$ is rooted and has an inner disc $\Delta_Q$.
Then the inner graph of $\gamma_Q$ contains at least one arc $Q_j$ of every $C_j$ with $j\in[i+1,s]$ such that $Q$ is not contained in the inner graph of $\gamma_{Q_j}$.
\end{observation}

In a later part of the proof we aim to prove that there is a constant $x$ such that no subpath of an $R$-candle caught in a rooted $\rho$-aligned curve starts on the vortex boundary, touches $C_x$, and then returns to the vortex.
\smallskip

Let $(G,\Lambda)$ be a society and $(G,R)$ be an annotated graph.
Moreover, let $\rho$ be a blank cylindrical rendition of $(G,\Lambda)$ in a disc $\Delta$ with a nest $\mathcal{C} = \{ C_1,\dots,C_s \}$ around the vortex $c_0$.
Finally, let $\gamma$ be a rooted $\rho$-aligned curve with base $S$ and let $(G_{\gamma},\Lambda_{\gamma})$ denote the restriction of $(G,\Lambda)$ to the inner graph of $\gamma$.

Consider an $R$-candle $\mathcal{X} = \{ P,Q\}$ in $(G_{\gamma},\Lambda_{\gamma})$.
We say that such an $R$-candle $\mathcal{X}$ is \emph{caught by $\gamma$}.
A path $L \subseteq P \cup Q$ is a \emph{$\mathcal{X}$-loop} if it is grounded and both of its endpoints belong to $N(c_0)$.
\smallskip

Our next goal is to prove the following lemma.

\begin{lemma}\label{lemma_CandleLoopsAreNotDeep}
There exists a universal constant $c > 0$ such that for every society $(G,\Lambda)$ where $(G,R)$ is an annotated graph, every blank cylindrical rendition $\rho$ of $(G,\Lambda)$ in a disc $\Delta$ with a nest $\mathcal{C} = \{ C_1,\dots,C_s \}$ around the vortex $c_0$ where $s \geq c+1$ the following holds:

If $\gamma$ is a rooted $\rho$-aligned curve and there exists an $R$-candle caught by $\gamma$ in $(G,\Lambda)$m then there also exists an $R$-candle $\mathcal{X}$ in $(G,\Omega)$ caught by $\gamma$ such that no $\mathcal{X}$ loop in $G$ intersects $C_{c+1}$.
\end{lemma}

In order to prove this lemma, we require a first intermediate result on the behaviour of $\mathcal{X}$-loops.
In the situation of \zcref{lemma_CandleLoopsAreNotDeep}, we say that an $R$-candle $\mathcal{X}$ caught by $\gamma$ is \emph{exhaustive} if for every $\mathcal{X}$-loop $L$, if $i \in [2,s]$ is the largest integer such that $L \cap C_i \neq \emptyset$, then there exists an $\mathcal{X}$-loop $L'$ which intersects $C_{i-1}$ and which is contained in the inner graph of the unique closed curve $\varphi_{L}$ defined by the trace of $L$ and the boundary of the vortex $c_0$ that does not contain $c_0$.
See \zcref{fig_ExhaustiveExample} for an illustration.
We refer to the disk bounded by $\varphi_L$ as the \emph{$L$-centre}.
We say that any $\mathcal{X}$-loop $L''$ which is contained in the $L$-centre is an \emph{inner} loop of $L$.

\begin{figure}[ht]
 \centering
 \begin{tikzpicture}

 \pgfdeclarelayer{background}
		\pgfdeclarelayer{foreground}
			
		\pgfsetlayers{background,main,foreground}

 \begin{pgfonlayer}{background}
 \pgftext{\includegraphics[width=8cm]{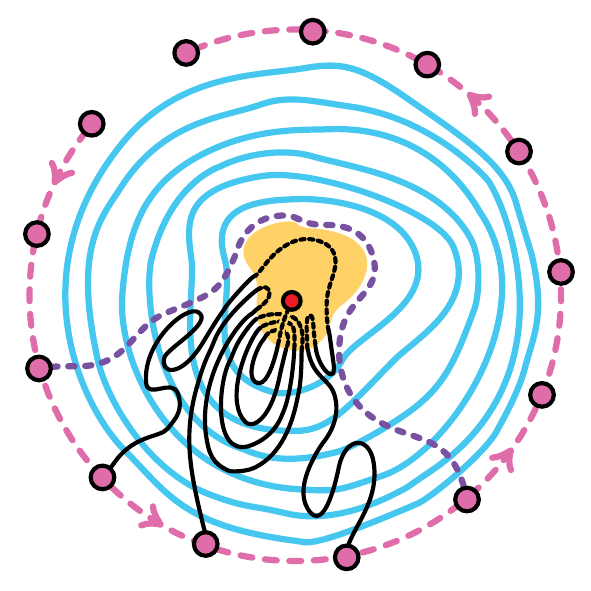}} at (C.center);
 \end{pgfonlayer}{background}
			
 \begin{pgfonlayer}{main}
 \node (C) [v:ghost] {};
 
 \end{pgfonlayer}{main}
 
 \begin{pgfonlayer}{foreground}
 \end{pgfonlayer}{foreground}

 \end{tikzpicture}
 \caption{A sketch of an exhaustive candle in a blank cylindrical rendition $\rho$ of a linear society with a nest (depicted in \textcolor{CornflowerBlue}{blue}).
 The \textcolor{DarkBananaYellow}{yellow} area marks the unique vortex of $\rho$.}
 \label{fig_ExhaustiveExample}
\end{figure}

\begin{lemma}\label{lemma_candlesAreExhaustive}
For every society $(G,\Lambda)$ where $(G,R)$ is an annotated graph, every blank cylindrical rendition $\rho$ of $(G,\Omega)$ in a disc $\Delta$ with a nest $\mathcal{C} = \{ C_1,\dots,C_s \}$ around the vortex $c_0$ the following holds:

If $\gamma$ is rooted $\rho$-aligned curve and there exists an $R$-candle caught by $\gamma$ in $(G,\Omega)$, then there also exists an exhaustive $R$-candle caught by $\gamma$.
\end{lemma}

\begin{proof}
The claim follows inductively from the following relatively simple argument.
Let $L$ be an $\mathcal{X}$-loop where $i \in [2,s]$ is the largest integer such that $L$ intersects $C_i$.
If no such $L$ exists, $\mathcal{X}$ is already exhaustive.
Moreover, if $L$ is not exhaustive, then we may choose $L$ such that no inner loop of $L$ intersects $C_{i-1}$.

Let $H$ be the inner graph of the $L$-centre.
Then notice that there exist disjoint subpaths $Q_1,\dots,Q_{\ell}$ of $C_{i-1}$, each with both endpoints in $L$ such that $\bigcup_{j \in [\ell]}Q_i = C_{i-1} \cap H$.
Since $L$ is an $\mathcal{X}$-loop, any subpath of a path of $\mathcal{X}$ contained in $H - V(\sigma(c_0))$ must be contained in some $\mathcal{X}$-loop.
Hence, from our assumption that there is no inner loop of $L$ that intersects $C_{i-1}$ we may infer that for each $j \in [\ell]$, none of the internal vertices of $Q_j$ belongs to $\mathcal{X}$.
Hence, if we consider $J \coloneqq L \cup \bigcup_{j\in [\ell]}Q_j$, then it becomes apparent that $J$ contains a path $L'$ with the same endpoints as $L$ but entirely disjoint from $C_i$ -- in fact, $L'$ contains only vertices from the inner graph of $C_{i-1}$.
Since both $L$ and the $Q_j$ are disjoint from the rest of $\mathcal{X}$ it follows that replacing $L$ with $L'$ in $\mathcal{X}$ results in a new $R$-candle caught by $\gamma$ where one of its loops intersects strictly less cycles of the nest while all other loops remain unchanged.

This last observation shows that, by choosing $\mathcal{X}$ to be an $R$-candle in $(G,\Omega)$ where the number of intersected cycles from the nest is minimal for each $\mathcal{X}$-loop, $\mathcal{X}$ must be an exhaustive $R$-candle as desired.
\end{proof}

Notice that the proof \zcref{lemma_candlesAreExhaustive} can actually be turned into an algorithm.
This algorithm will touch each of the at most $||G||$ loops at most $s$ times -- that is once for each cycle of the nest where the loop has its peak.
A naive implementation of the ``pushing'' procedure for a single loop takes $|G|$ time, so in total we get the following corollary towards our algorithmic application.

\begin{corollary}\label{cor_AlgoFindExhaustiveCandle}
There exists an algorithm that takes as input a society $(G,\Lambda)$ where $(G,R)$ is an annotated graph, a blank cylindrical rendition $\rho$ of $(G,\Lambda)$ in a disc $\Delta$ with a nest $\mathcal{C} = \{ C_1,\dots,C_s \}$ around the vortex $c_0$ where $s \geq c+1$, a rooted $\rho$-aligned curve $\gamma$ and an $R$-candle $\mathcal{X}$ caught by $\gamma$, and outputs an exhaustive $R$-candle caught by $\gamma$ in time $\mathbf{O}(|G|(s + ||G||))$.
\end{corollary}

We are now ready to proceed with the proof of \zcref{lemma_CandleLoopsAreNotDeep}.

\begin{proof}[Proof of \zcref{lemma_CandleLoopsAreNotDeep}]
Let $\beta$ be the function from \zcref{prop:VitalLinkage}.
We set $c \coloneqq 2 \cdot \beta(4,2) + 1$.

Let $\mathcal{X}$ be an $R$-candle in $(G,\Lambda)$ caught by $\gamma$.
By \zcref{lemma_candlesAreExhaustive} we may assume that $\mathcal{X}$ is exhaustive.
This means that one of two cases holds:
\smallskip

\textbf{Case 1:}
No $\mathcal{X}$-loop intersects $C_{c+1}$ and we are done.
\smallskip

\textbf{Case 2:}
There exists a $\mathcal{X}$-loop $L_{c+1}$ which intersects $C_{c+1}$.

In this case, let $H_0$ be the graph obtained from the union of all $C_i$-arcs and complete cycles $C_i$ in the inner graph of $\gamma$ and the two paths in $\mathcal{X}$.
Let then $H_1$ be obtained from $H_0$ by iteratively performing the following operations until none of them can be applied
\begin{enumerate}
    \item if there is a vertex $v \in V(H_0)$ of degree $0$, delete $v$,
    \item if there is $v\in V(H)$ of degree $1$, contract its unique incident edge.
    \item if there is $v \in V(H)$ of degree $2$, contract one of its two incident edges.
\end{enumerate}
As a result, any $C_i$-arc of $H_0$ that intersected $\mathcal{X}$ is still present -- be it a bit smaller possibly -- in $H_1$.
Moreover, every vertex of $H_1$ is now of degree $3$ or $4$.
Now notice that each $\mathcal{X}$-loop that intersects some cycle of $\mathcal{C}$ in $H_0$ corresponds to an $\mathcal{X}$-loop in $H_1$.
Indeed, this means that there is now a sequence $\langle L_1 ,\dots, L_{c+1}\rangle$ of $\mathcal{X}$ loops where $L_i$ intersects an arc of $C_j$ for each $j \in [i]$ and $L_{i-1}$ is an inner loop of $L_i$ for all $i \in[2,c+1]$.
This means that whenever some $L_i$ meets an arc of any $C \in \mathcal{C}$, then so does every $L_j$ for which $L_i$ is an inner loop.
Thus, for each $i\in[c+1]$ we may select an arc $A_i$ of $C_i$ which is intersected by all $C_j$ with $j \in [i,c+1]$.
Then let $B_i \coloneqq A_i \cup L_i$.
Since all of the $A_i$ are pairwise disjoint and all of the $L_i$ are, each vertex of $\bigcup_{i \in [c+1]}B_i$ belongs to at most two $B_i$.
Moreover, any two $B_i$ intersect in at least one vertex.
Hence, $\mathcal{B} \coloneqq \{ B_i \colon i\in[c+1] \}$ is a bramble.
And since any vertex can meet at most two elements of $\mathcal{B}$, the order of $\mathcal{B}$ is at least $\lceil \nicefrac{c+1}{2}\rceil$.
With $c = 2 \cdot \beta(4,2) + 1$ this implies by \zcref{prop_brambles} that the treewidth of $H_1$ is at least $\beta(4,2) + 1$.

Let $\mathcal{X}= \{ P,Q \}$ where $P,Q \subseteq H_1$ are the two paths of what remains of the candle $\mathcal{X}$ in $H_1$.
Let $p_1,p_2$ be the endpoints of $P$ and $q_1,q_2$ be the endpoints of $Q$.
Then, by \zcref{prop_DisjointPathsIrrelevant}, there exists a vertex $v \in V(H_1)$ such that there still exist disjoint paths $P'$ and $Q'$ such that $P'$ joins $p_1$ and $p_2$ while $Q'$ joins $q_1$ and $q_2$ in $H_1 - v$.
By our choice of $H_1$ we know that $v$ has degree $3$ or $4$.
In either case, there are at most two edges incident with $v$ in $H_1$ that do not belong to the nest.
Let $H_2$ be the graph obtained from $H_1$ by removing all such edges and then applying the three reduction rules above.
It follows, that also $H_2$ still contains a candle rooted on the same four vertices $p_1$, $p_2$, $q_1$, $q_2$.
Moreover, still there is a bijection between the arcs of the cycles in the nest in $H_0$ and those in $H_1$ -- with the exception of those arcs that do not intersect with the candle at all.

From here on we repeat the steps above as follows:
Suppose we have constructed $H_i$ for some $i \geq 2$.
Then we know that $H_i$ has a candle rooted on the vertices $p_1,p_2,q_1$, and $q_2$.
By \zcref{lemma_candlesAreExhaustive} we may assume this candle -- let us call it $\mathcal{X}_i$ -- to be exhaustive.
If no $\mathcal{X}$-arc intersects $C_{c+1}$ we have won.
Otherwise we again use \zcref{prop_DisjointPathsIrrelevant} to find an irrelevant vertex $v_i$, remove all edges incident with $v_i$ that do not belong to the nest, and apply our three reduction rules in order to remove vertices of degree less than $3$.

Since each iteration removes at least one edge, but we always ensure that there still is a candle afterwards, this process must eventually come to a halt.
Hence, after at most $||H_0||$ steps we have found the desired candle.
\end{proof}

Notice that we have used a non-constructive version of \zcref{prop:VitalLinkage} in the proof of \zcref{lemma_CandleLoopsAreNotDeep}.
There exists a constructive variant of \zcref{prop:VitalLinkage} in the work of Cavallaro, Gorsky, Kreutzer, Thilikos, and Wiederrecht \cite{CavallaroGKTW2026Optimal}.
In order to save on notation, we will not state the full theorem but rather provide a specialisation of this theorem to the \textsc{$k$-Disjoint Paths Problem}.

In this context, given an annotated graph $(G,R)$ and an integer $k$, we say that a vertex $v \in V(G) \setminus R$ is \emph{strongly $(k,R)$-irrelevant} if for every choice of at most $k$ terminals pairs $\mathcal{T} = \{(s_1,t_1),\dots,(s_{k'},t_{k'})\}$ where $s_i,t_i \in R$ for all $i\in[k']$, $(G,\mathcal{T})$ is a \textsf{yes}-instance of the \textsc{$k$-Disjoint Paths} problem if and only if $(G-v,\mathcal{T})$ is a \textsf{yes}-instance of the \textsc{$k$-Disjoint Paths} problem.
The original version of the theorem gives an even stronger conclusion, allowing for up to $d$ additional terminals that are ``floating'', that is unspecified.
This problem is known as the $(k,d)$\textsc{-Folio} problem and here we only state the version where $d=0$.
See \cite{CavallaroGKTW2026Optimal} for more information.

\begin{proposition}[Cavallaro, Gorsky, Kreutzer, Thilikos, and Wiederrecht \cite{CavallaroGKTW2026Optimal}]\label{prop_DisjointPathsIrrelevant}
There exists a function $f_{\ref{prop_DisjointPathsIrrelevant}} \colon \mathbb{N}^2 \to \mathbb{N}$ and an algorithm that takes as input a annotated graph $(G,R)$ of bidimensionality at most $b$, and either correctly determines that $\mathsf{tw}(G) \leq f_{\ref{prop_DisjointPathsIrrelevant}}(k,b)$ or finds a vertex $v\in V(G)\setminus R$ such that $v$ is strongly $(k,R)$-irrelevant in time $2^{2^{\poly(b)}\cdot \poly(k)}\cdot |G|$.
Moreover, $f_{\ref{prop_DisjointPathsIrrelevant}}(k,b) \in 2^{\poly(b)}\cdot \poly(k)$.
\end{proposition}

With \zcref{prop_DisjointPathsIrrelevant} and \zcref{prop_TwoPaths}, the proof of \zcref{lemma_CandleLoopsAreNotDeep} can be made constructive and thus, such candles can be found efficiently.

\begin{corollary}\label{cor_FindShallowCandles}
Let $c > 0$ be the constant from \zcref{lemma_CandleLoopsAreNotDeep}.
There exists an algorithm that takes as input a society $(G,\Lambda)$ where $(G,R)$ is an annotated graph, a blank cylindrical rendition $\rho$ of $(G,\Lambda)$ in a disc $\Delta$ with a nest $\mathcal{C} = \{ C_1,\dots,C_s \}$ around the vortex $c_0$ where $s \geq c+1$ and a rooted $\rho$-aligned curve $\gamma$, and finds one of the following outcomes in time $\mathbf{O}(|G|^5)$:
\begin{enumerate}
    \item An $R$-candle $\mathcal{X}$ captured y $\gamma$ in $(G,\Lambda)$ such that no $\mathcal{X}$-loop intersects $C_{c+1}$, or
    \item the correct conclusion that $(G,\Lambda)$ does not contain an $R$-candle captured by $\gamma$.
\end{enumerate}
Moreover, if we are given any $R$-candle $\mathcal{X}'$ captured by $\gamma$ as an additional input, the algorithm finds an $R$-candle $\mathcal{X}$ captured by $\gamma$ as above with the same roots as $\mathcal{X}'$.
\end{corollary}

\begin{proof}
We start by simply looking for an $R$-candle captured by $\gamma$ greedily as follows:
For each possible choice of $r \in R$ and $x_1,x_2,x_3 \in N(\gamma)$ where $x_1,x_2,x_3$ appear on $\Lambda$ in the order listed, we run the algorithm from \zcref{prop_TwoPaths} on $(G_{\gamma},\{ (r,x_2),(x_1,x_3) \})$ where $G_{\gamma}$ is the inner graph of $\gamma$.
If the answer is yes for one of those $|G|^4$ calls, then we have found an $R$-candle $\mathcal{X}$ captured by $\gamma$.
Otherwise we may conclude that no such candle exists and stop the algorithm.

Hence, after at most $\mathbf{O}(|G|^5)$ time we may now assume to have found an $R$-candle $\mathcal{X}_0$ in $(G,\Lambda)$.
In case we were provided with such an $R$-candle $\mathcal{X}'$ beforehand, this is now where the algorithm starts.

Next we use the algorithm from \zcref{cor_AlgoFindExhaustiveCandle} to find an exhaustive $R$-candle $\mathcal{X}_1$ captured by $\gamma$ in time $\mathbf{O}(|G|(s + ||G||)) \in \mathbf{O}(s \cdot |G|^3)$.

From here we follow a similar route as the one explained in the proof of \zcref{lemma_CandleLoopsAreNotDeep}.
If some $\mathcal{X}_1$-loop intersects $C_{c+1}$, then the treewidth of the union $H$ of the $C_i$-arcs for all $i\in[s]$ and $\mathcal{X}_1$ must be above the threshold needed for an application of \zcref{prop_DisjointPathsIrrelevant} with $k=4$ and $b=2$.
This yields an irrelevant vertex $v$.

Notice that the maximum degree of $H$ is at most $4$.
If $v$ is a vertex without neighbours, simply delete $v$ and let $H_1$ be the resulting graph.
In case $v$ has degree $1$, let $H_1$ be the graph obtained from $H_0\coloneqq H$ by contracting the unique edge incident with $v$.
Similarly, if $v$ has degree $2$, contract one of its two incident edges.
Finally, if $v$ has degree larger than $2$, then at least one of its incident edges belongs to $\mathcal{X}_1$ but not to any of the cycles of $\mathcal{C}$.
In this case let $H_1$ be the graph obtained from $H_0$ by deleting all edges incident with $v$ that do not belong to $\mathcal{C}$.

Notice that this process keeps all arcs of the cycles in $\mathcal{C}$ intact -- it might shrink them though.
Also, $||H_0|| \leq 4|H_0|$ due to the bound on the maximum degree.

Hence, by repeating this process at most $\mathbf{O}(|G|)$ times -- each time either contracting an edge or removing edges incident with an irrelevant vertex that do not belong to the cycles --  we eventually must end in a graph $H_{\ell}$ which has bounded treewidth.
The treewidth of $H_{\ell}$ is in fact at most $c$.

Notice that $H_{\ell}$ still contains an $R$-candle $\mathcal{X}_{\ell}$ on the same roots as $\mathcal{X}_1$.
We may even find an exhaustive one, call it $\mathcal{X}_{\ell}$, using the algorithm from \zcref{cor_AlgoFindExhaustiveCandle}.
However, it now follows from the arguments in the proof of \zcref{lemma_CandleLoopsAreNotDeep}, that no $\mathcal{X}_{\ell}$-loop can intersect what remains of $C_{c+1}$ as otherwise we would have a witness that the treewidth of $H_{\ell}$ is still at least $c+1$. 

There are $\mathcal{O}(|G|)$ edges in $H_0$ and each iteration of this process takes $\mathbf{O}(|G|)$ time.
Thus, after spending another $\mathbf{O}(|G|^2)$ time we may call \zcref{prop_TwoPaths} and \zcref{cor_AlgoFindExhaustiveCandle} to find $\mathcal{X}_{\ell}$ as described above.
\end{proof}

\paragraph{Arranging candles into a mesh.}
In the previous steps we have shown how to find candles and how to enforce some additional structure.
Before we are able to move on to the next step, we need to show that in the case where we have many consecutive candles rooted on a society with a large nest, there exists a large candle mesh.

Let $(G,\Lambda)$ be a society with a cylindrical rendition $\rho$ and vortex $c_0$.
We say that a curve $\gamma$ is \emph{semi $\rho$-aligned} if it can be divided into three -- possibly empty -- pieces:
Two of them are $\rho$-aligned curves with one end in $\Lambda$, while the last starts and ends on points of $N(c_0)$ and is otherwise entirely contained in the interior of $c_0$.
This last part is called the \emph{inner piece} of $\gamma$.

For a semi $\rho$-aligned curve, the vortex society $\Lambda_{c_0}$ can be partitioned into two segments:
\begin{enumerate}
    \item $S_1\subseteq V(\Lambda_{c_0})$ such that, if $\gamma'$ is the rooted $\rho$-aligned curve obtained from $\gamma$ by replacing the inner piece of $\gamma$ with any $\rho$-aligned curve $\gamma''$ that intersects $\rho$ precisely in $S_1$, then the inner disc of $\gamma'$ does not contain $c_0$, and
    \item $S_2 \coloneqq V(\Lambda_{c_0}) \setminus S_1$.
\end{enumerate}
We call the curve $\gamma'$ the \emph{projection of $\gamma$} and $S_1$ the \emph{inner segment} of $\gamma$.
If there does not exist a path in $\sigma(c_0)$ from $S_1$ to $S_2$ we say that $\gamma$ is a \emph{proper semi $\rho$-aligned curve}.
The \emph{inner graph} of a proper semi $\rho$-aligned curve $\gamma$ is defined as the union of the inner graph of $\gamma'$ together with all components of $\sigma(c_0)$ that contain a vertex from $S_1$.

Notice that whenever we consider a proper semi $\rho$-aligned curve $\gamma$, we may delete everything in $\sigma(c_0)$ that does not belong to the inner graph of $\gamma$, resulting in the society $(G',\Lambda)$.
If we then shrink the vortex $c_0$ in order to be entirely in the interior of $\gamma$, we obtain a cylindrical rendition $\rho'$ of $(G',\Lambda)$ from $\rho$ where $\gamma$ is a rooted $\rho'$-aligned curve whose inner disc contains the vortex $c_0'$ of $\rho'$.
Notice further that any nest in $G$ and $\rho$ around $c_0$ remains a nest around the vortex of $\rho'$ in $G'$.
This gives us the right to treat proper semi $\rho$-aligned curves as rooted $\rho$-aligned curves by implicitly referring to the graph $G'$ and the rendition $\rho'$.
Hence, any definition and result for rooted $\rho$-aligned curves referring only to their inner graphs can be lifted to also apply to proper semi $\rho$-aligned curves.

\begin{lemma}\label{lemma_ForcingCandlesIntoAMesh}
Let $c$ be the constant from \zcref{lemma_CandleLoopsAreNotDeep}.
For every integer $k \geq 1$ and every society $(G,\Lambda)$ where $(G,R)$ is an annotated graph, if the following requirements are met:
\begin{itemize}
    \item $(G,\Lambda)$ has a blank cylindrical rendition in a disc with a nest $\mathcal{C} = \{ C_1,\dots,C_s\}$ where $s \geq c + 3k +1$, and
    \item there exist $k$ pairwise disjoint $R$-candles $\mathcal{X}_1,\dots,\mathcal{X}_k$ in $(G,\Lambda)$ such that $\Lambda$ can be partitioned into $k$ segments $S_1,\dots,S_k$ appearing in this order and the roots of $\mathcal{X}_i$ belong to $S_i$ for all $i\in[k]$,
    \item for every $i \in [k]$ there exists a proper semi $\rho$-aligned curve $\gamma_i$ whose inner graph contains $\mathcal{X}_i$ and are disjoint except for $V(\gamma_i \cap \gamma_j)$ for all $i \neq j \in[k]$, and
    \item for each $i\in[k]$, no $\mathcal{X}_i$-loop intersects $C_{c+1}$.
\end{itemize}
then $(G,R)$ contains a $k$-candle mesh.
\end{lemma}

\begin{proof}
For each $i\in[k]$ let $\mathcal{X}_i \coloneqq \{ P_i,Q_i\}$ where $Q_i$ is the candle stick of $\mathcal{X}_i$.
We now select three paths $X_i,Y_i,Z_i$ from each $\mathcal{X}_i$ as follows.
Let $x_i^1,y_i,x_i^2$ be the three roots of $\mathcal{X}_i$ on $\Lambda$ in order of their appearance.
Then let $X_i$ be the longest subpath of $P_i$ that contains $x_i^1$, has its other endpoint on $C_{c+1}$, and does not contain a vertex of the vortex $c_0$.
Similarly, let $Y_i$ be the longest subpath of $Q_i$ that contains $y_i$, has its other endpoint on $C_{c+1}$, and does not contain a vertex of the vortex $c_0$.
Finally, let $Z_i$ be the longest subpath of $P_i$ that contains $x_i^2$, has its other endpoint on $C_{c+1}$, and does not contain a vertex of the vortex $c_0$.

Since no $\mathcal{X}_i$-loop intersects $C_{c+1}$ it follows that $X_i \cup Z_i$ contain all vertices in $V(P_i) \cap V(C_{c+1})$ while $Y_i$ contains all vertices in $V(Q_i) \cap V(C_{c+1})$.
Note that we are given the right to talk about $\mathcal{X}_i$-loops in the first place by the third assumption of the assertion that provides us with a proper semi $\rho$-aligned curve for each candle.
Moreover, $X_i$ and $Z_i$ must be vertex-disjoint since $\rho$ is a blank rendition and thus all vertices of $R$ are contained in $\sigma(c_0)$.
Let $P_i'$ be the subpath of $P_i$ joining the two non-$\Lambda$-endpoints of $X_i$ and $Y_i$ as well as $Q_i'$ be the subpath of $Q_i$ joining the flame of $\mathcal{X}_i$ and the non-$\Lambda$-endpoint of $Y_i$.

Now, for every $i \in[c+1,s]$ there exists a subpath $T_i$ of $C_i$ starting on $X_1$, ending on $Z_k$, and intersecting all $X_i, Y_i, Z_i$ for $i \in [k]$.
From here on we consider the graph $H \coloneqq \bigcup_{i\in[k]} (X_i \cup Y_i \cup Z_i) \cup \bigcup_{i \in[c+1,s]} T_i$.
Our goal is to show that within $H$ we can replace each $T_i$ with a path $T_i'$ whose intersection with any of the $X_i'$, $Y_i'$, and $Z_i'$ is a single path where the $X_i',Z_i',Y_i'$ are replacements of their respective paths while keeping their endpoints.
If this is true, then $H' \coloneqq \bigcup_{i\in[k]} (X_i' \cup Y_i' \cup Z_i') \cup \bigcup_{i \in[c+1,s]} T_i'$ is a mesh and $H' \cup \bigcup_{i \in [k]} (P_i' \cup Q_i')$ is easily seen to contain a $k$-candle mesh since $s - c \geq 3k$.
\smallskip

In order to find $H'$ we start by first ''pushing'' the $F_i$ as much towards $Z_k$ as possible for all $F \in \{ X,Y,Z\}$ and $i \in [k]$.
What we mean by that is as follows.
Notice that $(H,\Lambda_H)$ inherits a vortex-free rendition $\rho_H$ in a disc $\Delta$ from $\rho$ where $\Lambda_H$ is the linear ordering obtained from first tracing along $X_1$ towards its endpoint on $\Lambda$, then iterating over the ends of the $F_i$ on $\Lambda$ until we reach the endpoint of $Z_k$.
Then we trace along $Z_k$ towards its other endpoints and now we list the remaining endpoints of the $F_i$ in reverse order compared to the order in which we visited their other endpoints on $\Lambda$.
Notice that the trace of any $F_i \notin \{ X_1,Z_k\}$ separates $X_1$ and $Z_k$ in $\Delta$.
Indeed, the trace of each such $F_i$ separates $\Delta$ into two discs, $\Delta^1_{F_i}$ which contains $X_1$ and $\Delta^2_{F_i}$ which contains $Z_k$.
We say that a subgraph $J \subseteq H$ is on the \emph{right side} of $F_i$ if it is contained in the restriction of $\rho_H$ to the disc $\Delta^2_{F_i}$.
For the case $F_i = X_1$, any subgraph of $H$ is on its right side, and for $F_i = Z_k$ no such $J$ is on its right side except for subgraphs of $Z_k$ itself.

Now proceed as follows:
Suppose there exists some path $U\subseteq H$ grounded in $\rho_H$ such that there is $F_i$ for which both endpoints of $U$ belong to $F_i$, $U$ is otherwise disjoint from \textsl{all} $F_i'$, has at least one edge not on $F_i$, and $U$ is on the right side of $F_i$.
We call $U$ a \emph{$F_i$-bumpy} and we say that $F_i$ \emph{has a bumpy} if there exists an $F_i$-bumpy. 

Let $U'$ be the unique subpath of $F_i$ joining the endpoints of $U$.
If we now replace $U'$ with $U$ in $F_i$, the disc $\Delta^1_{F_i}$ for the new $F_i$ has strictly increased.
Otherwise, however, we are still able to recover the full candle $\mathcal{X}_i$ with the changed subpath and $F_i$ remains vertex-disjoint from all other.
Hence, as long as there exists some $F_i$ that has a bumpy we may proceed with this replacement strategy.
As the size of the discs $\Delta^1_{F_i}$ strictly increases with each exchange and each bumpy introduces a fresh edge, this process must halt after at most $||H||$ steps.
Thus, from here on we may assume that no $F_i$ has a bumpy.
\smallskip

Next we deal with paths that extend to the left.
Let $i \in[c+1,s]$ and $U \subseteq T_i$.
We say that $U$ is an \emph{$T_i$-curl} at $F_j$ if $U$ is grounded, both endpoints lie on $F_j$ where $j \in[1,k]$ and $F \in \{ X,Y,Z\}$, $U$ is otherwise disjoint from $F_j$, and $U$ is contained in the restriction of $\rho_H$ to $\Delta^1_{F_j}$. 
We say that a $T_i$-curl at $F_j$ is \emph{free} if the subpath of $F_j$ joining the endpoints of $U$ -- we refer to this path as the \emph{$U$-bridge} -- does not share an internal vertex with any of the $T_{i'}$, $i' \in[c+1,s]$.

Whenever we find some $T_i$-curl $U$ at some $F_j$, we may replace $U$ on $T_i$ with the $U$-bridge of $F_j$.
This clearly maintains that each $T_i$ is an $X_1$-$Z_k$-path, as well as the pairwise disjointness of the $T_i$.
Moreover, this operation cannot introduce new bumpies.

Notice that for all $i\in[c+1,s]$, there does not exist any $T_i$-curl on $Z_k$.
Moreover, it is also easy to see that every $T_i$-curl at $Y_k$ is free.
Hence, by performing the operation above for each such curl at $Y_k$, we may assume that there is also no curl at $Y_k$.

We say that $F_i$ is \emph{right of} $F_j'$ if either $i = j$ and $F = X$, $F' \in \{ Y,Z\}$ or $F = Y$, $F' = Z$ or $i < j$.
We now iterate through the $F_j$ from right to left -- where $Z_k$ and $Y_k$ are our starting points -- and use the same replacement trick as the one used for $Y_k$ to ensure that $F_j$ has no $T_i$-curls for any $i \in [c+1,s]$.
To see that this can be done simply consider the following situation.

Suppose we have arrived at $F_j$ and already know that there is no $i\in[c+1,s]$ such that there is a $T_i$-curl at $F'_{j'}$ for any $F'_{j'} \neq F_j$ that is to the right of $F_j$.

\textbf{We claim} that every $T_i$-curl at $F_j$ must be free.
Suppose there is such a curl $U$ which is not free.
Then there is a vertex $x$ on the $U$-bridge of $F_j$ that belongs to $T_{i'}$ for some $i' \in[c+1,s]$.
We may select $x$ and a subpath $A$ of $T_{i'}$ such that $x$ is an endpoint of $A$, the other endpoint of $A$ belongs to some $F'_{j'}$, and $A$ is internally disjoint from all of the $F''_{j''}$.
We distinguish three cases.
\medskip

\textbf{Case 1:}
$A$ is contained in a $T_{i'}$-curl at $F_j$.

In this case we have that the $A$-bridge on $F_j$ is strictly shorter than the $U$-bridge on $F_j$ and thus we may start over with $A$ instead of $U$ in our case distinction.
This case cannot repeat indefinitely since in each step the length of the bridge is strictly reduced.
\smallskip

\textbf{Case 2:}
$A$ is not a $T_{i'}$-curl at $F_j$ but its other endpoint is also on $F_j$.

In this case, the only other option is that $A$ is an $F_j$-bumpy.
Since we assumed that $F_j$ has no bumpies, this is a contradiction.
\smallskip

\textbf{Case 3:}
The other endpoint of $A$ lies of $F'_{j'} \neq F_j$.

In this case, since $A$ is not contained in a $T_{i'}$-curl at $F_j$ it follows that $F'_{j'}$ must be to the right of $F_j$.
In this case, we start tracing along $T_{i'}$ starting at $x$ but avoiding $A$.
At some point $T_{i'}$ must leave $F_j$ and then meet some $F'''_{j'''}$.
Let $B$ be the subpath of $T_{i'}$ joining the point where we leave $F_j$ and the point where we meet $F'''_{j'''}$ for the first time.
There are two cases; $F'''_{j'''} \in \{ F_j,F'_{j'} \}$ since $B$ cannot intersect $A$.

\textbf{Case 3.1:}
$F'''_{j'''} = F_j$.

In this case, $B$ is either a curl or a bumpy.
If it is a curl, we discard $A$ and continue with $B$ by the same arguments as those from \textbf{Case 1}.
Otherwise $B$ is a bumpy and we have a contradiction.

\textbf{Case 3.2:}
$F'''_{j'''} = F'_{j'}$.

In this case we can see that $A$ and $B$ are joined by a path in $T_{i'} \cap F_j$, resulting in a single path $W$ with both ends on $F'_{j'}$.
It is now clear that $W$ is a $T_{i'}$-curl at $F'_{j'}$ which cannot exist by our assumption.
\medskip

It follows that every $T_i$-curl at $F_j$ is free for all $i\in[c+1,s]$.
Hence, by replacing such free curls with subpaths of $F_j$ we restructure all of the $T_i$ in a way that maintains our structure but removes all $T_i$-curl at $F_j$.
Therefore, we may assume that there do not exist $i \in[c+1,s]$ and $F_j$ such that there is a $T_i$-curl at $F_j$.
\medskip

From the absence of curls it follows that whenever one of the $T_i$ enters an $F_j$, leaves $F_j$, and returns, this must happen within $\Delta^2_{F_j}$.
Moreover, before $T_i$ returns it must intersect $F'_{j'}$ directly to the right of $F_j$.
Moreover, each time such a situation arises, $T_i$ must then follow along $F_j$ and, once it leaves, it must return to $F'_{j'}$.
This however means that there is a $T_i$-curl at $F'_{j'}$ which is a contradiction.
Hence, the only option left is, that at this point the $T_i$ and $F_j$ form a $((s-c) \times k)$-mesh.
As discussed at the beginning, this suffices to complete the proof.
\end{proof}

\subsection{Part 2: Cleaning a vortex}\label{subsec_Step2}
We are now ready to move on to the second step.
Our goal is to prove that, whenever we are given a society with a large nest in a cylindrical rendition around some vertex, we either find a large candle mesh -- which by \zcref{lemma_2outerInCandles} implies that the $\mathsf{depth}_2$ is big -- or are able to delete a small set of vertices to remove all candles from the vortex.
Moreover, we wish to further strengthen this in order to be able to apply \zcref{thm_FacialSocieties} to say that those parts of the vortex which do not allow for a cross are entirely facial while those which do contain a cross must be blank.

\paragraph{Two types of candle-free societies.}
Towards this goal, we first start by proving a companion lemma to \zcref{thm_FacialSocieties} that allows to distinguish candle-free societies into two types.
\smallskip

Given a society $(G,\Lambda)$ where $(G,R)$ is an annotated graph, we say that $(G,\Omega)$ has \emph{access to a red vertex} if there exists an $R$-$V(\Lambda)$-path in $G$.
The core observation is that, in the presence of a non-empty nest, a cross together with access to a red vertex imply the existence of a candle.

\begin{lemma}\label{lemma_RedCross}
Let $c > 0$ be the constant from \zcref{lemma_CandleLoopsAreNotDeep}.
Let $(G,\Lambda)$ be a society where $(G,R)$ is an annotated graph, and let $\rho$ be a blank rendition of $(G,\Lambda)$ in a disc $\Delta$ with a nest $\mathcal{C} = \{ C_1,\dots,C_s \}$ around the vortex $c_0$ of $\rho$ where $s \geq c+1$.
Then one of the following is true
\begin{enumerate}
    \item There exists an $R$-candle $\mathcal{X}$ in $(G,\Lambda)$ such that no $\mathcal{X}$-loop intersects $C_{c+1}$,
    \item $R$ is facial in $(G,\Lambda)$, or
    \item $(G,\Lambda)$ has no access to a red vertex.
\end{enumerate}
Moreover, there exists an algorithm that takes $(G,\Lambda)$, $R$, $\mathcal{C}$, and $\rho$ as above as input and finds one of the three outcomes in time $\mathbf{O}(|G|^5)$.
\end{lemma}

\begin{proof}
Suppose that $(G,\Lambda)$ has access to a red vertex and $R$ is not facial in $(G,\Lambda)$.
Then, by \zcref{thm_FacialSocieties}, either $(G,\Lambda)$ contains an $R$-candle, or $(G,\Lambda)$ has a cross.
Our goal is to show that in both cases, the first outcome of the lemma holds, that is:
There exists an $R$-candle $\mathcal{X}$ in $(G,\Omega)$ such that no $\mathcal{X}$-loop intersects $C_{c+1}$.
\medskip

\textbf{Case 1:}
$(G,\Lambda)$ contains an $R$-candle.
\smallskip

In this case, all we have to do is to choose $\gamma$ to be a curve whose ends are the minimum and the maximum of $\Lambda$ and  which is otherwise disjoint from $\rho$ and then apply \zcref{cor_FindShallowCandles}.
Since $\mathcal{X}$ trivially captured by $\gamma$, the second outcome of \zcref{cor_FindShallowCandles} cannot apply and thus, we find an $R$-candle $\mathcal{X}'$ such that no $\mathcal{X}'$-loop intersects $C_{c+1}$ as desired.
\medskip

\textbf{Case 2:}
$(G,\Lambda)$ has a cross.
\smallskip

Let $s_1,s_2,t_1,t_2$ be the endpoints of the two paths $P_1$ and $P_2$ that form the cross in $(G,\Lambda)$.
Notice that we can find such a cross in linear time due to \zcref{prop_TwoPaths,prop_TwoPathsProper}.
Then let $r \in R$ be any vertex such that there exists a path $Q'$ in $G$ connecting $r$ to $V(G)$.
Let us now find a path from $r$ to $P_1 \cup P_2$ as follows:
Since $\mathcal{C}$ contains at least one cycle and $\rho$ is a cylindrical rendition in a disc, each of $P_1$ and $P_2$ must have an edge in the vortex $c_0$.
Similarly, $r \in V(\sigma(c_0))$.
Hence, $P_1$, $P_2$, and $Q'$ all intersect $C_1$.
Let $Q'$ be the path found by starting at $r$, then following along $Q''$ until either, we meet a vertex of $V(P_1) \cup V(P_2)$ -- in which case we stop -- or we meet $C_1$.
In the later case we follow along $C_1$ until we meet the first vertex of $V(P_1) \cup V(P_2)$.
Without loss of generality let us assume that $Q'$ ends on $P_1$.
Let now $Q$ be the $r$-$t_1$-subpath of $P_1 \cup Q'$.
Then $\mathcal{X} = \{P_2,Q\}$ is an $R$-candle in $(G,\Omega)$.
Hence, we have reduced \textbf{Case 2} to \textbf{Case 1} and our proof is complete.
\end{proof}

A core feature of the proof of \zcref{lemma_RedCross} is that we actually do not need access to the entire nest to make it work.
At its core the main argument we use is that the three paths $P_1$ and $P_2$ of the cross as well as the path $Q''$ connecting $R$ to $V(\Lambda)$ are all connected by a subpath of $C_1$.
This means that we may use the same proof to derive a slightly stronger variant of \zcref{lemma_RedCross} which is the key tool we are after.
We state the stronger variant below but omit its almost identical but slightly more tedious proof.

\begin{lemma}\label{lemma_StrongRedCross}
Let $c > 0$ be the constant from \zcref{lemma_CandleLoopsAreNotDeep}.
Let $(G,\Lambda)$ be a society where $(G,R)$ is an annotated graph, let $\rho$ be a blank rendition of $(G,\Lambda)$ in a disc $\Delta$ with a nest $\mathcal{C} = \{ C_1,\dots,C_s \}$ around the vortex $c_0$ of $\rho$ where $s \geq c+1$, and let $\gamma$ be a proper semi $\rho$-aligned curve.
Moreover, let $(G_{\gamma},\Lambda_{\gamma})$ be the society obtained from the inner graph $G_{\gamma}$ of $\gamma$ by restricting $\Lambda$ to $G_{\gamma}$.
Then one of the following is true
\begin{enumerate}
    \item There exists an $R$-candle $\mathcal{X}$ in $(G,\Omega)$ such that no $\mathcal{X}$-loop intersects $C_{c+1}$,
    \item $R$ is facial in $(G_{\gamma},\Lambda_{\gamma})$, or
    \item $(G_{\gamma},\Lambda_{\gamma})$ has no access to a red vertex.
\end{enumerate}
Moreover, there exists an algorithm that takes $(G,\Lambda)$, $R$, $\mathcal{C}$, $\rho$, and $\gamma$ as above as input and finds one of the three outcomes in time $\mathbf{O}(|G|^5)$.
\end{lemma}

Notice that the running time of $\mathbf{O}(|G|^5)$ in \zcref{lemma_StrongRedCross} is really due to our greedy way to find an $R$-candle by guessing all four endpoints of the two paths.
The running time can actually be sped up by a lot simply through using \zcref{prop_TwoPathsProper} and the insights gained through the proof of \zcref{lemma_RedCross}.
Given a society with a nest around a vortex, we can check for the existence of a cross in linear time.
If there is a cross and the society has access to red, then we easily create an $R$-candle out of the cross.
Otherwise either the society does not have access to red, or it has a vortex-free rendition in a disc.
In the later case, the proof of \zcref{thm_FacialSocieties} shows that in order to check for the existence of an $R$-candle we can use a simple auxiliary graph and another call to \zcref{prop_TwoPathsProper} to either find the candle of a witness that the society is facial.

\begin{corollary}\label{cor_checkForCandle}
Let $(G,\Lambda)$ be a society where $(G,R)$ is an annotated graph.
Suppose there exists a cylindrical rendition $\rho$ of $(G,\Lambda)$ in a disc with a non-empty nest around the vortex $c_0$.
Then we can find in linear time one of the following three outcomes:
\begin{enumerate}
    \item an $R$-candle in $(G,\Lambda)$,
    \item a vortex-free rendition $\rho'$ for $(G,\Lambda)$ together with an ultimate frontier $\gamma$ witnessing that $R$ is facial in $(G,\Lambda)$, or
    \item the correct conclusion that $(G,\Lambda)$ doe not have access to a red vertex.
\end{enumerate}
\end{corollary}

\paragraph{Snuggling into a nest.}
So far we have been almost entirely focussed on handling candles in societies.
However, in order to progress, we also require some strengthening of the structure given by a nest in a cylindrical rendition.
This is, because later on we will work with societies whose vortex has bounded depth, but we also require the society itself to be of bounded depth.
By pushing a nest of order $s$ as close as possible towards the vortex of depth $d$, it is possible to restrict to the society defined by the inner graph of the outer-most cycle of the nest and to observe that the depth of this society is now at most $\mathbf{O}( s + d)$.
Towards this observation we require some preparation.
\smallskip

Let $(G,\Lambda)$ be a society with a cylindrical rendition $\rho$ in a disc $\Delta$ together with a nest $\mathcal{C}$ around the vortex $c_0$.
We say that $\mathcal{C}$ is \emph{snug} if there does not exist a cycle $C \in \mathcal{C}$ together with a grounded path $P \subseteq G$ of length at least one such that $P$ has both endpoints on $C$, is otherwise disjoint from $\mathcal{C}$, and $P$ is contained in the inner graph of $C$.
In particular, $P$ cannot share edges with any cycle in $\mathcal{C}$.
\smallskip

The following is a variation of the ``cozy nest lemma'' from the work of Gorsky, Seweryn, and Wiederrecht \cite{GorskySW2025Polynomial}.
The main difference is that in their work, Gorsky, Seweryn, and Wiederrecht want the nest to be as far away from the vortex as possible instead of it being as close to the vortex as possible.
Since any disc becomes an annulus when the interior of the unique vortex is removed, this is simply a matter of perspective which does not add any further complications.
Thus, we omit the nearly identical proof here and refer the interested reader to the original paper.

\begin{proposition}[Gorsky, Seweryn, and Wiederrecht \cite{GorskySW2025Polynomial}]\label{prop:SnugNest}
Let $s\geq 1$ be an integer, $(G,\Lambda)$ be a society, $\rho$ be a cylindrical rendition of $(G,\Lambda)$ in a disc with nest $\mathcal{C} = \{ C_1,\dots,C_s\}$.
Then there exists a snug nest $\mathcal{C}' = \{ C_1',\dots,C_s'\}$ of order $s$ in $(G,\Lambda)$ such that the inner graph of $C_s$ contains $\bigcup \mathcal{C}'$.

Moreover, there exists an algorithm that finds $\mathcal{C}$ in time $\mathbf{O}(s||G||^2)$.
\end{proposition}

With this, we are now able to proceed to showing that the society defined by the inner graph of the outer-most cycle of a snug nest has an upper bound on its depth depending only on the depth of the vortex and the size of the nest.

\begin{lemma}\label{lemma_SnugDepthBounded}
Let $s\geq 1$ be an integer, $(G,\Lambda)$ be a society, $\rho$ be a cylindrical rendition of $(G,\Lambda)$ in a disc with snug nest $\mathcal{C} = \{ C_1,\dots,C_s\}$ around the vortex $c_0$ of depth at most $d$.

Then, if $(G',\Lambda')$ is a linear society defined by the inner graph of $C_s$ where $N \coloneqq V(\Lambda') = N(\rho) \cap V(C_s)$ and $\Lambda'$ is a linearisation of a cyclic order of $N$ obtained by tracing along $C_s$, the depth of $(G',\Lambda')$ is at most $(2s + d + 2)$. 
\end{lemma}

\begin{proof}
Let $\Delta$ denote the disc bounded by the trace of $C_s$ and $\rho'$ denote the restriction of $\rho$ to $\Delta$.

Suppose towards a contradiction that there exists a segment $S$ of $\Lambda'$ such that there is a linkage $\mathcal{L}$ of order $2s + d + 3$ from $S$ to $V(\Lambda') \setminus S$.
We may assume that $S$ contains the minimum of $\Lambda'$ since, for the sake of understanding the depth, we should consider $(G,\Omega')$ as a society where $\Omega'$ is the cyclic order obtained from $\Lambda'$ by making the minimum of $\Lambda'$ be the successor of the maximum of $\Lambda'$.

Now let $S'$ be a maximal subsegment of $S$ such that, if $\mathcal{L}' \subseteq \mathcal{L}$ is the set of all paths in $\mathcal{L}$ with one end in $S'$, every path in $\mathcal{L}'$ has an edge in $\sigma(c_0)$.
Nice that, since $\rho'$ is a cylindrical rendition, this means that no path in $\mathcal{L} \setminus \mathcal{L}'$ has an edge in $\sigma(c_0)$.
we order the paths in $\mathcal{L}$ according to the appearance of their endpoints in $S$ with respect to $\Lambda'$ -- this is why we prefer a linear order.
Let $L_1 \in \mathcal{L}'$ be the smallest path in $\mathcal{L}'$ and $L_2$ be the largest w.\@r.\@t.\@ $\Lambda'$.
Let further, for each $L \in \mathcal{L}'$ denote by $L'$ the shortest $S'$-$N(c_0)$-subpath of $L$.
Then there exists a segment $U$ of the boundary of $c_0$ such that $U$ together with the traces of $L_1'$ and $L_2'$ and the subsegment of $S'$ between their endpoints bounds a disc $\Delta'$ that does not contain $c_0$.
We may now observe that every path $L \in \mathcal{L}' \setminus \{ L_1,l_2\}$ must contain a subpath that starts on $N(U)$, ends on $N(c_0) \setminus U$, and is entirely contained within $\sigma(c_0)$.
This is because each such $L$ goes from $S$ to $V(\Lambda')\setminus S$ and in order to reach there it must somehow leave $\Delta'$.
The only way to do so, however, is by passing through $c_0$.
Since the depth of $c_0$ is at most $d$, we now have that $|\mathcal{L}'| \leq d + 2$.

This leaves us with $\mathcal{L}'' \coloneqq \mathcal{L} \setminus \mathcal{L}'$ which now has at least $2s_1$ paths, all of which are edge-disjoint from $\sigma(c_0)$ and thus grounded in $\rho'$.

Let $\gamma$ be a rooted $\rho'$-aligned curve whose endpoints are the two extrema of $S$ and such that $U \subseteq V(\gamma)$ is the unique intersection of $\gamma$ with $N(c_0)$.
It follows that every path in $\mathcal{L}''$ must intersect $\gamma$.
Indeed, $\gamma$ has two disjoint subcurves $\gamma_1$ and $\gamma_2$, each of them joining one extremum of $S$ with an extremum of $U$ and being otherwise disjoint from $U$.

We claim that for each $i\in[2]$, at most $s$ paths of $\mathcal{L}'$ can intersect $\gamma_i$.
To see this, note that there cannot be a path $L \in \mathcal{L}$ such that there exists $j \in [s]$ and $L$ has a grounded subpath $Q$ with both endpoints on $C_j$ such that $Q$ is contained in the inner graph of $C_j$ but disjoint from all other cycles in $\mathcal{C}$.
This is because such a path would contradict our assumption that $\mathcal{C}$ is snug.

Let $\mathcal{L}_1$ be the collection of all paths from $\mathcal{L}$ that intersect $\gamma_1$.
The argument for $\gamma_2$ is analogous and thus we omit it here.
There exists a natural order for the paths in $\mathcal{L}_1$ as follows.
Each path $P \in \mathcal{L}_1$ is grounded and its trace separates $\Delta$ into two discs, precisely one of them contains $c_0$, let $\Delta_P$ be the disc that doesn't contain $c_0$.
Moreover, if $P \neq P'$, then either $\Delta_P \subseteq \Delta_{P'}$ or $\Delta_{P'} \subseteq \Delta_P$.
Thus, we may index $\mathcal{L}_1 = \{ P_1,\dots,P_{\ell}\}$ such that $\Delta_{P_i} \subseteq \Delta_{P_{i+1}}$ for all $i \in [\ell-1]$.

By our observation on the properties of snug nests above, we now know that for each $i\in[\ell]$ there exists $j_i \in [s]$ such that $P_i$ intersects $C_{j_i}$ but no grounded $V(C_{j_i})$-subpath of $P_i$ is contained in the inner graph of $C_{j_i}$.
This also implies that if $i < i' \in[\ell]$, then $j_{i} < j_{i'}$ since $\Delta_{P_i} \subsetneq \Delta_{P_{i'}}$.
Hence, $|\mathcal{L}_1| \leq s$ as claimed.
\smallskip

This, however, means that |$\mathcal{L}| \leq 2s + d + 2$ contradicting our assumption and completing the proof.
\end{proof}

\paragraph{Cleaning a vortex.}
We are finally ready to actually clean a vortex of bounded depth.
That is, we wish to show that for any vortex of bounded depth surrounded by a big enough nest in a blank rendition, we can either find a $k$-candle mesh or delete a small set of vertices such that the remaining society is $k'$-shallow for some $k' \leq k$.
That is, our next goal is to prove the following theorem.

\begin{theorem}\label{thm_CleanVortex}
There exist a functions $f_{\ref{thm_CleanVortex}}\colon \mathbb{N} \to \mathbb{N}$ and $g_{\ref{thm_CleanVortex}} \colon \mathbb{N}^2 \to \mathbb{N}$ such that for all non-negative integers $s,d,k$ and every society $(G,\Lambda)$ where $(G,R)$ is an annotated graph, if there exists a blank cylindrical rendition $\rho$ of $(G,\Lambda)$ in a disc with a nest $\mathcal{C} = \{ C_1,\dots,C_s\}$ around the vortex $c_0$ where $s \geq f_{\ref{thm_CleanVortex}}(k)$, and $c_0$ has depth at most $d$, one of the following is true:
\begin{enumerate}
    \item $(G,R)$ contains a $k$-candle mesh, or
    \item there is a set $A \subseteq V(G)$ of size at most $g_{\ref{thm_CleanVortex}}(d,k)$ such that $(G-A,\Lambda)$ has a $k$-shallow cylindrical rendition in a disc of depth at most $2f_{\ref{thm_CleanVortex}}(k) +d +2$.
\end{enumerate}
Moreover, there exists an algorithm that takes as input a society $(G,\Lambda)$, a set $R \subseteq V(G)$, a blank rendition $\rho$, and a nest $\mathcal{C}$ as above and finds one of the three outcomes in time $\mathbf{O}(|G|^6)$.
And it holds that $f_{\ref{thm_CleanVortex}}(k) \in \poly(k)$, and $g_{\ref{thm_CleanVortex}}(d,s,k) \in \poly(d + k)$.
\end{theorem}

\paragraph{Isolating a vortex segment.}
Let $(G,\Lambda)$ be a society of depth at most $d$ and let $S$ be a segment of $\Omega$.
By definition of depth, there cannot be $d+1$ disjoint paths starting on $S$ and ending on $V(\Lambda) \setminus S$.
Hence, by Menger's Theorem, there exists a set $X_S$ of size at most $d$ such that in $G-X_S$ there does not exist an $S$-$(V(\Lambda) \setminus S)$-path in $G-X_S$.
We say that any minimum size set $X_S \subseteq V(G)$ such that there is no $S$-$(V(\Lambda)\setminus S)$-path in $G-X_S$ is an \emph{$S$-isolator} and note that such a set can be found in time $\mathbf{O}(d \cdot ||G||)$.
\smallskip

Let now $(G,\Lambda)$ be a society of depth at most $d$ where $(G,R)$ is an annotated graph.
Our goal is to iteratively separate segments $S_1,\dots,S_{\ell}$ of $\Lambda$ from their complement.
By doing so, we will be collecting an increasing set $A_i \subseteq V(G)$ where $i$ keeps track of the number of iterations.
Moreover, we are going to maintain a second increasing set $B_i$ which will ``guard'' the ``attachments'' inside each of the $A_i$ as follows:
For all $i \leq j \in[\ell]$, if we consider the $S_i$-crop $(G_i,\Lambda_i)$ of $(G - A_j,\Lambda)$, then there exists an $R$-candle in $(G_i,\Lambda_i)$.
However, $(G_i - B_j,\Lambda_i)$ does not have any $R$-candles.
\medskip

Let $(G,\Lambda)$ be a society of depth at most $d$ where $(G,R)$ is an annotated graph.
A \emph{$h$-segmentation} of $(G,\Lambda)$ is iteratively defined as follows:
Let $\Lambda = \langle v_1,\dots,v_n \rangle$.

We define $S_1$ to be the shortest segment of $\Lambda$ containing $v_1$ such that there is a set $Y_1 \subseteq V(G)$ of size at most $d$ for which
\begin{itemize}
    \item there is no $S_1$-$V(\Lambda)\setminus S_1$-path in $G-Y_1$, and
    \item the $S_1$-crop $(G_1^1,\Lambda_1^1)$ of $(G-Y_1,\Lambda)$ has an $R$-candle.
\end{itemize}
Now let $S_1'$ be the segment of $\Lambda$ obtained from $S_1$ by removing the very last vertex and let $Z_1$ be any $S_1'$-isolator of order at most $d$.
Then, by choice of $S_1$ we know that the $(S_1')$-crop of $(G-(Y_1 \cup Z_1),\Lambda)$ does not contain an $R$-candle.

We set $A_1 \coloneqq Y_1$ and $B_1 \coloneqq Y_1 \cup Z_1$.

Now suppose the first $i \in [h-1]$ segments $S_1,\dots,S_i$ have already been defined together with the sets $Y_j,Z_j,A_j$, and $B_j$ for $j \in[i]$.
We describe how the next iteration is constructed.

Let $S_{i+1}$ be the shortest segment of $\Lambda$ disjoint from $S \coloneqq \bigcup_{j \in [i]}S_j$ starting with the smallest vertex in $\Lambda$ not in $S$ such that there is an $S_{i+1}$-isolator $Y_{i+1}$ in $(G-A_i,\Lambda)$ of size at most $d$ for which the $S_{i+1}$-crop of $(G-(A_i \cup Y_{i+1}),\Lambda)$ contains an $R$-candle.
Let further $S_{i+1}'$ be the segment obtained from $S_{i+1}$ by removing its last vertex.
Then for any $S_{i+1}'$-isolator $Z_{i+1}$ in $(G-A_i,\Lambda)$ of size at most $d$ it holds that the $S_{i+1}'$-crop of $(G-(A_i \cup Y_{i+1} \cup Z_{i+1}),\Lambda)$ contains no $R$-candle.
We set $A_{i+1} \coloneqq A_i \cup Y_{i+1}$ and $B_{i+1} \coloneqq B_i \cup Y_{i+1} \cup Z_{i+1}$.

The \emph{$h$-segmentation} is now the list $S_1,\dots,S_h$ together with the sets $A_h$ and $B_h$.
Moreover, a \emph{segmentation} of $(G,\Lambda)$ is an $h$-segmentation $S_1,\dots,S_h$ together with the sets $A_h$ and $B_h$ such that $\bigcup_{i\in[h]} S_i = V(\Lambda)$.
\smallskip

Since -- as it is well known -- Menger's Theorem for finding either $d$ disjoint paths or a separator of size at most $d$ in time $\mathbf{O}( d\cdot |G|)$, and by \zcref{cor_FindShallowCandles} we can decide whether there exists an $R$-candle in any fixed society $(G',\Lambda')$ in time $\mathbf{O}(|G'|^5)$, we can find an  $h$-segmentation for any society $(G,\Lambda)$ of depth at most $d$ in time $\mathbf{O}(h \cdot d \cdot |G|^5)$.

\begin{observation}\label{obs:FindHSegmentation}
Let $(G,\Lambda)$ be a society of depth at most $d$ where $(G,R)$ is an annotated graph.
Then an $h$-segmentation of $(G,\Lambda)$ can be found in time $\mathbf{O}(h \cdot d \cdot |G|^5)$.
\end{observation}

Finally, we combine all of the above results into a proof of \zcref{thm_CleanVortex}.

\begin{proof}[Proof of \zcref{thm_CleanVortex}]
We start by setting up the two functions.
Let $c$ be the constant from \zcref{cor_FindShallowCandles}.
\begin{align*}
    f_{\ref{thm_CleanVortex}}(k) & \coloneqq c + 3k + 2\\
    g_{\ref{thm_CleanVortex}}(d,k) & \coloneqq (k-1) \cdot (4f_{\ref{thm_CleanVortex}}(k) + 2d + 4)
\end{align*}

We start with a short preprocessing phase.
First, in case $s > f_{\ref{thm_CleanVortex}}(k)$, let $(G'',\Lambda'')$ be the society induced by $\rho$ and the inner graph of $C_{f_{\ref{thm_CleanVortex}}(k) + 1}$.
Otherwise set $(G'',\Lambda'') \coloneqq (G,\Lambda)$.
Let $\mathcal{C}'' = \{ C_1,\dots,C_{f_{\ref{thm_CleanVortex}}(k)}\}$.
Let further $\Delta''$ be the disc bounded by the trace of $C_{f_{\ref{thm_CleanVortex}}(k) + 1}$ and let $\rho''$ be the restriction of $\rho$ to $\Delta''$.

Next, we apply \zcref{prop:SnugNest} to $(G,''\Lambda'')$ and $\mathcal{C}''$ in order to be able to obtain a snug nest $\mathcal{C}' = \{ C'_1,\dots,C_{f_{\ref{thm_CleanVortex}}(k)}'\}$.
Let now $(G',\Lambda')$ be the society induced by $\rho''$ and the inner graph of $C'_{f_{\ref{thm_CleanVortex}}(k)}$.
Let further $\Delta'$ be the disc bounded by the trace of $C'_{f_{\ref{thm_CleanVortex}}(k)}$ and let $\rho'$ be the restriction of $\rho''$ to $\Delta'$.

By \zcref{lemma_SnugDepthBounded} and because $\mathcal{C}'$ is snug, we have that the depth of $(G',\Lambda')$ is at most $\theta \coloneqq 2f_{\ref{thm_CleanVortex}}(k) + d + 2$.

From here we distinguish two cases:
\begin{description}
    \item[Case 1:] $(G',\Lambda')$ has a $k$-segmentation, or
    \item[Case 2:] there exists an $h$-segmentation $\mathcal{S}$ of $(G',\Lambda')$ such that $h \leq k-1$ and $\mathcal{S}$ is a segmentation of $(G',\Lambda')$.
\end{description}
\medskip

\textbf{Case 1:} $(G',\Lambda')$ has a $k$-segmentation.
\smallskip

Let $S_1,\dots,S_{k}$, $A_{k}$, $B_{k}$ be a $k$-segmentation of $(G',\Lambda')$.
For each $i\in[k]$ let $(G_i,\Lambda_i)$ denote the $S_i$-crop of $(G'-A_i,\Lambda')$.
Then, by definition of $k$-segmentations we have that for each $i\in[k]$ there exists an $R$-candle $\mathcal{X}_i$ in $(G_i,\Lambda_i)$.
Since the $G_i$ are pairwise vertex-disjoint, it follows that the $\mathcal{X}_i'$ are so as well.
Let $A' \coloneqq A_k \cap V(\sigma(c_0))$.

Indeed, by the definition of $k$-segmentations and the fact that each $S_i$ was selected under the choice of an $S_i$-isolator $Y_i$ in $G' - A_{i-1}$, we find that we may associate to each $\mathcal{X}_i$ a proper semi $\rho$-aligned curve $\gamma_i$ whose inner graph contains $\mathcal{X}_i$ and whose inner graphs are disjoint except for $V(\gamma_i \cap \gamma_j)$ for all $i \neq j \in[k]$

By \zcref{cor_FindShallowCandles} we may now assume that every $\mathcal{X_i'}$-loop is disjoint from $C_{c+1}$ where $c$ is the constant from \zcref{cor_FindShallowCandles}.
This takes $\mathbf{O}(|G|^5)$ time.

This gives us the right to apply \zcref{lemma_ForcingCandlesIntoAMesh} to the candles $\mathcal{X}_{i}$, $i\in[k]$, which yields a $k$-candle mesh in $(G,R)$ and thus meets the first outcome of our theorem.
\medskip

\textbf{Case 2:} 
There exists an $h$-segmentation $\mathcal{S}$ of $(G',\Lambda')$ such that $h \leq k-1$ and $\mathcal{S}$ is a segmentation of $(G',\Lambda')$.

In this case we have that $|B_{h}| \leq (k-1) \cdot 2\theta$ since $(G',\Lambda')$ has depth at most $\theta$ and $B_{h}$ consists of $h \leq k-1$ pairs of two sets, each of size at most $\theta$.
Since $\mathcal{S}$ is a segmentation of $(G',\Lambda')$ this means that the $S_i$, $i\in[h]$, form a partition of $\Lambda'$ and thus $(G' - B_h,\Lambda')$ does not contain a single $R$-candle.
Indeed, if we denote by $(G_i,\Lambda_i)$ the $S_i$-crop of $(G'-B_h,\Lambda')$, then also for each $i \in [h]$, $(G_i,\Lambda_i)$ does not contain an $R$-candle.
By \zcref{lemma_RedCross} this means that $(G_i,\Lambda_i)$ either has no access to a red vertex, or that $R$ is facial in $(G_i,\Lambda_i)$.
This, however, implies that $(G'-B_h,\Lambda')$ is $h$-shallow and there exists a cylindrical rendition $\varphi$ of depth at most $\theta$ of $(G'-B_h,\Lambda')$ in the disc $\Delta'$.
\medskip

Since $(G',\Lambda')$ was ultimately obtained by ``cropping'' $(G,\Lambda)$ along the trace of $C'_{f_{\ref{thm_CleanVortex}}(k)}$, we may now replace the restriction of $\rho$ to $\Delta'$ with a single cell $c_1$ containing $G'-B_h$ such that $N(c_1)$ is precisely the set of nodes of $C_{f_{\ref{thm_CleanVortex}}(k)}'$.
Thereby, we obtain a cylindrical rendition $\varphi'$ of $(G-B_h,\Lambda)$ of depth at most $\theta$ which is $h$-shallow.
\end{proof}

\paragraph{The proof of \zcref{thm_localstructure}.}
The proof of \zcref{thm_localstructure} is now relatively straightforward.
We simply call upon \zcref{prop:localstructure} where we ask the nests of the vortices to be at least of order $f_{\ref{thm_CleanVortex}}(k)$ each.
Then we apply \zcref{thm_CleanVortex} to each vortex together with its nest in the resulting $\Sigma$-rendition $\rho$.
Each such application to a vortex $v$ yields a set $B_v$ of size at most $g_{\ref{thm_CleanVortex}}(\mathsf{depth}_{\ref{prop:localstructure}}(f_{\ref{thm_CleanVortex}}(k), t, r),k) \in \poly(k + t + r)$ which we add to the apex set.
Since there are only $\mathbf{O}(t^2 + r^2)$ vortices, this increases the size of the apex set by only a polynomial amount.
Adjusting $\rho$ by inserting the new renditions for the societies of the inner graphs of the outer-most cycles from the nests of vortices that were created by \zcref{thm_CleanVortex} finally yields the desired $k$-shallow $\Sigma$-rendition.

\section{Annotated graphs with vital solutions}\label{sec_VitalTheorem}

The goal of this section is to prove a variant of the Vital Linkage Function function for the $k$\textsc{-Spanning Disjoint Paths} problem.
That is, we want to show that instances of $k$\textsc{-SDP} with a unique solution that uses all vertices must either have large $\mathsf{depth}_2$ or bounded treewidth.

We say that an instance $(G,R,\mathcal{T} = \{ (s_1,t_1),\dots,(s_k,t_k) \})$ of the $k$-\textsc{Spanning Disjoint Paths} problem is \emph{vital} if there exists a unique linkage $\mathcal{L} = \{ L_1,\dots,L_k\}$ where $L_i$ joins $s_i$ and $t_i$ for each $i \in[k]$ and $R \subseteq V(\bigcup \mathcal{L})$ in $G$ and $V(G) = V(\mathcal{L})$.

\begin{theorem}\label{thm_VitalSpanningLinkage}
There exists a function $\beta_{\ref{thm_VitalSpanningLinkage}} \colon \mathbb{N}^2 \to \mathbb{N}$ such that for every $r,k \geq 0$, if $(G,R,\mathcal{T})$ is a vital instance of the $k$\textsc{-Spanning Disjoint Paths} problem where $\mathsf{depth}_2(G) \leq r$, then $\mathsf{tw}(G) \leq \beta_{\ref{thm_VitalSpanningLinkage}}(r,k)$.
Moreover, $\beta_{\ref{thm_VitalSpanningLinkage}}(r,k) \in 2^{\poly(r + k)}$.
\end{theorem}

The strategy to prove \zcref{thm_VitalSpanningLinkage} follows along three main steps.
The first two of these steps are precisely the same as already established by Robertson and Seymour:
Let $(G,R,\mathcal{T})$ be a vital instance of $k$\textsc{-SDP} with $\mathsf{depth}_2$ at most $r$, then
\begin{description}
    \item[Step 1:] We use a result of Protopapas, Thilikos, and Wiederrecht \cite{ProtopapasThilikosWiederrecht2025ColorfulMinors} together with a result of Robertson and Seymour \cite{RobertsonSeymour1995DisjointPaths} in order to prove that $G$ cannot contain a clique minor of size $\Omega(r)$ as otherwise there would be an irrelevant vertex.
    \item[Step 2:] In the absence of big clique minors, but when the graph has large treewidth, we may apply \zcref{thm_localstructure}.
    Moreover, the assumption of vitality allows us to completely remove all apices.
    \item[Step 3:] Finally, we arrive in an apex-free situation. Here, we now proceed in two substeps:
    First, we show that we only need to care about a set $R' \subseteq R$ of size $\poly( r + k )$ and may remove all other red vertices.
    Second: This will then allow us to interpret $(G',R',\mathcal{T})$ as a vital instance $(G',\mathcal{T}')$ of $\poly( r + k )$\textsc{-Disjoint Paths}.
    Here $G'$ is obtained through a series of small operations in order to reduce $R$ to $R'$.
    Hence, in this last step, \zcref{prop_DisjointPathsIrrelevant} shows that $G$ must have small treewidth as desired.
\end{description}

Notice that the most technical part of this roadmap by far is \textbf{Step 3}.
In this step we need to introduce new tools to deal with spanning linkages and also implement several technical statements from the work on vital linkages by Cavallaro, Gorsky, Kreutzer, Thilikos, and Wiederrecht \cite{CavallaroGKTW2026Optimal} in our setting.
In most cases, we will actually make much stronger assumptions on our instance and  in a fourth and final step we will argue why these assumptions were justified.
Indeed, during \textbf{Step 3} we will assume that the entire graph consists only of the linkage $\mathcal{L}$ itself and a collection of cycles embedded in a concentric way in some disk in the surface.
Our main goal will be to show that the inner-most cycle of this collection is disjoint from $\mathcal{L}$ which will then be the final contradiction to the assumption of vitality.

\subsection{What to do with a clique}\label{subsec_cliqueMinor}
As mentioned before, the original statements for irrelevant vertices and the \textsc{$k$-Disjoint Paths} problem are often given for a more general problem known as the $(k,d)$\textsc{-Folio} problem.
In essence, $(k,d)$\textsc{-Folio} is a common generalisation of \textsc{Minor Containment} and $k$\textsc{-Disjoint Paths} that allows for the minors to be partially rooted at fixed vertices.
Since $(k,d)$\textsc{-Folio} is much more general than what is necessary for our purposes, we simplify the statements sightly.
In this sense we provide a slight weakening of the folio definition in order to state the statements required to handle the case of large clique minors.
\smallskip

Let $G$ be a graph and $Z \subseteq V(G)$ be a set of vertices.
The \emph{weak $k$-folio relative to $Z$}, denoted by $k\text{-}\mathsf{WFolio}_G(Z)$, is the collection of all sets $\{ s_1,\dots,s_{k'},t_1,\dots,t_{k'} \} \subseteq Z$ of at most $k$ vertex pairs $\{ (s_1,t_1),\dots,(s_{k'},t_{k'}) \}$, $k' \leq k$, such that $(G,\mathcal{T})$ is a \textsf{yes}-instance of $k$\textsc{-Disjoint Paths} for every $\mathcal{T} \in k\text{-}\mathsf{WFolio}_G(Z)$ .

A vertex $v \in V(G)$ is said to be \emph{irrelevant} for the weak $k$-folio relative to $Z$ if and only if $k\text{-}\mathsf{WFplio}_{G-v}(Z) = k\text{-}\mathsf{WFplio}_{G}(Z)$.
\smallskip
In in particular, it is worthwhile pointing out that the weak $k$-folio relative to $Z$ defined above is essentially the special case where the detail $d = 0$ in the original statement.

Under these definitions, the following theorem of Robertson and Seymour acts as the backbone of our engine for handling clique minors.

\begin{proposition}[Robertson and Seymour \cite{RobertsonSeymour1995DisjointPaths}]\label{thm_cliqueirrelevantvertex}
 Let $G$ be a graph, let $k \in \mathbb{N}$ be a non-negative integer, and let $Z \subseteq V(G)$ with $|Z| \leq k$.
 Further let $\theta \coloneqq \lfloor \nicefrac{5}{2} \cdot k \rfloor + 1$, let $\varphi$ be a minor model of $K_{\theta}$ in $G$, and let $(A,B)$ be a separation in $G$ such that
 \begin{enumerate}
 \item \label{item:cliqueright} $A \cap \varphi(v) = \emptyset$ for some $v \in V(K_{\theta})$,
 \item \label{item:terminalleft} $Z \subseteq A$,
 \item \label{item:smallsep} subject to \zcref{item:cliqueright} and \zcref{item:terminalleft}, $(A,B)$ is of minimum order, and
 \item subject to \zcref{item:cliqueright}, \zcref{item:terminalleft}, and \zcref{item:smallsep}, $A$ is maximal.
 \end{enumerate}
 Then for any $v \in B \setminus A$, $v$ is irrelevant for the weak $k$-folio of $G$ relative to $Z$.

 In particular, given $G$, $Z$, and $\varphi$ as above, such a vertex $v$ can be found in $\poly(k)||H||$-time.
\end{proposition}

We will also require the following result from \cite{ProtopapasThilikosWiederrecht2025ColorfulMinors} which refines some of the tools from \cite{RobertsonSeymour1995DisjointPaths}.

\begin{proposition}[Protopapas, Thilikos, and Wiederrecht \cite{ProtopapasThilikosWiederrecht2025ColorfulMinors}]\label{thm:colorfulclique}
 Let $t,k \in \mathbb{N}$ be positive integers with $k \geq \lfloor \nicefrac{3}{2} \cdot q t \rfloor + t$.
 Let $(G,R)$ be an annotated graph such that $G$ contains a minor model $\varphi$ of $K_k$.
 Then one of the following is true:
 \begin{enumerate}
 \item There exists a red-minor model $\psi$ of a red $K_t$ in $(G,R)$ such that $\mathcal{T}_\psi$ is a truncation of $\mathcal{T}_\varphi$, or

 \item There exists a set $S \subseteq V(G)$ of size at most $t - 1$ such that the $\mathcal{T}_\varphi$-big component of $G - S$ is blank.
 \end{enumerate}
 Furthermore, there exists an algorithm that takes as input $t$, $(G,\chi)$, and $\varphi$ as above and finds one of the two outcomes in time $\poly(t)|\!|G|\!|$.
\end{proposition}

With this, we are ready to state the main lemma of this subsection.

\begin{lemma}\label{lemma_WahtToDoWithClique}
Let $r,k,t \geq 1$ be positive integers where $t \geq \lfloor \nicefrac{5}{2} ( 2(r+1)^2 + 2k ) \rfloor + 1 $, let $(G,R)$ be an annotated graph with $\mathsf{depth}_2(G,R) \leq r$, and $\mathcal{T} = \{ (s_1,t_1),\dots,(s_k,t_k)\}$ such that $(G,R,\mathcal{T})$ is an instance of $k$\textsc{-Spanning Disjoint Paths}.
If there exists a minor model $\varphi$ of $K_t$ in $G$, then there exists $v \in V(G) \setminus R$ such that $v$ is irrelevant for $(G,R,\mathcal{T})$.

Moreover, there exists an algorithm that takes as input $(G,R,\mathcal{T})$ and $\varphi$ as above and finds the irrelevant vertex $v$ in time $\poly( r + k )|G|$.
\end{lemma}

\begin{proof}
Let us first note that we may assume that $s_i \neq s_j$, $t_i \neq t_j$, and $s_i \neq t_j$ for all $ i \neq j \in[k']$.
For this, whenever we have, say $s_i = s_j$ for $i \neq j$, we simply add a copy $s_i'$ of $s_i$ to the graph.
That is, $s_i'$ is a fresh vertex with neighbourhood precisely $N_G(s_i)$.
We then replace $s_j$ by $s_i'$ in the pair $(s_j,t_j)$.

Next, we apply \zcref{thm:colorfulclique} to $(G,R)$, looking for a red $r^2$-clique minor.
In case we find this clique minor, it is clear that $(G,R)$ also contains the $2$-outer-annotated $((r+1) \times (r+1))$-grid as a red-minor, thereby implying that $\mathsf{depth}_2(G,R) \geq r+1$ which is a contradiction to our assumption.
Hence, we must find the second outcome of \zcref{thm:colorfulclique}:
A set $S \subseteq V(G)$ of size at most $(r+1)^2 - 1$ such that the $\mathcal{T}_{\varphi}$-big component of $G-S$ is blank.

Let now $G'$ be obtained from $G$ by duplicating every vertex of $S$ precisely once and let $S'$ be the set consisting of all members of $S$ and their duplicates.
Finally, let $H'$ be the $\mathcal{T}_{\varphi}$-big component of $G' - S'$ and let $H \coloneqq G'[V(H') \cup S']$.
Finally, let $Z \coloneqq (\{ s_1,\dots,s_{k},t_1,\dots,t_{k}\} \cap V(H)) \cup S'$.
Then $|Z| \leq 2(r+1)^2 + 2k$.

Let $\mathcal{L}$ be any solution for $(G,R,\mathcal{T})$.
Then every component of $H \cap (\bigcup \mathcal{L})$ is a path between a pair of distinct vertices from $Z$In particular, if some path from $H \cap (\bigcup \mathcal{L})$ contains a vertex $s$ from $S$ as an internal vertex, we split this path into two, one ending in $s$ and the other starting in the copy $s'$ of $S$ introduced in the previous step. 
Moreover, any two of these paths are vertex-disjoint.
It follows that any solution for $(G,R,\mathcal{T})$ induces an instance of $k'$\textsc{-Disjoint Paths} on $H$ where $k' \leq (r + 1)^2 + k$ and all terminals are chosen from $Z$.
Hence, if a vertex $v \in V(H)$ is irrelevant for the weak $((r + 1)^2 + k)$-folio relative to $Z$, then it is also irrelevant for $(G,R,\mathcal{T})$.

By applying \zcref{thm_cliqueirrelevantvertex}, not only do we know that such a vertex exists, we can also find it in the required time.
\end{proof}

\subsection{Desolate instances and their loci}\label{subsec_locus}
From here on, we are mostly interested in a special kind of instances for the $k$\textsc{-Spanning Disjoint Paths} problem.
These are instances that adhere to \zcref{thm_localstructure} while meeting additional requirements.
We refer to such instances as ``desolate'' and we will later on see, that any vital instance can be transformed into a desolate one.
Before we introduce the full definition, we require some additional terminology.

\paragraph{Snug cycles.}
Let $(G,\Omega)$ be a society with a vortex-free rendition $\rho$ in a disc $\Delta$ and let $\mathcal{C} = \{ C_1,\dots,C_s\}$ be a collection of cycles in $G$.
We say that $\mathcal{C}$ is \emph{concentric} in $\rho$ if
\begin{enumerate}
    \item $C_i$ is grounded in $\rho$ for each $i\in[s]$, and
    \item if $\Delta_i$ denotes the disc bounded by the trace of $C_i$ for each $i\in[s]$, then $\Delta_1 \subsetneq \Delta_2 \subseteq \dots \subsetneq \Delta_s$.
\end{enumerate}
We say that a set $\mathcal{C} = \{ C_1,\dots,C_s\}$ of concentric cycles in $\rho$ is \emph{snug}, if
\begin{enumerate}
    \item the inner graph of $C_1$ has no grounded cycle apart from $C_1$ itself, and
    \item if $\rho'$ is the cylindrical rendition of $(G,\Omega)$ obtained by removing all cells in $\Delta_1$, and declaring $\Delta_1 \setminus N(C_1)$ a vortex, then $\{ C_2,\dots,C_s\}$ is snug in $\rho'$.
\end{enumerate}
Notice that, due to \zcref{prop:SnugNest}, we know that, starting with any set of concentric cycles $\mathcal{C} = \{ C_1,\dots,C_s\}$ where $\Delta_s$ denotes the disc bounded by the trace of $C_s$, we can find a snug set $\mathcal{C'_1,\dots,C'_s}$ of concentric cycles where, if $\Delta'_s$ denotes the disc bounded by the trace of $C'_s$, then $\Delta'_s \subseteq \Delta_s$.

Let now $(G,R,\mathcal{T})$ be an instance of the $k$\textsc{-Spanning Disjoint Paths} problem and $\rho$ be a blank $\Sigma$-rendition of $G$.
We say that a collection $\mathcal{C} = \{ C_1,\dots,C_s\}$ is \emph{snug} in $\rho$ if
\begin{enumerate}
    \item $C_s$ is grounded in $\rho$ and its trace bounds the disc $\Delta_s$,
    \item $\Delta_s$ contains no vortex of $\rho$, and the inner graph of $\Delta_s$ is disjoint from $\mathcal{T}$, and
    \item if $\rho'$ is the restriction of $\rho$ to $\Delta_s$ and $(G',\Omega')$ is the society where $G'$ is the inner graph of $\Delta_s$ and $\Omega'$ is the cyclic ordering of the nodes of $C_s$ obtained by tracing along the boundary of $\Delta_s$ in clockwise direction, then $\mathcal{C}$ is snug in $\rho'$.
\end{enumerate}

\paragraph{Dry wells.}
Let $s$ be a positive integer.
An \emph{$s$-well} $\mathcal{W} = (G, \Omega, \rho, \mathcal{C}, \mathcal{P})$ is a society $(G,\Omega)$ with a vortex-free, cylindrical rendition $\rho$ in a disc $\Delta$ such that $G = \bigcup \mathcal{C} \cup \bigcup \mathcal{P}$ for a set of pairwise internally disjoint $V(\Omega)$-paths $\mathcal{P}$ and $\mathcal{C} = \{ C_1, \ldots, C_s \}$ is a snug set of cycles.

Since $\mathcal{C}$ is snug, we know that for each $P \in \mathcal{P}$, we may let $T_P$ be the trace of $P$ in $\rho$ and there exists a unique component $\Delta_P$ of $\Delta - T_P$ such that the intersection of $\Delta_P$ and the trace of $C_1$ is empty.
We call the closure of $\Delta_P$ the \emph{interior of $P$}.
See \zcref{fig_Well} for an illustration of a well.

\begin{figure}[ht]
 \centering
 \begin{tikzpicture}

 \pgfdeclarelayer{background}
		\pgfdeclarelayer{foreground}
			
		\pgfsetlayers{background,main,foreground}

 \begin{pgfonlayer}{background}
 \pgftext{\includegraphics[width=8cm]{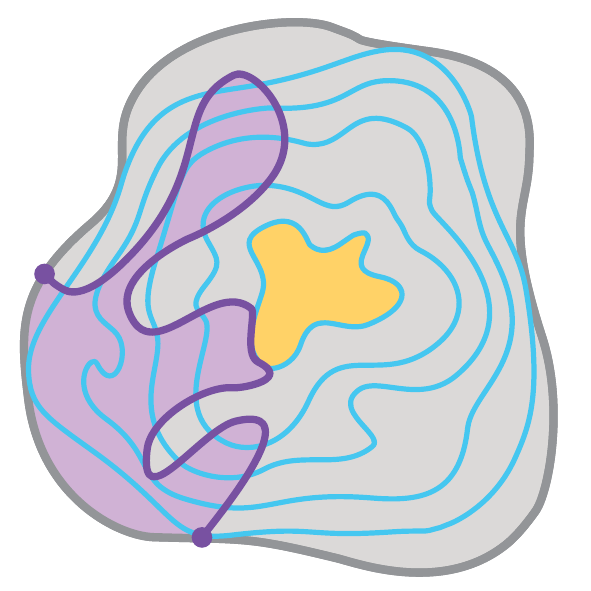}} at (C.center);
 \end{pgfonlayer}{background}
			
 \begin{pgfonlayer}{main}
 \node (C) [v:ghost] {};
 
 \end{pgfonlayer}{main}
 
 \begin{pgfonlayer}{foreground}
 \end{pgfonlayer}{foreground}

 \end{tikzpicture}
 \caption{An example of a well where cycles are depicted in \textcolor{CornflowerBlue}{blue} and the \textcolor{Amethyst}{purple} path $P$ marks an example of a path in the well. The shaded area is the interior of $P$ while the \textcolor{DarkBananaYellow}{yellow} are is the interior of the trace of $C_1$.}
 \label{fig_Well}
\end{figure}

We say that an $s$-well $\mathcal{W} = (G, \Omega, \rho, \mathcal{C}, \mathcal{P})$ is \emph{drained} if either $|\mathcal{P}| \leq 1$ or
\begin{itemize}
 \item for each $P \in \mathcal{P}$ we have that for all $i \in [s-1]$, if $P \cap C_i$ is non-empty, then there exists $Q \in \mathcal{P} \setminus \{ P \}$ such that the trace of $Q \cap C_{i+1}$ is non-empty and found in $\Delta_P$, and
 
 \item for all $i \in [s]$, no $C_i$-path $P'$ in $P$ that is disjoint from $V(\mathcal{C} \setminus \{ C_i \})$ is contained in a cell $c \in C(\rho)$ such that $\sigma(c)$ contains an edge of $C_i$.
\end{itemize}

The core property of drained wells is, that we know for any path $P$ of $\mathcal{P}$ that enters deeply into the well, it cannot be pushed closer to the part of the boundary of the well contained in $\Delta_P$ because it is blocked by an entire sequence of paths that exhaust the cycles of the well that lie further outside.

\begin{proposition}[Cavallaro, Gorsky, Kreutzer, Thilikos, and Wiederrecht \cite{CavallaroGKTW2026Optimal}]\label{prop_fullbucketsindrainedwells}
Let $\mathcal{W} = (G, \Omega, \rho, \mathcal{C}, \mathcal{P})$ be a drained $s$-well and let $P \in \mathcal{P}$ be a path such that $C_i \cap P$ is non-empty for some $i \in [s]$.
Then there exist $(s - i) + 1$ distinct paths $P_s, P_{s-1}, \ldots, P_{i+1}, P_i = P$ such that $P_j \cap C_j$ is non-empty for all $j \in [i,s]$ and we have $\Delta_{P_{i+1}} \subsetneq \ldots \subsetneq \Delta_{P_{s}} $.
\end{proposition}

Indeed, simply draining a well is not strong enough for our purposes, we require a well to entirely ``dry''.
Indeed, the fact that the cycles of a well form a snug family allows us to deduce a set of very strong properties for the interaction between the paths in $\mathcal{P}$ and the cycles in $\mathcal{C}$.

We say that a well $\mathcal{W} = (G, \Omega, \rho, \mathcal{C} = \{ C_1, \ldots, C_s\}, \mathcal{P})$ is \emph{dry} if $\mathcal{W}$ is drained and for all $P \in \mathcal{P}$ we have that
\begin{enumerate}
 \item the graph $C_1 \cap P$ is either empty or a single path,

 \item there exists a unique $i \in [s]$ such that $C_i \cap P$ is a single path,

 \item for all $j \in [i-1]$ the graph $C_j \cap P$ is empty, and

 \item for all $j \in [i+1,s]$ the graph $C_j \cap P$ consists of exactly two paths.
\end{enumerate}

\begin{proposition}[Cavallaro, Gorsky, Kreutzer, Thilikos, and Wiederrecht \cite{CavallaroGKTW2026Optimal}]\label{prop_drywell}
 Let $s$ be a positive integer and let $\mathcal{W} = (G, \Omega, \rho, \mathcal{C}, \mathcal{P})$ be an $s$-well.
 Then there exists a dry $s$-well $\mathcal{W}' = (G', \Omega, \rho', \mathcal{C}, \mathcal{P}')$ such that $G' \subseteq G$, $\rho'$ is the rendition of $G'$ induced by $\rho$, $|\mathcal{P}| = |\mathcal{P}'|$, and for each $P \in \mathcal{P}$ there exists a $P' \in \mathcal{P}'$ with the same endpoints.
\end{proposition}

\paragraph{Desolate instances.}
We are now ready to give a full list of the properties we wish to assume for \textbf{Step 2} the first part of \textbf{Step 3}.

Let $G$ be a graph $\mathcal{L}$ be a linkage in $G$, and $H \subseteq G$ be a subgraph of $G$.
We denote by $\mathcal{L}[H]$ the linkage $\mathcal{L}'$ consisting of the collection of all maximal subpaths $L'$ of paths in $\mathcal{L}$ that are entirely contained in $L$.

Let $g,s,d$, and $b$ be non-negative integers.
Let $\mathcal{G} = (G,R,\mathcal{T})$ be a vital instance of $k$\textsc{-Spanning Disjoint Paths}.
We say that $\mathcal{G}$ is \emph{$(g,s,d,b,w)$-desolate} if there exist
\begin{itemize}
    \item a surface $\Sigma$ of Euler-genus at most $g$,
    \item an $s$-shallow $\Sigma$-rendition $\rho$ of $G$ with at most $b$ vortices, all of which have depth at most $d$, and
    \item a collection $\mathcal{C} = \{ C_1,\dots,C_w\}$ of vertex-disjoint cycles in $G$
\end{itemize}
such that, if $\mathcal{L}$ denotes the unique solution of the vital instance $\mathcal{G}$,
\begin{enumerate}
    \item $G = \bigcup \mathcal{C} \cup \bigcup \mathcal{L}$
    \item $\mathcal{C}$ is snug in $\rho$,
    \item no two vortices of $\rho$ intersect on their boundaries,
    \item every terminal vertex is either contained in $\sigma(c)$ for some vortex $c$ of $\rho$ or belongs to $N(\rho)$, and
    \item if $\Delta_w$ denotes the disc bounded by the trace of $C_w$ and $(G_w,\Omega_w)$ is the society obtained from the inner graph $G_w$ of $C_w$ where $\rho_w$ denotes the restriction of $\rho$ to $\Delta_w$ and $G_w$, then $(G_w,\Omega_w,\rho_w,\mathcal{C},\mathcal{L}[G_w])$ is a dry $w$-well.
\end{enumerate}
We say that the rendition $\rho$ \emph{witnesses} the desolation of $(G,R,\mathcal{T})$.

With this, we finally have the correct terminology to talk about the type of instances we are interested in.
The intermediate goal is now to show that and $(g,s,d,b,w)$-desolate instance $(G,R,\mathcal{T})$ of $k$\textsc{-Spanning Disjoint Paths} can be transformed into a $(g,s,d,b,w)$-desolate instance $(G',R',\mathcal{T}')$ of $\poly(g + b + k)$\textsc{-Spanning Disjoint Paths} where every vertex of $R'$ is also a terminal.
Since this will imply an upper bound on $|R'|$ it means that we may completely ignore the ``spanning'' aspect of the problem.
We will use this observation to deduce that, in this case, $w$ must be bounded by $2^{\poly(g + b + k)}$.

\subsection{Bringing desolation to a vital instance}\label{subsec_apexRemoval}
We start by showing that any vital instance of large enough treewidth can be reduced to one that is desolate.
This is, in some sense, the most crucial part of the proof and it is also where we get rid of apices, which is a fundamental step that allows our future arguments.

\begin{lemma}\label{lemma_theDesolationOfAnInstance}
There exist functions $\mathsf{genus}_{\ref{lemma_theDesolationOfAnInstance}},\mathsf{depth}_{\ref{lemma_theDesolationOfAnInstance}},\mathsf{breadth}_{\ref{lemma_theDesolationOfAnInstance}},\mathsf{terminals}_{\ref{lemma_theDesolationOfAnInstance}} \colon \mathbb{N}^2 \to \mathbb{N}$ and $\mathsf{treewidth}_{\ref{lemma_theDesolationOfAnInstance}} \colon \mathbb{N}^3 \to \mathbb{N}$ such that, for every vital instance $\mathcal{G} = (G,R,\mathcal{T})$ of $k$-\textsc{Spanning Disjoint Paths} where $\mathsf{depth}_2(G,R) \leq r$ and every integer $w \geq 1$, one of the following is true
\begin{enumerate}
    \item $\mathsf{tw}(G) \leq \mathsf{treewidth}_{\ref{lemma_theDesolationOfAnInstance}}(r,k,w)$, or
    \item there exists a $(\mathsf{genus}_{\ref{lemma_theDesolationOfAnInstance}}(r,k),r,\mathsf{depth}_{\ref{lemma_theDesolationOfAnInstance}}(r,k),\mathsf{breadth}_{\ref{lemma_theDesolationOfAnInstance}}(r,k),w)$-desolate instance $\mathcal{G}'$ of $k'$\textsc{-Spanning Disjoint Paths} where $k' \leq \mathsf{terminals}_{\ref{lemma_theDesolationOfAnInstance}}(r,k)$.
\end{enumerate}
Moreover, it holds that $\mathsf{genus}_{\ref{lemma_theDesolationOfAnInstance}}(r,k),\mathsf{depth}_{\ref{lemma_theDesolationOfAnInstance}}(r,k),\mathsf{breadth}_{\ref{lemma_theDesolationOfAnInstance}}(r,k),\mathsf{terminals}_{\ref{lemma_theDesolationOfAnInstance}}(r,k) \in \poly(r + k)$, and $\mathsf{treewidth}_{\ref{lemma_theDesolationOfAnInstance}}(r,k,w) \in \poly(r + k + w)$.
\end{lemma}

\begin{proof}
We start by giving some insight to the functions involved.
Let 
\begin{align*}
    \theta(r,k) \coloneqq \left\lfloor \frac{5}{2}(2(r+1)^2 + 2k) \right\rfloor + 1.
\end{align*}
We set
\begin{align*}
    \mathsf{genus}_{\ref{lemma_theDesolationOfAnInstance}}(r,k) & \coloneqq 9\theta(r,k)^2 - 1\\
    \mathsf{depth}_{\ref{lemma_theDesolationOfAnInstance}}(r,k) & \coloneqq \mathsf{depth}_{\ref{thm_localstructure}}(\theta(r,k), r)\\ 
    \mathsf{breadth}_{\ref{lemma_theDesolationOfAnInstance}}(r,k) & \coloneqq \nicefrac{3}{2} \cdot ( \theta(r,k) -1) \cdot (3 \theta(r,k) -4) + r(r-1)-3\\
    \mathsf{terminals}_{\ref{lemma_theDesolationOfAnInstance}}(r,k) & \coloneqq k + 2 \cdot \mathsf{apex}_{\ref{thm_localstructure}}(\theta(r,k) , r)\\
    \mathsf{treewidth}_{\ref{lemma_theDesolationOfAnInstance}}(r,k,w) & \coloneqq \mathsf{mesh}_{\ref{prop_GridThm}} \left( \mathsf{mesh}_{\ref{thm_localstructure}}(r, \theta(r,k), \left\lceil \sqrt{2\mathsf{apex}_{\ref{thm_localstructure}}(\theta(r,k) , r) + 2k + 1} \right\rceil \cdot 2(w+2)) \right).
\end{align*}

Now, we first observe that, under the assumption that $\mathcal{G}$ is vital, due to \zcref{lemma_WahtToDoWithClique} we know that $G$ does not contain $K_{\theta(r,k)}$ as a minor.

In case $\mathsf{tw}(G) \leq \mathsf{treewidth}_{\ref{lemma_theDesolationOfAnInstance}}(r,k,w)$ we are done immediately, so we may assume that $\mathsf{tw}(G) > \mathsf{treewidth}_{\ref{lemma_theDesolationOfAnInstance}}(r,k,w)$.
Hence, by an application of \zcref{prop_GridThm}, we get that there exists a $\mathsf{mesh}_{\ref{thm_localstructure}}(r, \theta(r,k), 2 \cdot \left\lceil \sqrt{2\cdot \mathsf{apex}_{\ref{thm_localstructure}}(\theta(r,k) , r) + 2k + 1} \right\rceil \cdot (w+2))$-mesh $M_0$ in $G$.
Moreover, by our assumption we have that $\mathsf{depth}_2(G,R) \leq r$.

So now, if we apply \zcref{thm_localstructure} to $(G,R)$ and $M_0$, we know that the first two outcomes of \zcref{thm_localstructure} are impossible since $G$ does not contain a $K_{\theta(r,k)}$-minor.
Thus, there exists a set $A \subseteq V(G)$ with $|A| \leq \mathsf{apex}_{\ref{thm_localstructure}}(\theta(r,k), r)$ and a surface $\Sigma$ of genus less than $9\theta(r,k)^2$ such that $(G - A, R \setminus A)$ has a blank and $r$-shallow $\Sigma$-rendition $\rho$ with breadth at most $\nicefrac{3}{2} \cdot (\theta(r,k)-1) \cdot (3\theta(r,k)-4) + r(r-1)-3$, depth at most $\mathsf{depth}_{\ref{thm_localstructure}}(\theta(r,k), r)$, and there exists a vortex-free, $\rho$-aligned disc $\Delta \subseteq \Sigma$ such that the restriction of $\rho$ to $\Delta$ contains a flat $2 \cdot \left\lceil \sqrt{2\cdot \mathsf{apex}_{\ref{thm_localstructure}}(\theta(r,k) , r) + 1} \right\rceil \cdot (w+2)$-wall $M_1$.
\smallskip

Since $\mathcal{G}$ is a vital instance, there exists a unique $\mathcal{T}$-linkage $\mathcal{L}$ in $G$.
Let $G_1$ be obtained from $G$ by deleting every edge $e$ incident with a vertex of $A$ but not in $E(\mathcal{L})$.
It follows that every vertex of $A$ has either degree $2$ or degree $1$ if it is a vertex of $\mathcal{T}$.
We return to the apex set in a moment, but first, we fix our attention to the flat wall $M_1$.

Notice that $M_1$ contains $2\cdot \mathsf{apex}_{\ref{thm_localstructure}}(\theta(r,k) , r) + 2k + 1$ subwalls, each of order $2(w+2)$ such that the compasses of these subwalls are pairwise vertex-disjoint.
It follows, that there must exist one of those subwalls, let us call it $M_2$, such that its compass does not contain a neighbour of any vertex in $A$ in $G_1$, nor a terminal from $\mathcal{T}$.
Let $\mathcal{C}_0 = \{ C^0_1,\dots,c^0_{w+2}\}$ be the collection of cycles obtained by iteratively choosing $C_{w+3 - j}$ to be the perimeter of the subwall $M_2^j$ obtained by removing the perimeter from $M_2^{j-1}$ for all $j \in [w+2]$ where $M_2^0 \coloneqq M_2$.
Then, any subpath of a path in $\mathcal{L}$ that starts in a neighbour of $A$ in $G_1$ and contains a vertex of $C_{w+1}$ must contain a node of $C_{w+2}$.

Next, consider the disc $\Delta^1$ bounded by the trace of $C_{w+2}$ in $\Sigma$.
By applying \zcref{prop:SnugNest} to the cycles $\{ C_1^0,\dots,C_{w+1}^0\}$, we obtain a new collection $\mathcal{C}_1 = \{ C^1_1,\dots,C^1_{w+1}\}$ of concentric cycles that are snug in $\rho$ and whose traces are contained within $\Delta^1$.
\smallskip

Let now $G_2$ be obtained from $G_1$ by deleting all edges that do not belong to $E(\mathcal{L})$ or $E(\mathcal{C}_1)$.
Let $\rho_1$ be the restriction of $\rho$ to $G_2$.
Due to \zcref{thm_localstructure} and our construction above, we already have that
\begin{enumerate}
    \item $G_2 = \bigcup \mathcal{L} \cup \bigcup \mathcal{C}_1$,
    \item $\mathcal{C}_1$ is snug in $\rho_1$, and
    \item no two vortices of $\rho_1$ intersect on their boundaries.
\end{enumerate}
However, we still have apices, some terminals might be in the interior of non-vortex cells, and we need to ensure that we can use $\mathcal{C}_1$ to create a dry well.
\smallskip

Let $A = \{ a_1,\dots,a_{\ell}\}$.
We iterate over all $i \in [\ell]$ and perform the following actions.
For each $i \in [\ell]$ let $L \in \mathcal{L}$ be the path containing $a_i$ and let $x$ and $y$ be the endpoints of $L$.
Moreover, for both $z \in \{ x,y\}$ let $L_z$ be the component of $L - a_i$ containing $z$ and let $z'$ be the other endpoint of $L_z$.
Then $z'$ is a neighbour of $a_i$.
Note that, in case $a_i$ is an endpoint of $L$, one -- or both -- of $x$, $y$ might not exist.
We then delete $a_i$ from $G_2$ and replace $(x,y) \in \mathcal{T}$ with the (up to) two terminal pairs $(x,x')$ and $(y,y')$.
Moreover, we replace $L$ in $\mathcal{L}$ with the paths $L_x$ and $L_y$.

Let $G_3 = G_2 - A$ and let $\mathcal{T}_1$ and $\mathcal{L}_1$ be the set of terminal pairs and the corresponding linkage in $G_3$ after all $\ell$ iterations of the above process have been completed.
Notice that $(G_3,R\setminus A,\mathcal{T}_1)$ is still a vital instance.

We have now reached a point where $\rho_1$ is a $\Sigma$-rendition of $G_3$ and no longer have any apex vertices.
Moreover, note that $|\mathcal{T}_1| \leq |\mathcal{T}| + 2|A|$.
\smallskip

Next, let $x$ be any terminal from a pair $(x,y) \in \mathcal{T}_1$ such that there is a non-vortex cell $c$ of $\rho_1$ and $x \in V(\sigma(c)) \setminus N(c)$.
Let $L \in \mathcal{L}_1$ be the path with endpoint $x$.
Recall out previous observation that any subpath of some path in $\mathcal{L}$ starting in $A$ and ending in a vertex of $C^1_{w+1}$ must contain a node of $C^1_{w+2}$.
By our construction, this is still true in $G_3$ in the sense that, any subpath of a path in $\mathcal{L}_1$ starting in a terminal and ending on $C^1_{w+1}$ must contain a node of $C^1_{w+2}$.
Hence, we know that either $L \in \sigma(c)\setminus N(c)$, or there exists an $x$-$N(c)$-path $L' \subseteq L$ in $\sigma(c)$ such that all internal vertices of $L'$ have degree $2$ in $G_3$.
In the first case, we may simply remove $L$ from $G_3$ and $\mathcal{L}_1$, and delete $(x,y)$ from $\mathcal{T}_1$ where $y$ is the other endpoint of $L$.
Since, in this case, $L$ is disjoint from all cycles in $\mathcal{C}_1$, this deletion does not interfere with our goal of finding a desolate instance.
In the second case, let $x' \in N(c)$ be the other endpoint of $L'$.
Now delete all internal vertices together with $x$ from $G_3$ and $L$, and replace $(x,y)$ in $\mathcal{T}_1$ with $(x',y)$.
As a result, one more terminal has become a node of $\rho_1$.
Let $G_4$ be the resulting graph and let $\rho_2$ be the restriction of $\rho_1$ to $G_4$.
Moreover, let $\mathcal{T}_2$ and $\mathcal{L}_2$ the new set of terminal pairs and the new linkage respectively.
Note that these alterations still maintain the fact that $(G_4,R \setminus A,\mathcal{T}_2)$ is a vital instance and $|\mathcal{T}_2| \leq |\mathcal{T}_1| \leq k + 2|A|$.
\smallskip

Finally, let $\Delta_w$ be the disc bounded by the trace of $C^1_w$ and let $(G_w,\Omega_w)$ be the corresponding society induced by $\Delta_w$ and the inner graph $G_w$ of $C^1_w$.
Let $\mathcal{P}_1 = \mathcal{L}_2[G_w]$ be the collection of all maximal subpaths of paths from $\mathcal{L}_2$ in $G_w$.
Since $\mathcal{C}_1$ is snug and $G_w$ does not contain any of the terminals from $\mathcal{T}_2$ by our construction it follows that $\mathcal{W}_1 = (G_w,\Omega_w,\rho_2,\mathcal{C}_1,\mathcal{P}_1)$ is in fact a $w$-well.
We now apply \zcref{prop_drywell} to find a dry $w$-well $(G_w',\Omega_w,\rho_3,\mathcal{C}_1,\mathcal{P}_2)$ such that for each $P \in \mathcal{P}_1$ there is a path $P' \in \mathcal{P}_2$ with he same endpoints.
Finally, if we were to replace the subpaths in $\mathcal{P}_1$ with those in $\mathcal{P}_2$ in $\mathcal{L}_2$, this would lead to a new linkage connecting the terminals in $\mathcal{T}_2$ and spanning $R \setminus A$.
Since $(G_4,R \setminus A,\mathcal{T}_2)$ is a vital instance, this operation cannot change $\mathcal{L}_2$ and thus, $\mathcal{W}_1$ was a dry well to begin with.
Hence, we have now entered the stage where $(G_4,R\setminus A,\mathcal{T}_2)$ is a $(\mathsf{genus}_{\ref{lemma_theDesolationOfAnInstance}}(r,k),r,\mathsf{depth}_{\ref{lemma_theDesolationOfAnInstance}}(r,k),\mathsf{breadth}_{\ref{lemma_theDesolationOfAnInstance}}(r,k),w)$-desolate instance of $k'$\textsc{-Spanning Disjoint Paths} where $k' \leq k + 2|A| \leq \mathsf{terminals}_{\ref{lemma_theDesolationOfAnInstance}}(r,k)$ as desired.
\end{proof}

\subsection{Entering the wasteland}\label{subsec_simpleLoops}
With \zcref{lemma_theDesolationOfAnInstance} we now know that any vital instance of large enough treewidth can be transformed into a desolate one at the cost of increasing the number of terminals but without changing the Euler-genus of the surface.
Moreover, it is worth pointing out that the parameter $w$ in \zcref{lemma_theDesolationOfAnInstance} which dictates the size of the well does not appear in $\mathsf{terminals}_{\ref{lemma_theDesolationOfAnInstance}}(r,k)$ and also does not have any influence on the number or depth of the vortices, or on the Euler-genus of the surface.
This means, vital instances of larger and larger treewidth can be transformed into desolate instances with deeper and deeper wells without introducing any additional cost on any of the other parameters involved.
Thus, ultimately, all we have to do towards proving \zcref{thm_VitalSpanningLinkage} is to show that desolate instances cannot have arbitrarily deep wells.

Towards this goal, we next want to prove that any desolate instance of $k$\textsc{-Spanning Disjoint Paths} can be transformed into a vital instance of $\poly(r + k)$\textsc{-Disjoint Paths} with a large dry well.
It will then easily follow from \zcref{prop:VitalLinkage} that such an instance cannot exist, thereby completing the proof of \zcref{thm_VitalSpanningLinkage}.
\smallskip

We say that a pair $(G,\mathcal{T})$ where $G$ is a graph and $\mathcal{T} = \{ (s_i,t_i) \mid i \in [k] \}$ is a set of $k$-terminal pairs is a \emph{vital instance} of $k$\textsc{-Disjoint Paths} if there exists a unique $\mathcal{T}$-linkage $\mathcal{L}$ in $G$ and $V(\mathcal{L}) = V(G)$.

Let $g,d$, and $b$ be non-negative integers.
Let $\mathcal{G} = (G,\mathcal{T})$ be a vital instance of $k$\textsc{-Disjoint Paths}.
We say that $\mathcal{G}$ is \emph{$(g,d,b,w)$-barren} if there exist
\begin{itemize}
    \item a surface $\Sigma$ of Euler-genus at most $g$,
    \item a $\Sigma$-rendition $\rho$ of $G$ with at most $b$ vortices, all of which have depth at most $d$, and
    \item a collection $\mathcal{C} = \{ C_1,\dots,C_w\}$ of vertex-disjoint cycles in $G$
\end{itemize}
such that, if $\mathcal{L}$ denotes the unique solution of the vital instance $\mathcal{G}$,
\begin{enumerate}
    \item $G = \bigcup \mathcal{C} \cup \bigcup \mathcal{L}$
    \item $\mathcal{C}$ is snug in $\rho$,
    \item no two vortices of $\rho$ intersect on their boundaries,
    \item every terminal vertex is either contained in $\sigma(c)$ for some vortex $c$ of $\rho$ or belongs to $N(\rho)$, and
    \item if $\Delta_w$ denotes the disc bounded by the trace of $C_w$ and $(G_w,\Omega_w)$ is the society obtained from the inner graph $G_w$ of $C_w$ where $\rho_w$ denotes the restriction of $\rho$ to $\Delta_w$ and $G_w$, then $(G_w,\Omega_w,\rho_w,\mathcal{C},\mathcal{L}[G_w])$ is a dry $w$-well.
\end{enumerate}
We say that the $\Sigma$-rendition $\rho$ \emph{witnesses} that $\mathcal{G}$ is barren.
Moreover, we refer to $(G_w,\Omega_w,\rho_w,\mathcal{C},\mathcal{L}[G_w])$ as the \emph{well of $\rho$}.

Our next goal is to prove the following lemma.

\begin{lemma}\label{lemma_FromDesolationToBarren}
There exist functions $\mathsf{well}_{\ref{lemma_FromDesolationToBarren}}\colon\mathbb{N}^3 \to \mathbb{N}$ and $\mathsf{terminals}_{\ref{lemma_FromDesolationToBarren}} \colon \mathbb{N}^4 \to \mathbb{N}$ such that for all non-negative integers $k,g,s,d,b$, and $w$, if there exists a $(g,s,d,b,w')$-desolate instance of $k$\textsc{-Spanning Disjoint Paths} where $w' \geq \mathsf{well}_{\ref{lemma_FromDesolationToBarren}}(k,g,s,d,w)$, then there also exists a $(g,d,b,w)$-barren instance $(G,\mathcal{T})$ with terminal set $T$ of $\mathsf{terminals}_{\ref{lemma_FromDesolationToBarren}}(k,g,s,b)$\textsc{-Disjoint Paths} such that $\mathsf{bidim}(G,T) \leq 2k + b$.
Moreover, $\mathsf{terminals}_{\ref{lemma_FromDesolationToBarren}}(k,g,s), \mathsf{well}_{\ref{lemma_FromDesolationToBarren}}(k,g,d,w) \in 2^{\poly(k + d)} \cdot w$.
\end{lemma}

To proceed with the proof of \zcref{lemma_FromDesolationToBarren}, we require some additional terminology.

\paragraph{More on shallow renditions.}
Let $(G,R,\mathcal{T})$ be a $(g,s,d,b,w)$-desolate instance of $k$\textsc{-Spanning Disjoint Paths} and let $\rho$ be a $\Sigma$-rendition witnessing the desolation of $(G,R,\mathcal{T})$.
Let $\mathcal{V}$ denote the set of all vortices of $\rho$ and for each vortex $v \in \mathcal{V}$ let $S^v_1,\dots,S^v_{s_v}$ be a partition of the vortex society $(\sigma(v),\Lambda^v)$ into deep and facial segments.
For each $i \in [s_v]$, we denote by $(G^v_i,\Lambda^v_i)$ the $S^v_i$-crop of $(\sigma(v),\Lambda^v)$ and, in case $S^v_i$ is facial, we denote by $\gamma^v_i$ an ultimate frontier of $(G^v_i,\Lambda^v_i)$ in the facial rendition $\rho^v_i$ of $(G^v_i,\Lambda^v_i)$.

Notice that there exists a $\Sigma$-rendition $\rho^{\star}$ with at most $b \cdot \nicefrac{s}{2}$ vortices, each of depth at most $d$, such that $\rho$ is a restriction of $\rho^{\star}$ and for each $v \in \mathcal{V}$ and $i \in [s_v]$ where $S^v_i$ is a facial segment, $\rho^v_i$ is the restriction of $\rho^{\star}$ to $(G^v_i,\Lambda^v_i)$.
Indeed, note that we may also assume that every terminal which belongs to some $G^v_i$ where $S^v_i$ is a facial segment is a node of $\rho^{\star}$.
This is, because here we may perform the same reduction as before in the last part of the proof of \zcref{lemma_theDesolationOfAnInstance} while possibly removing some vertices of $R$.
We refer to $\rho^{\star}$ as a \emph{true rendition} of $(G,R,\mathcal{T})$ obtained from $\rho$.
The vortices of $\rho^{\star}$ can be seen to correspond precisely to the deep segments of the vortices of $\rho$.

\paragraph{Links and loops.}
The most technical part of the proof of \zcref{lemma_FromDesolationToBarren} will be to show that the number of maximal grounded subpaths of the unique solution $\mathcal{L}$ that interact with the ultimate frontiers of some facial segments of some vortices while also entering the well is bounded by a polynomial function only depending on the Euler-genus of the surface, the number and depth of the vortices, and the number of facial and deep segments in the shallow rendition.
Indeed, in order to show this it suffices to control only those subpaths start have at least one end on the ultimate frontier of some facial segment, and those that enclose discs which contain neither other vortices, nor terminals, nor other segments.
This is, because once these subpaths are controlled, all other subpaths are guaranteed to contain only terminals but no vertices of $R$.
This means that, once we reach this state, we have found an instance where all of $R$ has been reduced to occur only as endpoints of a bounded number of paths, meaning that we are now allowed to forgo the ``spanning condition'' of our \textsc{Spanning Disjoint Paths} problem entirely and focus only on the underlying \textsc{Disjoint Paths} problem.
\smallskip

To this end let $\mathcal{G}=(G,R,\mathcal{T})$ be a $(g,s,d,b,w)$-desolate instance of $k$\textsc{Spanning Disjoint Paths} whose desolation is witnessed by the $\Sigma$-rendition $\rho$.
Let also $\rho^{\star}$ be a true rendition of $\mathcal{G}$ obtained from $\rho$.
Let $G'$ be obtained from $G$ by deleting $\sigma_{\rho^{\star}}(v) - N(v)$ for every vortex of $\rho^{\star}$.
Let $X$ denote the set of all vertices of $N(\rho^{\star})$ that are either terminals, belong to the boundary of some vortex, or to some ultimate frontier of one of the facial segments of $\rho$.
and let $\mathcal{L}^{\star}$ denote the set of all maximal $\rho$-grounded subpaths of paths in $\mathcal{L}$ with both endpoints $X$ but with no interior vertices in $X$.
We call the paths in $\mathcal{L}^{\star}$ the \emph{arcs} of $\mathcal{L}$.
It follows that every path in $\mathcal{L}^{\star}$ is grounded in $\rho^{\star}$.
We are now going to classify the paths in $\mathcal{L}^{\star}$ into different types.
For this, we slightly modify $\Sigma$ by removing the interiors of all closures of the vortices of $\rho^{\star}$ as well as the interiors of $\rho^{\star}$-aligned discs $\Delta_x$ of radius $\varepsilon > 0$ for some small $\varepsilon$ such that, for every terminal $x \in N(\rho^{\star})$ where $x$ is not on the boundary of a vortex, $x$ is the only intersection of $\Delta_x$ with $\rho^{\star}$.
Finally, for every ultimate frontier $\gamma^v_i$ of $\rho^{\star}$, we remove the interior of a $\rho^{\star}$-aligned disk $\Delta^i_v$ which intersects $\rho^{\star}$ precisely in $N(\gamma^v_i)$.
Let $\Sigma^{\star}$ be the resulting surface.
As a consequence, we now have that every path in $\mathcal{L}^{\star}$ now has both endpoints on the boundary of $\Sigma^{\star}$.
Moreover, the number of boundary components of $\Sigma^{\star}$ is at most $2k + sb$.
We refer to $\Sigma^{\star}$ as the \emph{true surface} under $\rho^{\star}$.
\smallskip

A path $L \in \mathcal{L}^{\star}$ is a \emph{loop} if both of its endpoints belong to the boundary of the same boundary component of $\Sigma^{\star}$.
Otherwise we say that $L$ is a \emph{link}.
See \zcref{fig_LinksAndLoops} for an illustration.

\begin{figure}[ht]
 \centering
 \begin{tikzpicture}

 \pgfdeclarelayer{background}
		\pgfdeclarelayer{foreground}
			
		\pgfsetlayers{background,main,foreground}

 \begin{pgfonlayer}{background}
 \pgftext{\includegraphics[width=9cm]{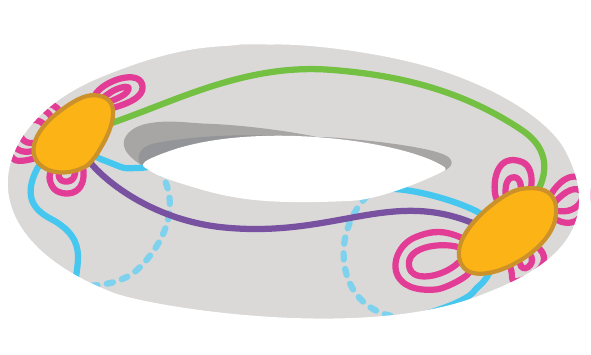}} at (C.center);
 \end{pgfonlayer}{background}
			
 \begin{pgfonlayer}{main}
 \node (C) [v:ghost] {};
 
 \end{pgfonlayer}{main}
 
 \begin{pgfonlayer}{foreground}
 \end{pgfonlayer}{foreground}

 \end{tikzpicture}
 \caption{A sketch of a torus-rendition $\rho$ of some graph with two vortices -- the \textcolor{DarkBananaYellow}{yellow} areas -- and examples of different types of loops and links.}
 \label{fig_LinksAndLoops}
\end{figure}

\textbf{Loops.}
We say that a loop $L$ of $\mathcal{L}^{\star}$ is \emph{simple} if both of its endpoints belong to the boundary of one of the original vortices of $\rho^{\star}$ and, if we were to identify the boundary containing the endpoints of $L$ into a single point, its trace bounds a disc in $\Sigma^{\star}$ which contains no boundary component.

A loop $L$ of $\mathcal{L}$ is called \emph{red} if there exists an ultimate frontier $\gamma^v_i$ such that both endpoints of $L$ belong to $N(\gamma^v_i)$ and the trace of $L$ together with some segment of $\gamma^v_i$ bound a disc in $\Sigma^{\star}$ which contains no other boundary component of $\Sigma^{\star}$ nor a terminal.

A loop $L$ that has both ends on an ultimate frontier $\gamma^v_i$, and the closed curve obtained from the union of $\gamma^v_i$ and the segment $\varphi$ of $\gamma^v_i$ between the endpoints of $L$ is contractible in $\Sigma^{\star}$, but $\varphi$ contains a terminal is called \emph{splitting}.

A loop $L$ is \emph{complex} if it has both endpoints on an ultimate frontier of $\rho^{\star}$ and is neither red nor splitting.
Any other loop that is neither simple, red, splitting, or complex is called \emph{common}.

Notice that every loop of $\mathcal{L}^{\star}$ is either simple, red, complex, or common.
Moreover, a loop is complex or common if and only if it forms -- possibly together with a vortex boundary or a segment of some ultimate frontier -- a non-contractible curve in $\Sigma^{\star}$
See \zcref{fig_TypesOfLoops} for an illustration of different types of loops.

\begin{figure}[ht]
 \centering
 \begin{tikzpicture}

 \pgfdeclarelayer{background}
		\pgfdeclarelayer{foreground}
			
		\pgfsetlayers{background,main,foreground}

 \begin{pgfonlayer}{background}
 \pgftext{\includegraphics[width=9cm]{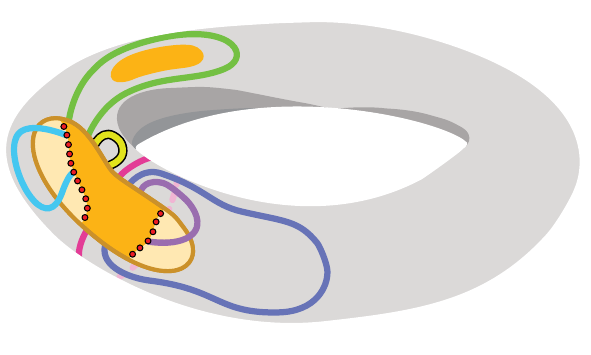}} at (C.center);
 \end{pgfonlayer}{background}
			
 \begin{pgfonlayer}{main}
 \node (C) [v:ghost] {};
 
 \end{pgfonlayer}{main}
 
 \begin{pgfonlayer}{foreground}
 \end{pgfonlayer}{foreground}

 \end{tikzpicture}
 \caption{A sketch of a shallow torus-rendition of some graph. The \textcolor{DarkBananaYellow}{yellow} areas mark the vortices where the lighter areas are those segments with final frontierts -- marked by the \textcolor{BostonUniversityRed}{red} vertices. The different colours mark different kinds of loops.}
 \label{fig_TypesOfLoops}
\end{figure}

\textbf{Links.}
Any path $L$ in $\mathcal{L}^{\star}$ that is not a loop is a link.
However, there are still several types of links.
A link is \emph{simple} if it runs between two terminal vertices both of which do not belong to a vortex boundary or an ultimate frontier or between one terminal vertex that does not belong to a vortex boundary or an ultimate frontier and some node on a vortex boundary or an ultimate frontier.
Notice that the number of simple links is at most $k$ and every link with an endpoint on the boundary of one of the discs $\Delta_x$ is simple.

A link is \emph{red} if both of its endpoints belong to ultimate frontiers of $\rho^{\star}$.
Notice that a red link might even have both endpoints on the same ultimate frontier if it forms a non-contractible curve.

A link is \emph{mixed} if one of its endpoints belongs to an ultimate frontier of $\rho^{\star}$ and the other endpoint belongs to the boundary of some vortex of $\rho^{\star}$.

Any link that is neither simple, red, or mixed is \emph{common}.
\medskip

We now have everything in place for the proof of \zcref{lemma_FromDesolationToBarren}.
To give another brief overview:
Our strategy is to first show that, by sacrificing some amount of cycles from the well, we may assume that no simple or red loop intersects any of the remaining cycles of the well.
This will then imply that we may as well assume that no simple or red loops exist.
In a second and third step, we then show that the total number of red and mixed links as well as complex loops are bounded.
When this is establish it will finally imply that we may assume that $|R|$ is bounded and thus, we have reached a point where we are looking at a regular instance of \textsc{Disjoint Paths} instead of its spanning variant.

\begin{proof}[Proof of \zcref{lemma_FromDesolationToBarren}]
Our goal is to step by step transform $(G,R,\mathcal{T})$ into the desired instance of \textsc{Disjoint Paths}.
Towards this goal let $\rho^{\star}$ be a true rendition of $(G,R,\mathcal{T})$ and consider the arcs $\mathcal{L}^{\star}$ of the unique solution $\mathcal{L}$ of $(G,R,\mathcal{T})$.
Moreover, let $\Sigma^{\star}$ denote the true surface under $\rho^{\star}$.

Let us set up the two functions involved:
\begin{align*}
\mathsf{well}_{\ref{lemma_FromDesolationToBarren}}(k,g,s,d,w) & \coloneqq 2\beta(2k + d,\lceil\sqrt{ 2k + d }\rceil) + 2\beta(2k + 2d,\lceil\sqrt{ 2k + 2d }\rceil) + w
\\
\mathsf{terminals}_{\ref{lemma_FromDesolationToBarren}}(k,g,s) & \coloneqq 2\cdot (4sb + 6g - 6) \cdot (8(2k + 1) + 2\beta(2k + 2d,\lceil\sqrt{ 2k + 2d }\rceil) - 4),
\end{align*}
where $\beta$ denotes the function from \zcref{prop:VitalLinkage}.
Let $w' \coloneqq \mathsf{terminals}_{\ref{lemma_FromDesolationToBarren}}(k,g,s)$.

\paragraph{Taking care of simple loops.}
We first are going to remove the simple loops.
For this let $L \in \mathcal{L}^{\star}$ be a simple loop such that $L \cap C_{w'} \neq \emptyset$.
Let $v$ be the vortex of $\rho^{\star}$ where $L$ has both endpoints.
Moreover, let us denote by $\Delta_L$ the disk bounded by $L$ and a segment $\psi$ of the boundary of $v$ which does not contain any terminals or vortices of $\rho^{\star}$.

\begin{claim}\label{claim_SimpleLoops}
The simple loop $L$ is disjoint from $C_{w' - 2\beta(2k + d,\lceil\sqrt{ 2k + d }\rceil) - 1}$.
\end{claim}

\begin{claimproof}
Suppose towards a contradiction that $L$ intersects $C_{w' - 2\beta(2k + d,\lceil\sqrt{ 2k + d }\rceil) - 1}$.

Notice that, since the vortex $v$ has depth at most $d$, there exists a set $Z \subseteq V(\sigma_{\rho^{\star}}(v))$ of size at most $d$ that separates $S$ from the complement of the vortex society of $v$.
Let $G_L$ be the union of the inner graph of $\Delta_L$ together with the set $Z$ and all components of $\sigma_{\rho^{\star}}(v) - Z$ that contain a vertex of $S$.
Consider the linkage $\mathcal{L}_L \coloneqq \mathcal{L}[G_L]$ and notice that $|\mathcal{L}| \leq k + \nicefrac{d}{2}$ since each path $P$ in $\mathcal{L}_L$ either has both endpoints in $Z$, or at least one endpoint of $P$ is a terminal from $\mathcal{T}$.

Let $T$ be the set of all endpoints of the paths in $\mathcal{L}_L$.
Now we may make a couple of observations.

First, notice that every path of $\mathcal{L}^{\star}$ that is contained in the inner graph of $\Delta_L$ is a subpath of some path from $\mathcal{L}_L$ and must also be a simple loop.

Second, no vertex of $G_L$ belongs to $R$.

Third, $\mathcal{L}_L$ is a vital $T$-linkage in $G_L$ since otherwise, we could find a different $T$-linkage $\mathcal{Q}$ with the same pattern in $G_L$.
By swapping the paths in $\mathcal{L}_L$ with those in $\mathcal{Q}$ within $\mathcal{L}$ we would then find a new spanning $\mathcal{T}$-linkage in $G$ with would contradict the vitality of $(G,R,\mathcal{T})$.
Hence, if we let $\mathcal{T}'$ be the collection of all endpoint pairs of the paths in $\mathcal{L}_L$, we get that $(G_L,\mathcal{T}')$ is a vital instance of $|\mathcal{L}_L|$\textsc{-Disjoint Paths}.

By construction, we have that $|\mathcal{L}_L| \leq k + \nicefrac{d}{2}$, which means that $|T| \leq 2k + d$.

Now recall that, by \zcref{prop_fullbucketsindrainedwells} and the definition of desolate instances, we know that there are paths $Q_{w'},Q_{w' - 1},\dots,Q_{w' - j} = L$ where $j \geq 2\beta(2k + d,\lceil\sqrt{ 2k + d }\rceil) + 1$ in $\mathcal{L}^{\star}$ such that each of them is a simple loop with both endpoint in the boundary of $v$ and such that, $\Delta_{Q_{i'}} \subsetneq \Delta_{Q_{i}}$ for all $i' \geq i \in [w'-j,w']$.
Indeed, this means that we have $j \geq 2\beta(2k + d,\lceil\sqrt{ 2k + d }\rceil) + 1$ such paths.
Moreover, for each $i \in [w'-j,w']$ we have that $Q_i$ intersects $C_{i'}$ for every $i' \geq i$.
Even stronger, since the well of $(G,R,\mathcal{T})$ in $\rho$ is dry, we know that for each $i \in [w'-j,w']$ there exists a subpath $U_i \subseteq C_i$ such that $Q_{i'}$ intersects $U_i$ for every $i' \in [w'-j,i]$ and $U_i \subseteq G_L$.

With this information, we are going to build a bramble in $G_L$ as follows:
For each $i \in [w'-j,w]$ let $B_i \coloneqq U_i \cup Q_i$ and let $\mathcal{B} \coloneqq \{ B_i \mid i\in[w'-j,w'] \}$.
From the construction it follows that any two $B,B' \in \mathcal{B}$ have a non-empty intersection and thus, $\mathcal{B}$ is indeed a bramble.
Moreover, any vertex of $G_L$ is contained in at most $2$ elements of $\mathcal{B}$.
Hence, the order of $\mathcal{B}$ is at least $\nicefrac{j}{2} \geq \beta(2k + d,\lceil\sqrt{ 2k + d }\rceil) + \nicefrac{1}{2}$.
By \zcref{prop_brambles} and \zcref{prop:VitalLinkage} this, however, implies that $(G_L,\mathcal{T}_L)$ contains an irrelevant vertex, thereby contradicting the vitality of $(G_L,\mathcal{T}_L)$.
Hence, $L$ must be disjoint from $C_{w' - 2\beta(2k + d,\lceil\sqrt{ 2k + d }\rceil) - 1}$.
\end{claimproof}

\paragraph{The case of red loops.}
Next, we will take care of the red loops.
For red loops, we bound on how deep they can go into the well of $\rho$ is much stronger.
Let now $L \in \mathcal{L}^{\star}$ be a red loop such that $L \cap C_{w'} \neq \emptyset$ and let $\gamma$ be the ultimate frontier where both endpoints of $L$ lie.
Moreover, let us denote by $\Delta_L$ the disk bounded by $L$ and a segment $\psi$ of $\gamma$ which does not contain any terminals or vortices of $\rho^{\star}$.

\begin{claim}\label{claim_RedLoops}
The red loop $L$ is disjoint from $C_{w' - 1}$.
\end{claim}

\begin{claimproof}
As before, we may assume that $L$ intersects $C_{w' - 1}$
Let $G_L$ be the inner graph of $\Delta_L$ together with $\sigma_{\rho^{\star}}(c)$ for all cells whose closure intersects the ultimate frontier $\gamma$ in at least two nodes.
Notice that, since $\gamma$ is an ultimate frontier, every path of $G$ that starts in $G_L$ and ends in $G - G_L$ must contain a vertex of $L$.
This implies that there exists some path $L' \in \mathcal{L}$ such that every path $Q' \in \mathcal{L}^{\star}$ that is contained in $G_L$ is a subpath of $L'$.

Now, since $L$ intersects $C_{w' - 1}$, and by \zcref{prop_fullbucketsindrainedwells} and the definition of desolate instances, there must exist some path $Q \in \mathcal{L}^{\star}$ which intersects $C_{w'}$ and which is entirely contained in $G_L$.
Then, by our observation above, $Q$ and $L$ are both subpaths of the same path $L' \in \mathcal{L}$.
Moreover, since every path starting in $G_L-L$ and ending in $G-G_L$ must intersect $L$ we know that there exists a maximal path $P \subseteq L'$ in $G_L$ which contains both $L$ and $Q$.
Thus, $L$ and $P$ can share at most one endpoint.
In fact, $L$ and $P$ must share an endpoint since every internal vertex of $L$ must have degree $2$ in $L'$ and every path from $G_l - L$ to $G -G_L$ must intersect $L$.
Let $x \in V(P)$ be the endpoint of $P$ that does not belong to $L$.
Then there exists a subpath $J$ of $L'$ starting in $x$ and ending in an endpoint $t$ of $L'$, which is a terminal, such that $J$ is internally disjoint from $P$.
Since $t \notin V(G_L)$ because $G_L$ does not contain any terminals, $J$ must contain a vertex of $L$.
However, by choice of $x$, this vertex cannot be an endpoint of $L$ which contradicts the choice of $J$ in the first place and our claim follows.
\end{claimproof}

\paragraph{When a loop splits the terminals.}
Our next concern are splitting loops.
Let $L \in \mathcal{L}^{\star}$ be a splitting loop such that $L \cap C_{w'} \neq \emptyset$ and let $\gamma$ be the ultimate frontier where both endpoints of $L$ lie.
Let further $\varphi$ be the segment of $\gamma$ between the endpoints of $L$.
Moreover, let us denote by $\Delta_L$ the disk bounded by $L$ and a segment $\psi$ of $\gamma$ which does not contain any terminals or vortices of $\rho^{\star}$.

Our strategy to deal with this case is similar to the strategies we will apply for the cases of red and mixed links.
So in some sense, the case of splitting loops is the most crucial.

A core property we have used in previous iterations was, that whenever we have a loop that intersects some cycle $C_{j}$ of the well, then there actually exist $w' - j + 1$ many distinct loops all ``stacked'' in the same disc and each intersecting one cycle of the well fewer than the previous.
This is really what \zcref{prop_fullbucketsindrainedwells} says about wells, and it follows because we may extend the disc $\Delta_{w'}$ defined by the trace of $C_{w'}$ along the trace of $L$ and the curve $\varphi$ in order to get a new well which contains all of the loops caught in the disc $\Delta_L$.
Applying \zcref{prop_drywell} to this well and using the fact that our instance is vital implies that each loop individually can intersect any cycle of the well in at most two subpaths.

\begin{claim}\label{claim_SplittingLoops}
The splitting loop $L$ is disjoint from $C_{w' - 18k - 8}$.
\end{claim}

\begin{claimproof}
Towards a contradiction, we assume that $L$ intersects $C_j$ with $j \leq w' - 18k - 8$.

Let us consider the following graph.
Given a loop $L$ as above, we let $G_L^{\star}$ be the graph obtained from the union of all paths from $\mathcal{L}^{\star}$ in $G_L$ together with all cycles in $\mathcal{C}$.
Let $\mathcal{L}_L$ be the set of all paths from $\mathcal{L}^{\star}$ in $G_L$.
Let further $\rho_L$ be the restriction of $\rho$ to $G_L^{\star}$.
Notice that $\rho_L$ can be understood as a rendition of $G_L^{\star}$ in a disc $\Delta_L^{\star}$ such that the nodes on the boundary of $\Delta_L^{\star}$ belong to $L$, $\varphi$ or $C_{w'}$, giving rise to a cyclic ordering $\Omega_L$.
Finally notice that $\mathcal{W}_L = (G_L^{\star},\Omega_L,\rho_L,\mathcal{C},\mathcal{L}_L)$ is a $w'$-well.
We call $\mathcal{W}_L$ the \emph{$L$-bay}.
Because the $L$-bay is dry due to \zcref{prop_drywell} and the vitality of our instance, we know that any loop in $G_L^{\star}$ intersects any cycle of $\mathcal{C}$ in at most two subpaths.

Since $L$ is a splitting loop, we know that every path in $\mathcal{L}_L$ must be either a red or a splitting loop.
By \zcref{claim_RedLoops} we know that none of the red loops can intersect $C_{w' - 1}$.

With $j \leq w' - 18k - 8$ we know by \zcref{prop_fullbucketsindrainedwells} that there are loops $P_{w'},P_{w'-1},\dots,P_{w' - 18k -7},P_{w' - 18k - 8} \in \mathcal{L}_L$ such that $P_i \cap C_i$ is non-empty for all $i \in [w'-18k-8,w']$ and we have $\Delta_{P_{w'}} \subsetneq \dots, \subsetneq \Delta_{P_{w' - 18k - 8}}$.

For every $i \in [2k + 1]$ let $\mathcal{P}_i \coloneqq \{ P_{i'} \mid i' \in [w' - (i-1)9 , w' - (i-1)9 - 8] \}$.
Then $|\mathcal{P}_i| = 9$ for all $i \in [2k + 1]$.
Moreover, for every $i \in [2k+1]$ if we let $\varphi^1_{i}$ and $\varphi^2_{i}$ be the two disjoint segments of $\gamma$ which both join one end of $P_{w' - (i-1)9}$ to an end of $P_{w' - (i-1)9 - 8}$, then the traces of $P_{w' - (i-1)9}$ and $P_{w' - (i-1)9 - 8}$ together with $\varphi^1_i$ and $\varphi^2_i$ form a closed curve bounding a disc $\Delta_i$.
See \zcref{fig_SplittingLoops} for an illustration of the situation.

\begin{figure}[ht]
 \centering
 \begin{tikzpicture}

 \pgfdeclarelayer{background}
		\pgfdeclarelayer{foreground}
			
		\pgfsetlayers{background,main,foreground}

 \begin{pgfonlayer}{background}
 \pgftext{\includegraphics[width=7cm]{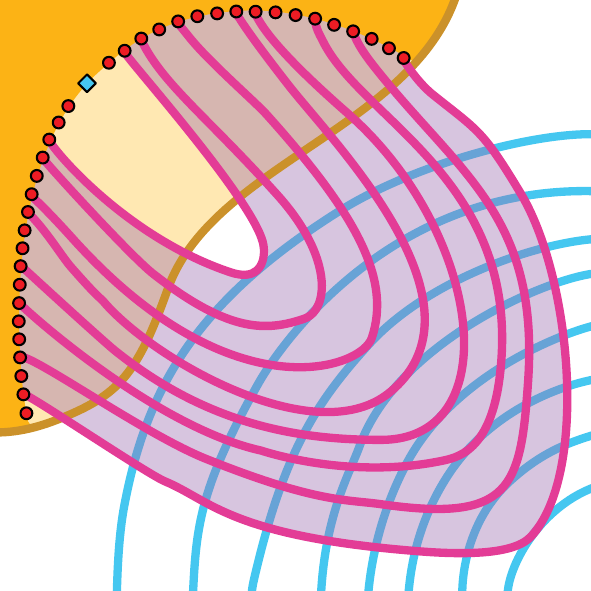}} at (C.center);
 \end{pgfonlayer}{background}
			
 \begin{pgfonlayer}{main}
 \node (C) [v:ghost] {};
 
 \end{pgfonlayer}{main}
 
 \begin{pgfonlayer}{foreground}
 \end{pgfonlayer}{foreground}

 \end{tikzpicture}
 \caption{A sketch of a family of splitting loops. The \textcolor{Amethyst}{purple} shaded area marks the disc $\Delta_i$.}
 \label{fig_SplittingLoops}
\end{figure}

Notice that two such discs $\Delta_i$ and $\Delta_{i'}$ may only intersect in their boundaries and there, only in nodes that are shared endpoints of paths from $\mathcal{P}_i$ and $\mathcal{P}_{i'}$.

Since we have $2k + 1$ such discs, there must be one, say $\Delta_i$, whose inner graph does not contain a terminal.
Let $G_i$ be the inner graph of $\Delta_i$ under $\rho_L$ in $G_L^{\star}$.
Notice that, by our choices of the paths in $\mathcal{P}_i$ and the fact that $G_i$ does not contain a terminal, it follows that no path of $\mathcal{L}_L \setminus \mathcal{P}_i$ that is contained in $G_i$ can be a splitting loop.
Hence, all such paths must be red loops and therefore disjoint from $C_{w'-1}$.

Next observe that there must be a subpath $Q$ of some path $L' \in \mathcal{L}$ such that $Q$ contains every path of $\mathcal{L}_L$ that belongs to $\Delta_i$ -- including the paths from $\mathcal{P}_i$.
Let us select the paths $P_{w' - (i-1)9 - 5}$, $P_{w' - (i-1)9 - 6}$, $P_{w' - (i-1)9 - 7}$, and $P_{w' - (i-1)9 - 8}$ and let $\Delta_i'$ be the disc obtained from $\Delta_i$ by cutting along the trace of $P_{w' - (i-1)9 - 5}$ and keeping the part whose inner graph contains $P_{w' - (i-1)9 - 6}$.
Let further $G_L'$ be the restriction of $G_L^{\star}$ to $\Delta_i'$ and $Q'$ be the shortest subpath of $Q$ that contains $P_{i'}$ for $i' \in[w' - (i-1)9 - 5,w' - (i-1)9 - 8]$.
Since, for each $i' \in[w' - (i-1)9 - 5,w' - (i-1)9 - 8]$, $P_{i'}$ intersects $C_{i'}$ and all $C_{i''}$ with $i'' \geq i'$, there exist subpaths $F_{j'}$ of $C_{j'}$ for $j' \in [w' - (i-1)9 - 1,w' - (i-1)9 - 4]$:
\begin{itemize}
    \item $F_{j'}$ is minimal with the property that it intersects each $P_{i'}$, $i' \in [w' - (i-1)9 - 5,w' - (i-1)9 - 8]$ in precisely one subpath,
    \item $F_{j'}$ intersects $P_{w' - (i-1)9 - 5}$ and $P_{w' - (i-1)9 - 8}$ precisely in its endpoints, and
    \item $\bigcup_{i' \in[w' - (i-1)9 - 5,w' - (i-1)9 - 8]} P_{i'} \cup \bigcup_{j' \in [w' - (i-1)9 - 1,w' - (i-1)9 - 4]} F_{j'}$ contains a unique $(4 \times 4)$-mesh $M$.
    See \zcref{fig_MeshInLoops} for an illustration of the construction of $M$.
\end{itemize}

\begin{figure}[ht]
 \centering
 \begin{tikzpicture}

 \pgfdeclarelayer{background}
		\pgfdeclarelayer{foreground}
			
		\pgfsetlayers{background,main,foreground}

 \begin{pgfonlayer}{background}
 \pgftext{\includegraphics[width=8cm]{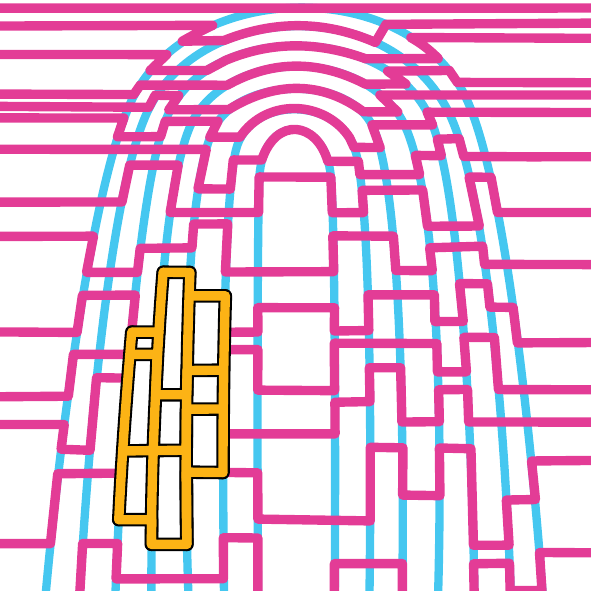}} at (C.center);
 \end{pgfonlayer}{background}
			
 \begin{pgfonlayer}{main}
 \node (C) [v:ghost] {};
 
 \end{pgfonlayer}{main}
 
 \begin{pgfonlayer}{foreground}
 \end{pgfonlayer}{foreground}

 \end{tikzpicture}
 \caption{The $(4 \times 4)$-mesh -- depicted in \textcolor{DarkBananaYellow}{yellow} -- made from the cycles of $\mathcal{C}$ -- depicted in \textcolor{CornflowerBlue}{blue} -- and the paths from the loops -- depicted in \textcolor{HotMagenta}{magenta}.}
 \label{fig_MeshInLoops}
\end{figure}

Since every loop in $G_L^{\star}$ belongs to $Q$ and no red loop intersects $C_{w' - 1}$, it follows that the intersection of $M$ with $Q$ is equal to the intersection of $M$ with $Q'$, which in turn is equal precisely to the four subpaths of the paths $P_{i'}$.
Therefore, $Q'$ can be divided into $7$ pairwise vertex-disjoint subpaths whose union makes up all of $Q'$ as follows:
For $j' \in [4]$, $Q_{2j' - 1}$ is the subpath of $P_{w' - (i-1)9 - 4 - j'}$ contained in $M$.
For $j' \in [3]$, the path $Q_{2j'}$ is the subpath of $Q'$ which is internally disjoint from $M$ and shares one endpoint with $Q_{2j' - 1}$ and the other with $Q_{2j' + 1}$.
Finally, notice that $M \cap \bigcup \mathcal{L} = M \cap Q'$ while $R \cap V(Q') = \bigcup_{j' \in[3]} R \cap V(Q_{2j'})$.

With all of these observations and terminologies at our disposal, we now claim that there exists a path $Q^*$ in $M \cup Q'$ with the same endpoints as $Q'$ such that $V(Q') \cap R = V(Q^{\star}) \cap R$, and such that $|Q^{\star}| < |Q'|$.
If this is true, then this would give a contradiction to the vitality of our instance since we could exchange $Q'$ for $Q^{\star}$ and obtain a feasible solution for $\mathcal{G}$ using less vertices.

To see that this claim is true we follow along $Q'$ starting with its endpoints on $P_{w' - 9(i-1) - 5}$, we follow along $Q$ until we reach the second endpoint of $Q_2$.
Now we follow along $M$ towards the endpoint of $Q_6$ while minimising our intersection with $P_{w' - 9(i-1) - 6}$.
This means, we travel along the first vertical path of $M$.
Then we follow along $Q_6$ until we reach $M$ again.
From here follow $P_{w' - 9(i-1) - 8}$ until we meet the second vertical path -- one of the subpaths of the cycles from the well -- of $M$.
Follow this vertical path until we reach $P_{w' - 9(i-1) - 6}$ again.
From here, trace along $P_{w' - 9(i-1) - 6}$ until we reach the endpoint of $Q_4$, then follow $Q_4$ to its other endpoint on $P_{w' - 9(i-1) - 7}$.
Finally, from the second endpoint of $Q_4$ we follow the fourth vertical path directly to the remaining endpoint of $Q'$ and call the path we just traces $Q^{\star}$.
See \zcref{fig_RerouteToAvoid} for an illustration of the re-routing.
Since $P_2j' \subseteq Q^{\star}$ we have that $V(Q') \cap R = V(Q^{\star}) \cap R$.
Moreover, $Q^{\star}$ does not contain the intersection of the third vertical path of $M$ with the paths $P_{w' - 9(i-1) - 7}$ and $P_{w' - 9(i-1) - 8}$.
Hence $|Q^{\star}| < |Q'|$ as desired and the proof of our claim is complete.
\end{claimproof}

\begin{figure}[ht]
 \centering
 \begin{tikzpicture}

 \pgfdeclarelayer{background}
		\pgfdeclarelayer{foreground}
			
		\pgfsetlayers{background,main,foreground}

 \begin{pgfonlayer}{background}
 \pgftext{\includegraphics[width=8cm]{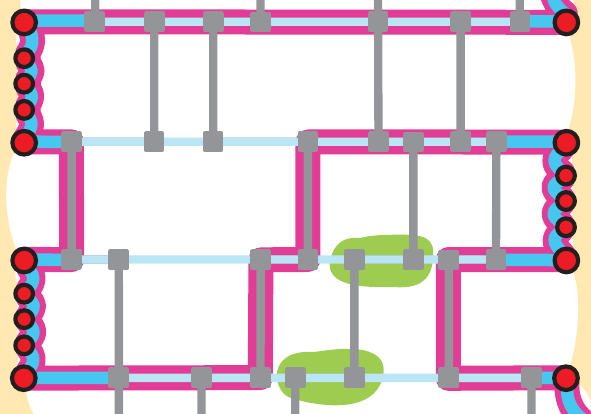}} at (C.center);
 \end{pgfonlayer}{background}
			
 \begin{pgfonlayer}{main}
 \node (C) [v:ghost] {};
 
 \end{pgfonlayer}{main}
 
 \begin{pgfonlayer}{foreground}
 \end{pgfonlayer}{foreground}

 \end{tikzpicture}
 \caption{The re-routing within the mesh in order to avoid some of the internal vertices. Notice that the figure is rotated by $90$ degrees and thus the paths that are horizontal here are vertical in the proof of \zcref{claim_SplittingLoops}.
 The parts marked in \textcolor{AppleGreen}{green} are those vertices freed up by the re-routing.}
 \label{fig_RerouteToAvoid}
\end{figure}

From here on, we will only consider $\mathcal{C}' \coloneqq  \{ C_1,\dots,C_{w''}\}$ where $w'' \coloneqq 2\beta(2k + 2d,\lceil\sqrt{ 2k + 2d }\rceil) + w$.
By \zcref{claim_SimpleLoops}, \zcref{claim_RedLoops}, and \zcref{claim_SplittingLoops} we know that all simple, red, and splitting loops are vertex disjoint from the cycles in $\mathcal{C}'$.
Sadly, this is where our ability ends to prove that paths in $\mathcal{L}^{\star}$ may be assumed disjoint from most cycles of the well.
However, by pivoting slightly, we can now make use of the ideas from the proof of \zcref{claim_SplittingLoops} in order to show that there is an upper bound on the number of remaining paths from $\mathcal{L}^{\star}$ that contain vertices from $R$ and enter deep into the well.

For this, we say that two paths $L_1,L_2 \in \mathcal{L}^{\star}$ are \emph{homotopic} if they are either both links of the same type, both common loops, or both complex loops and their traces in $\Sigma^{\star}$ are homotopic.
We say that a path in $\mathcal{L}^{\star}$ is \emph{essential} if it is neither a simple, red, or splitting loop.
Then the notion of homotopy above defines an equivalence relation $\sim$ on the set of essential paths.
Indeed, the total number of equivalence classes under this relation is at most $4sb + 6g - 6$ by \zcref{prop_TypeCounting}.
We refer to the equivalence classes of $\sim$ as the \emph{homotopy classes} of $\mathcal{L}^{\star}$.
In what follows we will consider equivalence classes of red and mixed links only.

Let $\mathcal{Q}$ be a homotopy class of $\mathcal{L}^{\star}$.
Then there exist two minimal segments $S_1$ and $S_2$ of the boundaries of $\Sigma^{\star}$ such that a path in $\mathcal{L}^{\star}$ belongs to $\mathcal{Q}$ if and only if it has one end on $S_1$ and the other on $S_2$.
If we assume that $|\mathcal{Q}| \geq 2$, then the minimality of $S_1$ and $S_2$ implies that there are two internally vertex-disjoint paths $P_1,P_2 \in \mathcal{Q}$ such that their traces together with $S_1$ and $S_2$ bound a disc $\Delta_{\mathcal{Q}}$ -- we refer to this disc as the \emph{strip of $\mathcal{Q}$} -- such that a path from $\mathcal{L}^{\star}$ intersects the inner graph of $\Delta_{\mathcal{Q}}$ if and only if it belongs to $\mathcal{Q}$.
The segments $S_1$ and $S_2$ will be referred to as the \emph{ends} of $\Delta_{\mathcal{Q}}$ and the paths $P_1$ and $P_2$ are the \emph{boundary paths} of $\mathcal{Q}$.

Let $\mathcal{Q}' \subseteq \mathcal{Q}$ with $|\mathcal{Q}'| \geq 2$.
We say that $\mathcal{Q}'$ is \emph{closed} if there exists a $\rho^{\star}$-aligned disc $\Delta' \subseteq \Delta_{\mathcal{Q}}$ such that a path $P$ of $\mathcal{L}^{\star}$ is contained in the inner graph of $\Delta'$ if and only if $P$ belongs to $\mathcal{Q}'$.
We call $\Delta'$ the \emph{strip} of $\mathcal{Q}'$ and the intersections of $S_1$ and $S_2$ with $\Delta'$ are the ends of $\Delta'$.
By the definition $\Delta'$ it follows that we may assume that there must exist two paths $P_1,P_2 \in \mathcal{Q}'$ such that the boundary of $\Delta'$ is the union of its ends and the traces of $P_1$ and $P_2$.
We call $P_1$ and $P_2$ the \emph{boundary paths} of $\mathcal{Q}'$

\paragraph{Of links and meshes.}
Our next goal is to formalise the idea used in the proof of \zcref{claim_SplittingLoops} to be more easily applicable for links and complex loops.
That is, we wish to show that either, only few members of a homotopy class of $\mathcal{L}^{\star}$ enter deep into the well, or we can find a closed subset of this homotopy class which contains a big mesh whose vertical paths are made from subpaths of the cycles from our well.

\begin{claim}\label{claim_DeepLinksMakeMesh}
Let $\mathcal{Q}$ be a closed subset of a homotopy class of $\mathcal{L}^{\star}$, and $j,k'\geq 1$ be integers with $j\in[w+7]$, then either at most $2k'-2$ members of $\mathcal{Q}$ intersect $C_j$, or there exists a closed set $\mathcal{Q}'=\{ Q_1,\dots,Q_{k'}\}$ of size ${k'}$ with boundary paths $Q_1$ and $Q_{k'}$ and strip $\Delta'$ and a $({k'} \times (w'' - j + 1))$-mesh whose horizontal paths $Q'_1,\dots,Q'_{k'}$ are such that $Q'_i \subseteq Q_i$ for each $i \in [{k'}]$ and whose vertical paths are $P_j,\dots,P_{w''}$ such that $P_i \subseteq C_i$ for each $i \in [j,w'']$.
\end{claim}

\begin{claimproof}
Let $P_1$ and $P_2$ be the two boundary paths of $\mathcal{Q}$.
For each $i \in [2]$ let $j_i$ be the smallest integer such that $P_i$ intersects $C_{j_i}$.
We may observe that, no path in $\mathcal{Q} \setminus \{ P_1,P_2\}$ can intersect $C_{\min\{ j_1,j_2\}}$.
Hence, if $j < \min \{ j_1,j_2\}$ there is nothing to show.
We distinguish two cases: $j < \min \{ j_1,j_2\} + {k'} - 1$ or $j \geq \min \{ j_1,j_2\} + {k'} - 1$.
\medskip

\textbf{Case 1:} $j < \min \{ j_1,j_2\} + {k'} - 1$.
\smallskip

In this case we may consider the closed set $\mathcal{Q}'' \subseteq \mathcal{Q}$ obtained by iteratively removing the boundary paths from both sides until each side has lost ${k'} - 1$ boundary paths.
Recall that $\mathcal{W}$ is dry and thus every arc of $\mathcal{L^{\star}}$ in $\Delta_{w''}$ -- the disc bounded by the trace of $C_{w''}$ -- follows the definition of dry wells.
Therefore, any arc of some path in $\mathcal{Q}$ that intersects a cycle beyond $C_{\min \{ j_1,j_2\} + {k'} - 1}$ -- i.\@e.\@ a cycle with smaller index -- must belong to one of the paths in $\mathcal{Q} \setminus \mathcal{Q}'$.
Moreover, $|\mathcal{Q}| - |\mathcal{Q}'| = 2{k'} - 2$.
Since $j < \min \{ j_1,j_2\} + {k'} - 1$ this implies that the first outcome of our claim holds.
\medskip

\textbf{Case 2:} $j \geq \min \{ j_1,j_2\} + {k'} - 1$
\smallskip

Similar to before we consider the closed set $\mathcal{Q}'' \subseteq \mathcal{Q}$ obtained by iteratively removing the boundary paths from both sides until each side has lost $\min \{ j_1,j_2\} + {k'} - 1$ boundary paths.
Let $\mathcal{Q}_i$ be the set of those paths removed from the side of $P_i$ for both $i \in[2]$.
Then $|\mathcal{Q}_i| = {k'}$ and for each $i \in [2]$, every path in $\mathcal{Q}_i$ intersects $C_{j_i + {k'} - 1}$.
Indeed, this means that there exists $i \in[2]$ where $j_i = \min\{ j_1,j_2\}$ such that every path in $\mathcal{Q}_i$ intersects $C_i$.
Notice that, by definition, $\mathcal{Q}' \coloneqq \mathcal{Q}_i$ is closed.
We enumerate $\mathcal{Q}' = \{ Q_1,\dots,Q_{k'}\}$ and denote its strip by $\Delta'$ such that $Q_1$ and $Q_{k'}$ are its boundary paths and for each $i \in [2,{k'}-1]$, the trace of $Q_i$ separates the traces of $Q_1,\dots,Q_{i-1}$ from the traces of $Q_{i+1},\dots,Q_{k'}$ in $\Delta'$.

Without loss of generality we may assume that $Q_1$ is one of the boundary paths of $\mathcal{Q}$.
We follow along $Q_{k'}$ from one of its endpoints, say $x$, until the first time $Q_{k'}$ intersects $C_j$, let $J^{k'}_1$ be the component of $Q_{k'} \cap C_j$ containing this intersection.
It follows that within $\Delta'$ there exists a path $N_1$ from $J^{k'}_1$ to $Q_1$ which is entirely contained in $C_j$.
Similarly, since the inner graph of $\Delta_{w''}$ -- the disc bounded by the trace of $C_{w''}$ -- has a vortex-free rendition in $\Delta_{w''}$, the same observation holds for every $j' \in[w'',j]$.
That is if start in $J^{k'}_1$ and follow along $Q_{k'}$ towards $x$, next we meet a component $J_{2}^{k'}$ of $C_{j+1} \cap Q_{k'}$, then we encounter a component $J_3^{k'}$ of $C_{j+2} \cap Q_{k'}$, and so on until we meet, for the first time, $C_{w''}$ in $J_{w'' - j + 1}^{k'}$.
For every $j' \in[2,w''-j+1]$ we then let $N_{j'}$ be the subpath of $C_{j + j' - 1}$ starting in $J^{k'}_{j'}$, ending on $Q_1$ and being fully contained in $\Delta'$.
Finally, for each $i' \in [{k'}]$ let $M_{i'}$ be the $N_1$-$N_{w'' - j + 1}$-subpath of $Q_{i'}$.
Notice that $M_{i'}$ is entirely contained in a single arc of $Q_{i'}$ in the inner graph of $\Delta_{w''}$ and thus, as $\mathcal{W}$ is a dry well, the intersection of $M_{i'}$ with $N_{j'}$ for each $i' \in [{k'}]$ and $j' \in [w'' - j + 1]$ is a single path.
Hence $\bigcup_{i' \in[{k'}]} M_{i'} \cup \bigcup_{j'\in[1,w''-j+1]}$ is a $({k'} \times (w'' - j +1))$-mesh as desired.
\end{claimproof}

\paragraph{Complex loops and red links.}
With \zcref{claim_DeepLinksMakeMesh} we have a straightforward way for extracting meshes from homotopy classes of $\mathcal{L}^{\star}$.
In the next step we apply this tool to control complex loops and red links.
It should be pointed out that, while technically different, for our purposes complex loops and red links behave in exactly the same way.
This is because for any homotopy class $\mathcal{Q}$ of either, both ends are segments of one or two ultimate frontiers of $\rho^{\star}$.

In particular, we have the following observation: 
Let $\mathcal{Q}$ be a closed subset of some homotopy class of complex loops or red links in $\mathcal{L}^{\star}$ and let $\Delta$ be the strip of $\mathcal{Q}$.
If the inner graph of $\Delta$ does not contain any terminal, then there exists a single path $L \in \mathcal{L}$ such that there is a subpath $L' \subseteq L$ which is precisely the union of all paths in $\mathcal{L}^{\star}$ that are contained in the inner graph of $\Delta$.

\begin{claim}\label{claim_ComplexLoopsAndRedLinks}
Let $\mathcal{Q}_0$ be a homotopy class of $\mathcal{L}^{\star}$ consisting of complex loops or red links.
Then at most $8(2k + 1) - 2$ members of $\mathcal{Q}_0$ intersect $C_{w'' - 4}$.
\end{claim}

\begin{claimproof}
We start by two applications of \zcref{claim_DeepLinksMakeMesh}.
First note that, if we assume towards a contradiction that at least $8(2k + 1) - 1$ members of $\mathcal{Q}_0$ intersect $C_{w''}$, then \zcref{claim_DeepLinksMakeMesh} provides us with a closed set $\mathcal{Q}_1$ of size $4(2k + 1)$ such that each member of $\mathcal{Q}_1$ intersects $C_{w''}$ and there is a mesh in the strip $\Delta_1$ of $\mathcal{Q}_1$ using all cycles $C_{{w''}-4},\dots,C_{w''}$ and subpaths of the paths in $\mathcal{Q}_1$.
We may then partition $\mathcal{Q}_1$ into $2k + 1$ closed sets, each containing precisely $4$ paths.
By the pigeonhole principle, one of those $2k + 1$ closed sets, call it $\mathcal{Q}_2$, has strip $\Delta_2$ with inner graph $G_2$ such that $G_2$ contains no terminal.
Hence, there exists a single path $L \in \mathcal{L}$ such that there is a subpath $L' \subseteq L$ which is precisely the union of all paths in $\mathcal{L}^{\star}$ that are contained in the inner graph of $\Delta_2$ by our discussion above.
Moreover, by \zcref{claim_DeepLinksMakeMesh} we know that there is a $(4 \times 4)$-mesh $M$ in $G_2$ whose horizontal paths $Q_1',\dots,Q_4'$ are subpaths of the paths in $\mathcal{Q}_2 = \{ Q_1,\dots,Q_4\}$ numbered accordingly.
Moreover, the vertical paths of $M$ are paths $P_1,\dots,P_4$ which are subpaths of the cycles $C_{w''}$, $C_{w''-1}$, $C_{w''-2}$, and $C_{w''-3}$ respectively.

Now let $L''$ be the shortest subpath of $L'$ that contains $Q_1', Q_2', Q_3'$, and $Q_4'$.
Then $L''$ can be divided into seven pairwise vertex-disjoint subpaths whose union makes up all of $L''$ as follows:
For $j' \in [4]$, $L_{2j' - 1}$ is the path $Q'_{j'}$.
For $j' \in [3]$, the path $L_{2j'}$ is the subpath of $L''$ which is internally disjoint from $M$ and shares one endpoint with $L_{2j' - 1}$ and the other with $L_{2j' + 1}$.
Finally, notice that $M \cap \bigcup \mathcal{L} = M \cap L''$ while $R \cap V(L'') = \bigcup_{j' \in[3]} R \cap V(L_{2j'})$.
From here, it is apparent that we have reached the same situation as in the end of the proof of \zcref{claim_SplittingLoops}.
Indeed, we may use the same routing strategy as before to observe that $L''$ can be replaced with a path $L^{\star}$ in $L'' \cup M$ in order to obtain a solution for $\mathcal{G}$ that avoids at least two vertices of $M$.
The situation is analogous to the one depicted in \zcref{fig_RerouteToAvoid}.
Since this is a contradiction to the vitality of $\mathcal{L}$ the claim is proven.
\end{claimproof}

\paragraph{Links that are mixed.}
The final case we need to discuss are mixed links as these are the only paths we have not touched yet which may contain vertices of $R$ in their interior.
If we now consider a homotopy class $\mathcal{Q}$ of $\mathcal{L}^{\star}$ consisting of mixed links, we still know that one of the ends of its strip is a segment of some ultimate frontier of $\rho^{\star}$.
However, the other end of the strip of $\mathcal{Q}$ is now a segment of some vortex.
This has the advantage that we still have some amount of control over the interaction of the set $R$ with the paths of $\mathcal{L}^{\star}$ caught in the strip, but the rerouting is no longer as simple because, once we have reduced to a subgraph without any terminals, we can no longer prove that there is always a single path of $\mathcal{L}$ containing all paths of $\mathcal{L}^{\star}$ involved in our situation.

\begin{claim}\label{claim_MixedLinks}
Let $\mathcal{Q}_0$ be a homotopy class of $\mathcal{L}^{\star}$ consisting of mixed links.
Then at most $8(2k + 1) - 2$ members of $\mathcal{Q}_0$ intersect $C_{w''-4}$.
\end{claim}

\begin{claimproof}
The first couple of steps are analogous to the proof of \zcref{claim_ComplexLoopsAndRedLinks}.

As before, we assume towards a contradiction that at least $8(2k + 1) - 1$ members of $\mathcal{Q}_0$ intersect $C_{w''-4}$, then \zcref{claim_DeepLinksMakeMesh} provides us with a closed set $\mathcal{Q}_1$ of size $4(2k + 1)$ such that each member of $\mathcal{Q}_1$ intersects $C_{w''-4}$ and there is a mesh in the strip $\Delta_1$ of $\mathcal{Q}_1$ using all cycles $C_{{w''}},\dots,C_{w''-4}$ and subpaths of the paths in $\mathcal{Q}_1$.
We then partition $\mathcal{Q}_1$ into $2k + 1$ closed sets, each containing precisely $4$ paths.
By the pigeonhole principle, one of those $2k + 1$ closed sets, call it $\mathcal{Q}_2$, has strip $\Delta_2$ with inner graph $G_2$ such that $G_2$ contains no terminal.
Moreover, by \zcref{claim_DeepLinksMakeMesh} we know that there is a $(4 \times 4)$-mesh $M$ in $G_2$ whose horizontal paths $Q_1',\dots,Q_4'$ are subpaths of the paths in $\mathcal{Q}_2 = \{ Q_1,\dots,Q_4\}$ numbered accordingly.
Moreover, the vertical paths of $M$ are paths $P_1,\dots,P_4$ which are subpaths of the cycles $C_{w''}$, $C_{w''-1}$, $C_{w''-2}$, and $C_{w''-3}$ respectively.

Let $S_1$ be the end of $\Delta_2$ which coincides with an ultimate frontier of $\rho^{\star}$ and let $S_2$ be the other end of $\Delta_2$.
Notice that $S_2$ may be understood as a segment of a vortex boundary of $\rho$ -- technically speaking, $S_2$ could also arise from one of the boundary components of $\Sigma^{\star}$ introduced where terminal vertices were present in the surface, but such boundary components can only have a single path of $\mathcal{L}^{\star}$ ending on them and we already know that $\mathcal{Q}_2$ contains $4$ paths.
We may now observe that every subpath $U$ of a path from $\mathcal{L}$ that exists within $G_2$ is either a simple loop or a path consisting of a -- possibly empty -- collection of red loops and two paths $Q,Q'$ from $\mathcal{Q}_2$.
Moreover, both endpoints of $P$ must lie on $S_2$ and there is no path in $\mathcal{Q}_2 \setminus \{ Q,Q'\}$ whose trace separates $Q$ and $Q'$ in $\Delta_2$.
Indeed, this means that there are precisely two such paths, say $U_1$ and $U_2$ such that $U_1$ contains $Q_1$ and $Q_2$ and $U_2$ contains $Q_3$ and $Q_4$.

Notice that, by construction, we know that the trace of $P_i$ separates $S_1$ and $S_2$ in $\Delta_2$ for all $i\in[4]$.
We may assume that for each $i\in[2,4]$, $P_1$ is contained in the same component as $S_2$ when deleting the trace of $P_i$.
By \zcref{claim_RedLoops}, \zcref{claim_SimpleLoops}, and \zcref{claim_SplittingLoops} we know that the only paths of $\mathcal{L}^*$ that intersect $M$ are the four paths containing $Q_1,\dots,Q_4$ respectively.
Therefore, there exist two subpaths $F_1$ and $F_2$ of $P_1$, each $F_i$ joining two vertices of $U_i$, such that there is a path $O_i \subseteq U_i$ with the same endpoints as $F_i$ and, if we were to replace $O_1$ and $O_2$ with $F_1$ and $F_2$ in the paths of $\mathcal{L}$ we would get a new solution for the instance $(G,\mathcal{T})$ of the $k$\textsc{-Disjoint Paths} problem.
Of course this is not yet enough since $O_1$ and $O_2$ contain all vertices of $R$ contained in $G_2$.
However, we may replace $F_1$ by a path $K$ as follows:
Follow along $U_1$ starting from the endpoint of $F_1$ on $Q_1$, until we reach the second component of $U_1 \cap P_4$.
Then follow along $P_4$ until we reach $U_2$.
From here we follow along $U_2$ in a way such that we pass through all vertices of $U_2$ on $S_1$ until we reach $P_3$ for the first time.
From here we move along $P_3$ until we meet $U_1$ again.
Finally, we move along $U_1$ until we reach the other endpoint of $F_1$ on $U_1$.
As before, replacing $O_1$ and $O_2$ with $K$ and $F_2$ gives rise to a new solution for instance $(G,\mathcal{T})$ of the $k$\textsc{-Disjoint Paths} problem.
This time, however, $K$ contains all vertices of $R$ contained in $U_1 \cup U_2$ and, at the same time, $K \cup F_2$ avoids the intersection of $U_2$ with $P_2$ entirely.
Hence, this is a contradiction to the vitality of $\mathcal{L}$ and the proof of our claim is complete.
\end{claimproof}

\paragraph{Common links.}
We have now dealt with all links and loops that could be dangerous in the sense that they might interfere with any subpath that contains a vertex of $R$ in its interior.
In this last step we are going to also tame common loops and links.
Since both of these are fundamentally the same, we only require one proof for dealing with both of them at the same time.
The core difference is, that in this case we will not be working with the graph $G'$ all the time.
Instead, finally, we will return to the graph $G$ itself because we  need some further information on the inside of vortices.
This is similar to how we took care of simple loops.

\begin{claim}\label{claim_CommonLinksAndLoops}
let $\mathcal{Q}_0$ be a homotopy class of $\mathcal{L}^{\star}$ consisting of common links or loops.
Then at most $2\beta(2k + 2d,\lceil\sqrt{ 2k + 2d }\rceil) - 2$ members of $\mathcal{Q}_0$ intersect $C_{w}$.
\end{claim}

\begin{claimproof}
Similar to previous cases, we assume towards a contradiction that at least $2\beta(2k + 2d,\lceil\sqrt{ 2k + 2d }\rceil) - 1$ members of $\mathcal{Q}_0$ intersect $C_{w}$.
Then \zcref{claim_DeepLinksMakeMesh} provides us with a closed set $\mathcal{Q}_1$ of size $\beta(2k + d,\lceil\sqrt{ 2k + d }\rceil)$ such that each member of $\mathcal{Q}_1$ intersects $C_{w}$ and there is a mesh in the strip $\Delta_1$ of $\mathcal{Q}_1$ using all cycles $C_{{w''}},\dots,C_{w}$ and subpaths of the paths in $\mathcal{Q}_1$.
It therefore follows that $\mathsf{tw}(G_2) > \beta(2k + 2d,\lceil\sqrt{ 2k + d }\rceil)$ which will be the source of our contradiction.

Let $S_1$ and $S_2$ be the two ends of $\Delta_2$ and let $v_1,v_2$ be the two vortices which contain $S_1$ and $S_2$ as segments respectively.
Notice that we may assume that both $S_1$ and $S_2$ belong to vortices because $|\mathcal{Q}_1| \geq 2$.

Since the vortex $v_i$ has depth at most $d$ for both $i\in[2]$, there exists a set $Z_i \subseteq V(\sigma_{\rho^{\star}}(v_i))$ of size at most $d$ that separates $S_i$ from the complement of the vortex society of $v_i$.
Let $G_3$ be the union of the inner graph of $\Delta_2$ together with the set $Z_i$ and all components of $\sigma_{\rho^{\star}}(v_i) - Z_i$ that contain a vertex of $S_i$ for each $i\in[2]$.
Consider the linkage $\mathcal{L}_{\mathcal{Q}} \coloneqq \mathcal{L}[G_3]$ and notice that $|\mathcal{L}| \leq k + d$ since each path $P$ in $\mathcal{L}_{\mathcal{Q}}$ either has both endpoints in $Z_1 \cup Z_2$, or at least one endpoint of $P$ is a terminal from $\mathcal{T}$.

Let $T$ be the set of all endpoints of the paths in $\mathcal{L}_{\mathcal{Q}}$.

Now notice that, $\mathcal{L}_{\mathcal{Q}}$ is a vital $T$-linkage in $G_{\mathcal{Q}}$ since otherwise, we could find a different $T$-linkage $\mathcal{Q}'$ with the same pattern in $G_3$ -- recall that $G_3$ does not contain a vertex of $R$.
By swapping the paths in $\mathcal{L}_{\mathcal{Q}}$ with those in $\mathcal{Q}'$ within $\mathcal{L}$ we would then find a new spanning $\mathcal{T}$-linkage in $G$ with would contradict the vitality of $(G,R,\mathcal{T})$.
Hence, if we let $\mathcal{T}'$ be the collection of all endpoint pairs of the paths in $\mathcal{L}_{\mathcal{Q}}$, we get that $(G_3,\mathcal{T}')$ is a vital instance of $|\mathcal{L}_{\mathcal{Q}}|$\textsc{-Disjoint Paths}.

By construction, we have that $|\mathcal{L}_{\mathcal{Q}}| \leq k + d$, which means that $|T| \leq 2k + 2d$.
With our previous observation that $\mathsf{tw}(G_2) > \beta(2k + 2d,\lceil\sqrt{ 2k + d }\rceil)$, this is a contradiction to \zcref{prop:VitalLinkage}.
\end{claimproof}

\paragraph{Transforming the instance.}
We have now finally reached the stage where we may conclude our proof.

Let $\mathcal{C}'' \coloneqq \{ C_1,\dots,C_w\}$ and let $G^{\star}_1$ be obtained from $G'$ by deleting every vertex and edge that does not either belong to a cycle of $\mathcal{C}''$ or some path from $\mathcal{L}^{\star}$ that intersects $C_w$.
Let $\widetilde{\mathcal{L}}$ be the collection of all remaining paths from $\mathcal{L}^{\star}$ in $G^{\star}$.
By \zcref{claim_RedLoops}, \zcref{claim_SimpleLoops}, and \zcref{claim_SplittingLoops} this includes all red, simple, and splitting loops.
Moreover, by \zcref{prop_TypeCounting} there are at most $4sb + 6g - 6$ homotopy types among the paths in $\widetilde{\mathcal{L}}$.
By \zcref{claim_ComplexLoopsAndRedLinks} and \zcref{claim_MixedLinks}, any homotopy class of $\widetilde{\mathcal{L}}$ consisting of complex loops, red links, or mixed links contains at most $8(2k + 1) - 2$ paths.
Moreover, by \zcref{claim_CommonLinksAndLoops}, any homotopy class consisting of common loops or links contains at most $2\beta(2k + 2d,\lceil\sqrt{ 2k + 2d }\rceil) - 2$ paths.
Hence, in total we have that
\begin{align*}
    |\widetilde{\mathcal{L}}| \leq (4sb + 6g - 6) \cdot (8(2k + 1) + 2\beta(2k + 2d,\lceil\sqrt{ 2k + 2d }\rceil) - 4).
\end{align*}
For each $L \in \widetilde{\mathcal{L}}$ let us denote by $s_l$ and $t_L$ the endpoints of $L$.
Let now $\mathcal{T}^{\star} \coloneqq \{ (s_,t_L) \mid L \in \widetilde{\mathcal{L}} \}$.
Then $(G^{\star},\mathcal{T}^{\star})$ is an instance of $|\mathcal{T}^{\star}|\textsc{-Disjoint Paths}$.
Notice that the inner graph $G_w$ of $C_w$ gives rise to a dry $w$-well as no path in $\mathcal{L}^{\star} \setminus \widetilde{\mathcal{L}}$ contained any vertex of $G_w$.
Moreover, if there is a vertex of $R$ in $G^{\star}$, then this vertex must be a terminal of $(G^{\star},\mathcal{T}^{\star})$.
Finally, we may observe that $(G^{\star},\mathcal{T}^{\star})$ is a vital instance because otherwise there was a $\mathcal{T}^{\star}$-linkage $\mathcal{Q}^{\star} \neq \mathcal{T}^{\star}$ in $G^{\star}$ and since $G^{\star}$ is disjoint from the vertices in $V(\mathcal{L}) \setminus V(\widetilde{\mathcal{L}})$, replacing the paths in $\widetilde{\mathcal{L}}$ with those in $\mathcal{Q}^{\star}$ would contradict the vitality of $\mathcal{G}$.
Thus, $(G^{\star},\mathcal{T}^{\star})$ is indeed a $(g,d,b,w)$-barren instance of $\mathsf{terminals}_{\ref{lemma_FromDesolationToBarren}}(k,g,s,b)$\textsc{-Disjoint Paths}.
Moreover, all terminals of $\mathcal{T}^{\star}$ lie on the boundaries of at most $b + 2k$ discs with empty interior and thus the bidimensionality of $(G^{\star},T^{\star})$, where $T^{\star}$ is the set of all terminals from $\mathcal{T}^{\star}$ is at most $2k + b$.
\end{proof}

\paragraph{Proof of the Vital Linkage Theorem for Spanning Linkages.}
We are finally ready for the proof of \zcref{thm_VitalSpanningLinkage}.

\begin{proof}[Proof of \zcref{thm_VitalSpanningLinkage}]
Let us first fix the function $\beta_{\ref{thm_VitalSpanningLinkage}}$.
\begin{align*}
    \beta_{\ref{thm_VitalSpanningLinkage}} (r,k) & \coloneqq \mathsf{treewidth}_{\ref{lemma_theDesolationOfAnInstance}} \left( r,k, \mathsf{well}_{\ref{lemma_FromDesolationToBarren}}\big(k,g,r,d,b, 2\beta( \mathsf{terminals}_{\ref{lemma_FromDesolationToBarren}}(t,g,r,b) , 2t + b ) + 1 \big) \right) 
\end{align*}
and
\begin{align*}
    g & \coloneqq \mathsf{genus}_{\ref{lemma_theDesolationOfAnInstance}}(r,k)\\
    d & \coloneqq \mathsf{depth}_{\ref{lemma_theDesolationOfAnInstance}}(r,k)\\
    b & \coloneqq \mathsf{depth}_{\ref{lemma_theDesolationOfAnInstance}}(r,k) \text{, and }\\
    t & \coloneqq \mathsf{terminals}_{\ref{lemma_theDesolationOfAnInstance}}(r,k).
\end{align*}
Now, assume towards a contradiction that there exists a vital instance $(G_0,R_0,\mathcal{T}_0)$ of the $k$\textsc{-Spanning Disjoint Paths} problem where $\mathsf{depth}_2(G,R) \leq r$ and $\mathsf{tw}(G) > \beta_{\ref{thm_VitalSpanningLinkage}} (r,k)$.

Now, an application of \zcref{lemma_theDesolationOfAnInstance} yields that there exists a $(g,r,d,b,w')$-desolate instance $(G_1,R_1,\mathcal{T}_1)$ of $t'$\textsc{-Spanning Disjoint Paths} where $t' \leq t$ where
\begin{align*}
    w' \coloneqq \mathsf{well}_{\ref{lemma_FromDesolationToBarren}}\big(k,g,r,d,b, 2\beta( \mathsf{terminals}_{\ref{lemma_FromDesolationToBarren}}(t,g,r,b),2t + b \big) + 1.
\end{align*}

Next, we apply \ref{lemma_FromDesolationToBarren} which yields the existence of a $(g,d,b, 2\beta( \mathsf{terminals}_{\ref{lemma_FromDesolationToBarren}}(t,g,r,b) , 2t + b ) + 1 )$-barren instance $(G_3,\mathcal{T}_3)$ of $k'$\textsc{-Disjoint Paths} where, if $T$ denotes the terminal set, $k' \leq \mathsf{terminals}_{\ref{lemma_FromDesolationToBarren}}(t,g,r,b)$ and $\mathsf{bidim}(G_3,T_3) \leq 2t + b$.

This means that there exists a surface $\Sigma$ of Euler-genus at most $g$ and a $\Sigma$-rendition $\rho$ which witnesses that $(G_3,\mathcal{T}_3)$ is barren.
Let $w \coloneqq 2\beta( \mathsf{terminals}_{\ref{lemma_FromDesolationToBarren}}(t,g,r,b) , 2t + b ) + 1$ and $\mathcal{W} = (G_w,\Omega_w,\rho_w,\mathcal{C},\mathcal{L}[G_w])$ be the well of $\rho$.
Recall that $\mathcal{W}$ is dry by definition.
Let $\mathcal{C} = \{ C_1,\dots,C_w\}$ and $\mathcal{L}$ be the unique $\mathcal{T}_3$-linkage in $G_3$.
Then we know that there exists at least one arc $P_1$ of some path $L \in \mathcal{L}$ such that $P_1$ has a non-empty intersection with $C_1$.

With \zcref{prop_fullbucketsindrainedwells} we know that there exist $w$ distinct paths $P_1,\dots,P_w$ such that $C_i \cap P_i \neq \emptyset$ for all $i \in[w]$ and it holds that $\Delta_{P_1} \subsetneq \Delta_{P_{2}} \subsetneq \dots \subsetneq \Delta_{P_w}$.
Let now for each $i \in[w]$ denote by $U_i$the subpath of $C_i$ in the inner graph of $\Delta_{P_1}$ which intersects all $P_j$ with $j \leq i$.
Finally, for each $i \in [w]$ let $B_i \coloneqq P_i \cup U_i$ and set $\mathcal{B} \coloneqq \{ B_i \mid i \in[w]\}$.
Then $\mathcal{B}$ is a bramble with $w$ elements.
Moreover, any vertex contained in some element of the bramble $\mathcal{B}$ is contained in at most two of its elements.
Hence, the order of $\mathcal{B}$ is at most $\nicefrac{w}{2} \geq \beta( \mathsf{terminals}_{\ref{lemma_FromDesolationToBarren}}(t,g,r,b), 2t + b) + 1$.
By \zcref{prop_brambles} this means that $\mathsf{tw}(G_3) \geq \beta( \mathsf{terminals}_{\ref{lemma_FromDesolationToBarren}}(t,g,r,b), 2t + b) + 1$ while $\mathsf{bidim}(G_3,T_3) \leq 2t + b$.
This is a contradiction to \zcref{prop:VitalLinkage}.
\end{proof}

\section{Irrelevant vertices for spanning linkages}\label{sec_Irrelevant}

With \zcref{thm_VitalSpanningLinkage} we have completed the core theoretical part of our algorithm.
However, in order to compute an irrelevant vertex, we need to make one more step.
Due to \zcref{thm_cliqueirrelevantvertex} we know how to find an irrelevant vertex in the presence of a large clique minor.
Next, we require a similar technique if we find a large wall, but no large clique minor associated with it.
In such a situation, the so-called \textsl{Flat Wall Theorem} is usually applied to find an area where the graph can be split into two: One part which contains the majority of the graph and the other which allows for a vortex-free rendition in a disc while hosting a large grounded wall.
The goal is to show that the centre of such a wall can always be avoided by a solution.

The way we achieve this is by proving an analogue of the ``Annulus Combing Lemma'' due to Golovach, Stamoulis, and Thilikos \cite{GolovachST2023Combing}.
To be more precise, we follow the proof of the same lemma due to Cavallaro, Gorsky, Kreutzer, Thilikos, and Wiederrecht \cite{CavallaroGKTW2026Optimal}.

The Annulus Combing Lemma says, roughly, that, given a graph with a rendition into the sphere with precisely two vortices, such that there is a large number of disjoint and grounded cycles separating the two vortices and a big flow connecting the two boundary cycles of this ``annulus'', then any (spanning) linkage can be forced to traverse through this annulus by using only the paths of this flow.
The advantage of this viewpoint is, that it allows for the linkage to be ``split'' along the flow, which allows us to decompose the linkage into two new linkages, each with a bounded number of terminals.
In a later step, we will show that in case one of the two vortices contains a large wall and has a vortex-free rendition into a disc itself, then the centre of this wall is irrelevant.

\subsection{Combing spanning linkages}\label{subsec_Combing}
We begin with the proof of the annulus combing lemma for spanning linkages.
The proof is almost identical to the original proof of Theorem 10.1 from \cite{CavallaroGKTW2026Optimal} and we include it here only for the sake of completeness since we do need to alter the statement slightly.
Before we get there, we need some additional terminology -- note that these definitions are taken directly from \cite{CavallaroGKTW2026Optimal}.

\paragraph{Flat Meshes.}
Let $n \ge 2$ be an integer.
Let $G$ be a graph, and let $M \subseteq G$ be an $n$-mesh.
We say that $M$ is a \emph{flat mesh} in $G$ if there exists a rendition $\rho$ of $G$ on the sphere with a single vortex $c_0$ such that $M$ is flat in $\rho$ and the trace of the perimeter of $M$ in $\rho$ separates all vertices in $N(\rho) \cap V(M)$ from $c_0$. 
We say that $\rho$ \emph{witnesses} the flatness of $M.$

\paragraph{Cylindrical meshes.}
In the context of the Annulus Combing Lemma we require a cylindrical variant of meshes.

Let $m,n$ be positive integers, let $M$ be a graph, and let $C_1, \ldots, C_m$ be cycles and $P_1, \ldots, P_n$ be paths in $M$ such that the following holds for all $i \in [m]$ and $j \in [n]$:
\begin{itemize}
 \item $C_1, \ldots, C_m$ are pairwise vertex-disjoint, $P_1, \ldots, P_n$ are pairwise vertex-disjoint, and $M = C_1 \cup \cdots \cup C_m \cup P_1 \cup \cdots \cup P_n$.

 \item $C_i \cap P_j$ is a path, and if $i \in \{ 1, m \}$ or $j \in \{ 1, n \}$, then $C_i \cap P_j$ has exactly one vertex,

 \item when traversing $C_i$ starting from an endpoint of $P_1\cap C_i$, then either the paths $P_1, \ldots, P_n$ are encountered in the order listed or the next $P_j$ one encounters is $P_n$ and from here the paths are encountered in the order $P_n,\dots,P_1$, and

 \item $P_j$ has one end in $C_1$ and the other in $C_m$, and when traversing $P_j$ starting from its endpoint on $C_1$, the cycles $C_1, \ldots, C_m$ are encountered in the order listed.
\end{itemize}
If the above conditions hold for $M$, we call $M$ an \emph{$(n \times m)$-cylindrical mesh}.
The cycles $C_1, \ldots, C_m$ are called the \emph{concentric cycles}, or \emph{cycles}, of $M$ and the paths $P_1, \ldots, P_n$ are called the \emph{radial paths}, or \emph{rails}, of $M$.
We also call $(n \times n)$-cylindrical meshes \emph{$n$-cylindrical meshes}.

\paragraph{The compass of a (cylindrical) mesh.}
Let $M$ be a flat mesh in a graph $G$ where $\rho$ is a sphere rendition of $G$ witnessing the flatness of $M$.
The \emph{compass} of $M$ is the inner graph of its perimeter.

Let now $G$ be a graph and $\rho$ be a $\Sigma$-rendition of $G$ for some surface $\Sigma$.
The \emph{compass} of a cylindrical mesh $M$ grounded in $\rho$ is the union of all graphs $\sigma(c)$ over the cells $c$ contained in the annulus defined by the traces of the outer- and inner-most cycle of $M$.

\paragraph{A partially embedded railed annulus.}
Let $w,r \geq 0$ be two integers.
A \emph{$(w,r)$-railed annulus} is a graph $A_{w,r}$ such that there exist pairwise vertex-disjoint cycles $C_1,\dots, C_w$, called the \emph{circles}, and pairwise vertex-disjoint paths $R_1,\dots R_r$, called the \emph{rails}, such that
\begin{enumerate}
 \item $A_{w,r} = \bigcup_{i\in [w]} C_i \cup \bigcup_{i \in [r]} R_i$ is a planar graph,
 \item $R_i$ has one endpoint on $C_1$, the other on $C_w$ for all $i\in[r]$, and is otherwise disjoint from $C_1\cup C_w$,
 \item for all $i\in[r]$, $j\in[w]$, $R_i \cap C_j$ is a path, and
 \item for each $i\in[r]$, $R_i$ visits the cycles $C_1, \dots,C_w$ in the order listed when traversing along $R_i$ starting from its endpoint on $C_1$.
\end{enumerate}
We call $C_1$ and $C_w$ the \emph{boundary circles} of $A_{w,r}$.
We say that a graph $A$ is a \emph{railed annulus} if there exist $r$ and $w$ such that $A$ is a $(w,r)$-railed annulus.

Now let $\Sigma$ be a surface and $G$ be a graph with a $\Sigma$-rendition $\rho$.
We say that a railed annulus $A \subseteq G$ is \emph{separating} in $\rho$ if
\begin{enumerate}
    \item $A$ is grounded in $\rho$,
    \item if $A$ has at least two circles and at least two rails, then there exists an annulus $\circledcirc \subseteq \Sigma$ such that the traces of the boundary circles of $A$ coincide with the two boundary components of $\circledcirc$, and the trace of every rail of $A$ is contained in $\circledcirc$,
    \item $\circledcirc$ is disjoint from the closures of the vortices of $\rho$, and
    \item for every circle $C$ of $A$, the trace of $C$ is a separating curve in $\Sigma$.
\end{enumerate}
We call the annulus $\circledcirc$ above the \emph{domain} of $A$ in $\rho$.
The \emph{crop of $G$ by $\circledcirc$} is the union of all graphs $\sigma(c)$ whose cell $c$ of $\rho$ is contained in $\circledcirc$.
Finally, if $(G,R)$ is an annotated graph with a $\Sigma$-rendition $\rho$ and a separating railed annulus $A$ with domain $\circledcirc$, we say that $A$ is \emph{blank} if the crop of $G$ by $\circledcirc$ does not contain any vertices from $R$.
In case $(G,R,\mathcal{T})$ is an instance of $k$\textsc{-Spanning Disjoint Paths} for some $k$, we also say that $A$ is \emph{terminal-free} of the crop of $G$ by $\circledcirc$ does not contain any terminals.

Let $p\geq 0$ and $q\geq 2$ be integers.
Let $A = A_{2p+q,r}$ be a railed annulus.
The \emph{$p$-trimming} of $A$ is the railed annulus $A_{q,r}$ obtained from $A$ by removing the first and last $p$ circles and restricting the rails to their minimal $C_{p+1}$-$C_{p+q}$-subpaths.

\paragraph{Combed linkages.}
Let $\Sigma$ be a surface, $G$ be a graph with a $\Sigma$-rendition $\rho$, and let $A \subseteq G$ be a separating railed annulus in $G$ with domain $\circledcirc$ in $\rho$.
Let $H$ be the crop of $G$ by $\circledcirc$.
A linkage $\mathcal{L}$ in $G$ is \emph{combed in $A$} if the intersection of $\bigcup\mathcal{L}$ and $H$ is a subgraph of the union of the rails of $A$.

With this, we are ready to state the spanning version of the Annulus Combing Lemma.

\begin{theorem}\label{thm_CombedAnus}
There exist functions $f_{\ref{thm_CombedAnus}},h_{\ref{thm_CombedAnus}},g_{\ref{thm_CombedAnus}}\colon \mathbb{N}^2 \to \mathbb{N}$ such that for every instance $(G,R,\mathcal{T})$ of $k$-\textsc{Spanning Disjoint Paths} with $\mathsf{depth}_2(G,R) \leq r$, every surface $\Sigma$, and every $\Sigma$-rendition $\rho$ of $G$, the following holds:

If $A = A_{a,s}$ is a blank, terminal-free, and separating railed annulus with $a \geq f_{\ref{thm_CombedAnus}}(k,r)$ and $s \geq g_{\ref{thm_CombedAnus}}(k,r)$ in $G$ and $\rho$ and $\mathcal{L}$ is a spanning $\mathcal{T}$-linkage, then there exists a spanning $\mathcal{T}$-linkage $\mathcal{L}^{\star}$ in $G$ such that
\begin{enumerate}
 \item $\mathcal{L}^{\star}$ is combed in the $h_{\ref{thm_CombedAnus}}(k,r)$-trimming of $A$, and
 \item if $\circledcirc$ is the domain of $A$ and $G'$ is the crop of $G$ by $\circledcirc$, then $\bigcup\mathcal{L}^{\star} - G' \subseteq \bigcup\mathcal{L} - G'$.
\end{enumerate}

Moreover, $f_{\ref{thm_CombedAnus}}(k,r), h_{\ref{thm_CombedAnus}}(k,r), g_{\ref{thm_CombedAnus}}(k,r) \in 2^{\poly(k + r)}$.
\end{theorem}

Before we dive into the proof, we require another small lemma form the work of Cavallaro et al.

Let $\Sigma$ be a surface. 
A \emph{topological linkage} $\mathcal{T}$ in $\Sigma$ with pattern $\tau$ is a set of pairwise disjoint simple curves such that $\mathcal{T}$ contains for each pair $\{a,b\} \in \tau$ a curve with endpoints $a$ and $b$.
We denote the \emph{pattern} of $\mathcal{T}$ by $\tau(\mathcal{T})$.

A pattern $\tau$ is \emph{realisable} (in $G$) if there is a linkage $\mathcal{L}$ in $G$ with $\tau(\mathcal{L}) = \tau$ and it is \emph{topologically feasible} (in $\Sigma$) if there is a topological linkage $\mathcal{T}$ in $\Sigma$ with $\tau(\mathcal{T}) = \tau$.

\begin{proposition}[Cavallaro, Gorsky, Kreutzer, Thilikos, and Wiederrecht \cite{CavallaroGKTW2026Optimal}]\label{prop_two-sided-cylinder}
 Let $G$ be a graph with a vortex-free rendition $\rho$ into a disc $\Delta$ and let $\mathcal{C} \coloneqq \{ C_1,
 \dots, C_t\}$ with $t \geq 2k$ a sequence of disjoint grounded cycles in $G$ such that the intersection of $\mathsf{bd}(\Delta)$ and $\rho$ is the set of nodes of $C_1$ and for all $i < j \in[2,t]$, the trace of $C_i$ separates the trace of $C_1$ from the trace of $C_j$.
 Let $\mathcal{P}$ be a $V(C_t){-}V(C_1)$-linkage of order $2k$ in $G$ such that $V(C_i) \cap V(P)$ is a path for all $1 \leq i \leq t$ and all $P\in \mathcal{P}$.
 Finally, let $\tau$ be a pattern of order $k$ on $(V(C_t) \cup V(C_1)) \cap V(\mathcal{P})$.

 If $\tau$ is topologically feasible then it is realisable in $G$.
\end{proposition}

\begin{proof}[Proof of \zcref{thm_CombedAnus}]
Let us start by defining our functions.
We set
\begin{align*}
 f_{\ref{thm_CombedAnus}}(k,r) & \coloneqq 16\beta_{\ref{thm_VitalSpanningLinkage}}(r,k) + 16\\
 h_{\ref{thm_CombedAnus}}(k,r) & \coloneqq 8\beta_{\ref{thm_VitalSpanningLinkage}}(r,k) + 8\text{, and}\\
 g_{\ref{thm_CombedAnus}}(k,r) & \coloneqq 6\beta_{\ref{thm_VitalSpanningLinkage}}(r,k) + 6.
\end{align*}
The claimed bounds now follow directly from \zcref{thm_VitalSpanningLinkage}.
\smallskip

Let $\mathcal{C} = \langle C_1,\dots, C_a\rangle$ be the circles of $A$ and let $G_0 \coloneqq \bigcup\mathcal{L} \cup \bigcup\mathcal{C}$ and notice that $G_0$ still contains all vertices of $R_0 \coloneqq R$.
We further assume that $\mathcal{L}$ is chosen to minimise $|E(\mathcal{L}) \setminus E(\mathcal{C})|$ among all spanning $\mathcal{T}$-linkages in $G$.
\smallskip

Let $\circledcirc$ be the domain of $A$ in $\Sigma$ and let $G_A$ be the crop of $G_0$ by $\circledcirc$.
Let further $\mathcal{S} \coloneqq \mathcal{L}[G_A]$ be the set of all $(V(C_1) \cup V(C_a))$-subpaths of the paths in $\mathcal{L}$ contained in $G_A$.
A path in $\mathcal{S}$ is called a \emph{river} if it is a $V(C_1)$-$V(C_a)$-path, a \emph{mountain} if it is a $V(C_1)$-path, and a \emph{valley} if it is a $V(C_a)$-path.

\paragraph{Obtaining a vital instance.}
Let $(G_1,R_1)$ be obtained from $(G,R)$ by iteratively picking a vertex $v \in V(G)\setminus R$ of degree exactly $2$ and contracting one of its two incident edges, whenever at one endpoint of the contracted edge belongs to $R$, we add the resulting contraction vertex to $R_1$.
Finally, every vertex of $R \cap V(G_2)$ is added to $R_1$.
In a second step, let $(G_2,R_2)$ be obtained from $(G_1,R_1)$ by contracting each remaining edge that belongs to some circle of $A$ \textsl{and} some path in $\mathcal{L}$ and setting $R_2 \coloneqq R_1$.
Notice that $G'$ still contains a sequence $\mathcal{C}' =\langle C'_1,\dots, C'_a \rangle$ of vertex-disjoint cycles such that $C_i'$ is obtained from $C_i$ for each $i\in[a]$.
Moreover, if $x$ is an endpoint of some path in $\mathcal{L}$, then $x$ is not a vertex of any circle of $A$ since $A$ is blank and thus, $x$ has degree $1$ in $G_1$.
It follows that $(G_2,R_2)$ contains a spanning linkage $\mathcal{L}_2$ such that $\tau(\mathcal{L}_2) = \tau(\mathcal{L})$.
This is true in particular since the construction step to go from $(G_1,R_1)$ to $(G_2,R_2)$ only contracted edges that already belonged to $\mathcal{L}$ and thus ensured that the paths in $\mathcal{L}$ remain paths.
Moreover, each such edge was an edge of $\mathcal{C}$ which means, since $G_A$ is blank, that all vertices of $R_1$ remain vertices of $G_2$.

Finally, notice that, since $(G_2,R_2)$ is a red-minor of $(G,R)$, we may inherit a $\Sigma$-rendition $\rho_2$ of $G_2$ from the restriction of $\rho$ to $G_0$.

\begin{claim}\label{claim_VitalAnus}
The spanning linkage $\mathcal{L}_2$ is vital in $(G_2,R_2)$.
\end{claim}

\begin{claimproof}
Suppose there is a vertex $v\in V(G_2)$ not used by $\mathcal{L}_2$.
Then $v \notin R_2$ and thus, there exists $i\in[a]$ such that $v \in V(C_i')\setminus V(\mathcal{L}_2)$.
But then, $v$ must have degree $2$ in $G_2$ which is ruled out by the construction of $G_2$ from $G_1$.
Hence, we obtain that $V(\mathcal{L}_2) = V(G_2)$.

Indeed, we may make a stronger observation:
Let $v\in V(C'_i)$ for some $i\in[\ell]$, then $v$ has degree $4$ in $G_2$ since $v$ has degree at least $3$, may be incident with at most four edges of $G_2$, and two of those edges are required to belong to $\mathcal{L}_2$.
Since, by construction, $\mathcal{L}_2$ cannot share an edge with $C_i'$, it follows that $v$ cannot be of degree $3$.

Now suppose that $\mathcal{L}_2$ is not unique in $G_2$.
Then there exists a spanning linkage $\mathcal{Q}$ in $G_2$ with the same pattern such that $\mathcal{L}_2 \neq \mathcal{Q}$.
Suppose that there exists an edge $e$ in $E(\mathcal{L}_2) \setminus E(\mathcal{Q}_2)$.
Then, since $(G_2,R_2)$ is a red-minor of $(G,R)$ with $|R_2| = |R|$ and all endpoints of the paths in $\mathcal{L}$ still have degree $1$ in $G_2$, there would be a spanning linkage $\mathcal{Q}'$ in $(G,R)$ of the same pattern as $\mathcal{L}$ such that $E(\mathcal{Q}) \subseteq (E(\mathcal{L}) \cup E(\mathcal{C})) \setminus \{ e \}$.
Hence, $|E(\mathcal{Q}) \setminus E(\mathcal{C})| < |E(\mathcal{L}) \setminus E(\mathcal{C})|$ which is a contradiction to the minimality of $E(\mathcal{L}) \cap E(\mathcal{C})$.
Therefore, $\mathcal{L}_2$ must be unique in $(G_2,R_2)$ must therefore be vital.
\end{claimproof}

\begin{claim}\label{claim_BidimensionalitySurvives}
We have that $\mathsf{depth}_2(G_2,R_2) \leq r$.
\end{claim}

\begin{claimproof}
The claim follows directly from the fact that $\mathsf{depth}_2$ is monotone under taking red-minors and $(G_2,R_2)$ is indeed a red-minor of $(G,R)$.
\end{claimproof}

We say that two rivers $J_1$ and $J_2$ are \emph{consecutive}, if deleting their traces from $\circledcirc$ leaves one component $\zeta'$, whose closure is homeomorphic to a disc, that does not intersect the trace of any other river in $\mathcal{S}$.
We define what it means for two rails of $A'$ to be \emph{consecutive} in the same way.
The closure $\zeta = \mathsf{cl}(\zeta')$ of such a disc obtained from the traces of two consecutive rails or rivers is called a \emph{slice} of $\circledcirc$.
Note that the rivers in $\mathcal{S}$ decompose $\circledcirc$ into $|\mathcal{S}|$ slices $\zeta_1,\dots,\zeta_{\ell}$ where $\ell$ is the number of rivers in $\mathcal{S}$.
Moreover, for every mountain or valley $P$ in $\mathcal{S}$ there exists a unique $i\in[\ell]$ such that the trace of $P$ is contained in $\zeta_i$.
For each $i\in [\ell]$ let us denote by $\mathcal{M}_i$ the set of all mountains whose trace is contained in $\zeta_i$ and by $\mathcal{V}_i$ the set of all valleys whose trace is contained in $\zeta_i$.

Our goal is to argue that mountains and valleys cannot enter deeply into the circles of $A$.
In order to be able to make this argument, we require some further structure for the ways rivers and mountains can interact with the circles of $A$.
To achieve this, we aim to apply \zcref{prop_drywell} which requires us to confine the curves of mountains and valleys into wells.
Since the arguments for dealing with valleys is analogous to those for mountains, we only discuss the mountain case below.

\paragraph{The wells in the mountains have run dry.}
Fix any $i\in[\ell]$.
Let $H_i$ be the graph $\bigcup \mathcal{M}_i \cup \mathcal{C}'$.
Then one can easily see that $H_i$ is a planar graph with a plane embedding $\Gamma_i$ inherited from $\rho$ and we may declare $\Omega_i$ to be the cyclic ordering of the vertices of $C'_1$ obtained by traversing along $C'_1$ in clockwise order.
Then $(H_i,\Omega_i,\rho,\mathcal{C}',\mathcal{M}_i)$ is a well and by \zcref{prop_drywell} we may assume this well to be dry.
We call $(H_i,\Omega_i,\rho,\mathcal{C}',\mathcal{M}_i)$ the \emph{well associated with $\mathcal{M}_i$}.
Indeed, notice that the new collection of paths returned by \zcref{prop_drywell} must still have all of its traces contained in $\zeta_i$ and is therefore guaranteed to stay disjoint from the rivers of $\mathcal{S}$.

This means that the trace of any mountain in $\mathcal{M}_i$ encloses a disc with the boundary of $\zeta_i$ which cannot contain the trace of any valley from $\mathcal{V}_i$.
Hence, we may treat valleys in exactly the same way and assume their respective wells to also be dry.

\paragraph{Mountains are flat and valleys are shallow.}
With the assumption that each of the wells associated with $\mathcal{M}_i$ or $\mathcal{V}_i$ is dry, we are now ready to venture deeper into the proof.

\begin{claim}\label{claim_FlatMountains}
Let $M$ be a mountain and $V$ be a valley in $\mathcal{S}$.
Then $M$ is disjoint from $C'_{2\beta(k,r) + 3}$ and $V$ is disjoint from $C'_{a - 2\beta(k,r) - 2}$.
\end{claim}

\begin{claimproof}
Since there is still no difference between mountains and valleys apart from the choice of indices, it suffices to make the arguments for mountains and valleys will follow analogously.

Suppose there exist $i\in[\ell]$ and a mountain $M \in \mathcal{M}_i$ which intersects $C'_{2\beta_{\ref{thm_VitalSpanningLinkage}}(r,k)+3}$.
Then, since the well $(H_i,\Omega_i,\rho,\mathcal{C}',\mathcal{M}_i)$ is dry, we know by \zcref{prop_drywell} that there exists a sequence $\langle M_1,\dots,M_{2\beta_{\ref{thm_VitalSpanningLinkage}}(r,k)+3} \rangle$ of pairwise vertex-disjoint mountains in $\mathcal{M}_i$ such that for each $j\in [2\beta_{\ref{thm_VitalSpanningLinkage}}(r,k)+3]$, $M_j$ intersects $C'_j$.
Notice that, whenever a mountain $M_j$ intersects the circle $C'_j$, then $M_j$ also intersects $C'_{j'}$ for all $j' \leq j$.

Now let us define, for each $j \in [2\beta_{\ref{thm_VitalSpanningLinkage}}(r,k) + 3]$ the graph $B_j$ to be $M_j \cup C_j'$ and let $\mathcal{B} \coloneqq \{ B_j \mid j \in [2\beta_{\ref{thm_VitalSpanningLinkage}}(r,k) + 3] \}$.
Notice that every member of $\mathcal{B}$ is connected, and any two members must intersect by the remark above.
Moreover, it follows from the fact that the $M_j$ are pairwise vertex-disjoint and the $C_j'$ are also pairwise vertex disjoint, that no vertex of $G'$ may be contained in more than two distinct members of $\mathcal{B}$.
Hence, $\mathcal{B}$ is a bramble of order at least $\beta_{\ref{thm_VitalSpanningLinkage}}(r,k) + 2$.
Therefore, \zcref{prop_brambles} says that $\mathsf{tw}(G') \geq \beta_{\ref{thm_VitalSpanningLinkage}}(r,k) + 1$ which contradicts the fact that $\mathcal{L}_2$ is vital by \zcref{thm_VitalSpanningLinkage}.
\end{claimproof}

\paragraph{Counting rivers.}
The only remaining members of $\mathcal{S}$ we have not yet taken care of are the rivers.
These must necessarily intersect all cycles in $\mathcal{C}'$, however, the same argument as the one from \zcref{claim_CommonLinksAndLoops} shows that the total number of rivers must be bounded.
Let $\mathcal{R} \subseteq \mathcal{S}$ be the set of all rivers in $\mathcal{S}$.

\begin{claim}\label{claim_FewRivers}
We have that $|\mathcal{R}| \leq 2\beta_{\ref{thm_VitalSpanningLinkage}}(r,k) + 2$.
\end{claim}

We omit this proof of \zcref{claim_FewRivers} since it is almost entirely analogue to the proofs of \zcref{claim_CommonLinksAndLoops} and \zcref{claim_FlatMountains}.
The core idea is, once again, to build a large order bramble which will then contradict the vitality of $\mathcal{L}_2$.

\paragraph{Straightening rivers.}
Notice that $G_2$ still contains an annulus $A_1$ with $q$ circles and $r$ rails which is obtained from $A$ through the contractions that transformed $G_0$ into $G_2$.

Let $A_2$ be the $(2\beta_{\ref{thm_VitalSpanningLinkage}}(r,k) + 2)$-trimming of $A_1$, let $\circledcirc_2$ be the domain of $A_2$, and let $\mathcal{F}$ be the set of rails of $A_2$.
Notice that the set of circles of $A_2$ is the set $$\mathcal{C}_2 \coloneqq \{ C_{2\beta_{\ref{thm_VitalSpanningLinkage}}(r,k) + 3}',\dots, C_{a - 2\beta_{\ref{thm_VitalSpanningLinkage}}(r,k) - 2}' \}.$$
Let us set $a_2 \coloneqq a - 4\beta_{\ref{thm_VitalSpanningLinkage}}(r,k) - 4$ and adjust indices to set $\mathcal{C}_2 = \{ C_1^2,\dots, C_{a_2}^2 \}$ such that the order of the indices is preserved.

Let $A_2^1$ be the railed annulus obtained from $A_2$ by first restricting its circles to the cycles $C_1^2,\dots,C_{6\beta_{\ref{thm_VitalSpanningLinkage}}(r,k) + 6}^2$, and then taking, as rails, from each $F \in \mathcal{F}$ its minimal $C_1^2$-$C_{6\beta_{\ref{thm_VitalSpanningLinkage}}(r,k) + 6}^2$-subpath, forming the set $\mathcal{F}_1$.
Let $\circledcirc_2^1$ be the domain of $A_2^1$.

Similarly, let $A_2^2$ be the railed annulus obtained from $A_2$ by taking as circles the cycles $C_{a_2 - 6\beta_{\ref{thm_VitalSpanningLinkage}}(r,k) - 5}^2,\dots, C_{a_2}^2$ and by taking as rails from each $F \in \mathcal{F}$ its minimal $C_{a_2 - 6\beta_{\ref{thm_VitalSpanningLinkage}}(r,k) - 5}^2$-$C_{a_2}^2$-subpath, thereby forming the set $\mathcal{F}_2$.
Let $\circledcirc_2^2$ be the domain of $A_2^2$.

Finally, let $\mathcal{S}_2$ be the collection of all $(V(C_1^2) \cup V(C_{a_2}^2))$-subpaths of the paths in $\mathcal{L}_2$ which are contained in the crop of $G_2$ by $\circledcirc_2$.
As before, we could partition $\mathcal{S}_2$ into mountains, valleys, and rivers, but by the construction of $A_2$ and \zcref{claim_FlatMountains} it follows that every member of $\mathcal{S}_2$ is a river.
Moreover, $|\mathcal{S}_2| \leq 2\beta_{\ref{thm_VitalSpanningLinkage}}(r,k) + 2$ by \zcref{claim_FewRivers}.

Let for each $i\in[2]$, $\mathcal{S}^i_2$ be the restriction of the paths in $\mathcal{S}_2$ to the crop of $G_2$ by $\circledcirc^i_2$.
We denote the crop of $G_2$ by $\circledcirc^i_2$ by $G_2^i$.
In addition, let $X_1$ be the set of endpoints of the paths in $\mathcal{S}^1_2$ on $C_1^2$ and let $X_2$ be the set of endpoints of the paths in $\mathcal{S}^2_2$ on $C_{a_2}^2$.

We say that a set of $\mathcal{D}$ rails is \emph{consecutive} if $\mathcal{D} = \{ D_1,\dots,D_{|\mathcal{D}|}\}$ such that $D_i$ and $D_{i+1}$ are consecutive for all $i\in [|\mathcal{D}|-1]$.
Let $Y_1$ be the set of endpoints of $2\beta_{\ref{thm_VitalSpanningLinkage}}(r,k) + 2$ consecutive rails of $A_2^1$ on $C_{6\beta_{\ref{thm_VitalSpanningLinkage}}(r,k) + 6}^2$.
Finally, let $Y_2$ be the set of endpoints of the rails of $A_2^2$ on $C_{a_2 - 6\beta_{\ref{thm_VitalSpanningLinkage}}(r,k) - 5}^2$.

\begin{claim}\label{claim_RailTheAnnulus1}
There exists a linkage $\mathcal{J}_1$ of size $|\mathcal{S}_2| = |X_1|$ between $X_1$ and $Y_1$ in $G_2^1$.
\end{claim}

\begin{claimproof}
Suppose there does not exist a $X_1$-$Y_1$-linkage of order $X_1$ in $G_2^1$.
By Menger's theorem, there exists a set $Z_1 \subseteq V(G^1_2)$ of size at most $|X_1| - 1$ such that in $G^1_2 - Z_1$ there is no $X_1$-$Y_1$-path.

Since $|X_1| \leq 2\beta_{\ref{thm_VitalSpanningLinkage}}(r,k) + 2$ as a consequence of \zcref{claim_FewRivers} and $|Y_1| = |\mathcal{F}| = s \geq 2\beta_{\ref{thm_VitalSpanningLinkage}}(r,k) + 2$ there must exist a rail $F$ of $A^1_2$ which completely avoids $Z_1$.
Moreover, there must also exist a river $S \in \mathcal{S}^1_2$ which is disjoint from $Z_1$ and finally, there must exist a circle $C$ of $A^1_2$ which is disjoint from $Z_1$.
Then, $R \cup S \cup C$ is a connected graph which contains a vertex of $X_1$ and a vertex of $Y_1$, which are endpoints of $S$ and $F$ respectively.
This is a contradiction to the fact that $Z_1$ is an $X_1$-$Y_1$-separator and the claim follows.
\end{claimproof}

Now let $Y_2'$ be the subset of $Y_2$ consisting of precisely those vertices in $Y_2$ that belong to members of $\mathcal{F}$ that contain an endpoint of some path of $\mathcal{J}_1$.
Then $|Y_2'| = |\mathcal{J}_1| = |X_2|$.
Moreover, let $H_2'$ be the subgraph of $G_2^2$ consisting only of the circles of $A_2^2$, the rails of $A_2^2$ with vertices in $Y_2^2$, and the paths in $\mathcal{S}_2^2$.

\begin{claim}\label{claim_RailTheAnnulus2}
There exists a linkage $\mathcal{J}_2$ of size $|\mathcal{S}_2| = |X_2|$ between $X_2$ and $Y_2'$ in $H_2'$.
Moreover, for every path $S \in \mathcal{S}_2$, if $x \in Y_1$ is the endpoint of the path in $\mathcal{J}_1$ containing an endpoint of $S$ and $y \in Y_2'$ is the endpoint of the path in $\mathcal{J}_2$ containing the other endpoint of $S$, then there exists a rail of $A_2$ which contains both $x$ and $y$.
\end{claim}

\begin{claimproof}
The proof of this claim follows along two steps.
First we route the paths in $\mathcal{S}'$ starting from $X_2$ onto the rails containing the end of $\mathcal{F}_1$ in $Y_1$, and then we use the remaining infrastructure in $A_2'$ in order to guarantee that the vertices of $X_2$ are connected to the rails in the way required by the claim.

Let now $B_1$ be the railed annulus obtained from $A^2_2$ by taking the cycles $C_{a'-2\beta_{\ref{thm_VitalSpanningLinkage}}(r,k)-1}^2,\dots,C_{a_2}^2$ together with the restriction of the rails of $A_2^2$ to these cycles.
Let $\circledcirc^1_3$ be the domain of $B_1$ and let $H_1$ be the crop of $G_2$ by $\circledcirc^1_3$.

Moreover, let $B_2$ be the railed annulus obtained from $A_2^2$ by first deleting $B_1$ and then iteratively removing vertices of degree $1$.
Let $\circledcirc^2_3$ be the domain of $B_2$ and let $H_2$ be the crop of $G_2$ by $\circledcirc^2_3$.
Notice that $B_2$ has $4\beta_{\ref{thm_VitalSpanningLinkage}}(r,k) + 4$ cycles.

We denote by $\mathcal{F}^{\star}$ the set of rails of $A_2$ containing vertices from $Y_2'$ which are precisely those that contain the endpoints of the paths from $\mathcal{J}_1$ in $Y_1$.
Now let $Z_1$ be the set of all ends of those rails of $B_1$ on $C_{a_2-2\beta_{\ref{thm_VitalSpanningLinkage}}(r,k)-1}^2$ which are subpaths of the paths in $\mathcal{F}^{\star}$.
Similarly, let $Z_2$ be the set of ends endpoints of those rails of $B_2$ on the cycle $C_{a_2-2\beta_{\ref{thm_VitalSpanningLinkage}}(r,k)-2}^2$ which are subpaths of the paths in $\mathcal{F}^{\star}$.

Finally, let $\mathcal{J}_2''$ be the collection of all $Z_1$-$Z_2$-subpaths of the paths in $\mathcal{J}^{\star}$.
\medskip

\textbf{Re-routing in $B_1$.}
We can now use arguments analogous to those from the proof of \zcref{claim_RailTheAnnulus1} to find an $X_2$-$Z_1$-linkage $\mathcal{J}_2'$ of size $|X_2|$ in $H_1$.
As before, the existence of this linkage is guaranteed by Menger's theorem and the fact that $|X_2| = |Z_1| \leq 2\beta_{\ref{thm_VitalSpanningLinkage}}(r,k) + 2$.
\smallskip

Indeed, it is possible to combine the linkages $\mathcal{J}_2'$ and $\mathcal{J}_2''$ along the rails in $\mathcal{F}^{\star}$ to create an $X_2$-$Z_2$-linkage of size $|X_2|$.
Hence, we only have to ensure that the second part of our claim is satisfied.
That is, we need to route the vertices in $Z_2$ to the correct vertices in $Y_2'$ as specified by $\mathcal{J}_1$.

\paragraph{Fixing the connections in $B_2$.}
Recall that $Z_2$ denotes the set of endpoints of the rails of $B_2$ on the cycle $C_{a'-2\beta_{\ref{thm_VitalSpanningLinkage}}(r,k)-2}^2$ which are subpaths of the paths in $\mathcal{F}^{\star}$.
Moreover, $Y_2'$ denotes the set of endpoints of the rails of $B_2$ on the cycle $C_{a_2-6\beta_{\ref{thm_VitalSpanningLinkage}}(r,k)-5}^2$.
Similarly, $X_2'$ denotes the set of the endpoints of the paths in $\mathcal{J}_2'$ on the cycle $C_{a_2}^2$ and $Z_1$ denotes the set of endpoints of the paths in $\mathcal{J}_2'$ on the cycle $C_{a_2-2\beta_{\ref{thm_VitalSpanningLinkage}}(r,k)-1}^2$ while $Y_1'$ denotes the endpoints of the paths in $\mathcal{J}_1$ on the cycle $C_{2\beta_{\ref{thm_VitalSpanningLinkage}}(r,k) + 2}^2$.

Let us now set $s^{\star} \coloneqq |\mathcal{F}^{\star}|$ and $\mathcal{F}^{\star} = \{ F^{\star}_1,\dots,F^{\star}_{r^{\star}} \}$ such that, when restricting the rails of $A_2$ to $\mathcal{F}^{\star}$, the rails $F^{\star}_i$ and $F^{\star}_{i+1}$ are consecutive.
We derive from this ordering indices for the elements of $Y_1'$, $Y_2'$, $Z_1$, and $Z_2$ as follows.
We number $Y_1' = \{ y^1_1,\dots,y^1_{s^{\star}}\}$ and $Y_2' = \{ y^2_1,\dots,y^2_{s^{\star}} \}$ such that $y^j_i$ lies on $F^{\star}_i$ for both $j\in[2]$ and all $i\in[s^{\star}]$.
We also number the vertices in $X_1 = \{ x^1_1,\dots,x^1_{s^{\star}}\}$ such that $x^1_i$ is the endpoint of the path in $\mathcal{J}_1$ that contains $y^1_i$ for each $i\in[s^{\star}]$.
Then, we number the vertices $X_2 = \{x^2_1,\dots,x^2_{s^{\star}}\}$ such that $x^1_i$ and $x^2_i$ belong to the same path, denoted by $S_i$, of $\mathcal{S}_2$ for each $i\in [s^{\star}]$.
Finally, we number $Z_1 = \{ z^1_1,\dots, z^1_{s^{\star}}\}$ and $Z_2 = \{ z^2_1,\dots, z^2_{s^{\star}}\}$ such that $z^1_i$ and $z^2_i$ belong to the path $F^{\star}_i$.

Let us now define the bijection $\pi \colon [s^{\star}] \to [s^{\star}]$ such that for each $i\in[s^{\star}]$, $z^2_i$ lies on $F^{\star}_{\pi(i)}$.

It is now entirely possible that $i \neq \pi(i)$, however, in order to complete our proof we need to find a $Z_2$-$Y_2'$-linkage in $B_2$ which connects $z^2_{\pi(i)}$ to $y^2_i$ for each $i\in[s^{\star}]$.
However, we claim that within $\circledcirc_3^2$, this linkage is topologically feasible.
To see this notice that we know we may use the paths from $\mathcal{J}_1$ and $\mathcal{J}_2'\cup\mathcal{J}_2''$ to redraw the curves obtained from the traces of the paths in $\mathcal{S}_2$ to pass through the points of $Y_2'$ and $Z_2$.
It follows now that for each of those curves, the remaining curve between the two points in $Y_2'$ and $Z_2$ may now be confined to be contained entirely in $\circledcirc^2_3$ and this is possible while keeping the collection of these curves being a topological linkage.

With the pattern of order $s^{\star}$ defined as $\big\{ \{ y^2_i,z^2_{\pi(i)} \} \mid i\in[s^{\star}] \big\}$ being topologically feasible in $\circledcirc$, and the fact that $B_2$ is a railed annulus with at least $2r^{\star}$ circles and at least $2s^{\star}$ rails, it follows from \zcref{prop_two-sided-cylinder} that it is realisable in $B_2$, leading to a linkage $\mathcal{J}_2'''$.
Now, $\mathcal{J}_2 \coloneqq \mathcal{J}_2' \cup \mathcal{J}_2'' \cup \mathcal{J}_2'''$ is our desired linkage.
\end{claimproof}

Let $Y_1'$ denote the subset of $Y_1$ consisting only of endpoints of paths from $\mathcal{J}_1$.
To complete the proof of the theorem, we now observe that $\mathcal{F}^{\star}$ contains a $Y_1'$-$Y_2'$-linkage $\mathcal{J}_3$ of size $|Y_1'| = |Y_2'| = |\mathcal{F}^{\star}|$ that is internally vertex-disjoint from both $\mathcal{J}_1$ and $\mathcal{J}_2$.
By combining $\mathcal{J}_1$, $\mathcal{J}_2$, and $\mathcal{J}_3$ into a single linkage $\mathcal{S}^{\star}$, we have now found a set of rivers that is combed in the $(6\beta_{\ref{thm_VitalSpanningLinkage}}(r,k) + 6)$-trimming $A^{\leftmoon}$ of $A_2$.
This means that, by replacing $\mathcal{S}_2$ with $\mathcal{S}^{\star}$ we may obtain a linkage $\mathcal{L}^{\leftmoon}$ which is combed in $A^{\leftmoon}$ and which satisfies $\tau(\mathcal{L}^{\leftmoon}) = \tau(\mathcal{L}_2)$.
Since $G_2$ is a minor of $G$ and $A^{\leftmoon}$ is a minor of the $(8\beta_{\ref{thm_VitalSpanningLinkage}}(r,k) + 8)$-trimming $A^{\star}$ of $A$, it follows that $\mathcal{L}^{\leftmoon}$ corresponds to a linkage $\mathcal{L}^{\star}$ in $G$ which is combed on $A^{\star}$ and has the same pattern as $\mathcal{L}$.
With this, the proof is complete.
\end{proof}

\subsection{Avoiding an insulated vertex}\label{subsec_Irrelevant}
With \zcref{thm_CombedAnus} we are now able to control the interaction of the graph enclosed by some flat wall with the outside and, in particular, with solutions to spanning linkage problems.
That is, when we are given a large enough flat wall, we know that -- apart from the apex set -- we only need to consider a bounded number of additional paths that enter the inner regions of the wall.
The next step is to turn this intuition into an actual algorithm, thereby proving that the centre of a large enough wall can always be avoided by a solution for $k$\textsc{-Spanning Disjoint Paths}, given that $\mathsf{depth}_2$ is bounded.

Before we continue, we require some more definitions.

\paragraph{Layers and central submeshes.}
Let $r \geq 3$ be an integer and $M$ be an $(r \times r)$-mesh.
We define the layers of $M$ recursively as follows:
The first \emph{layer} of $M$ is its perimeter.
For $i \in [2,\lfloor \nicefrac{r}{2} \rfloor]$, the \emph{$i$th layer} of $M$ is the $(i-1)$th layer of the $((r-1) \times (r-1))$-mesh obtained from $M$ by removing its perimeter and iteratively removing vertices of degree $1$.

Given an integer $q \in [3,r]$ such that $r-q$ is even, the \emph{$q$-central submesh} of $M$ is the $(q \times q)$-mesh obtained from $M$ by removing the first $\nicefrac{r-q}{2}$ layers of $M$ and then iteratively deleting vertices of degree at most $1$.

\paragraph{Bridges and attachments.}
Let $H$ be a subgraph of a graph $G$.
An \emph{$H$-bridge} in $G$ is a connected subgraph $J$ of $G$ such that $E(J) \cap E(H) = \emptyset$ and either $E(J)$ consists of a unique edge with both ends in $H$, or 
$J$ is constructed from a component $C$ of $G - V(H)$ and the non-empty set of edges $F \subseteq E(G)$ with one end in $V(C)$ and the other in $V(H)$, by taking the union of $C$, the endpoints of the edges in $F$, and $F$ itself.
Notice that the $H$-bridges induce a partition of $E(G)\setminus E(H)$.
The vertices in $V(J) \cap V(H)$ are called the \emph{attachments} of $B$ and the set $V(J)\setminus V(H)$ is called the \emph{interior} of $J$.

\paragraph{Homogeneous meshes.}
The following definitions are due to Gorsky, Seweryn, and Wiederrecht \cite{GorskySW2026Price} who recently showed that homogeneity of flat meshes with respect to a set of $q$ colours can be achieved within polynomial bounds.

Let $r\geq 3$ be an integer, $M$ be an $(r \times r)$-mesh in $G$ and let $C_P$ be the perimeter of $M$.
Let $C\subseteq M$ be any cycle of $M$.
The \emph{compass} of $C$ is the union of $C$ together with all $C$-bridges in $G$ that are entirely contained in the compass of $M$.
We denote the compass of $C$ with respect to $\rho$ by $\mathsf{compass}_M(C)$ and we drop the $\rho$ in the subscript if it is understood from the context.
The \emph{interior} of a cycle $C\subseteq M$, denoted by $\mathsf{int}_M(C)$ is the graph $\mathsf{compass}_{\rho}(C)-C$.
We also drop the $M$ in the subscript if the wall is understood from the context.
\smallskip

A \emph{$q$-colorful graph} is a pair $(G,\chi)$ where $G$ is a graph and $\chi\colon V(G) \to 2^{[q]}$ assigns to each vertex a, possibly empty, set of at most $q$ colours.
Notice that the case $q=1$ is precisely the case of annotated graphs.
We explicitly allow for $q$ to be $0$.

Given an arbitrary $q$-colorful graph $(G,\chi)$, a set $X\subseteq V(G)$ and a subgraph $H\subseteq G$, we define
\begin{align*}
 \chi(X) \coloneq \bigcup_{v\in X} \chi(v),
\end{align*}
as well as $\chi(H) \coloneqq \chi (V(H))$.
Note that this should not be confused with the chromatic number of $H$, which is of no concern to us.
Finally, we write $(H,\chi)$ for the \emph{$q$-colorful subgraph} of $G$, where we implicitly restrict $\chi$ to the vertex set of $H$.
\medskip

Let $q$ be a non-negative integer and $(G,\chi)$ be a $q$-colorful graph.
Let $r\geq 3$ be an integer and $M$ be an $(r \times r)$-mesh in $G$ such that $M$ is flat in $G$ witnessed by a sphere-rendition $\rho$.
We say that $M$ is \emph{homogeneous} in $(G,\chi)$ if for every cycle $C\subseteq M$ it holds that
\begin{align*}
 \chi( \mathsf{int} (C) ) = \chi( \mathsf{compass} (M) ).
\end{align*}

We stress that the definition for homogeneity above, as well as the following definitions, are slightly more general than they would need to be for our purposes.
However, the results we will be using from the literature have been proven in this general setting and to adapt them to simpler variants sufficient for our setting appears tedious.
While such an adaptation is clearly possible, it would not provide any substantial improvement for the involved functions as the general bound of $2^{\poly(k + r)}$ is already best possible up to the degree of the polynomial.
Therefore, we opt for a minimisation in technical proofs in exchange for slightly overloaded definitions.
The one simplification we are going to take, however, is the same we have done for our statement of \zcref{prop_DisjointPathsIrrelevant}.
That is, we are not taking full advantage of the \textsc{Folio} problem and instead constrain ourselves to the definition of strong $(k,R)$-irrelevancy as given right before the statement of \zcref{prop_DisjointPathsIrrelevant}.
\smallskip

The main problem now is to provide a definition for the colours we assign to our vertices in order to model the properties required for our case.
The definitions below are much more general than necessary as they provide enough additional information to extend from \textsc{Spanning Disjoint Paths} to a folio variant.
It might well be possible to remove any dependencies on the respective collections of rooted minors inside the cells, but as mentioned before, to remove these dependencies one would need to open the black box of the Graph Minors Algorithm further.

\paragraph{Rooted graphs.}
Let $k\geq 0$ be an integer.
A \emph{$k$-rooted graph} is a pair $(G,\mathcal{R})$ where $\mathcal{R} = \{ r_1,r_2,\dots,r_k \} \subseteq V(G)$ is a multiset of $k$ not necessarily distinct vertices.
We sometimes write $(G,r_1,\dots,r_k)$ instead of $(G,\mathcal{R})$.
The $r_i$ are called the \emph{roots} of $(G,\mathcal{R})$.
A pair $(G,\mathcal{R})$ is a \emph{rooted graph} if it is a $k$-rooted graph for some $k$.

We stress that there is a huge difference between rooted graphs and annotated graphs.
While at first glance both concepts might look fairly similar, in an annotated graph, all annotated vertices are -- essentially -- interchangeable.
In some sense one may interpret a $k$-rooted graph $(G,\mathcal{R})$ as a $k$-colorful graph where for each $i\in[k]$ there is a unique vertex of $G$ that carries colour $i$.
In this sense, the idea of rooted minors is fairly similar to the idea of red-minors, or more general of colorful minors of $k$-colorful graphs.
There is one key difference though and that is that we are no longer allowed to forget that some vertices were annotated -- or carried a colour.
This means, for rooted minors, all roots are required to be preserved in some way.

\paragraph{The $d$-folio of a rooted graph.}
A rooted graph $(H,q_1,\dots,q_k)$ is a \emph{rooted minor} of another rooted graph $(G,r_1,\dots,r_k)$ if there exists a minor model $\varphi$ of $H$ in $G$ such that for every $i\in[k]$ it holds that $r_i \in \varphi(q_i)$.

Given an integer $d\geq 0$ and a rooted graph $(G,\mathcal{R})$, the \emph{$d$-folio} of $G$, denoted as $d\text{-}\mathsf{folio}(G,\mathcal{R})$, is the set of all rooted graphs $(H,\mathcal{Q})$ of detail at most $d$ that are rooted minors of $(G,\mathcal{R})$.

\paragraph{Augmenting cells.}
Let $a,r \in \mathbb{N}$ with $r \geq 3$.
Let $G$ be a graph, $A \subseteq V(G)$ be a set of at most $a$ vertices, and let $M$ be an $(r \times r)$-mesh in $G$ that is flat in $G-A$ witnessed by the sphere rendition $\rho$.
We denote by $C(M)$ the collection of all cells contained in the disc defined by the trace of the perimeter of $M$ which contains all non-perimeter vertices of $M$ of degree $3$.

For each cell $c \in C(M)$ we consider a labelling $\lambda_c \colon N(c) \to [3]$ such that the set of labels assigned by $\lambda_c$ to $N(c)$ is one of $[1]$, $[2]$, or $[3]$.
In addition, we consider a bijection $\alpha_A \colon A \to [a]$.

The labellings $\lambda_c$ and $\alpha_A$ will be used to fix labels for the boundaries of the annotated graphs $(G[V(\sigma(c)) \cup A],N(c) \cup A)$ to translate them into rooted graphs.

Let $c \in C(M)$, we define the ordering $\Omega_c \coloneqq \langle x_1,\dots,x_{\ell}\rangle$ with $\ell \leq 3$ of the vertices of $N(c)$ such that the following holds:
\begin{enumerate}
 \item $\langle x_1,\dots,x_q\rangle$ is a counter clockwise cyclic ordering of the vertices of $N(c)$ as they appear on the boundary of $c$ in $\rho$ and
 \item for each $i\in[q]$ we have that $\lambda_c(x_i) = i$.
\end{enumerate}

For each cell $c \in C(W)$ we fix $\pi_c \colon A \cup N(c) \to [a + |N(c)|]$ such that $\pi_c(x) = \alpha_A(x)$ for all $x \in A$ and $\pi_c(x) - a = \lambda_c(x)$ for all $x \in N(c)$.
Finally, the \emph{augmented cell} $c$ is defined as the rooted graph
\begin{align*}
 \mathbf{R}_c \coloneqq (G[V(\sigma(c) \cup A)],r_1,\dots,r_{a + |N(c)|})
\end{align*}
where $r_i = \pi_c^{-1}(i)$ for all $i \in [ a + |N(c)| ]$.
We denote the graph $G[V(\sigma(c) \cup A)]$ by $R_c$.

\paragraph{Palettes of augmented cells.}
With each $c \in C(M)$ we associate a set of graphs as follows.

Let $d \geq 0$ be an integer.
We denote by $d\text{-}\mathsf{palette}(c)$ the set $d\text{-}\mathsf{folio}(\mathbf{R}_c)$ which is referred to as the \emph{$d$-palette} of $c$.

Notice that $| d\text{-}\mathsf{palette}(c) | \in 2^{\mathbf{O}((a + d)^2)}$.
In the following, the set of all rooted graphs of detail at most $d$ with at most $a + 3$\, roots is treated as the set of colours for our colorful graph.
This means, in particular, that $q \in 2^{(a + d)^2}$.

\paragraph{A $d$-folio homogeneous wall.}
Let $a,r \in \mathbb{N}$ with $r \geq 3$.
Let $G$ be a graph $A \subseteq V(G)$ a set of at most $a$ vertices and let $M$ be an $(r \times r)$-mesh in $G$ that is flat in $G-A$ witnessed by the sphere rendition $\rho$.
We define the \emph{$(d,A,M,\rho)$-folio augmentation} of $G$, denoted by $(G^{(d,A,M,\rho)},\phi)$, as the colorful graph obtained from $G$ by introducing, for each $c \in C(M)$, a vertex $x_c$ adjacent to exactly the vertices of $N(c)$, setting $\phi(v) \coloneqq \emptyset$ for all $v\in V(G)$ and $\phi(c) \coloneqq d\text{-}\mathsf{palette}(c)$ for every $c \in C(M)$.
Notice that $\rho$ may be extended to a sphere rendition $\rho'$ of $G^{(d,A,M,\rho)} - A$ such that $\rho$ is the restriction of $\rho'$ to $G - A \subseteq G^{(d,A,M,\rho)} - A$ simply by setting $\sigma_{\rho'} \coloneqq \sigma(c) + x_c$ for all $c\in C(M)$ and $\sigma_{\rho'}(c) \coloneqq \sigma(c)$ for all $c \in C(\rho) \setminus C(M)$.
We call $\rho'$ the \emph{$\rho$-augmentation} for $G^{(d,A,M,\rho)}$.

In particular, $M$ is still flat in $G^{(d,A,M,\rho)} - A$ as witnessed by $\rho'$.

Now, $M$ is said to be \emph{$(d,A)$-folio homogeneous} in $G$ if $M$ is homogeneous in $(G^{(d,A,W,\rho)} - A,\phi)$.

\paragraph{Suspended meshes.}
Let $r\geq 4$.
Given a flat mesh $M$ in a graph $G$ witnessed by a sphere rendition $\rho$, we say that a cell $c\in C(M)$ is an \emph{inner cell} of $M$ if the second and third layer of $M$ separate $N(c)$ from the perimeter of $M$ in $G$.
\smallskip

Let $r \geq 4$ be an integer.
Let $G$ be a graph and $M \subseteq G$ be an $(r \times r)$-mesh in $G$.
We say that $M$ is a \emph{suspended $(r\times r)$-mesh} -- or simply a \emph{suspended mesh} -- in $G$ if there exists a sphere-rendition $\rho$ of $G-A$ such that $M$ is flat in $\rho$ and for every cell $c\in C(M)$ the following two properties hold:
\begin{enumerate}
 \item for every non-empty set $S \subseteq N(c)$ there exists a connected subgraph of $\sigma(c)$ containing all vertices from $S$ and avoiding all vertices from $N(c)\setminus S$ and
 \item there exist $|N(c)|$ vertex-disjoint paths from $N(c)$ to the perimeter of $M$.
\end{enumerate}
We say that $\rho$ \emph{witnesses} that $M$ is suspended in $G$.

\begin{proposition}\label{prop_irrelevantVertex}
There exists a function $f_{\ref{prop_irrelevantVertex}} \colon \mathbb{N}^4 \to \mathbb{N}$ such that for all integers $d,k,b,r \in \mathbb{N}$, if $(G,R)$ is an annotated graph of $\mathsf{bidim}(G,R) \leq b$ with $A \subseteq V(G)$ containing an $(f_{\ref{prop_irrelevantVertex}}(b,d,k,r) \times f_{\ref{prop_irrelevantVertex}}(b,d,k,r))$-mesh such that $M$ is suspended and $(4d,A)$-folio homogeneous in $G - A$ as witnessed by the sphere rendition $\rho$ of $G - A$ such that the compass of $M$ does not contain a vertex of $R$, then the following holds:

If $\Delta_r$ is the disc bounded by the perimeter of the $r$-central submesh $M'$ of $M$ which contains all nodes of $M'$ in $\rho$ and $$X \coloneqq \bigcup_{\substack{c \in C(\rho) \\ c \subseteq \Delta_r}} V(\sigma(c)),$$ then every $v \in X$ is strongly irrelevant for $(k,d)\text{-}\mathsf{folio}(G,R)$.

Moreover, $f_{\ref{prop_irrelevantVertex}}(b,d,k,r) \in 2^{\poly(b + d)} \cdot \poly(k + r)$.
\end{proposition}

\paragraph{Finding an irrelevant vertex.}
We are now ready to combine all of the imported technology above with \zcref{thm_CombedAnus} to prove that the centre of any large enough suspended flat and blank mesh is irrelevant for the \textsc{Spanning Disjoint Paths} problem.

\begin{theorem}\label{thm_findIrrelevantVertexFlatWall}
There exists a function $f_{\ref{thm_findIrrelevantVertexFlatWall}} \colon \mathbb{N}^3 \to \mathbb{N}$ that takes as input an instance $(G,R,\mathcal{T})$ of $k$\textsc{-Spanning Disjoint Paths} with $\mathsf{depth}_2(G,R) \leq r$, a set $A \subseteq V(G)$ with $|A| \leq a$, and a suspended, blank, and $(0,A)$-folio homogeneous flat $(m \times m)$-mesh $M$ in $G-A$ witnessed by the sphere-rendition $\rho$ of $G-A$, such that the compass of $M$ does not contain any terminal, and $m \geq f_{\ref{thm_findIrrelevantVertexFlatWall}}(r,a,k)$ and finds a vertex $v \in V(G)\setminus R$ such that $(G,R,\mathcal{T})$ is a \textsc{yes}-instance if and only if $(G-v,R,\mathcal{T})$ is a \textsc{yes}-instance in time $2^{\poly(r + a + k)}|G|$.
\end{theorem}

\begin{proof}
We proceed in the following steps.
First, we show that $M$ contains a large submesh $M_1$ which is still flat in $G-A$ and whose compass is blank and free of terminals.
Then, as the second step, we divide $M_1$ into $2a+2$ regions as follows:
The first $2a+1$ regions are made up from families of layers of $M_1$ while the $2a+2$nd region is a central submesh $M_2$ of $M_1$.
Finally, we show that the central region of $M_2$ is irrelevant to our instance.
The purpose of the first $2a+1$ regions is to find a railed annulus whose compass avoids the interaction of a fixed solution with the apex set entirely.
This is needed in order to apply \zcref{thm_CombedAnus}.

Before we dive in, let us clarify the function $f_{\ref{thm_findIrrelevantVertexFlatWall}}$:
\begin{align*}
f_{\ref{thm_findIrrelevantVertexFlatWall}}(r,a,k) \coloneqq (4a+2)(f_{\ref{thm_CombedAnus}}(k,r) + 1) + 4g_{\ref{thm_CombedAnus}}(k,r) + f_{\ref{prop_irrelevantVertex}}(2a+1,0, 2a + 2\beta_{\ref{thm_VitalSpanningLinkage}}(2a + k,r) + 2,3).
\end{align*}

\paragraph{Insulating the central region.}
Let us fix $m_0 \coloneqq (2a+1)(f_{\ref{thm_CombedAnus}}(k,r) + 1)$ and let $\mathcal{C}_0 = \{ C_{1} ,\dots , C_{m_0}\}$ be the first $m_0$ layers of $M$ numbered in order such that $C_1$ is the perimeter of $M$.

Notice that there exists a set $\mathcal{P}_0'$ of $n_0 \coloneqq g_{\ref{thm_CombedAnus}}(k,r)$ paths in $M$ that start on the perimeter of the $(2g_{\ref{thm_CombedAnus}}(k,r)) + f_{\ref{prop_irrelevantVertex}}(2a+1,0, 2a + 2\beta_{\ref{thm_VitalSpanningLinkage}}(2a + k,r) + 2,3))$-central submesh $M_1$ of $M$, end on $C_1$, and such that $C_i \cap P$ is a path for every $i \in [m_0]$.
Let us denote by $\mathcal{P}_0$ the set of the $C_1$-$C_{m_0}$-subpaths of the paths in $\mathcal{P}_0'$.
Now $\mathcal{C}_0$ together with $\mathcal{P}_0$ forms an $(m_0,n_0)$-railed annulus with domain $\circledcirc_0$.
Moreover, if we consider the restriction $\rho_0$ of $\rho$ to $\circledcirc_0$, we can see that $A_0$ is a partially embedded, blank, and its compass does not contain any terminal of $\mathcal{T}$.

Let $m_1 \coloneqq f_{\ref{thm_CombedAnus}}(k,r)$.
Then there exist $2a+1$ pairwise vertex-disjoint $(m_1,n_0)$-railed annuli $A_1,\dots,A_{2a+1}$ in $A_0$ where, for every $i \in [2a + 1]$, every circle of $A_i$ is a member of $\mathcal{C}$ and every rail of $A_i$ is a subpath of some path in $\mathcal{P}_0$.
Moreover, the $A_i$ may be chosen in a way such that there are $2a$ cycles $U_1,\dots,U_{2a} \in \mathcal{C}$ where $U_i = C_{i \cdot f_{\ref{thm_CombedAnus}}(k,r) + 1}$ and the trace of $U_i$ separates $A_1,\dots,A_{i}$ from $A_{i+1},\dots,A_{2a+1}$ in $\circledcirc_0$.

Finally, let $M_2$ be the $(f_{\ref{prop_irrelevantVertex}}(2a+1,0, 2a + 2\beta_{\ref{thm_VitalSpanningLinkage}}(2a + k,r) + 2,3))$-central submesh of $M_1$.

\paragraph{The central vertex.}
Let $\Delta_3$ be the disc bounded by the perimeter of the $3$-central submesh $M_3$ of $M_2$ which contains all nodes of $M_3$ in $\rho$.
Let furthermore
\begin{align*}
    X \coloneqq \bigcup_{\substack{c \in C(\rho) \\ c \subseteq \Delta_3}} V(\sigma(c)),
\end{align*}
and let $v \in X$ be any vertex.
\medskip

We claim that $(G,R,\mathcal{T})$ is a \textsf{yes}-instance if and only if $(G-v,R,\mathcal{T})$ is a \textsf{yes}-instance.
\smallskip

To show that this claim is true, suppose towards a contradiction that $(G-v,R,\mathcal{T})$ is a \textsf{no}-instance while there exists a spanning $\mathcal{T}$-linkage $\mathcal{L}$ in $(G,R,\mathcal{T})$.

\paragraph{Finding an insulating annulus.}
Notice that every vertex $x \in A$ has degree $0$, $1$, or $2$ in $\mathcal{L}$, where we say that $x$ has degree $0$ if either $x = s_i = t_i$ for some $(s_i,t_i) \in \mathcal{T}$, or $x \notin V(\mathcal{L})$.
In either of these cases, no vertex contained in $X$ can interact with $x$ and thus we may disregard all vertices of $A$ with degree $0$ in $\mathcal{L}$.

For all other vertices from $A$ it might happen that their neighbour(s) in $\mathcal{L}$ belong to the compass of one of the $A_i$.
However, since every $x \in A$ has degree at most $2$ in $\mathcal{L}$, there can be at most $2a$ many indices $i \in[2a+1]$ such that $A_i$ contains some vertex of $\mathcal{L}$ that is adjacent -- in $\mathcal{L}$ --  to a vertex from $A$.
Hence, there must be $j \in [2a + 1]$ such that the compass of $A_j$ does not contain a vertex from $N_{\mathcal{L}}(x)$ for any $x \in A$.
From this discussion it follows that we may take $A = \{ x_1,\dots,x_a\}$ and then perform the same sort of operation as we did in the proof of \zcref{thm_VitalSpanningLinkage} in order to reduce to an instance without apices.
That is, starting from $\mathcal{T}$ we iterate over all $i \in [a]$ and whenever we reach $x_i$, if $x_i$ has degree $0$ in $\mathcal{L}$ we just delete $x_i$ and possibly all corresponding terminal pairs.
If $x_i$ has degree $1$ in $\mathcal{L}$, then there is $(s,t) \in \mathcal{T}$ with $x \in \{ s,t\}$, without loss of generality assume $x = s$.
Now let $y$ be the unique neighbour of $x_i$ in $\mathcal{L}$.
We delete $x_i$ and replace $(x,t)$ with $(y,t)$ in $\mathcal{T}$.
Finally, if $x_i$ has degree $2$, then there exist neighbours $y$ and $z$ of $x_i$ in $\mathcal{L}$ and $(s,t) \in \mathcal{T}$ such that $y$ separates $s$ and $x_i$ in $\mathcal{L}$ while $z$ separates $x_i$ and $t$ in $\mathcal{L}$.
In this case, we delete $x_i$ and replace $(s,t)$ with the pairs $(s,y)$ and $(z,t)$.
Let $(G-A,R',\mathcal{T}')$ be the resulting instance of $|\mathcal{T}'|$\textsc{-Spanning Disjoint Paths} where $R' \coloneqq R\setminus A$, and let $T'$ denote the set of terminals.

\paragraph{Reducing the instance.}
By our choice of $A_j$, we know that in $(G-A,R',\mathcal{T}')$, the compass of $A_j$ contains neither a vertex of $R'$ nor of $T'$.
Moreover, $A_j$ is separating in $G-A$.
Hence, we may now apply \zcref{thm_CombedAnus}.

This yields a spanning $\mathcal{T}'$-linkage $\mathcal{L}^{\star'}$ in $G-A$ such that
\begin{enumerate}
 \item $\mathcal{L}^{\star'}$ is combed in the $h_{\ref{thm_CombedAnus}}(k,r)$-trimming $A_j^{\star}$ of $A_j$, and
 \item if $\circledcirc_1$ denotes the domain of $A_j$ and $G_j$ is the crop of $G-A$ by $\circledcirc_1$, then $\bigcup\mathcal{L}^{\star'} - G_j \subseteq \bigcup\mathcal{L} - G_j$.
\end{enumerate}
Moreover, we know from \zcref{claim_FewRivers} that the total number of rails in $A_j^{\star}$ that contain vertices from $\mathcal{L}^{\star'}$ is at most $2\beta_{\ref{thm_VitalSpanningLinkage}}(2a + k,r) + 2$.
We recover a spanning $\mathcal{T}$-linkage $\mathcal{L}^{\star}$ by re-inserting the vertices of $A$ into the paths of $\mathcal{L}^{\star'}$ according to the way $\mathcal{T}'$ was created.

We now further edit our instance as follows:
Let $C^{\fullmoon}$ be the circle of $A^{\star}$ with the highest index when seen as a circle of $A_0$.
That is, no circle of $A^{\star}$ separates $C^{\fullmoon}$ from $M_3$ and $C^{\fullmoon}$ separates all other circles of $A^{\star}$ from $M_3$ in $G-A$.

Let $G^{\fullmoon'}$ be the inner graph of $C^{\fullmoon}$ in $\rho$ and let $G^{\fullmoon} \coloneqq G[V(G^{\fullmoon'}) \cup A]$.
Then, for each path rail $F$ of $A_j$ that contains a vertex of $\mathcal{L}^{\star}$, select some vertex $y_F$ to be the endpoint of $F$ on $C^{\fullmoon}$.
Notice that, if we let $\Omega^{\fullmoon}$ be the cyclic ordering of $V(C^{\fullmoon})$ obtained by traversing along $C^{\fullmoon}$ in clockwise direction, then $(G^{\fullmoon}-A,\Omega^{\fullmoon})$ has a vortex-free rendition $\rho^{\fullmoon}$ in the disc $\Delta^{\fullmoon}$.

Let $T^{\fullmoon}$ be the collection of all vertices of $T\cup A$ in $G^{\fullmoon}$ together with all vertices $y_F$ as above.
Then $|T^{\fullmoon}| \leq 2\beta_{\ref{thm_VitalSpanningLinkage}}(2a + k,r) + 2 + a$ since no vertex of $T$ is contained in the compass of $M$.
Moreover, since all but at most $a$ vertices of $T^{\fullmoon}$ are contained in $C^{\fullmoon}$ it follows that $\mathsf{bidim}(G^{\fullmoon},T^{\fullmoon}) \leq a + 1$.
Now notice that for each path $L \in \mathcal{L}^{\star}$, every component $L'$ of $G^{\fullmoon} \cap L$ is a $T^{\fullmoon}$-path.
That is, both endpoints of $L'$ belong to $T^{\fullmoon}$, let us call them $s_{L'}$ and $t_{L'}$.
Let 
\begin{align*}
\mathcal{T}^{\fullmoon} \coloneqq \{ (s_{L'},t_{L'}) \mid \text{ there exists } L\in\mathcal{L}^{\star} \text{ such that } L' \text{ is a component of } L \cap G^{\fullmoon} \}.
\end{align*}

\paragraph{Recovering the linkage.}
Notice that, if $\mathcal{Q}$ is any $\mathcal{T}^{\fullmoon}$-linkage in $G^{\fullmoon}$, then we may replace $\mathcal{L}^{\star} \cap G^{\fullmoon}$ with $\mathcal{Q}$ and obtain a spanning $\mathcal{T}$-linkage in $G$.

Recall that $M_2$ is the $(f_{\ref{prop_irrelevantVertex}}2a+1,0, 2a + 2\beta_{\ref{thm_VitalSpanningLinkage}}(2a + k,r) + 2,3)$-central submesh of $M_1$.
Moreover, by our assumption, $M_1$ is $(0,A)$-folio homogeneous.
Since $|T^{\fullmoon}| \leq 2a + 2\beta_{\ref{thm_VitalSpanningLinkage}}(2a + k,r) + $ and $\mathsf{bidim}(G^{\fullmoon},T^{\fullmoon}) \leq 2a + 1$ we may now apply \zcref{prop_irrelevantVertex} to $(G^{\fullmoon},\mathcal{T}^{\fullmoon})$ and $M_2$.
This, however, means that every vertex in $X$, including our vertex $v$, is irrelevant
Hence, there exists a $\mathcal{T}^{\fullmoon}$-linkage $\mathcal{Q}$ in $G^{\fullmoon}$ that does not contain $v$,
By our discussion above this means that there exists a spanning $\mathcal{T}$-linkage $\mathcal{Q}^{\star}$ in $G$ that avoids $v$.
Hence, $\mathcal{Q}^{\star}$ is also a spanning $\mathcal{T}$-linkage in $G-v$ and thus, $(G - v,R,\mathcal{T})$ is in fact a \textsf{yes}-instance of $k$\textsc{-Spanning Disjoint Paths}.
As this is a clear contradiction to our assumptions, our proof is now complete.
\end{proof}

\section{The Spanning Linkage Algorithm}\label{sec_algorithm}

We now have everything in place to finalise the algorithmic part of our main theorems.
Indeed, \zcref{lemma_WahtToDoWithClique} and \zcref{thm_findIrrelevantVertexFlatWall} provide the foundation for finding the irrelevant vertex.
All we need to do is to ensure that, in any graph of large treewidth we can either find a large clique minor or a large suspended and homogeneous flat mesh efficiently.
Once this is done, the data reduction part of our algorithm is established.
Finally, we require a routine that solves $k$\textsc{-Spanning Disjoint Paths} in the regime of bounded treewidth.

We begin by describing the algorithm overall and then, in the subsequent subsections, fill in the two missing steps outlined above.
We describe the algorithm in slightly more detail, elaborating on some crucial subroutines as we go.
One addition we are making is the use of a theorem due to Reed and Wood \cite{ReedW2009Lineartime} which allows us to quickly find a large clique minor or to determine that the number of edges in our instance is linear in the number of vertices.
The only reason we have this step is to ensure that our final running time is $\mathbf{O}(n^2)$ rather than $\mathbf{O}(n \cdot m)$.

\begin{proposition}[Reed and Wood \cite{ReedW2009Lineartime}]\label{prop_quicklyfindingaclique}
 Let $t$ be a positive integer and let $G$ be a graph with $|\!|G|\!| \geq 2^{t-3}|G|$.
 Then one can find a $K_t$-minor model in $G$ in $\mathbf{O}(t|G|)$-time. 
\end{proposition}

\paragraph{The Spanning Linkage Algorithm.}
~

\textsl{Input:} An annotated graph $(G,R)$ together with a set $\mathcal{T} = \{ (s_i,t_i) \mid i\in[k] \}$ of $k$ terminal pairs and an integer $r \geq 1$.

\begin{description}
    \item[Step 1.] \textsl{Input:} An instance $(G,R,\mathcal{T})$ of $k$\textsc{-Spanning Disjoint paths}.

    \textsl{Procedure:} Check if $G$ contains a $K_t$-minor model $\varphi$ with $t \coloneqq \lfloor \nicefrac{5}{2} (2(r+1)^2 + 2k) \rfloor + 1$ using \zcref{prop_quicklyfindingaclique}. If the outcome is \textsc{yes}, forward $(G,R,\mathcal{T})$ and $\varphi$ to \textbf{Step 1.1}, otherwise proceed to \textbf{Step 2}.

    \begin{description}
        \item[Step 1.1.] \textsl{Input:} An instance $(G,R,\mathcal{T})$ of $k$\textsc{-Spanning Disjoint paths} where $\mathsf{depth}_2(G,R) \leq r$ together with a $K_t$-minor model $\varphi$ in $G$.

        \textsl{Procedure:} Apply \zcref{lemma_WahtToDoWithClique} to either find an irrelevant vertex $v \in V(G)\setminus R$ or determine that $\mathsf{depth}_2(G,R) > r$.
        In the second case \textbf{reject} the instance and in the first case forward $(G-v,R,\mathcal{T})$ to \textbf{Step 1}.
    \end{description}

    \item[Step 2:] \textsl{Input:} An instance $(G,R,\mathcal{T})$ of $k$\textsc{-Spanning Disjoint paths}.

    \textsl{Procedure:} Apply the algorithm from \zcref{prop_FindGoodFlatWall} to find one of the following outcomes:
    \begin{enumerate}
        \item a $K_t$-minor model $\varphi$ in $G$,
        \item a tree-decomposition $\mathcal{D}$ of width at most $f_{\ref{prop_FindGoodFlatWall}}(t) \cdot f_{\ref{thm_IrrelevantInMesh}}(r,k)$ for $G$, or
        \item a set $A \subseteq V(G)$, where $|A| \leq g_{\ref{prop_FindGoodFlatWall}(t)}$ together with an $(f_{\ref{thm_IrrelevantInMesh}}(r,k) \times f_{\ref{thm_IrrelevantInMesh}}(r,k))$-mesh $M$ that is flat and suspended in $G-A$ as witnessed by the sphere-rendition $\rho$ and the treewidth of $M$ in $G$ is at most $f_{\ref{prop_FindGoodFlatWall}}(t) \cdot f_{\ref{thm_IrrelevantInMesh}}(r,k)$.
    \end{enumerate}
    In case i), forward $(G,R,\mathcal{T})$ and $\varphi$ to \textbf{Step 1.1}.
    In case ii), forward $(G,R,\mathcal{T})$ to \textbf{Step 4}, and in case iii) forward $(G,R,\mathcal{T})$, $A$, $M$, and $\rho$ to \textbf{Step 3}.

    \item[Step 3.] \textsl{Input:} An instance $(G,R,\mathcal{T})$ of $k$\textsc{-Spanning Disjoint paths}, together with a set $A \subseteq V(G)\setminus R$ where $|A| \leq g_{\ref{prop_FindGoodFlatWall}(t)}$ together with an $(f_{\ref{thm_IrrelevantInMesh}}(r,k) \times f_{\ref{thm_IrrelevantInMesh}}(r,k))$-mesh $M$ that is flat and suspended in $G-A$ as witnessed by the sphere-rendition $\rho$ and the treewidth of $M$ in $G$ is at most $f_{\ref{prop_FindGoodFlatWall}}(t) \cdot f_{\ref{thm_IrrelevantInMesh}}(r,k)$.

    \textsl{Procedure:}
    Run the algorithm from \zcref{thm_IrrelevantInMesh} to either (i) determine that $\mathsf{depth}_2(G,R) > r$, in which case we \textbf{reject} the instance, or (ii) find a vertex $v \in V(G) \setminus R$ such that $(G,R,\mathcal{T})$ is a \textsc{yes}-instance if and only if $(G-v,R,\mathcal{T})$.
    In case (ii) forward $(G-v,R,\mathcal{T})$ to \textbf{Step 2}.

    \item[Step 4.]
    \textsl{Input:} An instance $(G,R,\mathcal{T})$ of $k$\textsc{-Spanning Disjoint paths} where $\mathsf{tw}(G) \leq f_{\ref{prop_FindGoodFlatWall}}(t) \cdot f_{\ref{thm_IrrelevantInMesh}}(r,k)$.

    \textsl{Procedure:} Apply the algorithm from \zcref{lemma_SmallTreewidthAlgorithm} to $(G,R,\mathcal{T})$ to determine whether $(G,R,\mathcal{T})$ is a \textsc{yes}- or \textsc{no}-instance. \textbf{Return} the outcome.
\end{description}

With the \textbf{Spanning Linkage Algorithm} established, we are ready to state the main theorem of this section.

\begin{theorem}\label{thm_SpanningLinkageAlgorithm}
The \textbf{Spanning Linkage Algorithm} correctly decides the $k$\textsc{-Spanning Disjoint Paths} problem on instances of $\mathsf{depth}_2$ at most $r$ in time $2^{2^{ \poly(r+k) }} \cdot n^2$.
\end{theorem}

Notice that the \textbf{Spanning Linkage Algorithm} closes a loop only if an irrelevant vertex has been found.
Hence, the algorithm iterates through at most $n$ loops.
Since all subroutines called throughout the algorithm have linear running times, this immediately yields the desired $\mathbf{O}(n^2)$ bound.
The double-exponential dependency on $r + k$ then follows from the individual subroutines.
Finally, to see that the algorithm is correct, one only needs to observe that -- under the assumption that $\mathsf{depth}_2(G,R) \leq 2$ -- each subroutine correctly finds an irrelevant vertex or produces the desired outcome.
Hence, all that is left to do is prove \zcref{thm_IrrelevantInMesh} and \zcref{lemma_SmallTreewidthAlgorithm} to complete the proof of \zcref{thm_SpanningLinkageAlgorithm}.

We stress that, although \zcref{thm_SpanningLinkageAlgorithm} is stated for the decision version of \textsc{$k$-Spanning Disjoint Paths}, the algorithm from \zcref{lemma_SmallTreewidthAlgorithm} can be modified to also find a solution in the same time and thus, the \textbf{Spanning Linkage Algorithm} is able to find a solution -- given that one exists.

\subsection{Finding an irrelevant vertex}\label{subsec_IrrelevantVertexFinding}
The goal of this subsection is to piece together several parts of the previously established toolset with some additional tools from the literature in order to find an irrelevant vertex whenever we are presented with a large flat mesh.

\begin{theorem}\label{thm_IrrelevantInMesh}
There exists a function $f_{\ref{thm_IrrelevantInMesh}} \colon \mathbb{N}^3 \to \mathbb{N}$, and an algorithm that takes as input an integer $r\geq 1$, an instance $(G,R,\mathcal{T})$ of $k$\textsc{-Spanning Disjoint Paths}, a set $A \subseteq V(G)$ of size at most $a$, and a flat and suspended $(f_{\ref{thm_IrrelevantInMesh}}(r,k,a) \times f_{\ref{thm_IrrelevantInMesh}}(r,k,a))$-mesh $M$ in $G-A$ witnessed by a sphere-rendition $\rho$ of $G-A$, such that the treewidth of $M$ in $G$ is at most $w$ and finds one of the following outcomes in time $2^{\poly(r + k + w + a)}(|G| + |\!|G|\!|)$:
\begin{enumerate}
    \item $\mathsf{depth}_2(G,R) > r$, or
    \item a vertex $v \in V(G)\setminus R$ such that $(G,R,\mathcal{T})$ is a \yes-instance if and only if $(G-v,R,\mathcal{T})$.
\end{enumerate}
Moreover, $f_{\ref{thm_IrrelevantInMesh}}(r,k,a) \in 2^{\poly(r + k + a)}$.
\end{theorem}

\paragraph{Quickly finding a suspended flat mesh.}
One subroutine we are still missing is a way to quickly find a suspended flat mesh.
Indeed, for this we require two additional results from the literature.

The first allows us to immediately reduce any flat mesh to one with a blank compass.

\begin{proposition}[Cavallaro, Gorsky, Kreutzer, Thilikos, and Wiederrecht \cite{CavallaroGKTW2026Optimal}]\label{lemma_BlankWall}
For all integers $r \geq 2$ and $b \geq 1$ and every annotated graph $(G,R)$, if $M$ is a flat $(w \times w)$-mesh with $w \geq (3b + 1)r + 1$ in $G$ as witnessed by the sphere-rendition $\rho$, then one of the following holds:
\begin{enumerate}
 \item $\mathsf{bidim}(G,R) \geq b$, or
 \item there exists an $(r \times r)$-submesh $M'$ of $M$ such that the compass of $M'$ does not contain a vertex from $R$.
\end{enumerate}
Moreover, if $\rho$ witnesses that $M$ is suspended in $G$, then $\rho$ also witnesses that $M'$ is suspended in $G$.
There also exists an algorithm that, given $G$, $\rho$, and $M$ as input finds either a red $((b+1) \times (b+1))$-grid as a red-minor in $G$ or the mesh $M'$ as above in time $\poly(b+r)|\!|G|\!|$.
\end{proposition}

The second is the one that actually finds the suspended flat mesh.

\begin{proposition}[Sau, Stamoulis, and Thilikos \cite{SauST2024More}]\label{prop_FindGoodFlatWall}
There exists a functions $f_{\ref{prop_FindGoodFlatWall}},g_{\ref{prop_FindGoodFlatWall}} \colon \mathbb{N} \to \mathbb{N}$ such that for every integer $t \geq 1$, all odd integers $r \geq 5$, and every graph $G$ one of the following holds:
\begin{enumerate}
 \item $G$ contains a $K_t$ minor,
 \item $G$ has a tree-decomposition of width at most at most $f_{\ref{prop_FindGoodFlatWall}}(t) \cdot r$, or
 \item there is a set $A \subseteq V(G)$, where $|A| \leq g_{\ref{prop_FindGoodFlatWall}(t)}$ together with an $(r \times r)$-mesh $M$ that is flat and suspended in $G-A$ as witnessed by the sphere-rendition $\rho$ and the treewidth of $M$ in $G$ is at most $f_{\ref{prop_FindGoodFlatWall}}(t) \cdot r$.
\end{enumerate}
Moreover, $f_{\ref{prop_FindGoodFlatWall}}(t) \in 2^{\mathbf{O}(t^2 \log t)}$, $g_{\ref{prop_FindGoodFlatWall}}(t) \in \mathbf{O}(t^{24})$, and there exists an algorithm that finds one of the three outcomes above in time $2^{2^{\mathbf{O}(t^2 \log t)} \cdot r \log r + \mathbf{O}(r^2)} \cdot |G| + 2^{ 2^{\mathbf{O}(t^2 \log t)} \cdot r^3 \log r }$.
\end{proposition}

\paragraph{Finding a folio homogeneous mesh.}
Let $G$ be a graph and $M \subseteq G$ be a flat wall in $G$ witnessed by the sphere rendition $\rho$.
The \emph{treewidth of $M$ in $G$} is the smallest integer $w$ such that the treewidth of the compass of $M$ in $\rho$ is at most $w$.

With this, we are finally ready to import the key tools from the work of Cavallaro, Gorsky, Kreutzer, Thilikos, and Wiederrecht \cite{CavallaroGKTW2026Optimal} that will drive our later machinery.

\begin{proposition}[Cavallaro, Gorsky, Kreutzer, Thilikos, and Wiederrecht \cite{CavallaroGKTW2026Optimal}]\label{prop_FolioHomoWall}
There exists a function $f_{\ref{prop_FolioHomoWall}} \colon \mathbb{N}^{3} \to \mathbb{N}$ such that for all integers $r\geq 3$ and $d \geq 0$ the following holds.
Let $G$ be a graph, $A \subseteq V(G)$ with $a \coloneqq |A|$, and $M_0$ be a flat $(f_{\ref{prop_FolioHomoWall}}(a,d,r) \times f_{\ref{prop_FolioHomoWall}}(a,d,r))$-mesh of in $G-A$ with a sphere rendition $\rho$ witnessing the flatness of $M$ in $G-A$.
Then there exists an $(r \times r)$-mesh $M_1 \subseteq M_0$ such that
\begin{enumerate}
 \item $M_1$ is flat in $G-A$ witnessed by $\rho$, and
 \item $M_1$ is $(d,A)$-folio homogeneous in $G$.
\end{enumerate}
In particular, if $M_0$ is suspended in $G-A$ witnessed by $\rho$, then so is $M_1$.
Moreover, $f_{\ref{prop_FolioHomoWall}}(a,d,r) \in 2^{\poly(a + d)} \cdot \poly(r)$ and there exists an algorithm that takes as input $G$, $A$, $M_0$, and $\rho$ as above and computes $M_1$ in time $2^{\poly(w + a + d)} \cdot \poly(r) \cdot (|G| + |\!|G|\!|)$ where $w$ denotes the treewidth of $M_0$ in $G$.
\end{proposition}

We are now ready to proceed with the proof of \zcref{thm_IrrelevantInMesh}.

\begin{proof}[Proof of \zcref{thm_IrrelevantInMesh}]
Let us first set up the function $f_{\ref{thm_IrrelevantInMesh}}$:
\begin{align*}
    f_{\ref{thm_IrrelevantInMesh}}(r,k,a) & \coloneqq f_{\ref{prop_FolioHomoWall}}(a,0, (3r + 4) \cdot f_{\ref{thm_findIrrelevantVertexFlatWall}}(r,a,k) + 1)\\
    & \in 2^{\poly(a)} \cdot \poly\big( (3r + 4) \cdot f_{\ref{thm_findIrrelevantVertexFlatWall}}(r,a,k) + 1) \big)\\
    & = 2^{\poly(r + k + a)}.
\end{align*}

We first apply \zcref{prop_FolioHomoWall} to $M$ in order to find a flat and suspended $\big((3r + 4) \cdot f_{\ref{thm_findIrrelevantVertexFlatWall}}(r,a,k) + 1) \times (3r + 4) \cdot f_{\ref{thm_findIrrelevantVertexFlatWall}}(r,a,k) + 1)\big)$-mesh $M_1 \subseteq M$ in $G-A$ which is $(0,A)$-folio homogeneous.
Notice that this takes $2^{\poly(w + a + d)} \cdot \poly((3r + 4) \cdot f_{\ref{thm_findIrrelevantVertexFlatWall}}(r,a,k) + 1)) \cdot (|G| + |\!|G|\!|)$ time.

Next, we apply \zcref{lemma_BlankWall} to $M_1$ in order to either find a red $((r+1) \times (r+1))$-grid as a red-minor in $(G-A,R)$, or a $(f_{\ref{thm_findIrrelevantVertexFlatWall}}(r,a,k) \times f_{\ref{thm_findIrrelevantVertexFlatWall}}(r,a,k))$-submesh $M_2$ of $M_1$ whose compass is blank.
Notice that the flatness of $M_2$ in $G-A$ is still witnessed by $\rho$ which means that $M_2$ remains both suspended and $(0,A)$-folio homogeneous.
Moreover, the treewidth of $M_2$ in $G$ remains at most $w$ as well since the compass of $M_2$ is a subgraph of the compass of $M$.

In case we find a red $((r+1) \times (r+1))$-grid as a red-minor in $(G-A,R)$, we have, in particular, found that $(G-A,R)$ contains the $2$-outer $((r+1) \times (r+1))$-grid as a red-minor, implying that $\mathsf{depth}_2(G,R) \geq r+1 > r$.
This is the first outcome of the theorem and thus, we may assume it does not happen.
Hence, we may assume to  find $M_2$ as above.

We may now apply \zcref{thm_findIrrelevantVertexFlatWall} to $M_2$ which yields the desired vertex $v \in V(G) \setminus R$ in time $2^{\poly(r + k + a)}|G|$ and completes our proof.
\end{proof}

\subsection{The case of bounded treewidth}\label{subsec_boundedTreewidth}
Deleting irrelevant vertices reduces every instance of bounded $\mathsf{depth}_2$
to an equivalent one whose treewidth is bounded in terms of $k$ and
$\mathsf{depth}_2$; it then remains to solve such a bounded-treewidth instance.
This subsection supplies the algorithm for that case.

The algorithm is a standard dynamic programming over a tree decomposition, and
the only point requiring care is that the table must witness the spanning
condition. We use a tree decomposition in which every bag contains the terminals
and has width $\mathbf{O}(k+\mathsf{tw}(G))$, and at each node we record, for every
partial linkage pattern on the bag, whether it admits a realisation covering
every vertex of $R$ already processed, the vertices of $R$ still isolated in the
current bag being left to be covered later. As there are
$2^{\poly(k+\mathsf{tw}(G))}$ such patterns per node, the decision
follows in time $2^{\poly(k+\mathsf{tw}(G))}\cdot|G|$.  
 
 \begin{lemma}\label{lemma_SmallTreewidthAlgorithm}
 There is an algorithm that given an instance  $(G,R,\Tcal)$
of $k$-\textsc{Spanning Disjoint Paths} decides whether  $(G,R,\Tcal)$ is a \textsf{yes}-instance in  $2^{\poly(k+\mathsf{tw}(G))}\cdot |G|$ time.
 \end{lemma} 

\begin{proof}
Let $V_{\Tcal}$ be the  terminals in the pairs of $\Tcal$. We first build a rooted tree decomposition $(T,\beta,r)$ of $G$ where $\beta(r)=V_{\Tcal}$ and where 
for each $t\in V(T)$, $V_{\Tcal}\subseteq\beta(t)$.
This decomposition can be constructed in $2^{\poly(w)}\cdot |G|$ steps using, e.g., the algorithm of Bodlaender in \cite{Bodlaender1996LinearTimeTreeDecompositions} (see also \cite{BodlaenderKloks1991BetterAlgorithmsTreewidth,Korhonen2021SingleExponentialTreewidthApproximation,BodlaenderJaffkeTelle2020OptimalTreeDecompositionsRevisited}).
Moreover, for each $t\in V(T)$, $|\beta(t)|\leq 2k+w+1 \in \mathbf{O}(k + w)$.
For the needs of the dynamic programming we may also assume that every non-leaf node of $T$
has one or two children in $(T,r)$ and that $|T| \in \mathbf{O}(|G|)$.

For every $t\in V(T)$ we denote by $\textsf{desc}_{T}(t)$ the set containing all descendants 
of $t$ (including $t$) in the rooted tree $(T,r)$ and we define $G_t = G\left[\bigcup_{t'\in\textsf{desc}_{T}(t)} \beta(t')\right].$
Also, define
\[
\Pcal(t)=\{Z \mid V(Z)=\beta(t)\text{ and } Z \text{ is a linear forest}\}.
\]

 Given a graph $Z\in \Pcal(t)$, we denote by $I(Z)$ the set of isolated vertices of $Z$. 
A \emph{$Z$-realisation} in $G_t$ is a family $\mathcal C_Z=\{P_e\mid e\in E(Z)\}$ such that, for each $e=xy\in E(Z)$, $P_e$ is an $x$-$y$ path in $G_t$ whose internal
vertices avoid $\beta(t)$ and two paths $P_e,P_f$ intersect only in their
common endpoint when the corresponding edges $e,f$ share an endpoint in $Z$ 
and are otherwise vertex-disjoint.

For every $t\in V(T)$, the dynamic programming computes 
the function $\mathsf{table}_{t}: \Pcal(t)\to\{0,1\}$ where 
$\mathsf{table}_{t}(Z)=1$ if and only if there is a 
 $Z$-realisation $\Ccal_{Z}$ in $G_t$ where 
\begin{eqnarray}
R\cap V(G_{i}) &  \subseteq & V(\bigcup\Ccal_{Z})\cup I(Z)\label{conf_kiop}
\end{eqnarray}

Notice now that each function $\mathsf{table}_{t}$ can be seen as a set of $2^{\mathbf{O}(|\beta(t)|\log |\beta(t)|)} \subseteq 2^{\poly(k+w)}$ pairs.
Moreover observe that from the pairs in $\mathsf{table}_{r}$ we may derive whether $(G,R,\Tcal)$
 is a \textsf{Yes}-instance of $\textsc{SDP}$.
 For this, set $Z_{\mathcal T}
\coloneqq
\left(V_{\mathcal T},\{\{s_i,t_i\}\mid i\in[k]\}\right).$
Then $(G,R,\mathcal T)\text{ is a \textsf{Yes}-instance}$ if and only if $
\mathsf{table}_r(Z_{\mathcal T})=1.$

 Notice now that it is trivial to compute $\mathsf{table}_{t}$ when $t$ is a leaf of $(T,r)$.
 Moreover, for every non-leaf node $t\in V(T)$, it is easy to compute 
 $\mathsf{table}_{t}$, provided that we know the corresponding functions of its child $t'$ or its children $t_{1},t_{2}$.
As this computation can be performed in $|\mathsf{table}_{t'}|$ or  $|\mathsf{table}_{t_1}|\cdot |\mathsf{table}_{t_2}|$ steps respectively, the overall dynamic programming can be performed bottom-up and runs in   

$2^{\poly(k+w)}\cdot |G|$ time, as claimed.
\end{proof}

The above algorithm may also keep track, for each $t\in V(T)$, 
of a set of $Z$-realisations in $G_t$, for all $Z\in\Pcal(t)$.
The knowledge of these sets makes is possible to construct 
a solution, i.e., the disjoint paths between the terminals of $\Tcal$ that span $R$,   in case $(G,R,\Tcal)$ is a \textsf{yes}-instance.

Notice that an easy modification of the above can also solve the optimisation 
version of \textsc{SDP}, that is to compute $\textsf{max}_\mathsf{SDP}(G,R,\Tcal)$:  the maximum size 
of a subset $R'$ of $R$ for which $(G,R,\Tcal)$ is a \textsf{Yes}-instance of  \textsc{SDP}.
The only that changes is that now, instead of $\mathsf{table}_{t}\colon \Pcal(t)\to\{0,1\}$, we compute $\mathsf{maxtable}_{t}: \Pcal(t)\to\Nbbb$
where $\mathsf{maxtable}_{t}(Z)$ is the maximum size of a subset $R'$ of $R\cap V(G_t)$
where there is a $Z$-realisation $\Ccal_{Z}$ in $G_t$ and \eqref{conf_kiop} holds 
 with $R$ replaced by $R'$.
  
\section{Irrevelant-vertex tightness}\label{sec_CounterExample}
Let us now denote by $\mathcal{X}$ the class of all annotated graphs $\mathbf{H}$ such that there exists a $2$-outer-annotated grid $\mathbf{G}$ such that $\mathbf{H}$ is a red-minor of $\mathbf{G}$.

With \zcref{thm_SpanningLinkageAlgorithm} and \zcref{thm_localstructure} we have established the positive part of our main result.
That is, in every red-minor-closed class $\mathcal{C}$, if it does not contain the entire class $\mathcal{X}$, then there exists an irrelevant vertex rule for the \textsc{$k$-Spanning Disjoint Paths} problem.
In particular, if $\mathcal{X} \not\subseteq \mathcal{C}$, then there exists a Vital Linkage Theorem for \textsc{$k$-Spanning Disjoint Paths}.
In this section we prove the reverse.
That is, we show that for every $w \in \mathbb{N}$, $\mathcal{X}$ contains an annotated graph $(G_w,R_w)$ such that there are vertices $s,t \in V(G_w)$ for which $(G_w,R_w,\{ (s,t)\})$ is a vital instance of \textsc{$1$-Spanning Disjoint Paths} while $\mathsf{tw}(G_w) \geq w$.
This shows that a variant of \zcref{thm_VitalSpanningLinkage} for $\mathcal{X}$ cannot hold and therefore implies, in particular, that $TC(\mathcal{X})$ does not have the irrelevant vertex property.

\begin{figure}[ht]
 \centering
 \begin{tikzpicture}

 \pgfdeclarelayer{background}
		\pgfdeclarelayer{foreground}
			
		\pgfsetlayers{background,main,foreground}

 \begin{pgfonlayer}{background}
 \pgftext{\includegraphics[width=8cm]{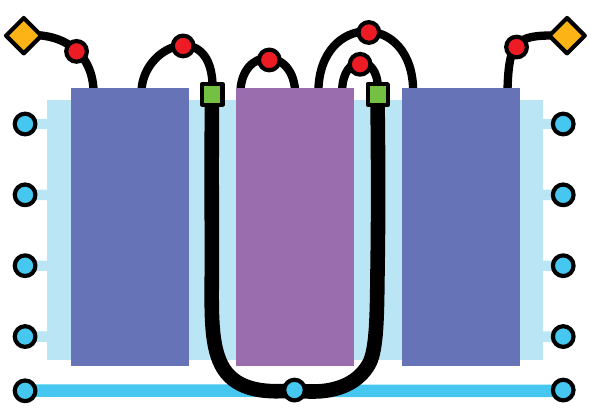}} at (C.center);
 \end{pgfonlayer}{background}
			
 \begin{pgfonlayer}{main}
 \node (C) [v:ghost] {};
 
 \end{pgfonlayer}{main}
 
 \begin{pgfonlayer}{foreground}
 \end{pgfonlayer}{foreground}

 \end{tikzpicture}
 \caption{A schematic illustration of the gadget $\mathsf{forward}_n$. The arrows indicate the direction the unique solution must traverse from the vertex $s$ to the vertex $t$.
 The two \textcolor{DarkBananaYellow}{yellow} vertices left and right mark $s$ and $t$.
 The \textcolor{Amethyst}{purple} box marks the gadget $\mathsf{backward}_{n-1}$ while the \textcolor{MidnightBlue}{blue} boxes are copies of the gadget $\mathsf{forward}_{n-1}$.}
 \label{fig_GenericForward}
\end{figure}

Our construction is inductive and finally requires two types of gadgets.
Each gadget will have a certain amount of ``entries'' and ``exists'' together with two ordered sets of interface vertices which exist mainly to propagate the construction.
These Interface vertices will be of degree $1$ and will not be part of the vital instance hidden inside.
Roughly speaking, we differentiate between two gadgets: A \textbf{forward} gadget and a \textbf{backward} gadget.
Both gadgets have an intrinsic orientation which we signify by \emph{left} and \emph{right}.
The \textbf{forward} gadget has a single entry on the left and a single exit on the right.
See \zcref{fig_GenericForward} for a schematic illustration.
The \textbf{backward} gadget on the other hand appears in two incarnations (which differ in a single edge).
The core \textbf{backward} gadget is a two way, it has an entry on the left that is connected to an exit on the right, and it has an entry on the right connected to an exist to the left.
The second incarnation of the \textbf{backward} gadget is obtained from the core construction by adding an edge -- subdivided with a red vertex -- connecting the exit on the left to the entry on the left.
Hence, this gadget has a single entry and a single exit, both on the right.
See \zcref{fig_GenericBackward} for an illustration.

\begin{figure}[ht]
 \centering
 \begin{tikzpicture}

 \pgfdeclarelayer{background}
		\pgfdeclarelayer{foreground}
			
		\pgfsetlayers{background,main,foreground}

 \begin{pgfonlayer}{background}
 \pgftext{\includegraphics[width=8cm]{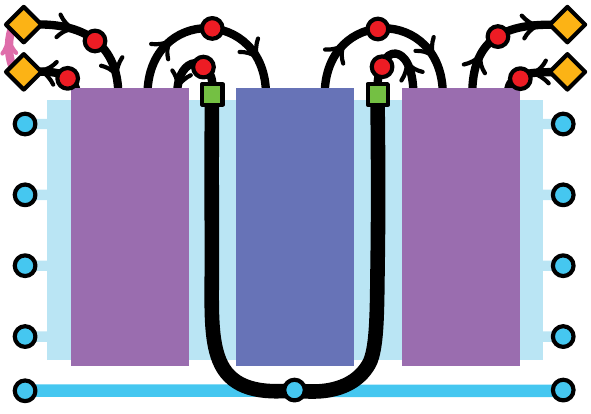}} at (C.center);
 \end{pgfonlayer}{background}
			
 \begin{pgfonlayer}{main}
 \node (C) [v:ghost] {};
 
 \end{pgfonlayer}{main}
 
 \begin{pgfonlayer}{foreground}
 \end{pgfonlayer}{foreground}

 \end{tikzpicture}
 \caption{A schematic illustration of the gadget $\mathsf{backward}_n$. The arrows indicate the direction the unique solution must traverse from the vertices $s_1$, $s_2$ to the vertices $t_1$ and $t_2$.
 The four \textcolor{DarkBananaYellow}{yellow} vertices left and right mark $s_1$, $s_2$, $t_1$ and $t_2$.
 The \textcolor{Amethyst}{purple} boxes are copies of the gadget $\mathsf{backward}_{n-1}$ while the \textcolor{MidnightBlue}{blue} box is a copy of the gadget $\mathsf{forward}_{n-1}$.
 The \textcolor{HotMagenta}{magenta} edge is optional.}
 \label{fig_GenericBackward}
\end{figure}

\subsection{Constructing a vital instance}\label{subsec_VitalInstanceConstruction}
Let us now formalise the construction.
As explained before we will have two main gadgets which we identify by its ``orientation''.
Each higher-order gadget is then constructed out of two copies of its own type one level lower plus one copy of the other kind from one level below.

\paragraph{Base constructions.}
We start with the base of our induction.
\medskip

The gadget $\mathbf{forward}_1$ is the tuple $(G_{\mathbf{f},1},R_{\mathbf{f},1},s,t,\ell_1,r_1)$ defined as follows.
The graph $G_{\mathbf{f},1}$ has vertex set $\{ s,t,x,\ell_1,r_1,z_1,z_2\}$ where $R_{\mathbf{f},1} \coloneqq \{ z_1,z_2\}$.
The edge set of $G_{\mathbf{f},1}$ is $\{ sz_1, z_1x, xz_2, z_2t, \ell_1x, xr_1\}$.

The vertex $s$ is the unique \emph{entry} of $\mathbf{forward}_1$ and the vertex $t$ is the unique \emph{exit} of $\mathbf{forward}_1$.
The \emph{left interface} is $\{ \ell_1\}$ and the \emph{right interface} is $\{ r_1\}$.
The path $P_1= \langle \ell_1,x,r_1\rangle$ is the \emph{$1$st row} of $\mathbf{forward}_1$.
\smallskip

The gadget $\mathbf{backward}_1$ is the tuple $(G_{\mathbf{b},1},R_{\mathbf{b},1},s_1,t_1,s_2,t_2,\ell_1,r_1)$ defined as follows.
The graph $G_{\mathbf{b},1}$ has vertex set $\{ s_1,t_1,s_2,t_2,x,\ell_1,r_1,y,z_1,z_2\}$ where $R_{\mathbf{b},1} \coloneqq \{ y,z_1,z_2\}$.
The edge set of $G_{\mathbf{b},1}$ is $\{ t_2z_1, z_1x, xz_2, z_2s_2, s_1y, yt_1, \ell_1x, xr_1\}$.

The vertices $s_1$ and $s_2$ are the \emph{entries} of $\mathbf{backward}_1$ and the vertices $t_1$ and $t_2$ are the \emph{exists} of $\mathbf{backward}_1$.
The \emph{left interface} is $\{ \ell_1\}$ and the \emph{right interface} is $\{ r_1\}$.
The path $P_1= \langle \ell_1,x,r_1\rangle$ is the \emph{$1$st row} of $\mathbf{backward}_1$.
See \zcref{fig_BaseGadgets} for illustrations of $\mathsf{forward}_1$ and $\mathsf{backward}_1$.

\begin{figure}[ht]
 \centering
 \begin{tikzpicture}

 \pgfdeclarelayer{background}
		\pgfdeclarelayer{foreground}
			
		\pgfsetlayers{background,main,foreground}
			
 \begin{pgfonlayer}{main}
 \node (C) [v:ghost] {};

 \node(L) [v:ghost] at (-3,0) {
 \begin{tikzpicture}

 \pgfdeclarelayer{background}
		 \pgfdeclarelayer{foreground}
			
		 \pgfsetlayers{background,main,foreground}

 \begin{pgfonlayer}{background}
 \pgftext{\includegraphics[width=5.5cm]{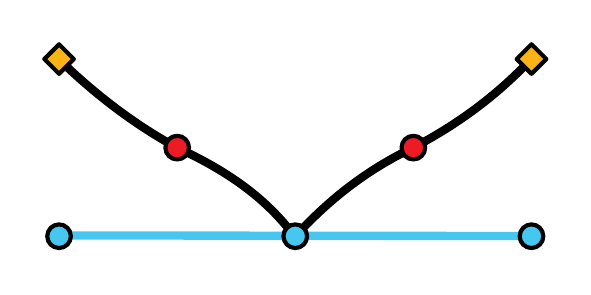}} at (C.center);
 \end{pgfonlayer}{background}
			
 \begin{pgfonlayer}{main}
 \node (C) [v:ghost] {};
 
 \end{pgfonlayer}{main}
 
 \begin{pgfonlayer}{foreground}
 \end{pgfonlayer}{foreground}

 \end{tikzpicture}
 };

 \node(M) [v:ghost] at (3,0) {
 \begin{tikzpicture}

 \pgfdeclarelayer{background}
		 \pgfdeclarelayer{foreground}
			
		 \pgfsetlayers{background,main,foreground}

 \begin{pgfonlayer}{background}
 \pgftext{\includegraphics[width=5.5cm]{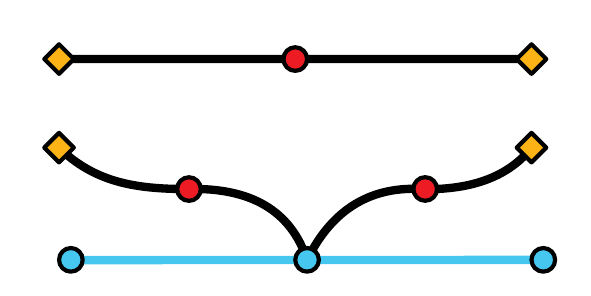}} at (C.center);
 \end{pgfonlayer}{background}
			
 \begin{pgfonlayer}{main}
 \node (C) [v:ghost] {};
 
 \end{pgfonlayer}{main}
 
 \begin{pgfonlayer}{foreground}
 \end{pgfonlayer}{foreground}

 \end{tikzpicture}
 };

 \node (i) [v:ghost] at (-3,-1.5) {$\mathsf{forward}_1$};
 \node (ii) [v:ghost] at (3,-1.5) {$\mathsf{backward}_1$};

 \end{pgfonlayer}{main}
 
 \begin{pgfonlayer}{foreground}
 \end{pgfonlayer}{foreground}

 \end{tikzpicture}
 \caption{The two base-level gadgets $\mathsf{forward}_1$ and $\mathsf{backward}_1$.}
 \label{fig_BaseGadgets}
\end{figure}
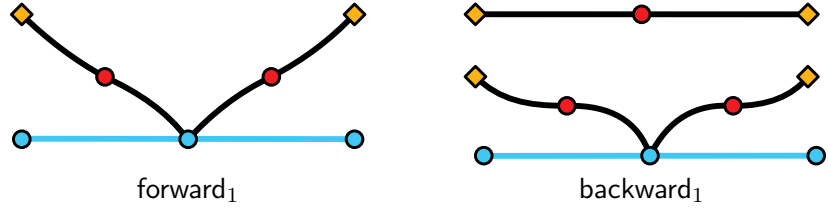

Generally speaking, in each level of our construction we will have a left and a right interface, denotes as $\{ \ell_1,\dots,\ell_n\}$ and $\{ r_1,\dots,r_n\}$ respectively.
For $i \in [n]$, the \emph{$i$th left (right) interface} will refer to the vertex $\ell_i$ ($r_i$).

Before we describe the full inductive construction, let us make one more illustration in form of the construction of $\mathbf{forward}_2$ and $\mathbf{backward}_2$.
\medskip

The gadget $\mathbf{forward}_2$ is the tuple $(G_{\mathbf{f},2},R_{\mathbf{f},2},s,t,\ell_1,\ell_2,r_1,r_2)$ obtained as follows.
Let $\mathbf{L} = (L,R_L,s_L,t_L,\ell^1_L,r^1_{L})$ and $\mathbf{R} = (R,R_R,s_R,t_R,\ell^1_R,r^1_{R})$ be two disjoint copies of $\mathbf{forward}_1$.
Moreover, let $\mathbf{M} = (M,R_M,s^1_M,t^1_M,s^2_M,t^2_M,\ell^1_M,r^1_{M})$ be a copy of $\mathbf{backward}_1$ disjoint from both $\mathbf{L}$ and $\mathbf{R}$.
The graph $G_{\mathbf{f},2}$is obtained by 
\begin{itemize}
    \item first identifying $r^1_L$ and $\ell^1_M$ into the vertex $a_1$ and $r^1_M$ and $\ell^1_R$ into the vertex $b_1$, then
    \item adding new vertices $\ell_2$, $x$, and $r_2$ together with the edges $\ell_2x$ and $xr_2$, then
    \item adding the edges $a_1x$ and $xb_1$, and finally,
    \item adding the edges $t_La_1$, $b_1s^2_M$, $t^2_Ms^1_M$, and $t^1_Ms_R$.
\end{itemize}
We set $s \coloneqq s_L$ and $t \coloneqq t_R$, and $R_{\mathbf{f},2} \coloneqq R_L \cup R_M \cup R_R$.
Moreover, we declare $\ell_1 \coloneqq \ell^1_L$ and $r_1 \coloneqq r^1_R$ such that $\{ \ell_1,\ell_2\}$ is the \emph{left interface} of $\mathbf{forward}_2$ and $\{ r_1,r_2\}$ is its \emph{right interface}.

Notice that the construction of $\mathbf{forward}_2$ crucially uses a modification of $\mathbf{backward}_1$ where the exit $t_2$ and the entry $s_1$ are being linked by an edge.
\smallskip

The gadget $\mathbf{backward}_2$ is the tuple $(G_{\mathbf{b},2},R_{\mathbf{b},2},s_1,t_1,s_2,t_2,\ell_1,\ell_2,r_1,r_2)$ obtained as follows.
Let $\mathbf{L} = (L,R_L,s^1_L,t^1_L,s^2_L,t^2_L\ell^1_L,r^1_{L})$ and $\mathbf{R} = (R,R_R,s^1_R,t^1_R,s^2_R,t^2_R,\ell^1_R,r^1_{R})$ be two disjoint copies of $\mathbf{backward}_1$.
Moreover, let $\mathbf{M} = (M,R_M,s_M,t_M,\ell^1_M,r^1_{M})$ be a copy of $\mathbf{forward}_1$ disjoint from both $\mathbf{L}$ and $\mathbf{R}$.
The graph $G_{\mathbf{b},2}$is obtained by 
\begin{itemize}
    \item first identifying $r^1_L$ and $\ell^1_M$ into the vertex $a_1$ and $r^1_M$ and $\ell^1_R$ into the vertex $b_1$, then
    \item adding new vertices $\ell_2$, $x$, and $r_2$ together with the edges $\ell_2x$ and $xr_2$, then
    \item adding the edges $a_1x$ and $xb_1$, and finally,
    \item adding the edges $s^2_La_1$, $t^1_Ls_M$, $b_1t^2_R$, and $t_Ms^1_R$.
\end{itemize}
We set $s_1 \coloneqq s^1_L$, $s_2 \coloneqq s^2_R$, $t_1 \coloneqq t^1_R$ and $t_2 \coloneqq t^2_L$, and $R_{\mathbf{b},2} \coloneqq R_L \cup R_M \cup R_R$.
Moreover, we declare $\ell_1 \coloneqq \ell^1_L$ and $r_1 \coloneqq r^1_R$ such that $\{ \ell_1,\ell_2\}$ is the \emph{left interface} of $\mathbf{backward}_2$ and $\{ r_1,r_2\}$ is its \emph{right interface}.
\smallskip

Notice that each of the level $1$ gadgets has a path joining $\ell_1$ to $r_1$ whose edges are not being used in the $R$-spanning $s$-$t$-path (or the two $R$-spanning paths in the backward gadget).
The level $2$ gadgets maintain this property and add a second path $P_2$ with similar properties connecting $\ell_2$ and $r_2$.
Moreover, the level $2$ constructions also contain a path $U$ from $a_1$ to $b_1$ which intersects $P_2$ precisely in the vertex $x$.
As we will see later, this path $U$ must always be used by any solution.
If we now interpret the area of the level $2$ gadget containing all red vertices as a vortex, then $U$ forms a simple loop on the boundary of this vortex.
Indeed, this is precisely the viewpoint we wish to induce.
Recall our proofs of \zcref{claim_SimpleLoops} and \zcref{claim_RedLoops}.
Here, it was crucial to observe that many ``stacked'' simple loops provide enough infrastructure to reroute a given solution, thereby showing that the instance was not vital.
Our goal is to stack more and more such simple loops onto $U$ by increasing the level, while also showing that these simple loops must always be used in every solution to the corresponding instance of \textsc{Spanning Disjoint Paths}.
If this succeeds, an argument similar to the one from \zcref{claim_SimpleLoops} will then find a large bramble in these gadgets that is unavoidable.

let now $n \geq 2$ be an integer.
We continue with the description of the general level $n$ construction.
\medskip

The gadget $\mathbf{forward}_n$ is the tuple $(G_{\mathbf{f},n},R_{\mathbf{f},n},s,t,\ell_1,\dots,\ell_n,r_1,\dots,r_n)$ obtained as follows.
Let $\mathbf{L} = (L,R_L,s_L,t_L,\ell^1_L,\dots \ell^{n-1}_L,r^1_{L},\dots,r^{n-1}_L)$ and $\mathbf{R} = (R,R_R,s_R,t_R,\ell^1_R,\dots,\ell^{n-1}_R,r^1_{R},\dots,r^{n-1}_R)$ be two disjoint copies of $\mathbf{forward}_{n-1}$.
In addition, we consider a copy of $\mathbf{backward}_{n-1}$ disjoint from both $\mathbf{L}$ and $\mathbf{R}$ denotes as $\mathbf{M} = (M,R_M,s^1_M,t^1_M,s^2_M,t^2_M,\ell^1_M,\dots,\ell^{n-1}_M,r^1_{M},\dots,r^{n-1}_M)$.
The graph $G_{\mathbf{f},n}$is obtained by 
\begin{itemize}
    \item first identifying $r^i_L$ and $\ell^i_M$ into the vertex $a_i$ and $r^i_M$ and $\ell^i_R$ into the vertex $b_i$ for each $i\in[n-1]$, then
    \item adding new vertices $\ell_n$, $x$, and $r_n$ together with the edges $\ell_nx$ and $xr_n$,
    \item adding the edges $a_ia_{i+1}$ and $b_ib_{i+1}$ for all $i \in [n-2]$, then
    \item adding the edges $a_{n-1}x$ and $xb_{n-1}$, and finally,
    \item adding the edges $t_La_1$, $b_1s^2_M$, $t^2_Ms^1_M$, and $t^1_Ms_R$.
\end{itemize}
We set $s \coloneqq s_L$ and $t \coloneqq t_R$, and $R_{\mathbf{f},2} \coloneqq R_L \cup R_M \cup R_R$.
Moreover, we declare $\ell_i \coloneqq \ell^i_L$ and $r_i \coloneqq r^i_R$ for all $i \in[n-1]$ such that $\{ \ell_1,\dots,\ell_n\}$ is the \emph{left interface} of $\mathbf{forward}_n$ and $\{ r_1,\dots,r_n\}$ is its \emph{right interface}.
Additionally, we call the path $U_n$ the \emph{$n$th loop} of $\mathbf{forward}_n$. If we denote by $U_1,\dots,U_{n-1}$ the $n-1$ loops of $\mathbf{M}$, then those together with $U_n$ become the \emph{$n$ loops of $\mathbf{forward}_n$}.
Similarly, the \emph{$n$th row} of $\mathbf{forward}_n$ is the path $P_n = \langle \ell_n,x,r_n\rangle$ and for $i \in [n-1]$, the \emph{$i$th row} of $\mathbf{forward}_n$ is the path $P_i$ obtained from the union of the $i$th rows of $\mathbf{L}$, $\mathbf{M}$, and $\mathbf{R}$.
\smallskip

The gadget $\mathbf{backward}_n$ is the tuple $(G_{\mathbf{b},n},R_{\mathbf{b},n},s_1,t_1,s_2,t_2,\ell_1,\dots,\ell_n,r_1,\dots,r_n)$ obtained as follows.
Let $\mathbf{L} = (L,R_L,s^1_L,t^1_L,s^2_L,t^2_L,\ell^1_L,\dots \ell^{n-1}_L,r^1_{L},\dots,r^{n-1}_L)$ and $\mathbf{R} = (R,R_R,s^1_R,t^1_R,s^2_R,t^2_R,\ell^1_R,\dots,\ell^{n-1}_R,r^1_{R},\dots,r^{n-1}_R)$ be two disjoint copies of $\mathbf{backward}_{n-1}$.
In addition, we consider a copy of $\mathbf{fowward}_{n-1}$ disjoint from both $\mathbf{L}$ and $\mathbf{R}$ denotes as $\mathbf{M} = (M,R_M,s_M,t_M,\ell^1_M,\dots,\ell^{n-1}_M,r^1_{M},\dots,r^{n-1}_M)$.
The graph $G_{\mathbf{b},n}$is obtained by 
\begin{itemize}
    \item first identifying $r^i_L$ and $\ell^i_M$ into the vertex $a_i$ and $r^i_M$ and $\ell^i_R$ into the vertex $b_i$ for each $i\in[n-1]$, then
    \item adding new vertices $\ell_n$, $x$, and $r_n$ together with the edges $\ell_nx$ and $xr_n$,
    \item adding the edges $a_ia_{i+1}$ and $b_ib_{i+1}$ for all $i \in [n-2]$, then
    \item adding the edges $a_{n-1}x$ and $xb_{n-1}$, and finally,
    \item adding the edges $s^2_La_1$, $t^1_Ls_M$, $b_1t^2_R$, and $t_Ms^1_R$.
\end{itemize}
We set $s_1 \coloneqq s^1_L$, $s_2 \coloneqq s^2_R$, $t_1 \coloneqq t^1_R$ and $t_2 \coloneqq t^2_L$, and $R_{\mathbf{b},2} \coloneqq R_L \cup R_M \cup R_R$.
Moreover, we declare $\ell_i \coloneqq \ell^i_L$ and $r_i \coloneqq r^i_R$ for all $i \in[n-1]$ such that $\{ \ell_1,\dots,\ell_n\}$ is the \emph{left interface} of $\mathbf{backward}_n$ and $\{ r_1,\dots,r_n\}$ is its \emph{right interface}.
Additionally, we call the path $U_n$ the \emph{$n$th loop} of $\mathbf{backward}_n$.
If we denote by $U_1,\dots,U_{n-1}$ the $n-1$ loops of $\mathbf{M}$, then those together with $U_n$ become the \emph{$n$ loops of $\mathbf{backward}_n$}.
Similarly, the \emph{$n$th row} of $\mathbf{backward}_n$ is the path $P_n = \langle \ell_n,x,r_n\rangle$ and for $i \in [n-1]$, the \emph{$i$th row} of $\mathbf{backward}_n$ is the path $P_i$ obtained from the union of the $i$th rows of $\mathbf{L}$, $\mathbf{M}$, and $\mathbf{R}$.

Before we move on, let us observe that $\mathbf{forward}_n$ and $\mathbf{backward}_n$ are indeed members of $\mathcal{X}$ for all $n \in \mathbb{N}_{\geq 1}$.

\begin{lemma}\label{lemma_GadgetsInGrid}
For every integer $n \geq 1$, the annotated graphs $(G_{\mathbf{f},n},R_{\mathbf{f},n})$ and $(G_{\mathbf{b},n},R_{\mathbf{b},n})$ from $\mathbf{forward}_n$ and $\mathbf{backward}_n$ are red-minors of the $2$-outer-annotated $(8 \cdot 3^{n-1} \times 8 \cdot 3^{n-1})$-grid.
\end{lemma}

\begin{proof}
We prove the claim by induction on $n$ where $n=1$ is the base case.
Moreover, the proofs for $\mathbf{forward}_n$ and $\mathbf{backward}_n$ are completely analogous, so we one present the inductive step for the case of $\mathbf{forward}_n$ here.

For the base of the induction, all we have to do is give minor models of $(G_{\mathbf{f},1},R_{\mathbf{f},1})$ and $(G_{\mathbf{b},1},R_{\mathbf{b},1})$ in the $2$-outer-annotated $(8 \times 8)$-grid.
In \zcref{fig_BaseGadgetModels}, we present such minor models with the additional property that all entries, exists, and interface vertices lie on the outer face.
Indeed, crucially, the embeddings of the minor model in the $2$-outer-annotated $(8 \times 8)$-grid mimic the natural embeddings of $(G_{\mathbf{f},1},R_{\mathbf{f},1})$ and $(G_{\mathbf{b},1},R_{\mathbf{b},1})$.
This is a property we will be using in the inductive step.

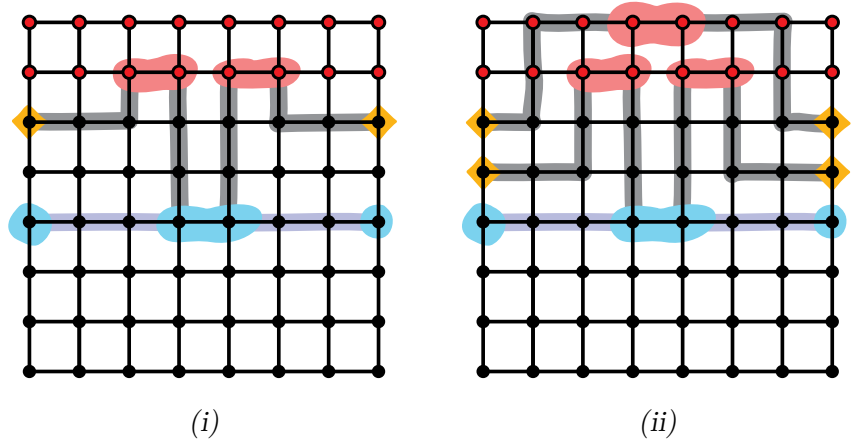
\begin{figure}[ht]
 \centering
 \begin{tikzpicture}

 \pgfdeclarelayer{background}
		\pgfdeclarelayer{foreground}
			
		\pgfsetlayers{background,main,foreground}
			
 \begin{pgfonlayer}{main}
 \node (C) [v:ghost] {};

 \node(L) [v:ghost] at (-3,0) {
 \begin{tikzpicture}

 \pgfdeclarelayer{background}
		 \pgfdeclarelayer{foreground}
			
		 \pgfsetlayers{background,main,foreground}

 \begin{pgfonlayer}{background}
 \pgftext{\includegraphics[width=5.5cm]{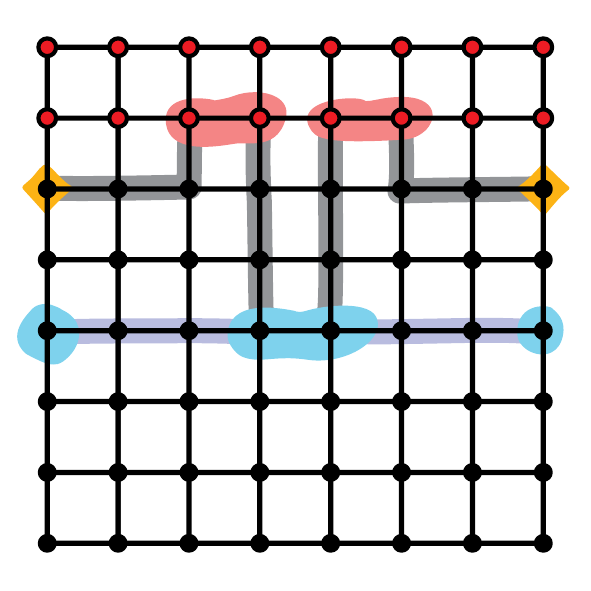}} at (C.center);
 \end{pgfonlayer}{background}
			
 \begin{pgfonlayer}{main}
 \node (C) [v:ghost] {};
 
 \end{pgfonlayer}{main}
 
 \begin{pgfonlayer}{foreground}
 \end{pgfonlayer}{foreground}

 \end{tikzpicture}
 };

 \node(M) [v:ghost] at (3,0) {
 \begin{tikzpicture}

 \pgfdeclarelayer{background}
		 \pgfdeclarelayer{foreground}
			
		 \pgfsetlayers{background,main,foreground}

 \begin{pgfonlayer}{background}
 \pgftext{\includegraphics[width=5.5cm]{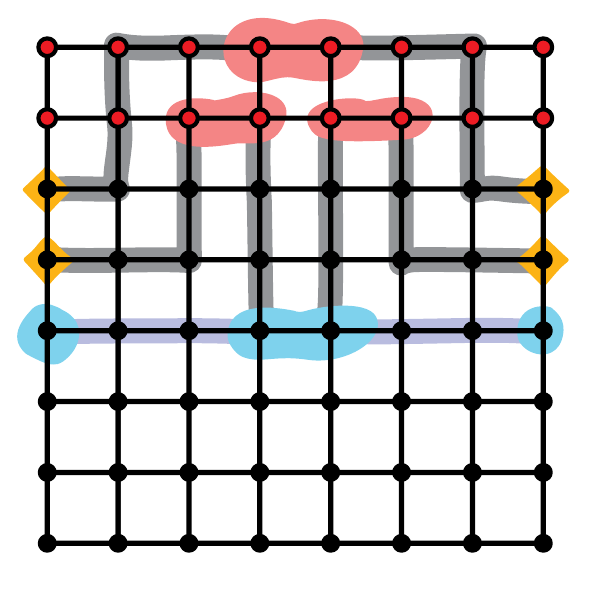}} at (C.center);
 \end{pgfonlayer}{background}
			
 \begin{pgfonlayer}{main}
 \node (C) [v:ghost] {};
 
 \end{pgfonlayer}{main}
 
 \begin{pgfonlayer}{foreground}
 \end{pgfonlayer}{foreground}

 \end{tikzpicture}
 };

 \node (i) [v:ghost] at (-3,-3) {\textsl{(i)}};
 \node (ii) [v:ghost] at (3,-3) {\textsl{(ii)}};

 \end{pgfonlayer}{main}
 
 \begin{pgfonlayer}{foreground}
 \end{pgfonlayer}{foreground}

 \end{tikzpicture}
 \caption{The annotated graphs (i) $(G_{\mathbf{f},1},R_{\mathbf{f},1})$ and (ii) $(G_{\mathbf{b},1},R_{\mathbf{b},1})$ as red-minors in the $2$-outer-annotated $(8 \times 8)$-grid.}
 \label{fig_BaseGadgetModels}
\end{figure}

Now, in order to construct a minor model of $(G_{\mathbf{f},n},R_{\mathbf{f},n})$ in the $2$-outer-annotated $(8 \cdot 3^{n-1} \times 8 \cdot 3^{n-1})$-grid, we assume inductively that we have red-minor models of $(G_{\mathbf{f},n-1},R_{\mathbf{f},n-1})$ and $(G_{\mathbf{b},n-1},R_{\mathbf{b},n-1})$ such that the models of the interface vertices, entries, and exists belong to the first and last row of the $2$-outer-annotated $(8 \cdot 3^{n-2} \times 8 \cdot 3^{n-2})$-grid in the natural way.
Indeed, we may assume that $s$ and $r$ of $\mathbf{forward}_{n-1}$ lie on the $3$rd row from the top, while $a_i$ and $b_i$ lie on the $(4+i)$th row from the top for all $i \in [n-1]$.
Additionally, for $\mathbf{backward}_{n-1}$, we may assume that $s_1$ and $t_1$ are realised on the $3$rd row from the top while $s_2$ and $t_2$ are realised on the $4$th row from the top.
Notice that this is also guaranteed in the base case as depicted in \zcref{fig_BaseGadgetModels}.

Now consider the $2$-outer-annotated $(8 \cdot 3^{n-1} \times 8 \cdot 3^{n-1})$-grid $\mathbf{G}$.
Now let $\mathbf{G}_1,\mathbf{G}_2$, and $\mathbf{G}_3$ be the three vertex-disjoint $2$-outer-annotated $(8 \cdot 3^{n-2} \times 8 \cdot 3^{n-2})$-grids contained in $\mathbf{G}$ indexed such that $\mathbf{G}_1$ intersects the first row of $\mathbf{G}$ and $\mathbf{G}_3$ intersects the last row of $\mathbf{G}$.
By our induction hypothesis, we know that we may find red-minor models of $(G_{\mathbf{f},n-1},R_{\mathbf{f},n-1})$ in $\mathbf{G}_1$ and $\mathbf{G}_3$ while $\mathbf{G}_2$ hosts a red-minor model of $(G_{\mathbf{b},n-1},R_{\mathbf{b},n-1})$.
Now, we merge the models of the right-sided interface vertices in $\mathbf{G}_1$ with their left sided counterparts in $\mathbf{G}_2$ along a matching, thereby creating models of the vertices $a_1,\dots,a_n$.
Similarly, we merge the models of the right-sided interface vertices in $\mathbf{G}_2$ with the models of the left-sided interface vertices in $\mathbf{G}_3$ along a matching in order to create models of the vertices $b_1,\dots,b_n$.
See \zcref{fig_Forward2Minor} for an example.

\begin{figure}[ht]
 \centering
 \begin{tikzpicture}

 \pgfdeclarelayer{background}
		\pgfdeclarelayer{foreground}
			
		\pgfsetlayers{background,main,foreground}

 \begin{pgfonlayer}{background}
 \pgftext{\includegraphics[width=15cm]{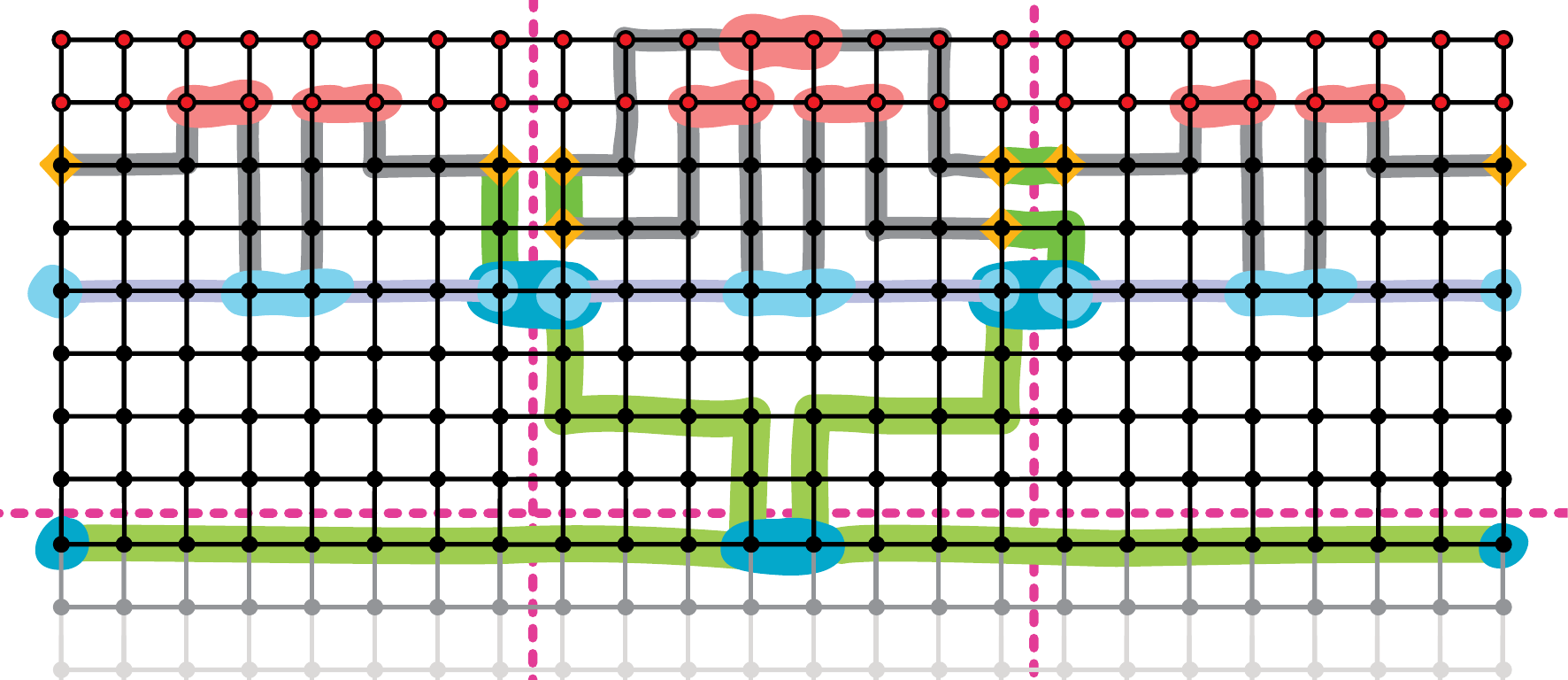}} at (C.center);
 \end{pgfonlayer}{background}
			
 \begin{pgfonlayer}{main}
 \node (C) [v:ghost] {};
 
 \end{pgfonlayer}{main}
 
 \begin{pgfonlayer}{foreground}
 \end{pgfonlayer}{foreground}
 
 \end{tikzpicture}
 \caption{The construction of a minor model of the annotated graph of $\mathbf{forward}_2$ inside the $2$-outer-annotated $(24 \times 24)$-grid in the proof of \zcref{lemma_GadgetsInGrid}. The three grids $\mathbf{G}_i$ are indicated by the dashed \textcolor{HotMagenta}{magenta} lines and the modifications to the constructions from \zcref{fig_BaseGadgetModels} are highlighted in darker colours and different shades of \textcolor{AppleGreen}{green}.}
 \label{fig_Forward2Minor}
\end{figure}
Next we mimic the construction steps for $\mathbf{forward}_n$ in the minor model.

We introduce a path between the model of $t_2$ and $s_1$ within the graph of $\mathbf{backward}_{n-1}$ in $\mathbf{G}_2$.
Also, we add a path between the models of $t_1$ in the graph of $\mathbf{backward}_{n-1}$ in $\mathbf{G}_2$ and the model of $s$ in the graph of $\mathbf{forward}_{n-1}$ in $\mathbf{G}_3$.

We introduce three new vertex models in the $(8 \cdot 3^{n-2})$th row $Q$ from the top as follows:
The model of $\ell_n$ will be the endpoint of $Q$ on the first column, $r_n$ will be the endpoint of $Q$ on the last column, and $x$ consists of the middle edge of $Q$ together with its endpoints.
The rest of $Q$ is used to model the edges $\ell_nx$ and $xr_n$.
We also add paths modelling the edges $a_nx$ and $xb_n$.

Next, we add paths for the edges $t$ -- taken from the graph of $\mathbf{forward}_{n-1}$ in $\mathbf{G}_1$ -- and $a_1$ in $\mathbf{G}_1$ and between $s_1$ -- taken from the graph of $\mathbf{backward}_{n-1}$ in $\mathbf{G}_2$ -- and $b_1$.
This path is simply the edge between the two branch sets joining $\mathbf{G}_2$ and $\mathbf{G}_3$.

Finally, for every $i\in[n-1]$ we add a path between the models of $a_i$ and $a_{i+1}$ and a path between $b_i$ and $b_{i+1}$ within $\mathbf{G}_2$ -- with the exception of the paths modelling the edges $a_{n-1}a_n$.

All of those additional paths are depicted in \textcolor{AppleGreen}{green} in \zcref{fig_Forward2Minor}.
It is now straightforward to check that the constructed red-minor model is indeed a red-minor model of the annotated graph $(G_{\mathbf{f},n},R_{\mathbf{f},n})$ as desired.
\end{proof}

\subsection{Treewidth of the instances}\label{subsec_InstanceTreewidth}

To see that the graphs of $\mathbf{forward}_n$ and $\mathbf{backward}_n$ have treewidth in $\mathbf{O}(n)$ is relatively straightforward.
In essence, we may use the loops and rows to construct a bramble similar to the way we have constructed brambles in several of the proof in earlier sections.

\begin{lemma}\label{lemma_GadgetTreewidth}
For every $n \in \mathbb{N}$, $n \geq 1$, the treewidth of the graphs of $\mathbf{forward}_n$ and $\mathbf{backward}_n$ is at least $\lfloor \nicefrac{n}{2} \rfloor - 1$.
\end{lemma}

\begin{proof}
We construct a bramble as follows:

For each $i \in[n]$ let $B_i$ be the graph obtained from the union of the $i$th loop and the $i$th row of $\mathbf{forward}_n$ or $\mathbf{backward}_n$.
From the construction of both gadgets it follows that the $i$th loop intersects the $j$th row for all $j \leq i$ and every $i\ in[n]$.
Moreover, every vertex belongs to at most one row and at most one loop.
Hence, $B_i$ and $B_j$ have a non-empty intersection and each $B_i$ is connected for all $i,j \in[n]$.
It follows that $\mathcal{B} \coloneqq \{ B_i ~\colon~ i\in[n] \}$ is indeed a bramble.
Moreover, the order of $\mathcal{B}$ is at least $\lfloor \nicefrac{n}{2} \rfloor$.
Hence, by \zcref{prop_brambles} the assertion follows.
\end{proof}

\subsection{Vitality of the instances}\label{subsec_InstanceVitality}

All that is left is to show that the instances of $1$\textsc{-Spanning Disjoint Paths} and $2$\textsc{-Spanning Disjoint Paths} generated by $\mathbf{forward}_n$ and $\mathbf{backward}_n$ are vital.

Let $G_{\mathbf{f}n,}^{\star}$ be the graph obtained from the graph $G_{\mathbf{f},n}$ of $\mathbf{forward}_n$ by deleting all of its interface vertices.
Similarly, let $G_{\mathbf{b},n}^{\star}$ be the graph obtained from the graph $G_{\mathbf{b},n}$ of $\mathbf{backward}_n$ by deleting all of its interface vertices.
Then, $(G^{\star}_{\mathbf{f},n},R_{\mathbf{f},n},\{ (s,t) \})$ is the instance of $1$\textsc{-Spanning Disjoint Paths} generated by $\mathbf{forward}_n$ while $(G^{\star}_{\mathbf{b},n},R_{\mathbf{b},n},\{ (s_1,t_1), (s_2,t_2)\})$ is the instance of $2$-\textsc{-Spanning Disjoint Paths} generated by $\mathbf{backward}_n$.

\begin{lemma}\label{lemma_VitalGadget}
For every integer $n \geq 1$, the instances $(G^{\star}_{\mathbf{f},n},R_{\mathbf{f},n},\{ (s,t) \})$ and $(G^{\star}_{\mathbf{b},n},R_{\mathbf{b},n},\{ (s_1,t_1), (s_2,t_2)\})$ are vital.
\end{lemma}

\begin{proof}
As before, we prove the claim by induction on $n$.
Moreover, since the arguments for $\mathbf{forward}_n$ and $\mathbf{backward}_n$ are almost identical, we only provide the proof for $\mathbf{forward}_n$ here.
However, it is necessary to assume that the statement holds for both $\mathbf{forward}_{n-1}$ and $\mathbf{backward}_{n-1}$ in order to make the inductive step go through.
\smallskip

For the base case consider the instance $(G_{\mathbf{f},1}^{\star},R_{\mathbf{f},1},\{ (s,t)\})$ generated by $\mathbf{forward}_1$.
Then here, $G_{\mathbf{f},1}^{\star}$ is simply a path with endpoints $s$ and $t$.
Hence, it is easy to see that $(G_{\mathbf{f},1}^{\star},R_{\mathbf{f},1},\{ (s,t)\})$ is vital.
\smallskip

Now let us assume that we have already shown that $(G^{\star}_{\mathbf{f},n-1},R_{\mathbf{f},n-1},\{ (s,t) \})$ and $(G^{\star}_{\mathbf{b},n-1},R_{\mathbf{b},n-1},\{ (s_1,t_1), (s_2,t_2)\})$ are vital.
Our goal is to show that $(G^{\star}_{\mathbf{f},n},R_{\mathbf{f},n},\{ (s,t) \})$ is also vital.
Moreover, we assume that the $n-1$st loop of $(G^{\star}_{\mathbf{f},n-1},R_{\mathbf{f},n-1},\{ (s,t) \})$ and $(G^{\star}_{\mathbf{b},n-1},R_{\mathbf{b},n-1},\{ (s_1,t_1), (s_2,t_2)\})$ respectively is part of the unique solution.

Let $U \coloneqq U_n$ be the $n$th loop of $\mathsf{forward}_n$ and let $P \coloneqq P_n$ denote its $n$th row.
Notice that in $G^{\star}_{\mathbf{f},n}$, $P$ consists of a single vertex which also belongs to $U$.
If we delete all internal vertices of $U$ from $(G^{\star}_{\mathbf{f},n},R_{\mathbf{f},n})$ we obtain three pairwise vertex-disjoint annotated graphs $\mathbf{G}_1$, $\mathbf{G}_2$, and $\mathbf{G}_3$ where both $\mathbf{G}_1$ and $\mathbf{G}_3$ may be seen as copies of $(G^{\star}_{\mathbf{f},n-1},R_{\mathbf{f},n-1},\{ (s,t) \})$
On the other hand, $\mathbf{G}_2$ has an edge $e \coloneqq s_1t_2$ whose deletion results in the annotated graph $\mathbf{G}_2'$ which may be seen as a copy of $(G^{\star}_{\mathbf{b},n-1},R_{\mathbf{b},n-1},\{ (s_1,t_1), (s_2,t_2)\})$.
Hence, by the induction hypothesis, $\mathbf{G}_1$, $\mathbf{G}_2'$, and $\mathbf{G}_3$ are vital.
It follows, that, if we interpret $\mathbf{G}_2$ as the instance $(G^{\star}_{\mathbf{b},n}+e,R_{\mathbf{b},n},\{ (s_2,t_1)\})$, then also $\mathbf{G}_2$ is vital.

We claim that any solution of $(G^{\star}_{\mathbf{f},n},R_{\mathbf{f},n},\{ (s,t) \})$ must contain solutions for all three $\mathbf{G}_i$.
To see this, it is enough to show that $U$ must be part of every $R_{\mathbf{f},n}$-spanning $s$-$t$-path in $(G^{\star}_{\mathbf{f},n},R_{\mathbf{f},n},\{ (s,t) \})$.
For this, notice that every vertex of the first row of $\mathbf{forward}_n$ -- except for the interface vertices and the endpoints of $U$ -- is adjacent to two red vertices, each of which has degree $2$.
This means that for any such vertex $v$ of the first row, none of the edges of the first row incident with $v$ can be part of any $R_{\mathbf{f},n}$-spanning $s$-$t$-path in $(G^{\star}_{\mathbf{f},n},R_{\mathbf{f},n},\{ (s,t) \})$.
Since we may assume by induction that the $n-1$st loops of the $\mathbf{G}_i$ are part of their respective unique solutions, we may now observe that the only option for a solution is to follow along $U$ until its other endpoint as desired.
This is because we may now inductively show that for every $i\in[2,n-1]$, no edge incident with $a_i$ or $b_i$ can be part of a solution.
This, however, means that any solution for $(G^{\star}_{\mathbf{f},n},R_{\mathbf{f},n},\{ (s,t) \})$ must indeed contain solutions for $\mathbf{G}_i$ for each $i\in[3]$.
Moreover, it must use $U$.
Since $G^{\star}_{\mathbf{f},n}$ is the union of $U$ and the graphs of the three $\mathbf{G}_i$, this means that $(G^{\star}_{\mathbf{f},n},R_{\mathbf{f},n},\{ (s,t) \})$ is indeed vital as desired.
\end{proof}

Combining \zcref{lemma_GadgetsInGrid,lemma_GadgetTreewidth,lemma_VitalGadget} no gives the following theorem on the non-existence of a variant of \zcref{thm_VitalSpanningLinkage}.

\begin{theorem}\label{thm_noVitalLinakgeFunction}
Let $\mathcal{X}$ be the class of all annotated graphs $\mathbf{H}$ for which there exists a $2$-outer-annotated grid $\mathbf{G}$ such that $\mathbf{H}$ is a red-minor of $\mathbf{G}$.

Then, for every $n \in \mathbb{N}$ there exists a vital instance $(G,R,\{ s,t\})$ of $1$\textsc{-Spanning Disjoint Paths} with $G \in \mathcal{X}$ and $\mathsf{tw}(G) \geq n$.
\end{theorem}

\section{Conclusion}
In this paper we identify $\mathsf{depth}_2$ as the precise threshold parameter governing the applicability of the Irrelevant Vertex Technique for \textsc{$k$-Spanning Disjoint Paths}. This yields an algorithm running in $2^{2^{\poly(k+d)}}\cdot n^2$ time, where $k\coloneqq|\Tcal|$ is the number of terminal pairs and $d\coloneqq\mathsf{depth}_2(G,R)$. We stress that our dichotomy is combinatorial rather than algorithmic: it delineates exactly those classes on which the Irrelevant Vertex Technique can be made to work, which is a statement about the reach of a particular method and not about the inherent complexity of the problem. In principle, a class could lie beyond the $\mathsf{depth}_2$ boundary and still admit a polynomial-time algorithm by some other route. To turn the combinatorial boundary into a complexity-theoretic one, one would therefore need to prove that \textsc{Spanning Disjoint Paths} is \textsf{NP}-complete when restricted to the class of all minors of the \emph{$2$-outer-annotated $(k\times k)$-grid}, for every $k\in\Nbbb$. As mentioned in the introduction, we cannot exclude that \textsc{Spanning Disjoint Paths} is solvable in polynomial time using other or new algorithmic techniques, and we leave this as an open question.

Another direction is to relax the goal from polynomial-time solvability to fixed-parameter or slice-wise tractability. Concretely, one may ask, for a suitably chosen ``tight'' annotated parameter $\p$, for an \textsf{XP} algorithm for \textsc{Spanning Disjoint Paths} with respect to $\p(G,R)$, the number of terminal pairs $k=|\Tcal|$, or both. In this direction, we conjecture that \textsc{Spanning Disjoint Paths} can be solved in $f(d_\infty)\cdot |G|^{f(k)}$, $f(k)\cdot |G|^{f(d_\infty)}$, or $|G|^{f(k,d_\infty)}$ time, where $d_\infty\coloneqq\mathsf{depth}_\infty(G,R)$; that is, with the exponent of $|G|$ controlled by only one of the two parameters, or by both together. Establishing any of these bounds would show that bidimensionality yields a tight algorithmic dichotomy, since \textsc{Spanning Disjoint Paths} is \textsf{NP}-complete on planar graphs even when $k=1$, so no bound with a $|G|$-exponent independent of $d_\infty$ can be expected.

As already mentioned, $\mathsf{depth}_2$ is only one of infinitely many ways to extend the notion of treewidth to annotated graphs. Recall that, for every $r\in\Nbbb\cup\{\infty\}$, the parameter $\mathsf{depth}_r$ collapses to treewidth once the graph is fully annotated, and that the two ends of this spectrum already carry names: for $r=1$ the parameter is called \textsl{monodimensionality}~\cite{sau2026modelcheckinglowmonodimensionality}, while for $r=\infty$ it is called \textsl{bidimensionality}~\cite{SauStamoulisThilikos2025CMSO,ThilikosWiederrecht2025Bidimensionality}. In this sense, as $r$ grows through the finite values, $\mathsf{depth}_r$ interpolates between the two, capturing gradually denser degrees of ``monodimensionality''. It is natural to ask whether these degrees demarcate dichotomies for other problems as well. A concrete candidate is the \textsc{Terminal $d$-Dominating Cyclability} problem, whose input is again an annotated graph $(G,R)$; here, instead of requiring the cycle to pass through every red vertex, we ask for a cycle $C$ whose vertex set $d$-dominates $R$, in the sense that every vertex of $R$ lies within distance at most $d$ from some vertex of $V(C)$. For $d=0$ this is exactly \textsc{Terminal Cyclability}, but already for larger $d$ our techniques do not extend to it under the $\mathsf{depth}_2$ parameter. Do higher degrees of monodimensionality suffice in that setting, and, if so, when is the resulting threshold tight, either combinatorially (with respect to the Irrelevant Vertex Technique) or algorithmically (with respect to a polynomial algorithm)? Answering this would require, for every $d\in\Nbbb$, a local structure theorem in the spirit of \cref{thm_localstructure}. Such a theorem appears considerably more intricate, and difficult even to formulate, than the one we prove here; for the extreme case $d=\infty$ a corresponding min-max theorem was obtained in \cite{ProtopapasThilikosWiederrecht2025ColorfulMinors}.

Our paper provides two algorithms, and it is worth separating their roles. The first is the algorithm of \cref{thm_localstructure}, which certifies the local structure of annotated graphs of bounded $\mathsf{depth}_2$ and runs in $\poly(t+r+w)\cdot |G|^6$ time. The second takes the local structure provided by \cref{thm_localstructure} as its starting point and solves \textsc{$k$-Spanning Disjoint Paths} in quadratic time by repeatedly removing irrelevant vertices. We made no attempt to optimise the exponent of $|G|$ in the first algorithm and leave this for further investigation; in fact, we believe it can be reduced to quadratic, which would bring the preprocessing step in line with the second algorithm.

A further direction would be to find a min-max graph parameter, defined through tree decompositions, that is equivalent to $\mathsf{depth}_2$. Informally, the torsos of such a decomposition should locally exhibit precisely the structure certified by \cref{thm_localstructure}, so that the global parameter records, in a single quantity, the local guarantees our theorem provides wall by wall. Such a min-max theorem could be established using standard local-to-global arguments, as developed in~\cite{KawarabayashiTW2021Quickly,RobertsonS1991GraphMinorsX}, at the cost of an additional factor of $|G|$ in the running time of the algorithm of \cref{thm_localstructure}. We avoided giving the proof of such a theorem, as it is not required for the results of this paper.

\bigskip

\noindent\textbf{Acknowledgement.} We wish to thank Petr Golovach for valuable early discussions on several  problems about how to span graphs inside graphs.

\newpage

\phantomsection
\addcontentsline{toc}{section}{References}

\newcommand{\etalchar}[1]{$^{#1}$}

\end{document}